\documentclass[12pt]{article}
\usepackage{amsmath,amssymb,amsfonts}
\usepackage{graphicx}
\usepackage[round,longnamesfirst]{natbib}
\usepackage{setspace}
\usepackage{lmodern}
\usepackage[T1]{fontenc}
\usepackage{times}
\usepackage{hyperref}
\usepackage{multimedia}
\usepackage{bm}
\usepackage{proba}
\usepackage{enumerate}
\usepackage{natbib}
\usepackage{lscape}
\usepackage{placeins}
\usepackage{longtable}
\usepackage{epstopdf}
\usepackage{graphics}
\usepackage{tikz}
\usepackage{float}
\usepackage{placeins}
\usepackage{longtable}
\usepackage{appendix}
\usepackage{caption}
\usepackage{subcaption}
\usepackage{amstext}
\usepackage{bbm}
\usepackage{latexsym}
\usepackage{bigints}
\usepackage{booktabs}
\usepackage{xr}
\usepackage[normalem]{ulem}
\usepackage[amsmath]{ntheorem}
\usepackage[latin1]{inputenc}
\usepackage{makeidx}
\usepackage{mathtools}
\usepackage[ruled,vlined,linesnumbered,resetcount]{algorithm2e}
\usepackage{xcolor}
\usepackage{bigints}
\usepackage{arydshln}
\usepackage[margin=1in]{geometry}
\usepackage{colortbl}

\renewcommand{\cite}{\citet}
 \bibpunct{(}{)}{,}{a}{,}{,}
\newtheorem{ass}{Assumption}
\newtheorem{lem}{Lemma}[section]
\newtheorem{cor}{Corollary}[section]

\newtheorem{thm}{Theorem}[section]

\definecolor{dgreen}{rgb}{0,0.5,0}
\definecolor{dblue}{rgb}{0,0,0.9}
\definecolor{dred}{rgb}{0.6,0.0,0.1}
\definecolor{dgold}{rgb}{0.5,0.3,0.0}
\definecolor{dvio}{rgb}{0.6,0.3,0.5}
\definecolor{gray}{rgb}{0.5,0.5,0.5}

\newcommand{\bee}{\begin{equation}}
\newcommand{\eee}{\end{equation}}

\newcommand{\wh}{\widehat}

\newcommand{\wtl}{\widetilde}
\newcommand{\lt}{\left}
\newcommand{\rt}{\right}
\newcommand{\EE}{\mathbf{E}}
\newcommand{\argmin}{\mathrm{argmin}}

\newcommand{\argmax}{\mathrm{argmax}}
\newcommand{\plim}{\mathrm{plim}}
\makeatletter \@addtoreset{equation}{section} \makeatother \renewcommand{\theequation}{\thesection.\arabic{equation}}
\bibpunct{(}{)}{,}{a}{,}{,}
\allowdisplaybreaks
\newtheorem{example1}{Example}

\newtheorem{example2}{Example}

\newtheorem{example3}{Example}

\newtheorem{example4}{Example}

\newtheorem{example5}{Example}

\newtheorem{example1case1}{Example}

\input{tcilatex}
\makeindex

\begin{document}

\title{\textbf{Testing for Endogeneity: \\A Moment-Based Bayesian Approach}%
\thanks{%
The views expressed here are our own and do not necessarily represent the
views of the Federal Reserve Bank of Philadelphia or the Federal Reserve
System.}}
\author{Siddhartha Chib\thanks{%
Olin Business School, Washington University in St. Louis, Campus Box 1133, 1
Bookings Drive, St. Louis, MO 63130. e-mail: chib@wustl.edu.} \\
%EndAName
Minchul Shin\thanks{
Research Department, Federal Reserve Bank of Philadelphia, 10 Independence
Mall, Philadelphia, PA 19106, e-mail: visiblehand@gmail.com.}\\
Anna Simoni\thanks{
CREST, CNRS, ENSAE, Ecole Polytechnique, Institut Polytechnique de Paris, 5
Avenue Henry Le Chatelier, 91120 Palaiseau, France, e-mail:
simoni.anna@gmail.com.} }
\maketitle

\begin{abstract}
A standard assumption in the Bayesian estimation of linear regression models
is that the regressors are exogenous in the sense that they are uncorrelated
with the model error term. In practice, however, this assumption can be
invalid. In this paper, using the exponentially tilted empirical likelihood framework, we develop a Bayes factor test for endogeneity that
compares a base model that is correctly specified under exogeneity but
misspecified under endogeneity against an extended model that is correctly
specified in either case. We provide a comprehensive study of the
log-marginal exponentially tilted empirical likelihood. We demonstrate that
our testing procedure is consistent from a frequentist point of view: as the
sample grows, it almost surely selects the base model if and only if
the regressors are exogenous, and the extended model if and only if the
regressors are endogenous. The methods are illustrated with simulated data,
and problems concerning the causal effect of automobile prices on automobile
demand and the causal effect of potentially endogenous airplane ticket
prices on passenger volume.
\end{abstract}

%\date{}

 \thispagestyle{empty}

\thispagestyle{empty}

\newpage

\doublespacing

\setcounter{page}{2}

\section{Introduction}

Consider the semiparametric linear regression model
\[
y = x^{\prime}\beta + z_{1}^{\prime}\gamma + \varepsilon,
\]
where $y \in \mathbb{R}$ is the outcome variable, $x \in \mathbb{R}^{d_x}$ is the treatment vector of interest,
$z_{1} \in \mathbb{R}^{d_{z_{1}}}$ is a vector of controls, and $\varepsilon$ is an unobserved disturbance.
A common assumption in Bayesian analysis is that the regressors $x$ are exogenous, meaning that they are uncorrelated with the error term $\varepsilon$. In many empirical settings this assumption is questionable.
If one has access to a set of valid instruments $z_{2} \in \mathbb{R}^{d_{z_{2}}}$, with dimension at least as large as that of $x$,
it becomes possible to conduct a Bayesian analysis that correctly accounts for endogeneity.
Such analysis can be formulated within both parametric and semiparametric frameworks,
as in the early contributions of \citet{Dreze1976} and the subsequent developments in \citet{KleibergenVanDijk1998},
\citet{CHAOPhillips1998}, \citet{KleibergenZivot2003}, and \citet{Schennach2005},
among many others.
Recent work, including \citet{HoogerheideKleibergenVanDijk200763}, \citet{liao2011},
\citet{FS2012JoE}, \citet{florens_simoni_2016}, \citet{FlorensSimoni2015}, \citet{kato2013},
\citet{Shin2014}, and \citet{ChibShinSimoni2018}, has extended these ideas
to semiparametric and likelihood-free settings.

An important question that has received little attention in the Bayesian literature concerns the testing of endogeneity.
Frequentist methods, such as the classical Durbin-Wu-Hausman test, offer an asymptotic procedure that assesses exogeneity by comparing estimators that are consistent under different assumptions.
These procedures, however, do not translate naturally into the Bayesian framework.
From a Bayesian standpoint, it is more straightforward to conceptualize the test for endogeneity as a comparison of models, rather than that of
parameters. Specifically, one can develop a test that is based on the relative support provided by the data for a model with exogeneity versus a model with endogeneity.

To develop this approach, and to avoid distributional assumptions, we proceed within a Bayesian framework for moment condition models.
We consider two competing specifications. The first is a base model defined by the moment conditions
\[
\mathbf{E}[\varepsilon(\theta)x]=0, \quad
\mathbf{E}[\varepsilon(\theta)z_1]=0, \quad
\mathbf{E}[\varepsilon(\theta)z_2]=0,
\]
where $\varepsilon(\theta) = y - x^{\prime}\beta + z_{1}^{\prime}\gamma$ and $\theta := (\beta,\gamma)$.
The second is an extended model that relaxes the exogeneity restriction and allows
\[
\mathbf{E}[\varepsilon(\theta)x] = v,
\]
where $v$ captures the covariance between the error term and the endogenous variable $x$.
We formulate the prior-posterior analysis of each model through the nonparametric exponentially tilted empirical likelihood (ETEL) and then compare the two models by marginal likelihoods and the Bayes factor. This approach offers several methodological advantages. The Bayes factor measures the strength of evidence for the two models on a continuous scale rather than through a strict accept or reject rule. In addition, the use of ETEL provides robustness to misspecification of the joint distribution of $(y, x, z_1, z_2)$ and allows us to obtain results that remain valid without specifying the distribution of the disturbances.

Using the \citet{Chib1995} marginal likelihood identity, we know that the log marginal likelihood decomposes into three parts:
the log ETEL, the log prior, and the negative log posterior ordinate.
We establish that this expression is asymptotically equal to a term bounded in probability, plus a term proportional to the Kullback-Leibler divergence between the true and the closest probability distribution satisfying the moment restrictions term, plus a penalty that corresponds to the ones of the Bayesian information criterion (BIC). The penalty arises from a change-of-variable transformation of the posterior density evaluated at the true or pseudo-true value of the parameters. The log of the Jacobian of this transformation constitutes the penalty, while the posterior density of the local parameter at zero is bounded in probability as $n$ increases. Accordingly, when $x$ is exogenous, the log-ETELs of the two models are asymptotically the same but the penalties differ. When $x$ is endogenous, the difference in the log-ETELs dominates, which leads to the selection of the extended model. Thus, the test correctly discriminates between the two data-generating processes in large samples. Our construction parallels the logic of the Hausman test, where one compares an estimator that is inconsistent under endogeneity with one that is not. Here, the comparison is between models that differ in the number of overidentifying moment conditions. In this sense, our test may be viewed as the Bayesian analogue of the Hausman specification test.

Compared to \cite{ChibShinSimoni2018}, our work builds on the same Bayesian ETEL framework but makes several key contributions. First, while \cite{ChibShinSimoni2018} explains how to test among different models, it does not address how to construct the specific models required to test hypotheses of interest in practical applications, %to be compared in concrete applications,
such as the endogeneity problem we examine here. In this paper, we explicitly construct the models necessary for testing endogeneity. Second, we introduce an assumption that guarantees the existence of the ETEL function, which, to our knowledge, is absent from the existing ETEL literature. This assumption ensures that the ETEL function exists at least in a suitable neighborhood of the true parameter value with probability approaching one. The issue arises because the ETEL function, as the solution to a constrained optimization problem, may have an empty feasible set for certain parameter values $\theta$. Without this assumption, derivatives of the ETEL function cannot be defined, posing challenges for both frequentist and Bayesian ETEL approaches. Third, we provide a more direct proof demonstrating that the ETEL function is asymptotically equivalent to a quadratic function. This result underpins our establishment of a Bernstein-von Mises theorem, which we then use to prove the consistency of our testing procedure. The direct proof leverages the linearity in $\theta$ of the moment restrictions implied in the instrumental variable (IV) regression problem. Along the same lines, the assumptions in this paper are weaker than those in \cite{ChibShinSimoni2018}, as they exploit the IV linear regression structure.\\
\indent Finally, as a by-product of proving the consistency of our testing procedure, we derive a new asymptotic representation of the log-marginal ETEL function, defined as the ETEL function integrated with respect to the prior distribution of the model parameter. We show that the log-marginal likelihood of each model can be asymptotically decomposed into a Kullback-Leibler (KL) divergence term (between the true distribution and the closest distribution satisfying the model's moment restrictions) plus a BIC-type penalty. We derive this penalty using a novel approach: by re-expressing the posterior ordinate at the true (or pseudo-true) parameter value via a local parameter change of variables, the resulting log-Jacobian yields the penalty term, while the posterior density of the local parameter evaluated at zero is $\mathcal{O}_{p}(1)$ as $n\rightarrow \infty$. This representation clarifies the mechanics of Bayes-factor testing in this context and leads to a more transparent proof of model-selection consistency than that in \cite{ChibShinSimoni2018}. Notably, we emphasize that the penalty plays a role in selecting the correct model only when $x$ is exogenous, in which case both models are correctly specified.

The rest of the paper proceeds as follows.
Section~\ref{sec:Preliminaries} summarizes Bayesian estimation and comparison of moment condition models using ETEL.
Section~\ref{sec:Models} presents the base and extended models and provides a simulated example to illustrate the procedure.
Section~\ref{sec:testing} develops the test for endogeneity and analyzes the large-sample behavior of the log-marginal likelihood, establishing consistency of the test.
Section~\ref{sec:examples} presents empirical examples, and Section~\ref{sec:conclusion} concludes.
An Appendix contains the proofs of the main results.

%------------------------------------------------------------------------------------------------------------------------------------------------------------------------
%------------------------------------------------------------------------------------------------------------------------------------------------------------------------
\section{Preliminaries}\label{sec:Preliminaries}

In this section we briefly provide the background on Bayesian estimation of
moment condition models using the exponentially tilted empirical likelihood
(ETEL). Further details can be found in \cite{Schennach2005} and \cite%
{ChibShinSimoni2018}.
%------------------------------------------------------------------------------------------------------------------------------------------------------------------------
\subsection{Moment restrictions and feasible distributions}\label{subsec:prelim_moments}
Let $w\in \mathbb{R}^{d_w}$ be a random vector, and let $\theta \in \Theta \subset \mathbb{R}^{p}$ denote a generic parameter vector. Let $\mathbb{M}$ denote the set of all probability distributions on $\mathbb{R}^{d_w}$. For a known vector of moment functions
$$g(w,\theta ):\mathbb{R}^{d_w}\times \Theta\rightarrow \mathbb{R}^{d},$$
the moment restriction is given by
\begin{equation}\label{eq:prelim_moment_restriction}
\mathbf{E}^{Q}[g(w,\theta )]=0,
\end{equation}
\noindent where $Q\in\mathbb{M}$ is a probability distribution under which the restriction is imposed and $\EE^Q[\cdot]$ denotes the expectation under $Q$. For each $\theta \in \Theta$, define the subset of distributions that satisfy the moment restriction by
\begin{equation}
\mathcal{Q}(\theta) := \left\{ Q\in \mathbb{M}:\,\mathbf{E}^{Q}[g(w,\theta)]=0\right\}.
\end{equation}
Suppose the data $w_{1:n}:=(w_{1},\ldots ,w_{n})$ are independently drawn from the true distribution $P$, which need not belong to $\mathcal{Q}(\theta)$ for some $\theta \in \Theta $. The expectation taken with respect to the true distribution $P$ is denoted by $\EE[\cdot] \equiv \EE^P[\cdot]$.
%
%------------------------------------------------------------------------------------------------------------------------------------------------------------------------
\subsection{Sample ETEL weights, tilted sample distribution, and likelihood}\label{subsec:prelim_etel_sample}

The empirical counterpart of \eqref{eq:prelim_moment_restriction} is the weighted restriction 
\begin{equation}\label{eq:prelim_sample_moment_restriction}
  \sum_{i=1}^n q_i\,g(w_i,\theta)=0, \qquad q_i>0, \qquad \sum_{i=1}^n q_i=1,
\end{equation}
where $\{q_i\}_{i=1}^n$ is a discrete distribution supported on $\{w_i\}_{i=1}^n$. Equivalently, any such weight vector induces a discrete probability measure on $\mathbb{R}^{d_w}$ with support $\{w_1,\ldots,w_n\}$
\begin{equation}\label{eq:prelim_Qn_def}
  Q_n(\cdot) := \sum_{i=1}^n q_i\,\delta_{w_i}(\cdot),
\end{equation}
where $\delta_{w_i}$ denotes the point mass at $w_i$.\\
\indent Since there might be no $\theta\in\Theta$ such that the uniform empirical distribution $q_i\equiv 1/n$ satisfy \eqref{eq:prelim_sample_moment_restriction}, we define the ETEL weights as the
solution to the following Kullback Leibler (KL)-closest feasible reweighting problem:
%
%$\sum_{i=1}^{n}q_{i}g(w_{i},\theta )=0$ with $\{\widehat{q}_{i}(\theta)\}_{i=1}^{n}$ denoting a discrete distribution depending on $\theta$ with support points $\{w_{i},\,i=1,\ldots ,n\}$. As it might be that the empirical distribution that places probability mass $\{\frac{1}{n}\}$ on each observation does not satisfy the moment restrictions, we consider the problem of finding the discrete distribution $\{\widehat{q}_{i}(\theta )\}_{i=1}^{n}$ that is the nearest in the Kullback-Leibler (KL) discrepancy to the empirical distribution and that satisfies the moment restriction.
%
%Enforcing the requirement that the moment restrictions are satisfied under this discrete distribution, the probability masses emerge as the solution to the optimization program
\begin{align}
  \{\widehat{q}_{i}(\theta )\}_{i=1}^n:=& \argmax_{q_{1}>0,\ldots ,q_{n}>0}\sum_{i=1}^{n}\left[ -q_{i}\log (nq_{i})\right]  \notag \\
  \text{subject to }\sum_{i=1}^{n}q_{i}=& 1,\qquad \sum_{i=1}^{n}q_{i}g(w_{i},\theta )=0 \label{eq:KLbasic}
\end{align}
\noindent which depends on the parameter vector $\theta\in\Theta$. Let $H_n\subset \Theta$ be the set of $\theta$ values for which the program \eqref{eq:KLbasic} is feasible, \textit{i.e.} the set of $\theta$s such that the interior of the convex hull of $\{g(w_{i},\theta ): i=1,\ldots,n\}$ contains zero. %, that is, the set of $\theta$s for which the problem is feasible. If $H_n$ is empty, there is no $\theta $ for which \eqref{eq:KLbasic} has a solution.
Assumption \ref{ass:feasibility} below ensures that $H_n$ is non-empty with probability approaching $1$.\\
\indent For $\theta\in H_n$, the ETEL (sample likelihood) is defined as the product of the ETEL weights:
%
%The ETEL is the likelihood constructed as the product of the probability masses solution to \eqref{eq:KLbasic}. It is defined for every $\theta\in H_n$ as
\begin{equation*}
  \widehat{q}(w_{1:n}|\theta ) = \dprod\limits_{i=1}^{n}\widehat{q}_{i}(\theta).
\end{equation*}
\noindent The ETEL arises as the integrated likelihood obtained by integrating out the unknown distribution $Q$ with respect to a particular nonparametric prior that imposes the moment restrictions \eqref{eq:prelim_sample_moment_restriction} conditional on a $\theta\in H_n$; see \cite{Schennach2005}.\\
\indent Given a prior density $\pi(\theta)$, the ETEL-based posterior is the truncated posterior 
\begin{equation}\label{eq:posterior}
  \pi^{n}(\theta |w_{1:n})\propto \pi(\theta)\,\widehat{q}(w_{1:n}|\theta )\,I[\theta \in H_n],
\end{equation}
where $I[A]$ denotes the indicator function. Since \eqref{eq:posterior} is not available in closed form, posterior summaries are obtained via tailored Markov chain Monte Carlo (MCMC) methods. Appendix \ref{sec:mcmc-ml-computation} describes computational details on MCMC sampling and related calculations.

\subsection{Dual representation and log-ETEL identities}\label{subsec:prelim_dual_identity}
A convenient way to compute $\{\widehat{q}_{i}(\theta )\}$ is via the dual representation of \eqref{eq:KLbasic}. Define the ETEL multiplier as, for every $\theta\in H_n$
\begin{equation*}
  \widehat{\lambda }(\theta)\equiv \widehat{\lambda }(w_{1:n},\theta) := \arg\min_{\lambda \in \mathbb{R}^{d}}\frac{1}{n}\sum_{i=1}^{n}\exp \left(\lambda'g(w_{i},\theta)\right).
\end{equation*}
Then, for every $\theta\in H_n$:
\begin{equation}
  \widehat{q}_{i}(\theta) = \frac{e^{\widehat{\lambda}(\theta)'g(w_{i},\theta)}}{\sum_{j=1}^{n}e^{\widehat{\lambda }(\theta)'g(w_{j},\theta )}},\qquad i\leq n.  \label{eq:ETEL:dual}
\end{equation}
It is useful to view \eqref{eq:ETEL:dual} as an exponential tilting of the uniform empirical distribution $P_n(\cdot) := \frac{1}{n}\sum_{i=1}^n \delta_{w_i}(\cdot)$, which places mass $1/n$ on each observation. 
For fixed $\theta\in H_n$, the ETEL weights $\{\wh q_i(\theta)\}$ therefore define a \textit{tilted sample distribution}
\begin{equation}\label{eq:prelim_Qn_hat_def}
  \wh Q_n(\cdot\mid\theta) := \sum_{i=1}^n \wh q_i(\theta)\,\delta_{w_i}(\cdot),
\end{equation}
that is absolutely continuous with respect to $P_n$. In particular, for each support point $w_i$,
\begin{equation}\label{eq:prelim_RN_Qn_hat}
\frac{d\wh Q_n(\cdot\mid\theta)}{dP_n}(w_i) = n\,\wh q_i(\theta) = \frac{e^{\widehat{\lambda}(\theta)'g(w_{i},\theta)}}{\frac{1}{n}\sum_{j=1}^{n}e^{\widehat{\lambda }(\theta)'g(w_{j},\theta )}},
\end{equation}
so $\wh Q_n(\cdot\mid\theta)$ is the exponential tilt of $P_n$ that enforces the sample moment restriction \eqref{eq:prelim_sample_moment_restriction}. The multiplier $\widehat\lambda(\theta)$ satisfies the sample first-order condition
\begin{equation}\label{eq:etel_sample_foc}
\sum_{i=1}^n \widehat q_i(\theta)\,g(w_i,\theta)=0,
\end{equation}
which is the sample analogue of \eqref{eq:prelim_moment_restriction} under the ETEL weights. For later use, we record the log-ETEL in a form amenable to expansions. Summing $\log\widehat q_i(\theta)$ yields the exact identity
\begin{equation}\label{eq:log_etel_master_identity}
  \log \widehat q(w_{1:n}\mid\theta) = \sum_{i=1}^n \widehat\lambda(\theta)'g(w_i,\theta) - n\log\!\Big(\sum_{j=1}^n \exp\{\widehat\lambda(\theta)' g(w_j,\theta)\}\Big),
\end{equation}
or equivalently,
\begin{equation}\label{eq:log_etel_master_identity_normalized}
  \log \widehat q(w_{1:n}\mid\theta) = -n\log n + \sum_{i=1}^n \widehat\lambda(\theta)'g(w_i,\theta) - n\log\!\Big(\frac{1}{n}\sum_{j=1}^n \exp\{\widehat\lambda(\theta)'g(w_j,\theta)\} \Big).
\end{equation}

\subsection{Population KL projection, exponential tilting, and pseudo-true values}\label{subsec:prelim_population_projection}
The population counterpart of $\{\widehat{q}_{i}(\theta )\}_{i=1}^n$ is the distribution $Q^*(\theta )\in \mathcal{Q}(\theta)$ that is the closest to $P$ in the KL divergence. For each $\theta$ such that $\mathcal{Q}(\theta) \neq \emptyset$, define
\begin{equation*}
Q^*(\theta) := \arg\inf_{Q\in \mathcal{Q}(\theta)}\mathrm{KL}(Q||P),
\end{equation*}
\noindent where 
$$\mathrm{KL}(Q||P):=\int \log \left( \frac{dQ}{dP}\right) dQ$$ 
if $Q$ is absolutely continuous with respect to $P$, and $\mathrm{KL}(Q||P) = +\infty $, otherwise. The population counterpart of the ETEL multiplier $\widehat{\lambda}(\theta)$ is
\begin{equation*}
  \lambda_*(\theta) := \arg \min_{\lambda \in \mathbb{R}^{d}} \mathbf{E}[e^{\lambda'g(w_{i},\theta )}]
\end{equation*}
\noindent for every $\theta\in\Theta$ such that $\mathcal{Q}(\theta) \neq \emptyset$. This induces the population exponential tilt
\begin{equation}\label{eq:prelim_population_tilt_def}
  q(w;\theta) := \frac{\exp\{\lambda_*(\theta)' g(w,\theta)\}}{\EE[\exp\{\lambda_*(\theta)' g(w,\theta)\}]}, \qquad \EE[q(w;\theta)]=1.
\end{equation}
Under mild regularity conditions, the KL projection $Q^*(\theta)$ admits the Radon-Nikodym derivative representation
\begin{equation}\label{eq:prelim_RN_Q_star}
  \frac{dQ^*(\theta)}{dP}(w) = q(w;\theta).
\end{equation}
Thus, if $P$ admits a Lebesgue density $p(w)$, then $Q^*(\theta)$ has density given by an exponential tilt of $p(w)$:
\[
q^*(w;\theta) = p(w)\, \frac{\exp\{\lambda_*(\theta)' g(w,\theta)\}}{\EE[\exp\{\lambda_*(\theta)' g(w,\theta)\}]} .
\]
By a change of measure,
\begin{equation}\label{eq:prelim_change_of_measure}
  \EE^{Q^*(\theta)}\!\big[g(w,\theta)\big] = \EE\!\big[q(w;\theta)\,g(w,\theta)\big].
\end{equation}
The right-hand side is the population tilted moment condition. It is the moment restriction expressed under $P$ rather than under $Q^*(\theta)$. %We also define the tilted second moment matrix
%\begin{equation}\label{eq:prelim_Omega_def}
%  \Omega(\theta) := \EE\!\big[q(w;\theta)\,g(w,\theta)g(w,\theta)'\big],
%\end{equation}
%which will govern the quadratic approximation of the log-ETEL and the limiting distribution of the tilted score process.
%
If one or more moment conditions are misspecified, then $Q^*(\theta)\neq P$ for all $\theta\in\Theta$, and the pseudo-true value $\theta_*$ is defined as the minimizer of the reversed KL divergence
\begin{equation}\label{eq:prelim_eta_star_def}
  \theta_* := \arg\min_{\theta; \mathcal{Q}(\theta) \neq \emptyset}\mathrm{KL}\!\big(P\|Q^*(\theta)\big),
\end{equation}
where
\begin{equation}\label{eq:prelim_KL_reverse_def}
  \mathrm{KL}\!\big(P\|Q^*(\theta)\big) :=  \int \log\!\Big(\frac{dP}{dQ^*(\theta)}\Big)\,dP
\end{equation}
whenever $P$ is absolutely continuous with respect to $Q^*(\theta)$. Under correct specification, there exists $\theta_\circ\in\Theta$ such that $P\in\mathcal{Q}(\theta_\circ)$, in which case $Q^*(\theta_\circ)=P$ and
$\theta_* = \theta_\circ$. Moreover, in that case $\lambda_*(\theta_\circ)=0$ and hence $q(w;\theta_\circ)\equiv 1$. Finally, when the dual representation holds, the pseudo-true value $\theta_*$ may also be expressed in terms of the population tilt as
\begin{equation}\label{eq:prelim_eta_star_dual}
  \theta_* = \arg\max_{\theta; \mathcal{Q}(\theta) \neq \emptyset} \EE\!\left[\log\!\Big(\frac{\exp\{\lambda_*(\theta)' g(w,\theta)\}} {\EE[\exp\{\lambda_*(\theta)' g(w,\theta)\}]}\Big) \right],
\end{equation}
where the term inside the logarithm is the Radon-Nikodym derivative $[dQ^*(\theta)/dP](w)$ in \eqref{eq:prelim_RN_Q_star}.
%
%If one or more moment conditions are misspecified, then $Q^{\ast }(\theta )\neq P$ for any $\theta \in \Theta $ and the pseudo-true value $\theta
%_{\ast }$ is defined as the value of $\theta $ that minimizes KL$(P||Q^{\ast}(\theta ))$ over $\Theta $. Notice the inversion of the probabilities in
%the $\mathrm{KL}$ discrepancies used to define $Q^{\ast }(\theta )$ and $\theta_*$. Assumption \ref{Ass_absolute_continuity} stated in Section \ref{sec:assumptions} guarantees that both KL$(P||Q^{\ast}(\theta ))$ and $\mathrm{KL}(Q||P)$ are well-defined. Under correct specification, $Q^{\ast }(\theta _{\circ})=P $, for some $\theta _{\circ }\in \Theta $, and $\theta _{\ast }=\theta
%_{\circ }$.
%
%When the dual representation of the optimization problem \eqref{eq:KLbasic} holds, the pseudo-true value can also be obtained as
%\begin{equation}
%\theta _{\ast }=\arg \max_{\theta \in H_n}\mathbf{E}\log \Bigl(\frac{e^{\lambda_*^{\prime }(\theta )g(w,\theta )}}{\mathbf{E}[e^{\lambda_{*}'(\theta )g(W,\theta )}]}\Bigr),
%\end{equation}%
%\noindent where in this case the term within the brackets is $[dQ^{\ast
%}(\theta )/dP](w)$.

\section{Models}\label{sec:Models}
In this section we specialize the generic moment-restriction framework in Section \ref{sec:Preliminaries} to the semiparametric linear regression setting introduced in the Introduction. %In this section, we link the general setting presented in Section \ref{sec:Preliminaries} to the semiparametric linear regression presented in Section 1.
\subsection{Data, regression structure, and target parameter}\label{subsec:models_setup}
Let $w:=(y,x,z_{1},z_{2})$ $\in \mathbb{R}^{d+1}$ be distributed according to an unknown probability distribution $P$, where $d := d_{x}+d_{z_{1}}+d_{z_{2}}$ and $d_w = d + 1$. Throughout, $\mathbf{E}[\cdot ]:=\mathbf{E}^{P}[\cdot]$ denotes expectation under $P$. We assume that under $P$, the random vector $w$ follows the regression model
\begin{equation}
y=\beta_{\circ }^{\prime }x+\gamma _{\circ }^{\prime }z_{1}+\varepsilon,\qquad \mathbf{E}[\varepsilon(\theta_{\circ })z_{j}]=0\quad \text{ for }j=1,2,  \label{eq:model}
\end{equation}%
\noindent where $\theta_{\circ }:=(\beta_{\circ },\gamma _{\circ })\in\Theta \subset\mathbb{R}^p$ is the true value of the regression coefficients, viewed as a functional of $P$: $\theta_{\circ} \equiv \theta_{\circ}(P)$ and $p=d_{x}+d_{z_{1}}$. In model \eqref{eq:model}, the vector $z_{1}$ contains exogenous controls (including an intercept), and $z_{2}$ contains instrumental variables. The object of interest is the causal effect of $x$ on $y$, represented by $\beta_{\circ}$. For any $\theta := (\beta',\gamma')'\in\Theta\subset\mathbb{R}^p$, with $p=d_{x}+d_{z_{1}}$ and $\wtl{w}_1 := (x',z_1')'$, define the regression residual
\begin{displaymath}
  \varepsilon(\theta) := y-\beta'x - \gamma'z_{1} \equiv y-\theta'\tilde{w}_1.
\end{displaymath}
If $x$ is endogenous under $P$, then $\EE[\varepsilon(\theta_\circ)x]\neq 0$. When $d_{z_2}\geq d_x$, the instruments $z_2$ help identify $\beta_\circ$ despite endogeneity.
%Under the assumption that $d_{z_{2}}\geq d_{x}$, the instruments help to identify $\beta _{\circ }$ when $\mathbf{E}[\varepsilon _{i}(\theta _{\circ })x_{i}]\neq 0$.
%
\subsection{Base model $\mathcal{M}_b$ (possibly misspecified)}\label{subsec:models_base}
The base model, denoted by $\mathcal{M}_b$, imposes the moment restrictions
\begin{equation}\label{eq:base:model}
  \EE^{Q}[g_{b}(w,\theta )] = 0,\qquad Q\in \mathcal{Q}_{b}(\theta),
\end{equation}
where the base-model moment function is
\begin{equation*}
  g_{b}(w,\theta ):=\varepsilon (\theta )%
\begin{pmatrix}
x \\
z_{1} \\
z_{2}
\end{pmatrix} \in\mathbb{R}^d,\qquad d = d_x + d_{z_1} + d_{z_2},
\end{equation*}
and the set of distributions satisfying the base-model restrictions is
\begin{equation}
  \mathcal{Q}_{b}(\theta) = \left\{ Q\in \mathbb{M};\,\mathbf{E}^{Q}[g_{b}(w,\theta )]=0\right\}.  \label{def:Q:b}
\end{equation}
Here, $\mathbb{M}$ denotes the set of all probability distributions on $\mathbb{R}^{d+1}$.  Under exogeneity, $\EE[\varepsilon(\theta_\circ)x]=0$ and the true distribution $P$ satisfies the base-model moments at $\theta_\circ$, so that $P\in\mathcal{Q}_b(\theta_\circ)$. Under endogeneity, $\EE[\varepsilon(\theta_\circ)x]\neq 0$ and therefore $P\notin\mathcal{Q}_b(\theta)$ for every $\theta\in\Theta$; in that case
$\mathcal{M}_b$ is misspecified. 
%The expectation in \eqref{eq:base:model} is with respect to a distribution $Q\in \mathcal{Q}_{b}(\theta)$ and if $x$ is endogenous under $P$ then, $P\notin \mathcal{Q}_{b}(\theta )$ which means that there is no $\theta $ that satisfies the moment conditions under $P$. 
In this case, the ETEL function, constructed from the sample $w_{1:n}$, is the empirical counterpart of the distribution $Q_{b}^{*}(\theta)$ that for every $\theta$ solves the moment conditions:
\begin{equation*}
\mathbf{E}^{Q_{b}^{\ast }(\theta )}[g_{b}(w,\theta )] = 0
\end{equation*}
%where for every $\theta $, $Q_{b}^{*}(\theta )$ is the distribution in the set $\mathcal{Q}_{b}(\theta )$ the closest to $P$ in the KL divergence, that is,
and that is the closest to $P$ in the KL divergence among all the distribution in the set $\mathcal{Q}_{b}(\theta )$, that is,
\begin{equation*}
  Q_{b}^{*}(\theta ):=\mathrm{arginf}_{Q\in \mathcal{Q}_{b}(\theta)}\mathrm{KL}(Q\Vert P).
\end{equation*}
Notice that $\mathrm{KL}(Q\Vert P)$ is set to $+\infty$ if $Q$ is not absolutely continuous with respect to $P$. In addition,
\begin{equation}
\theta_* := \arg \max_{\theta\in\Theta; \mathcal{Q}_b(\theta)\neq \emptyset}\mathbf{E}\log \Bigl(\frac{e^{\lambda_*^{\prime }(\theta )g_b(w,\theta )}}{\mathbf{E}%
[e^{\lambda_{\ast }^{\prime }(\theta )g_b(W,\theta )}]}\Bigr)
\label{eq:thetastar}
\end{equation}
denotes the pseudo-true value in the base model. Assumption \ref{Ass_absolute_continuity} given in Section \ref{sec:assumptions} below guarantees that this value exists. On the other hand, if $x$ is exogenous, then $Q_{b}^{\ast }(\theta _{\ast})=P$ and $\theta_* = \theta_\circ$, where $\theta_\circ$ denotes the true value of $\theta$ as defined above. In the following we denote the ETEL for the base model by $\wh{q}(w_{1:n}|\theta,\mathcal{M}_b) := \prod_{i=1}^n \wh{q}_i(\theta|\mathcal{M}_b)$, where $\wh{q}_i(\theta|\mathcal{M}_b)$ is constructed as in \eqref{eq:KLbasic} with $g(w_i,\theta)$ replaced by $g_{b}(w_i,\theta)$.
\subsection{Extended model $\mathcal{M}_e$ (correctly specified)}\label{subsec:models_extended}
The extended model, denoted by $\mathcal{M}_e$, augments the base model by explicitly parameterizing the endogeneity component
\[
  v:=\EE[\varepsilon(\theta)x]\in\mathbb{R}^{d_x}.
\]
Let $\mathcal{V}\subset\mathbb{R}^{d_x}$ and define the extended parameter
\begin{equation}\label{eq:psi_def}
  \psi:=(\theta,v)\in\Psi, \qquad \Psi:=\Theta\times\mathcal{V}.
\end{equation}
\noindent The extended-model moment function is
\begin{equation}\label{eq:ge_def}
  g_e(w,\psi) := g_e(w,\theta,v) := \varepsilon(\theta)
\begin{pmatrix}
x\\
z_1\\
z_2
\end{pmatrix}
-
\begin{pmatrix}
v\\
0\\
0
\end{pmatrix}
= g_b(w,\theta)-
\begin{pmatrix}
v\\
0\\
0
\end{pmatrix},
\end{equation}
and the model imposes the moment restrictions
\begin{equation}\label{eq:moment:restrictions}
  \EE^{Q}\!\big[g_e(w,\psi)\big]=0, \qquad Q\in\mathcal{Q}_e(\psi),
\end{equation}
where
\begin{equation}\label{eq:Qe_def}
  \mathcal{Q}_e(\psi) :=  \Big\{Q\in\mathbb{M}:\ \EE^{Q}\!\big[g_e(w,\psi)\big]=0 \Big\}.
\end{equation}
By construction, $\mathcal{M}_e$ is correctly specified under both exogeneity and endogeneity of $x$. Indeed, let
\begin{equation}\label{eq:v_circ_def}
  v_\circ := \EE[\varepsilon(\theta_\circ)x].
\end{equation}
\noindent Under \eqref{eq:model}, we have $\EE[\varepsilon(\theta_\circ)z_1]=0$ and $\EE[\varepsilon(\theta_\circ)z_2]=0$, and therefore
\begin{equation}\label{eq:extended_correct_spec}
  \EE\!\big[g_e(w,\psi_\circ)\big]=0, \qquad \psi_\circ:=(\theta_\circ,v_\circ).
\end{equation}
Consequently, $P\in\mathcal{Q}_e(\psi_\circ)$ and the KL projection satisfies
\begin{equation}\label{eq:Qe_star_equals_P}
  Q_e^*(\psi_\circ)=P.
\end{equation}
%We define the extended model, denoted by $\mathcal{M}_e$, by the moment conditions
%\begin{equation}
%  \mathbf{E}^{Q}[g_{e}(w,\psi )]=0,\qquad Q\in \mathcal{Q}_{e}(\psi), \label{eq:moment:restrictions}
%\end{equation}%
%where
%\begin{equation*}
%g_{e}(w,\psi ):=\varepsilon (\theta )
%\begin{pmatrix}
%x \\
%z_{1} \\
%z_{2}
%\end{pmatrix}
%-
%\begin{pmatrix}
%v \\
%0 \\
%0
%\end{pmatrix}%
%=g_{b}(w,\theta )-%
%\begin{pmatrix}
%v \\
%0 \\
%0%
%\end{pmatrix}%
%,
%\end{equation*}
%$\psi :=(\theta ,v)\in \Psi $, $\Psi :=\Theta \times \mathcal{V}$ with $\mathcal{V}\subset \mathbb{R}^{d_x}$, and $\mathcal{Q}_{e}(\psi) := \left\{Q\in \mathbb{M};\,\mathbf{E}^{Q}[g_{e}(w,\psi )]=0\right\} $. In this model, $v:=\mathbf{E}[\varepsilon (\theta )x]$ is the covariance between the error
%and $x$.
%Note that the extended model is correctly specified under both endogeneity and exogeneity of $x$. For instance, if $\mathbf{E}[\varepsilon (\theta )x]\neq 0$ for every $\theta \in \Theta $, while $\mathbf{E}[\varepsilon(\theta _{\circ })(z_{1}^{\prime },z_{2}^{\prime })^{\prime }]=0$, then $v$ will be equal to $\mathbf{E}[\varepsilon (\theta _{\circ })x]$, and \eqref{eq:moment:restrictions} is satisfied for this $\theta _{\circ }$. In the following, we use the notation $v_{\circ } = \mathbf{E}[\varepsilon (\theta _{\circ })x]$. Therefore,
In the extended model, the minimizer, $Q_{e}^{\ast }(\psi )=\arg\inf_{Q\in \mathcal{Q}_{e}(\psi)} \mathrm{KL}(Q\Vert P)$, is equal to $P$, and the population moment conditions in the extended model are
\begin{equation*}
\mathbf{E}^{P}[g_{e}(w,\psi _{\circ })]=0.
\end{equation*}%
Moreover,
\begin{equation}
  \psi_{\circ} := \arg \max_{\psi\in\mathcal{V}; \mathcal{Q}_e(\psi)\neq \emptyset}\mathbf{E}\log \left( \frac{e^{\lambda_*^{\prime }(\psi )g_{e}(w,\psi )}}{\mathbf{E}[e^{\lambda_{\ast }^{\prime }(\psi )g_{e}(W,\psi )}]}\right) ,  \label{eq:theta:circ}
\end{equation}%
\noindent where $\lambda_*(\psi ):=\arg \min_{\lambda \in \mathbb{R}^{d}}\mathbf{E}[e^{\lambda ^{\prime }g_{e}(w,\psi )}]$ for every $\psi\in\mathcal{V}$ such that $\mathcal{Q}_e(\psi)\neq \emptyset$. In the following we denote the ETEL for the extended model by $\wh{q}(w_{1:n}|\theta,\mathcal{M}_e) := \prod_{i=1}^n \wh{q}_i(\theta|\mathcal{M}_e)$, where $\wh{q}_i(\theta|\mathcal{M}_e)$ is constructed as in \eqref{eq:KLbasic} with $g(w_i,\theta)$ replaced by $g_{e}(w_i,\psi )$.

\subsection{Numerical illustration} \label{sec:numerical_illustration}

To illustrate the fitting of the base and extended models, consider first
the base model under endogeneity. Let the data-generating process (DGP) be
\begin{align*}
y_{i}& =\gamma _{0}+x_{i}\,\beta +z_{1i}\,\gamma _{1}+\varepsilon _{i} \\
x_{i}& =\delta _{0}+z_{1i}\,\delta _{1}+z_{2i}\,\delta _{2}+u_{i} \\
z_{1i}& =v_{i} \\
z_{2i}& =\omega_{i}
\end{align*}%
for $i=1,\ldots ,n$, where $n\in \{250,500,1000,2000\}$. Suppose that the $%
(u_{i},v_{i},\omega_{i})$ are marginally Gaussian, that $\varepsilon _{i}$
is marginally a skewed Gaussian mixture $0.5\mathcal{N}(0.5,0.5^{2})+0.5%
\mathcal{N}(-0.5,1.118^{2})$, that $(\varepsilon _{i},u_{i},v_{i})$ have a
joint distribution induced by a Gaussian copula with covariance matrix $R=%
\begin{psmallmatrix}
1 & 0.6 & 0 \\
0.6 & 1 & 0 \\
0 & 0 & 1%
\end{psmallmatrix}$ and that the covariance of $\omega_{i}$ with each of the
other errors is zero. Also assume that each parameter is one (except for $%
\delta _{1}$, which is .5). Under this DGP, $z_{1i}$ is uncorrelated with $%
\varepsilon _{i}$ and correlated with $x_{i}$ but since $\omega_{i}$ is
uncorrelated with the other shocks, $z_{2i}$ is a valid instrument that is
also relevant for $x_{i} $. For each of the four sample sizes, the posterior
of $\theta :=(\beta ,\gamma _{0},\gamma _{1})$ is calculated from the four
moment conditions
\begin{equation*}
\mathbf{E}\bigl[(y_{i}-x_{i}\,\beta -\gamma _{0}-z_{1i}\,\gamma _{1})%
\begin{pmatrix}
x_{i} \\
1 \\
z_{1i} \\
z_{2i}%
\end{pmatrix}%
\bigr]=%
\begin{pmatrix}
0 \\
0 \\
0 \\
0%
\end{pmatrix}%
.
\end{equation*}%
The ETEL posterior is sampled by the tailored one-block Metropolis-Hastings (M-H) algorithm %
\citep{ChibGreenberg1995} for 20,000 iterations beyond a burn-in of  1,000
cycles. The marginal posterior density of $\beta $ for each sample size is
computed from these MCMC sampled draws. Kernel smoothed versions of the
posterior densities are given in Figure \ref{fig:betabasebyn}. As shown in Figure \ref{fig:betabasebyn}, as the sample
size increases, the posterior density of $\beta$ under the base model concentrates on a value quite different from the true value of $\beta $, indicating misspecification due to neglected endogeneity. 
\begin{figure}[h]
\centering
\includegraphics[width = .95\textwidth]{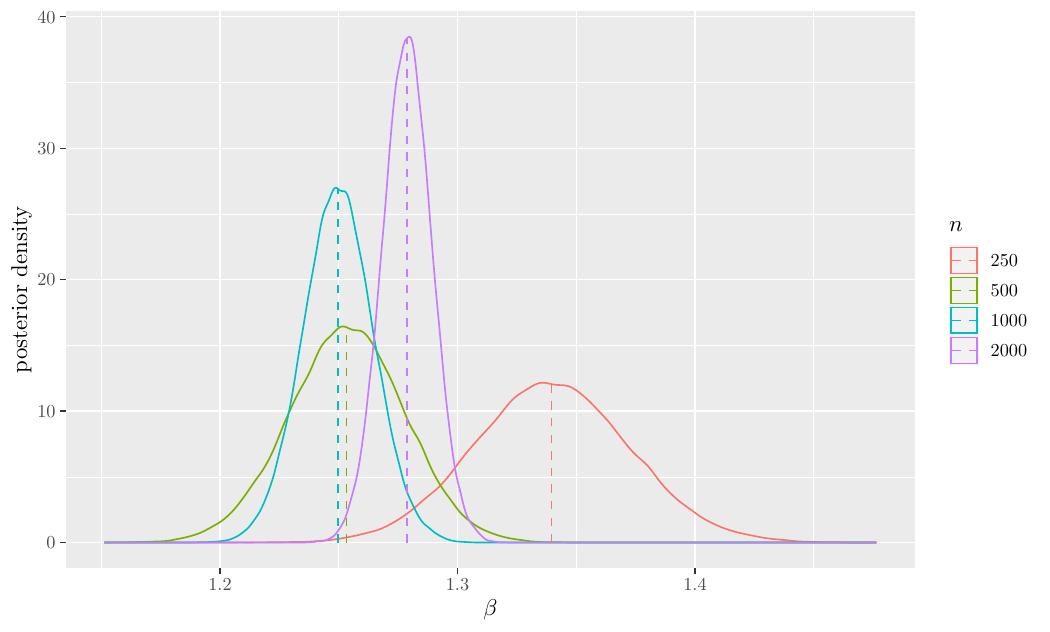}
\caption{Base model under neglected endogeneity: Marginal posterior
densities of $\protect\beta $ for different sample sizes. Posterior mean is
indicated by dashed vertical line.}
\label{fig:betabasebyn}
\end{figure}
In the extended (correctly specified) model we have
\begin{equation*}
\mathbf{E}\bigl[(y_{i}-x_{i}\,\beta -\gamma _{0}-z_{1i}\,\gamma _{1})%
\begin{pmatrix}
x_{i} \\
1 \\
z_{1i} \\
z_{2i}%
\end{pmatrix}%
\bigr]=%
\begin{pmatrix}
v \\
0 \\
0 \\
0%
\end{pmatrix}.
\end{equation*}%
The parameter of interest is now $\psi := (\beta ,\gamma _{0},\gamma _{1},v)$.
We use a default student-t prior on $v$ centered at the Generalized Method of Moments (GMM) estimate and spread given by 4 times the GMM asymptotic variance. The prior of $\theta $
is the same as in the base model. The ETEL posterior for each of the four different sample sizes is sampled by the tailored one-block M-H method for
20,000 iterations beyond a burn-in of 1,000 cycles. The marginal posterior densities of $\beta $ are given in Figure \ref{fig:betaextendedbyn} and
those of $v$ are in Figure \ref{fig:vbyn}. One can see that the posterior of $\beta $, even for $n=250$, is close to the true value of $\beta $, and, for
$n=2,000$, is essentially centered around the true value. In addition, the posterior of $v$, the $\text{cov}(x,\varepsilon)$, tends to concentrate
around the true value of 0.6.

\begin{figure}[h]
\centering
\includegraphics[width = .95\textwidth]{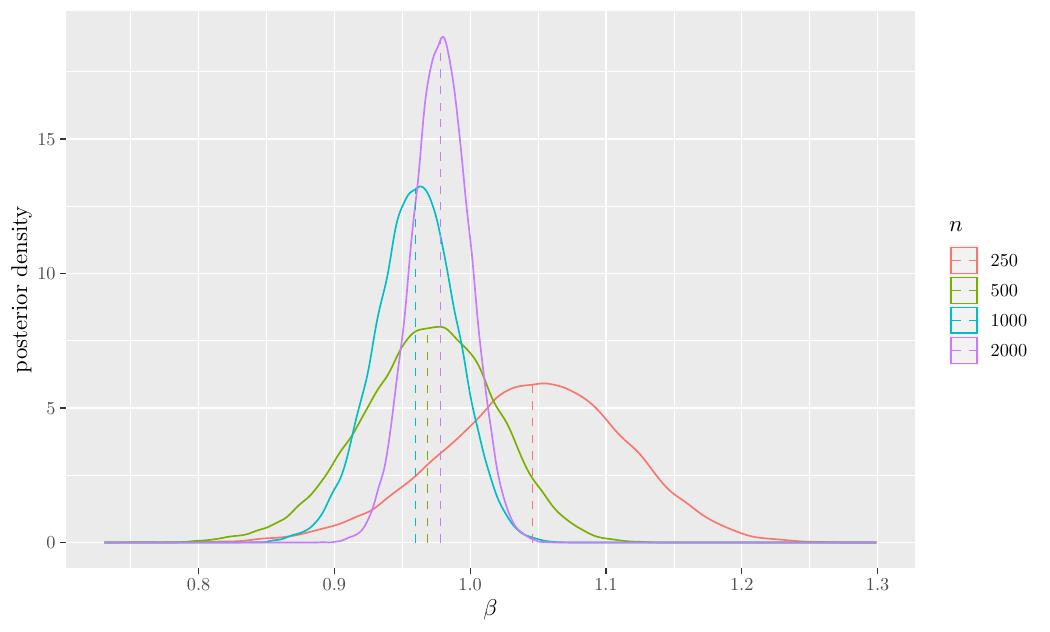}
\caption{Extended model ($x_{i}$ moment is inactive): Marginal posterior
densities of $\protect\beta $ for different sample sizes. Posterior mean is
indicated by dashed vertical line.}
\label{fig:betaextendedbyn}
\end{figure}

\begin{figure}[h]
\centering
\includegraphics[width = .95\textwidth]{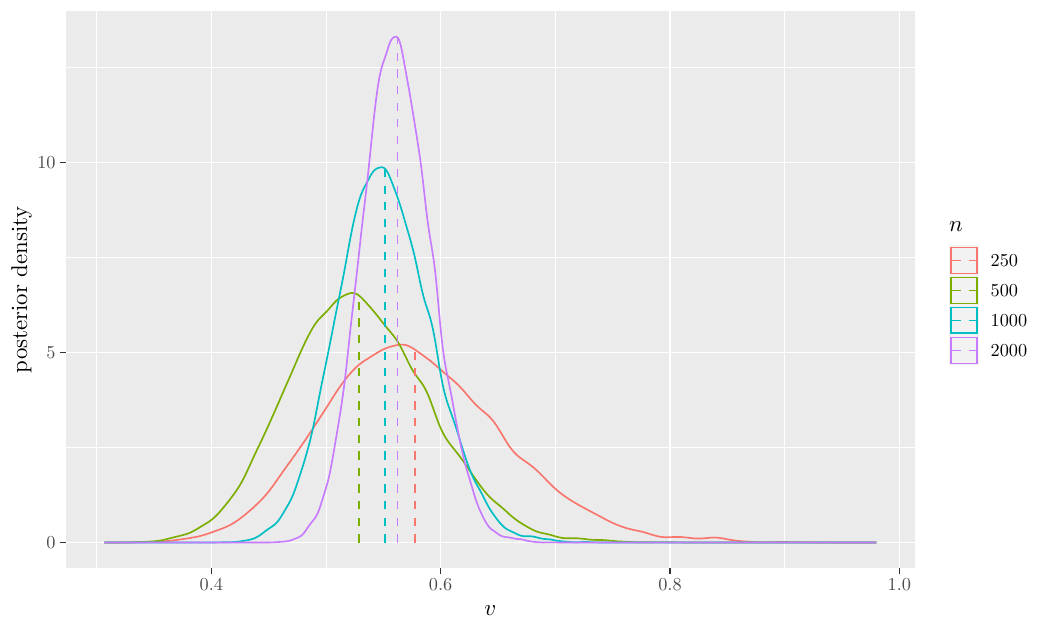}
\caption{Extended model under neglected endogeneity: Marginal posterior
densities of $v = \text{cov}(x,\protect\varepsilon)$ for different sample
sizes. Posterior mean is indicated by dashed vertical line.}
\label{fig:vbyn}
\end{figure}

\section{Testing procedure}

\label{sec:testing}

\subsection{Bayes factor}
Our Bayesian test of endogeneity is given by the Bayes factor of model $\mathcal{M}_{e}$ versus model $\mathcal{M}_{b}$ defined as:
\begin{equation*}
\text{BF}_{eb} = \frac{m(w_{1:n}|\mathcal{M}_{e})}{m(w_{1:n}|\mathcal{M}_{b})},
\end{equation*}%
where $m(w_{1:n}|\mathcal{M}_{b}):=\int_{H_n} \widehat{q}(w_{1:n}|\theta ,\mathcal{M}_{b})\pi (\theta )d\theta $ and $m(w_{1:n}|\mathcal{M}_{e}) := \int_{H_n} \widehat{q}(w_{1:n}|\psi,\mathcal{M}_{e})\pi (\psi )d\psi $ are the model marginal likelihoods arising from the ETEL functions (also called marginal ETEL
functions later on). We compute these by the method of \citet{Chib1995}, as
extended to general M-H chains in \citet{ChibJeliazkov2001}. We select $%
\mathcal{M}_{e}$ over $\mathcal{M}_{b}$ if $\log ($BF$_{eb})>0$, and select $%
\mathcal{M}_{b}$ otherwise.

According to the theory in \cite{ChibShinSimoni2018}, for valid comparisons
of moment condition models, the contending models must arise from a common
encompassing model and should have the same number of moment conditions. We
have ensured that this condition is met by including the $\mathbf{E}%
[\varepsilon _{i}(\theta )z_{2,i}]=0$ restriction in the base model, and not
excluding the $\mathbf{E}[\varepsilon _{i}(\theta )x_{i}]=v$ condition from
the extended model.

Intuitively, the Bayes factor picks the correct model because $\mathcal{M}%
_{b}$ is correctly specified when $x$ is exogenous and misspecified when $x$
is endogenous; however, $\mathcal{M}_{e}$ is correctly specified in both the
cases. Therefore, from \cite{ChibShinSimoni2018}, it follows that $\mathcal{M%
}_{b}$, which has $(d-p)$ overidentifying restrictions, rather than $M_{e}$,
which has $(d-p-d_{x})$ overidentifying restrictions, would be preferred by
the Bayes factor when $x$ is exogenous (because it has more overidentifying
restrictions than $\mathcal{M}_{e}$), whereas $\mathcal{M}_{e}$ would be
preferred when $x$ is endogenous (because $\mathcal{M}_{b}$ in that case
would be misspecified).

\subsection{Rationale}
% beginning of \inpu{theory_ML_new}
In this section we explain the rationale behind our testing procedure. The hypothesis that we want to test is the following:
$$H_{miss}: \quad P \textrm{ is such that }\nexists \theta\in\Theta\subset\mathbb{R}^p \textrm{ such that }\mathbf{E}[\varepsilon_{i}(\theta)x_{i}]= 0\qquad (endogeneity)$$
against
$$H_{cs}: \quad P \textrm{ is such that }\exists \theta\in\Theta\subset\mathbb{R}^p \textrm{ such that }\mathbf{E}[\varepsilon_{i}(\theta)x_{i}]= 0\qquad (exogeneity).$$
Here, the subscripts $miss$ and $cs$ are for misspecification and correct specification, respectively. The previous hypothesis can equivalently be written as $H_{miss}': v\neq 0$ and $H_{cs}': v = 0$. Our approach based on $\text{BF}_{eb}$ is equivalent to a Bayes test for $H_{miss}'$ versus $H_{cs}'$ based on a mixture prior on $v$ of the type $\pi_0\delta_0(v) + (1 - \pi_0)\pi(v)$, where $\delta_0(\cdot)$ denotes a Dirac mass on zero, $\pi_0\in[0,1]$, and $\pi(\cdot)$ is a continuous distribution. The two Bayes factors for these two approaches are numerically the same. The testing procedure works as follows: if $\text{BF}_{eb} \geq 1$, we conclude that $x$ is endogenous (\textit{i.e.} accept $H_{miss}$); if $\text{BF}_{eb} < 1$, we conclude that $x$ is exogenous (\textit{i.e.} accept $H_{cs}$).\\
\indent The next theorem shows that $H_{miss}$ and $H_{cs}$ can be expressed in terms of Kullback-Leibler divergences between $P$ and the set $\mathcal{Q}_{b}(\theta)$ of distributions that satisfy the moment restriction that we want to test as well as additional moment restrictions that are known to hold for $P$.
\begin{thm}\label{thm:equivalence}
  Suppose that there is a $\theta\in\Theta$ such that $\mathbf{E}[\varepsilon_{i}(\theta)(z_{1,i}'z_{2,i}')']= 0$ and that Assumption \ref{Ass_absolute_continuity} holds true. Consider the following statements:
  \begin{enumerate}
    \item[(i).] $P$ is such that $\nexists \theta\in\Theta$ such that $\mathbf{E}[\varepsilon_{i}(\theta)x_{i}] = 0$.
    \item[(ii).] $\mathrm{KL}(P||Q_b^*(\theta_*))>0$.
    \item[(iii). ] $P$ is such that $\exists \theta\in\Theta$ such that $\mathbf{E}[\varepsilon_{i}(\theta)x_{i}] = 0$.
    \item[(iv).] $\mathrm{KL}(P||Q_b^*(\theta_*))=0$.
  \end{enumerate}
  Then, \textit{(i)} is equivalent to \textit{(ii)}, and \textit{(iii)} is equivalent to \textit{(iv)}.
\end{thm}

This theorem makes clear that to test $H_{miss}$ and $H_{cs}$ one can equivalently focus on the Kullback-Leibler divergence $\mathrm{KL}(P||Q_b^*(\theta_*))$. Our Bayes test is based on Bayes factor and comparison of marginal likelihoods. There is a strict link between marginal likelihood and the Kullback-Leibler divergence: log-marginal likelihood of the base model behaves asymptotically as $- n \mathrm{KL}(P||Q_b^*(\theta_*))$ plus a penalty term, where the penalty depends on the number of parameters to estimate, and similarly for the log-marginal likelihood of the extended model. We are going to demonstrate this fact in the rest of this section.\\
\indent From the \citet{Chib1995} identity, we have for the base model: $\forall \theta\in H_n\subset\mathbb{R}^p$,
$$\log m(w_{1:n}|\mathcal{M}_b) = \log \pi(\theta|\mathcal{M}_b) + \log \wh{q}(w_{1:n}|\theta,\mathcal{M}_b) - \log \pi^n(\theta|w_{1:n},\mathcal{M}_b),$$
and similarly for the extended model. Because this identity is true for every $\theta\in H_n$, it is true for $\theta = \theta_*$ under Assumptions \ref{Ass_absolute_continuity} and \ref{ass:feasibility}: $\log m(w_{1:n}|\mathcal{M}_b) = \log \pi(\theta_*|\mathcal{M}_b) + \log \wh{q}(w_{1:n}|\theta_*,\mathcal{M}_b) - \log \pi^n(\theta_*|w_{1:n},\mathcal{M}_b)$. Next, let us introduce the local parameters $h_{\theta}:=\sqrt{n}(\theta - \theta_*)$ and $h_{\psi}:=\sqrt{n}(\psi - \psi_\circ)$, so that by the formula for transformations of random variables: $\pi^n(\theta|w_{1:n},\mathcal{M}_b) = \pi_{h_{\theta}}^n(\sqrt{n}(\theta - \theta_{*})|w_{1:n},\mathcal{M}_b)n^{p/2}$ and $\pi^n(\psi|w_{1:n},\mathcal{M}_e) = \pi_{h_{\psi}}^n(\sqrt{n}(\psi - \psi_{\circ})|w_{1:n},\mathcal{M}_e)n^{(p+d_x)/2}$, where $\pi_{h_{\theta}}^n(\cdot|w_{1:n},\mathcal{M}_b)$ and $\pi_{h_{\psi}}^n(\cdot|w_{1:n},\mathcal{M}_e)$ denote the posterior density of $h_{\theta}$ and $h_{\psi}$, respectively. By replacing this in the expression of the marginal likelihoods we obtain: $\forall \theta\in H_n\subset\mathbb{R}^p$,
% \begin{multline}
%   \log m(w_{1:n}|\mathcal{M}_b) = \log \pi(\theta|\mathcal{M}_b) + \log \wh{q}(w_{1:n}|\theta,\mathcal{M}_b) - \log \pi_{h_{\theta}}^n(\sqrt{n}(\theta - \theta_{*})|w_{1:n},\mathcal{M}_b)\\
%   - \frac{p}{2}\log(n) = \pi(\theta_*|\mathcal{M}_b) + \log \wh{q}(w_{1:n}|\theta_*,\mathcal{M}_b) - \log \pi_{h_{\theta}}^n(0|w_{1:n},\mathcal{M}_b) - \frac{p}{2}\log(n),\label{log_ML_base_initial}
% \end{multline}
% and
% \begin{multline}
%   \log m(w_{1:n}|\mathcal{M}_e) = \log \pi(\psi|\mathcal{M}_e) + \log \wh{q}(w_{1:n}|\psi,\mathcal{M}_e) - \log \pi_{h_{\psi}}^n(\sqrt{n}(\psi - \psi_{\circ})|w_{1:n},\mathcal{M}_e) \\ \hfill
%   - \frac{p+d_x}{2}\log(n)\\
%   = \log \pi(\psi_\circ|\mathcal{M}_e) + \log \wh{q}(w_{1:n}|\psi_\circ,\mathcal{M}_e) - \log \pi_{h_{\psi}}^n(0|w_{1:n},\mathcal{M}_e) - \frac{p+d_x}{2}\log(n)\label{log_ML_extended_initial}.
%\end{multline}
\begin{multline}
  \log m(w_{1:n}|\mathcal{M}_b) = \log \pi(\theta|\mathcal{M}_b) + \log \wh{q}(w_{1:n}|\theta,\mathcal{M}_b) - \log \pi_{h_{\theta}}^n(\sqrt{n}(\theta - \theta_{*})|w_{1:n},\mathcal{M}_b)
  - \frac{p}{2}\log(n) \\
  = \log \pi(\theta_*|\mathcal{M}_b) + \log \wh{q}(w_{1:n}|\theta_*,\mathcal{M}_b) - \log \pi_{h_{\theta}}^n(0|w_{1:n},\mathcal{M}_b) - \frac{p}{2}\log(n),\label{log_ML_base_initial}
\end{multline}
and, $\forall \psi\in H_n\subset\mathbb{R}^{p+d_x}$,
\begin{multline}
  \log m(w_{1:n}|\mathcal{M}_e) = \log \pi(\psi|\mathcal{M}_e) + \log \wh{q}(w_{1:n}|\psi,\mathcal{M}_e) \\
  \quad - \log \pi_{h_{\psi}}^n(\sqrt{n}(\psi - \psi_{\circ})|w_{1:n},\mathcal{M}_e)   - \frac{p+d_x}{2}\log(n) \\
  = \log \pi(\psi_\circ|\mathcal{M}_e) + \log \wh{q}(w_{1:n}|\psi_\circ,\mathcal{M}_e) - \log \pi_{h_{\psi}}^n(0|w_{1:n},\mathcal{M}_e) - \frac{p+d_x}{2}\log(n). \label{log_ML_extended_initial}
\end{multline}
The intuition for expressing the posterior of $\theta$ in terms of the posterior of the local parameter is that the Jacobian of the transformation makes explicit the role played by the dimension of the model, while the local parameter has a posterior distribution that is approximately Gaussian. This is true in both cases \textit{(i)} and \textit{(iii)} of Theorem \ref{thm:equivalence}. Hence, the Jacobian induces an explicit dimension-dependent penalty through posterior concentration.\\
\indent Therefore, the log-marginal likelihood decomposes into two terms that are bounded in probability as $n\rightarrow\infty$ (the prior ordinate evaluated at the pseudo-true value and the posterior ordinate of the local parameter) and two terms that diverge with $n$: the log-ETEL term and a model-dimension penalty of order $\tfrac{1}{2}\log n$ per parameter. Asymptotically, the marginal likelihood behaves like a penalized log-ETEL criterion, where the penalty arises endogenously from posterior concentration via the local reparameterization, rather than being imposed ad hoc.\\

Of course, for a testing procedure based on marginal likelihoods to be valid, it is necessary to establish that 
$\pi_{h_{\theta}}^n(\sqrt{n}(\theta - \theta_{*}) \mid w_{1:n},\mathcal{M}_b)$ and 
$\pi_{h_{\psi}}^n(\sqrt{n}(\psi - \psi_{\circ}) \mid w_{1:n},\mathcal{M}_e)$ 
are bounded in probability as $n \rightarrow \infty$. This requirement can be quite challenging to verify, particularly in non-standard settings such as the one considered here, where there is no parametric likelihood and the models may be misspecified. We establish these results in Theorems \ref{thm_BvM_misspecified} and \ref{thm_BvM:extended} in the Online Appendix, which refine Theorems~1 and~2 of \cite{ChibShinSimoni2018}. 

A critical step in proving these results is to show that the log-ETEL function satisfies a stochastic local asymptotic normality (LAN) property. While the remainder of the Bernstein--von Mises argument follows standard lines, establishing stochastic LAN is challenging because the ETEL function is itself a random likelihood. In this paper, we provide a new and more direct proof of the LAN property for the log-ETEL function (see Theorems \ref{thm:stochasticLAN:endogeneity}, \ref{thm:stochasticLAN:exogeneity}, and \ref{thm:stochasticLAN:extended} in the Online Appendix). Our proof leverages the specific structure of the IV regression problem: due to linearity, each term in the Mean Value Theorem expansion of the log-ETEL function around $\theta_*$ can be controlled more directly and uniformly in $h_{\theta}$ over compact sets. This allows us to avoid the empirical process theory used in \cite{ChibShinSimoni2018}. 

The final step toward understanding the asymptotic behavior of the marginal
likelihood is provided by Theorems \ref{lem_likelihood_ET_approx} and
\ref{lem_likelihood_ET_approx_extended}, which derive stochastic expansions
of the log-ETEL function in the base and extended models. These expansions,
which were not made explicit in \cite{ChibShinSimoni2018}, are new to the best
of our knowledge.

Our starting point is the exact master identity for the log-ETEL,
\eqref{eq:log_etel_master_identity_normalized}. Evaluating this identity at
$\theta=\theta_*$ and writing
$\widehat\lambda(\theta_*)=\lambda_*(\theta_*)+
(\widehat\lambda(\theta_*)-\lambda_*(\theta_*))$,
we obtain the expansion by performing a second-order Taylor expansion in
$\lambda$ around $\lambda_*(\theta_*)$. The stochastic LAN property delivers
a linear representation for $\sqrt n(\widehat\lambda(\theta_*)-\lambda_*(\theta_*))$,
which, when substituted back into the master identity, yields the quadratic
empirical-process terms reported below.

The assumptions under which the results hold are collected in
Section~\ref{sec:assumptions}. We use the notation
$\EE_n[\cdot]=n^{-1}\sum_{i=1}^n(\cdot)$ for the empirical mean and
$\mathbb G_n f=\sqrt n(\EE_n[f]-\EE[f])$ for the centered empirical process.

\begin{thm}[Base model: stochastic expansion of log-ETEL]
\label{lem_likelihood_ET_approx}
Let Assumptions \ref{Ass_absolute_continuity}--\ref{Ass_2_NS} and
\ref{Ass_3}\textit{(d)}--\textit{(f)} hold. Then,
\begin{multline}
\log \widehat q(w_{1:n}\mid\theta_*,\mathcal M_b)
= -n\log n
+ \sum_{i=1}^n
\log\!\left(
\frac{\exp\{\lambda_*(\theta_*)'g_b(w_i,\theta_*)\}}
{\EE_n[\exp\{\lambda_*(\theta_*)'g_b(w_j,\theta_*)\}]}
\right)
\\
 - \mathbb G_n\!\left[
\tau_i^{\dagger}(\lambda_*,\theta_*)\,g_b(w_i,\theta_*)'
\right]
\Omega_*^{\dagger}(\theta_*)^{-1}
\mathbb G_n\!\left[g_b(w_i,\theta_*)\right]
+ n(\wh\lambda(\theta_*) - \lambda_*(\theta_*))'\EE[g_b(w_i,\theta_*)]
\\
+ \frac12\,
\mathbb G_n\!\left[
\tau_i^{\dagger}(\lambda_*,\theta_*)\,g_b(w_i,\theta_*)'
\right]
\Omega_*^{\dagger}(\theta_*)^{-1}
\mathbb G_n\!\left[
\tau_i^{\dagger}(\lambda_*,\theta_*)\,g_b(w_i,\theta_*)
\right]
+ o_p(1).
\label{lem_likelihood_eq:1}
\end{multline}
Moreover,
\[
\mathbb G_n\!\left[
\tau_i^{\dagger}(\lambda_*,\theta_*)\,g_b(w_i,\theta_*)
\right]
\xrightarrow{d}
\mathcal N\!\left(0,\Omega_*^{\dagger}(\theta_*)\right),
\]
where
$\tau_i^{\dagger}(\lambda_*,\theta_*):=[dQ^*(\theta_*)/dP](w_i)$ and
$\Omega_*^{\dagger}(\theta_*):=
\EE^{Q^*(\theta_*)}[\varepsilon_i(\theta_*)^2\widetilde w_i\widetilde w_i']$.
In addition,
\[
\sqrt n(\widehat\lambda(\theta_*)-\lambda_*(\theta_*))
= -\Omega_*^{\dagger}(\theta_*)^{-1}
\mathbb G_n\!\left[
\tau_i^{\dagger}(\lambda_*,\theta_*)\,g_b(w_i,\theta_*)
\right]
+ o_p(1).
\]
\end{thm}

This theorem establishes a decomposition of the log-ETEL function for the base model $\mathcal{M}_b$, characterizing its asymptotic behaviour. This decomposition is used to prove the consistency of our Bayes factor testing procedure. Theorem \ref{lem_likelihood_ET_approx} states that $\log \wh{q}(w_{1:n}|\theta_*,\mathcal{M}_{b}) + n\log n$ can be expressed, up to an $o_p(1)$ term, as the sum of four random components. The third and fifth terms on the right-hand side of equation  \eqref{lem_likelihood_eq:1} are both of order $\mathcal{O}_p(1)$, while the second and fourth terms are of order $\mathcal{O}_p(n)$ and $\mathcal{O}_p(\sqrt{n})$, respectively, when $x_i$ is endogenous, and equal to zero when $x_i$ is exogenous. The fourth term is linear, with its rate $\mathcal{O}_p(\sqrt{n})$ following from the last part of the theorem, whereas the fifth term in \eqref{lem_likelihood_eq:1} is quadratic. The term $n\log n$ does not affect the comparison, as it cancels out with the corresponding term in the log-ETEL function of the extended model, as shown in the next theorem.\\
\indent For the extended model, we recall that $\psi_\circ = (\theta_\circ',v_\circ')'$ denotes the true value of the parameter in the extended model with $v_{\circ} = \EE[\varepsilon(\theta_\circ)x]$.
\begin{thm}[Extended model.]\label{lem_likelihood_ET_approx_extended}
  Let Assumptions \ref{Ass_0_NS}, \ref{ass:feasibility}, \ref{Ass_1_NS} with $\Theta$ replaced by $\Psi$ and Assumptions \ref{Ass_2_NS} and \ref{Ass_3} \textit{(d)}-\textit{(f)} with $\theta_*$ and $\lambda_*(\theta_*)$ replaced with $\theta_\circ$ and $0$, respectively, hold. Then,
  \begin{equation}\label{eq_lem_likelihood_ET_approx_extended}
    \log \wh{q}(w_{1:n}|\psi_\circ,\mathcal{M}_{e}) =   - n\log n - \frac{1}{2}\mathbb{G}_n\left[g_e(w_i,\psi_\circ)'\right] \Omega_{\psi_{\circ}}^{-1} \mathbb{G}_n\left[g_e(w_i,\psi_\circ)\right] + o_p(1),%\\
    %- 2\sqrt{n}\mathbb{G}_n\left[g_b(w_i,\theta_\circ)'\right]\Omega_{\psi_\circ}^{-1}\wtl{v}_{\circ} + n\wtl{v}_\circ' \Omega_{\psi_\circ}^{-1}\wtl{v}_\circ + o_p(1),
  \end{equation}
\noindent where $\Omega_{\psi_{\circ}} := \EE[g_e(w_i,\psi_\circ)g_e(w_i,\psi_\circ)']$, and $\mathbb{G}_n\left[g_e(w_i,\psi_\circ)'\right] \Omega_{\psi_\circ}^{-1} \mathbb{G}_n\left[g_e(w_i,\psi_\circ)\right]\xrightarrow{d}\chi_{d}^2$, where $\chi_{d}^2$ denotes a chi square distribution with $d$ degrees of freedom. %If $\EE[\varepsilon_i(\theta_\circ) x_i] = 0$, then
%\begin{equation}
%  \log \wh{q}(w_{1:n}|\psi_\circ,M_{b}) = - n\log n - \frac{1}{2}\mathbb{G}_n\left[g_b(w_i,\theta_\circ)'\right] \Omega_{\circ}^{-1} \mathbb{G}_n\left[g_b(w_i,\theta_\circ)\right] + o_p(1)
%\end{equation}
%\noindent and $\mathbb{G}_n\left[g_b(w_i,\psi_\circ)'\right] \Omega_{\psi_\circ}^{-1} \mathbb{G}_n\left[g_b(w_i,\psi_\circ)\right]\xrightarrow{d}\chi_{d}^2$.
\end{thm}

Theorem \ref{lem_likelihood_ET_approx_extended} establishes the asymptotic behaviour of the log-ETEL function of the extended model $\mathcal{M}_e$. Unlike the base model, $\log \wh{q}(w_{1:n}|\psi_\circ,\mathcal{M}_{e}) + n\log n$ is equal, up to an asymptotically negligible term, to a quadratic random term that remains bounded in probability as $n\rightarrow \infty$.\\
\indent If $\EE[\varepsilon_i(\theta_\circ) x_i] = 0$ (exogenous case), so that the assumptions in Theorem \ref{lem_likelihood_ET_approx} hold with $\theta_*$ replaced by $\theta_\circ$ and $\lambda_*(\theta_*) = \lambda_*(\theta_\circ) = 0$, then the log-ETEL function in the base model simplifies as
\begin{equation}\label{eq_lem_likelihood_ET_approx}
  \log \wh{q}(w_{1:n}|\theta_\circ,\mathcal{M}_{b}) = -n\log n - \frac{1}{2}\mathbb{G}_n\left[g_b(w_i,\theta_\circ)'\right] \Omega_\circ^{-1} \mathbb{G}_n\left[g_b(w_i,\theta_\circ)\right] + o_p(1),
\end{equation}
\noindent where $\Omega_\circ := \EE[\varepsilon_i(\theta_\circ)^2 \wtl{w}_i\wtl{w}_i']$, $\mathbb{G}_n\left[g_b(w_i,\theta_\circ)'\right] \Omega_\circ^{-1} \mathbb{G}_n\left[g_b(w_i,\theta_\circ)\right]\xrightarrow{d}\chi_{d}^2$, and $\chi_{d}^2$ denotes a chi square distribution with $d$ degrees of freedom. For the extended model, if $\EE[\varepsilon_i(\theta_\circ) x_i] = 0$ then $g_e(w_i,\psi_\circ) = g_b(w_i,\theta_\circ)$ and the log-ETEL function slightly simplifies as:
\begin{equation}
  \log \wh{q}(w_{1:n}|\psi_\circ,\mathcal{M}_{e}) = - n\log n - \frac{1}{2}\mathbb{G}_n\left[g_b(w_i,\theta_\circ)'\right] \Omega_{\circ}^{-1} \mathbb{G}_n\left[g_b(w_i,\theta_\circ)\right] + o_p(1),
\end{equation}
\noindent where $\mathbb{G}_n\left[g_b(w_i,\psi_\circ)'\right] \Omega_{\psi_\circ}^{-1} \mathbb{G}_n\left[g_b(w_i,\psi_\circ)\right]\xrightarrow{d}\chi_{d}^2$. Hence, when $x$ is exogenous, $\log \wh{q}(w_{1:n}|\theta_\circ,\mathcal{M}_{b}) $ and $\log \wh{q}(w_{1:n}|\psi_\circ,\mathcal{M}_{e})$ are equal asymptotically and they cancel in the comparison of the marginal likelihoods.\\
\indent In case of endogeneity, instead, $\log \wh{q}(w_{1:n}|\theta_\circ,\mathcal{M}_{b}) $ and $\log \wh{q}(w_{1:n}|\psi_\circ,\mathcal{M}_{e})$ are different and they play a central role in the comparison of marginal likelihoods. In this case, it is important to consider the behaviour of the average log-ETEL function $\frac{1}{n}\log \wh{q}(w_{1:n}|\theta_*,\mathcal{M}_b)$ which stays bounded asymptotically. The following two corollaries demonstrates that asymptotically the average log-ETEL functions behave as the Kullback-Leibler divergence, up to a $\log(n)$ term. While these results are implicit in the definition of the ETEL, we provide here a formal and explicit statement and in the Appendix their proof. This allows us to understand the behaviour of the marginal likelihood.
\begin{cor}[Base model.]\label{thm_limit_ETEL}
Suppose Assumptions \ref{Ass_absolute_continuity} - \ref{Ass_2_NS} and \ref{Ass_3} \textit{(d)}-\textit{(f)} hold. Then, as $n\rightarrow \infty$,
  \begin{equation}
    \frac{1}{n}\log \wh{q}(w_{1:n}|\theta_*,\mathcal{M}_b) + \log(n) \xrightarrow{p} \mathbf{E}\left[\log(dQ_b^*(\theta_*)/dP)\right],
  \end{equation}
  \noindent where $\mathbf{E}\left[\log(dQ_b^*(\theta_*)/dP)\right] = \EE\left[\log\Bigl(\frac{e^{\lambda_{*}(\theta_*)'g_b(w,\theta_*)}}{\mathbf{E}[e^{\lambda_*(\theta_*)'g_b(w,\theta_*)}]}\Bigr)\right] = -\mathrm{KL}(P||Q_b^*(\theta_*))$.
\end{cor}

\begin{cor}[Extended model.]\label{thm_limit_ETEL_extended_model}
Suppose Assumptions \ref{Ass_0_NS} -\ref{Ass_1_NS} with $\Theta$ replaced by $\Psi$, \ref{Ass_2_NS} and \ref{Ass_3} \textit{(d)}-\textit{(f)} with $\theta_*$ and $\lambda_*(\theta_*)$ replaced with $\theta_\circ$ and $0$, respectively, hold. Then, as $n\rightarrow \infty$,
  \begin{equation}
    \frac{1}{n}\log \wh{q}(w_{1:n}|\psi_\circ,\mathcal{M}_e) + \log(n) \xrightarrow{p} \mathbf{E}\left[\log(dQ_e^*(\psi_\circ)/dP)\right],
  \end{equation}
  \noindent where $\mathbf{E}\left[\log(dQ_e^*(\psi_\circ)/dP)\right] = \EE\left[\log\Bigl(\frac{e^{\lambda_{*}(\psi_\circ)'g(w,\psi_\circ)}}{\mathbf{E}[e^{\lambda_*(\psi_\circ)'g(w,\psi_\circ)}]}\Bigr)\right] = \mathrm{KL}(P||Q_e^*(\psi_\circ))$.
\end{cor}
Notice that $\mathbf{E}\left[\log(dQ_e^*(\psi_\circ)/dP)\right] = 0$ since the extended model is correctly specified and so $dQ_e^*(\psi_\circ)/dP = 1$. \\

\indent From Theorems \ref{lem_likelihood_ET_approx} and \ref{lem_likelihood_ET_approx_extended}, and Theorems \ref{thm_BvM_misspecified} and \ref{thm_BvM:extended} in the Online Appendix and from \eqref{log_ML_base_initial}-\eqref{log_ML_extended_initial}, there exists an $N$ such that for every $n>N$:
\begin{eqnarray}
  \log m(w_{1:n}|\mathcal{M}_b) & = & - n\log n + \sum_{i=1}^n \log\left(\frac{e^{\lambda_*(\theta_*)'g_b(w_i,\theta_*)}}{\EE_n[e^{\lambda_*(\theta_*)' g_b(w_j,\theta_*)}]}\right)\nonumber\\
  && \hfill \qquad + n(\wh\lambda(\theta_*) - \lambda_*(\theta_*))'\EE[g_b(w_i,\theta_*)] - \frac{p}{2}\log(n) + \mathcal{O}_p(1),\label{log_ML_base}\\
  \log m(w_{1:n}|\mathcal{M}_e) & = & -n\log(n) - \frac{p+d_x}{2}\log(n) + o_p(1).\label{log_ML_extended}
\end{eqnarray}
These log-marginal likelihoods quantify the overall validity of the model. In fact, from these expressions one sees that when $x_i$ is exogenous, that is, $\EE[\varepsilon_i(\theta_\circ) x_i] = 0$, then $\lambda_*(\theta_*) = 0$ and $\sum_{i=1}^n \log\left(\frac{e^{\lambda_*(\theta_*)'g_b(w_i,\theta_*)}}{\EE_n[e^{\lambda_*(\theta_*)' g_b(w_j,\theta_*)}]}\right) = 0$ for every $n\in\mathbb{N}$. Therefore, it is clear that asymptotically $\log m(w_{1:n}|\mathcal{M}_b)$ is larger than $\log m(w_{1:n}|\mathcal{M}_e)$.\\
\indent On the other hand, when there is no $\theta\in\Theta$ such that $\EE[\varepsilon_i(\theta) x_i] = 0$, then $\lambda_*(\theta_*) \neq 0$ and $\sum_{i=1}^n \log\left(\frac{e^{\lambda_*(\theta_*)'g_b(w_i,\theta_*)}}{\EE_n[e^{\lambda_*(\theta_*)' g_b(w_j,\theta_*)}]}\right) - n\log n$ diverges to $-\infty$ faster than the last two terms in \eqref{log_ML_base}, so that asymptotically $\log m(w_{1:n}|\mathcal{M}_b)$ is smaller than $\log m(w_{1:n}|\mathcal{M}_e)$. This is the main intuition of the consistency results in Theorems \ref{thm_consistency_selection_misspecification_2} and \ref{thm_consistency_selection_correct_spectification} in the next section. The proof of these theorems, which is provided in the Appendix, is more involved than this argument because the theorems provide an `if and only if' statement, which is stronger than consistency.

\subsection{Consistency of the testing procedure}
% Beginning of \input{Model_Selection_Consistency.tex}
%\input{Model_Selection_Consistency}
We now use the preceding theory to establish consistency of our testing procedure based on the Bayes factor constructed from the marginal ETEL functions. The theorems below establish that, as the sample size increases, $BF_{eb}$ selects $\mathcal{M}_{b}$ if and only if $x$ is exogenous, and selects $\mathcal{M}_{e}$ if and only if $x$ is endogenous, with probability approaching one. %We add the following assumption

%\begin{ass}\label{Ass_prior}
%  (a) $\pi $ is a continuous probability measure that admits a density with respect to the Lebesgue measure; (b) $\pi $ is positive on a neighborhood of $\theta_*$.
%\end{ass}

\begin{thm}\label{thm_consistency_selection_misspecification_2}
Let Assumptions \ref{Ass_absolute_continuity} - \ref{ass:feasibility}, \ref{Ass_2_NS}, \ref{Ass_identification_base_model}, \ref{Ass_identification_extended_model} hold and let Assumptions \ref{Ass_1_NS} and \ref{Ass_3} hold for $\theta_*$ and $\lambda_*(\theta_*)$ and also for $\theta_*$ and $\lambda_*(\theta_*)$ replaced with $\theta_\circ$ and $0$, respectively. Let the priors on $\theta$ and $\psi$ be continuous probability measures that admit densities with respect to the Lebesgue measure and that are positive on a neighborhood of $\theta_*$ and $\psi_{\circ}$, respectively. Let us consider the comparison of models $\mathcal{M}_{b}$ and $\mathcal{M}_{e}$.
Then,
\begin{equation*}
\lim_{n\rightarrow \infty }P\left( \log m(w_{1:n}|\mathcal{M}_{e})>\log
m(w_{1:n}|\mathcal{M}_{b})\right) =1
\end{equation*}
if and only if there is no $\theta\in\Theta$ such that $\mathbf{E}[\varepsilon_{i}(\theta)x_i]=0 $ holds, and the limit is zero otherwise.
\end{thm}

As we show in the proof, the failure of the necessary and sufficient condition $\mathbf{E}[\varepsilon _{i}(\theta )x_{i}]=0$ for any $\theta \in\Theta$, is equivalent to the inequality $\mathrm{KL}(P||Q_{e}^{\ast }(\psi ))<\mathrm{KL}(P||Q_{b}^{\ast }(\theta ))$, where $\mathrm{KL}(P||Q_{e}^{\ast }(\psi_\circ)) = 0$. In this case, the log-ETEL function dominates the other components of the log-marginal ETEL function so that the build-in penalty does not play any role. Thus, as in the general result in \cite[Theorem 3.2]{ChibShinSimoni2018} for moment condition models, comparing the
log-marginal likelihoods of the base and extended models, and selecting the one with the higher value, in the limit, selects the model that is closest in the KL divergence to the true model. In the framework of the present paper, this means that the comparison of marginal likelihoods allows to correctly conclude that $x_i$ is endogenous.\\
\indent Next, we show what happens when the variables $x_i$ are exogenous so that the moment restriction $\mathbf{E}[\varepsilon _{i}(\theta )x_{i}]=0$ holds for a particular value $\theta_\circ$ and the two models under comparison are correctly specified. The next theorem states that in this case the base model is selected. This is understandable through an argument of parsimony: the base model has the smaller number of parameters to estimate and so it is the preferred one when it is correctly specified.
\begin{thm}\label{thm_consistency_selection_correct_spectification}
Let Assumptions \ref{Ass_absolute_continuity} - \ref{ass:feasibility}, \ref{Ass_2_NS}, \ref{Ass_identification_base_model}, and \ref{Ass_identification_extended_model} hold and let Assumptions \ref{Ass_1_NS} and \ref{Ass_3} hold for both $\theta_*$ and $\lambda_*(\theta_*)$ and also for $\theta_*$ and $\lambda_*(\theta_*)$ replaced with $\theta_\circ$ and $0$, respectively. Let the priors on $\theta$ and $\psi$ be continuous probability measures that admit densities with respect to the Lebesgue measure and that are positive in a neighborhood of $\theta_*$ and $\psi_{\circ}$, respectively. Let us consider the comparison of models $\mathcal{M}_{b}$ and $\mathcal{M}_{e}$. Then,
\begin{equation*}
  \lim_{n\rightarrow \infty }P\left( \log m(w_{1:n}|\mathcal{M}_{b})>\log m(w_{1:n}|\mathcal{M}_{e})\right) = 1
\end{equation*}
if and only if there is a $\theta_\circ\in\Theta$ such that $\mathbf{E}[\varepsilon_{i}(\theta_\circ)x_i]=0 $ holds.
\end{thm}

When $x$ is exogenous, as in the previous theorem, both the log-ETEL function and the build-in penalty plays a role in selecting the correct model.
% Ending of \input{Model_Selection_Consistency.tex}

\paragraph{Discussion.} In this and the previous subsection, we demonstrate that our model selection criteria favor a model with a smaller Kullback-Leibler Information Criterion (KLIC). When two models share the same KLIC, our procedure opts for the model with a greater number of overidentifying restrictions, \textit{i.e.}, a more parsimonious or less flexible model. Interestingly, this aligns with the goal of \cite{sin1996information}'s penalized likelihood criteria for a parametric model. Consequently, our proposed model selection procedure in this paper and \cite{ChibShinSimoni2018} can be viewed as a fully Bayesian semi-parametric version of consistent model selection criteria, applied specifically to an endogeneity testing problem. Unlike other frequentist procedures, the `penalty' term required for consistency is inherently built into our Bayesian calculation. This point was not stressed in \cite{ChibShinSimoni2018} and it is a contribution of this paper.

\cite{Andrews1999}, \cite{AndrewsLu2001}, and \cite{HongPrestonShum2003} have proposed and studied model selection criteria for moment condition models, even though a formal likelihood function is not defined. These criteria involve a penalization term that is attached to the Generalized Method of Moments (GMM) and, more broadly, the Generalized Empirical Likelihood (GEL) objective function, rather than the likelihood function. Examples of such frequentist model selection approaches based on GMM estimation can be found in Online Appendix \ref{App:comparisonGMM}. However, the relationship between these model selection criteria and the KLIC minimization principle of \cite{sin1996information} for potentially misspecified parametric models is not immediately apparent.

It is noteworthy that our procedure exhibits the same asymptotic behavior as \cite{hong2012bayesian}'s generalized empirical likelihood Bayes factor. They impose a separate prior on the Lagrangian multiplier that is independent of $\theta$, which does not guarantee the imposition of moment restrictions. In contrast, we introduce an additional parameter $v$ to the `inactive' moment restriction, ensuring that our prior on $\theta$ and $v$ respects the moment restrictions.
%In order to show consistency of our Bayesian test of endogeneity, we first demonstrate that the marginal likelihood can be approximated by the projection (as defined by the Kullback-Leibler divergence) of the true distribution on the set of distributions that satisfy the exogeneity restrictions of the base model. This result is knew to the best of our knowledge.\\

\subsection{Testing among more than two models}
Our testing procedure can be extended to settings in which more than two models are compared.
Consider the case in which only a subset of the variables in $x$ is endogenous. To start, suppose that $d_x = 2$ and that only $x_1$ is endogenous, whereas $x_2$ is exogenous. That is, $\EE[\varepsilon(\theta_{\circ})(x_2,z_1')']=0$, while there exists no $\theta\in\mathbb{R}^p$ such that $\EE[\varepsilon(\theta)x_1]=0$. If we compare only the base and extended models, $\mathcal{M}_b$ and $\mathcal{M}_e$, we could erroneously conclude that $x_1$ and $x_2$ are both endogenous. Instead, it is more appropriate to consider the following models: $\mathcal{M}_b$, $\mathcal{M}_e$, $\mathcal{M}_{e_1}$ and $\mathcal{M}_{e_2}$, where, for $i=1,2$, $\mathcal{M}_{e_i}$ is the model defined by the moment condition
\begin{equation}
\mathbf{E}^{Q}[g_{e_i}(w,\psi )]=0,\qquad Q\in \mathcal{Q}_{e_i}(\psi), \label{eq:moment:restrictions:extended:1}
\end{equation}%
where
\begin{equation*}
g_{e_i}(w,\psi_i) := \varepsilon(\theta )
\begin{pmatrix}
x \\
z_{1} \\
z_{2}
\end{pmatrix}
- \begin{pmatrix}
v^{(i)} \\
0 \\
0%
\end{pmatrix} = g_{b}(w,\theta) - \begin{pmatrix}
v^{(i)} \\
0 \\
0
\end{pmatrix},
\end{equation*}
%$e_i$ denotes the $d$-dimensional vector with a $1$ in the $i$-th coordinate and zeros elsewhere, $v_i$ denotes the $i$-th component of $v\in\mathbb{R}^{d_x}$,
$\psi_i :=(\theta ,v^{(i)})\in \Psi $, $\Psi :=\Theta \times \mathcal{V}$, $\mathcal{V}\subset \mathbb{R}^{2}$ with $v^{(i)} = (v_1^{(i)}, v_2^{(i)})$ and $v_j^{(i)}\in\mathbb{R}$ for $j=1,2$, and $\mathcal{Q}_{e_i}(\psi_i) := \left\{Q\in \mathbb{M};\,\mathbf{E}^{Q}[g_{e_i}(w,\psi_i)] = 0\right\} $. We enforce the restriction that one component of $x$ is treated as exogenous by setting the corresponding
element of $v^{(i)}$ to zero. Specifically, define $v^{(1)} := (0,v_{2}^{(1)})$ and
$v^{(2)} := (v_{1}^{(2)},0)$. Model $\mathcal{M}_{e_{1}}$ treats $x_{1}$ as exogenous and allows $x_2$
to be endogenous, while model $\mathcal{M}_{e_{2}}$ treats $x_{2}$ as exogenous and allows $x_1$ to be
endogenous.

This construction allows a direct application of our baseline-versus-extended comparison. Treat
$\mathcal{M}_e$ as the common reference extended model and compare $\mathcal{M}_{e_i}$ against $\mathcal{M}_e$.
If the marginal likelihood of $\mathcal{M}_{e_2}$ exceeds that of $\mathcal{M}_e$, then the data support the
restriction $v_2=0$, suggesting that $x_2$ is exogenous while $x_1$ is treated as endogenous. Similarly, if the
marginal likelihood of $\mathcal{M}_{e_1}$ exceeds that of $\mathcal{M}_e$, then the data support $v_1=0$,
suggesting that $x_1$ is exogenous while $x_2$ is treated as endogenous.

More generally, when more than two models are under consideration, one can compare them via their marginal
likelihoods. In our context, these candidate
models are obtained from the extended model by setting a subset of the elements of $v$ to zero. In total, there
are $2^{d_x}$ models, including the base and extended models. Each model corresponds to a configuration of which
elements of $x$ are treated as endogenous: $x_i$ is treated as endogenous when the associated $v_i$ is unrestricted,
whereas $v_i=0$ corresponds to $x_i$ being treated as exogenous. Marginal-likelihood comparison over this finite
model set selects the configuration most supported by the data.

Moreover, our endogeneity testing can be enriched by comparing different model specifications. For example, suppose
$x$ is scalar and consider linear versus quadratic specifications,
\[
y_{i} = \gamma_{0} + \beta_{1} x_{i} + \beta_{2} x_{i}^{2} + \gamma_{1} z_{1i} + \varepsilon_{i}.
\]
Then endogeneity of $x_i$ can be assessed under each specification, leading to four candidate models, {linear-exogenous, linear-endogenous, quadratic-exogenous, quadratic-endogenous}. A marginal
likelihood comparison can be used to select the best model among these candidates. We apply this idea in our real data example (BLP model). In that setting, we consider four candidate models that jointly vary the functional form and the endogeneity status of price, namely {linear-exogenous, linear-endogenous, nonlinear-exogenous, nonlinear-endogenous}. Marginal likelihood comparison over these four candidates simultaneously assesses endogeneity within each specification and delivers a unified ranking across specifications.

Appendix~\ref{sec:additional-mc} reports Monte Carlo experiments for these two use cases; see Sections~\ref{sec:mc-illustration-1} and \ref{sec:mc-illustration-2}.

\subsection{Assumptions}\label{sec:assumptions}
\indent We provide the assumptions that we use to prove the results in the previous sections. The first assumption ensures that the set of distributions satisfying the moment conditions is non-empty, which is necessary for the ETEL to be well-defined. It guarantees that the dual representation of the optimization problem \eqref{eq:KLbasic} holds even when $P\notin\mathcal{Q}_{b,\theta}$ for every $\theta\in\Theta$. In fact, in the latter case it is possible that $Q_b^*(\theta )$ and $P$ do not have a common support for any $\theta $, in which case, the equality in \eqref{eq:ETEL:dual} does not hold; see \cite{Sueishi2013} for a discussion on this point.

\begin{ass}[Non-emptyness.]\label{Ass_absolute_continuity} When $\mathbf{E}[\varepsilon_{i}(\theta)x_{i}]\neq 0$ for every $\theta\in\Theta$, there exists $Q\in \bigcup_{\theta\in\Theta}\mathcal{Q}_{b,\theta }$ such that $Q$ is mutually absolutely continuous with respect to $P$, where $\mathcal{Q}_{b,\theta }$ is defined in \eqref{def:Q:b}.
\end{ass}

\noindent This assumption implies that there is a $\theta $ for which $\mathcal{Q}_{b,\theta }$ is non-empty, that $dQ_b^{*}(\theta )/dP = \Bigl(\frac{e^{\lambda_{*}(\theta )'g(w,\theta)}}{\mathbf{E}[e^{\lambda_*(\theta )'g(w,\theta )}]}\Bigr)$ and that $\theta_{\ast }$ is identified by \eqref{eq:thetastar}. In addition, by assuming mutual absolute continuity, it ensures that both $\mathrm{KL}(Q_b^{*}(\theta)||P)$ and $\mathrm{KL}(P||Q_b^{*}(\theta))$ are well defined. We then assume that $\theta_*$ is unique.

\begin{ass}[Identification.]\label{Ass_0_NS}
  The maximizer $\theta_*$ defined as the minimizer of $\mathrm{KL}(P||Q^*(\theta))$ with respect to $\theta\in\Theta$ is unique and is in the interior of $\Theta$, where the interior is defined with respect to the topology in $\mathbb{R}^p$.
\end{ass}

Since under Assumption \ref{Ass_absolute_continuity}, $\theta_*$ coincides with the minimizer in \eqref{eq:thetastar}, then the previous assumption implies uniqueness also of the latter.\\
\indent As we have pointed out in Section \ref{sec:Preliminaries}, the ETEL is not defined at the $\theta$s for which the optimization problem \eqref{eq:KLbasic} does not have a feasible solution, that is, at the $\theta$s that do not belong to the set $H_n$ in \eqref{eq:posterior}. Assumption \ref{Ass_absolute_continuity} guarantees that asymptotically the optimization problem \eqref{eq:KLbasic} is feasible at $\theta_*$. Similarly, if the model is correctly specified, then the optimization problem \eqref{eq:KLbasic} is feasible at $\theta_{\circ}$ (or $\psi_{\circ}$ depending on which model we consider). However, this is not enough for our asymptotic analysis. Instead, we need to assume that \eqref{eq:KLbasic} has a solution for every $\theta$ that is sufficiently close to $\theta_*$ with probability approaching $1$. The required notion of how close is specified in the next assumption, for which we introduce the following ball. For any sequence $M_n\rightarrow \infty$ as $n\rightarrow \infty$, and any $\wtl\theta\in\Theta$, define the $n^{-1/2}$-ball around $\wtl\theta$ as $B(\wtl\theta,M_n n^{-1/2}) := \{\theta\in\Theta:\,\|\theta - \wtl\theta\| \leq M_n n^{-1/2}\}.$ The ball $B(\wtl\theta,M_n n^{-1/2})$ shrinks to $0$ slightly slower than $n^{-1/2}$ depending on how fast $M_n$ goes to $\infty$. The $n^{-1/2}$-ball $B(\wtl\psi,M_n n^{-1/2})$ around some $\wtl\psi\in\Psi$ is defined similarly. In addition, denote by $int \Delta_n := \{q := (q_1,\ldots,q_n)'; \sum_{i=1}^n q_i = 1, \; q_1>0,\ldots,q_n>0\}$ the interior of the $n-1$ simplex. Finally, for any sequence $M_n\rightarrow \infty$: $B_{*,n} := \bigcap_{\{M_n; M_n\rightarrow \infty\} }B(\theta_*, M_n n^{-1/2})$ and $B_{\circ,n} := \bigcap_{\{M_n; M_n\rightarrow \infty\} }B(\psi_{\circ},M_n n^{-1/2})$.
%The next assumption guarantees that asymptotically the set $H$ contains at least the set $B(\theta_{\circ},C n^{-1/2})$ in the extended model and $B(\theta_*,C n^{-1/2})$ in the base model.

\begin{ass}\label{ass:feasibility}
  As $n\rightarrow \infty$, for every $\theta \in B_{*,n}$:
  \begin{equation}
    P\left(\sum_{i=1}^n g_b(w_i,\theta)q_i = 0 \textrm{ for at least one } q \in int \Delta_n \right) \rightarrow 1
  \end{equation}
  \noindent and $\EE^{Q_*}[g_b(w,\theta)]=0$. Similarly, as $n\rightarrow \infty$ and for every $\psi \in B_{\circ,n}$:
  \begin{equation}
    P\left(\sum_{i=1}^n g_e(w_i,\psi)q_i = 0  \textrm{ for at least one } q \in int \Delta_n \right) \rightarrow 1
  \end{equation}
  and $\EE[g_e(w,\psi)]=0$.
\end{ass}
This assumption is much weaker than requiring that \ref{eq:KLbasic} has a solution for every $\theta\in\Theta$ with probability approaching $1$. Even if not explicitly stated in the literature, a similar assumption is necessary for the frequentist ETEL, EL and ET estimators.

%The next assumption is standard when studying the frequentist properties of Bayesian procedures. First, it requires that the prior density of $\theta$ is continuous so that $\pi(\theta_* + h/\sqrt{n})$ behaves like the constant $\pi(\theta_*)$ for $n$ large, for very bounded $h\in\mathbb{R}^p$. Second, it requires that $\theta_*$ lies in the support of the prior. This is a necessary (though, of course, not sufficient) condition for the posterior distribution to concentrate on $\theta _{\ast }$ as $n$ becomes large.

%\begin{ass}\label{Ass_prior}
%(a) $\pi $ is a continuous probability measure that admits a density with respect to the Lebesgue measure; (b) $\pi $ is positive on a neighborhood of $\theta_*$.
%\end{ass}

The next fourth assumptions concern the model. Compared to the assumptions in \cite{ChibShinSimoni2018}, our assumptions are weaker due to the linearity in $\theta$ of the moment functions $g_b(w,\theta)$ and $g_e(w,\theta)$. Consequently, assumptions regarding the continuity of the moment functions and their derivatives are automatically satisfied. Recall the notation $w_i:= (y_i, x_i', z_i')'$, $z_i := (z_{1,i}',z_{2,i}')'$ and $\tilde w_{1,i} := (x_i', z_{1,i}')'$. Moreover, we denote $\tilde w_i := (x_i', z_i')'$, $\|\cdot\|_2$ the Euclidean norm and $\|\cdot\|_{op}$ the operator norm.

\begin{ass}\label{Ass_1_NS}
(a) $w_i$, $i=1,\ldots, n$ are i.i.d. observable random variables with each one taking values in a complete probability space $(\mathcal{W},\mathfrak{B}_{\mathcal{W}},P)$, where $\mathcal{W}\subseteq \mathbb{R}^{d + 1}$, $\mathfrak{B}_{\mathcal{W}}$ is the associated $\sigma$-field and $P$ is a probability distribution satisfying model \eqref{eq:model}; (b) $\Theta\subset \mathbb{R}^p$ is compact and
connected; (c) for every $\lambda$ in a neighborhood of $\lambda_*(\theta_*)$, the matrix $\mathbf{E}[e^{\lambda'\widetilde w_{i} \varepsilon_i(\theta_*)}\varepsilon_i(\theta_*)^2\widetilde w_i \widetilde w_i']$ has the smallest (resp. largest) eigenvalue bounded away from zero (resp. infinity).
\end{ass}

\begin{ass}\label{Ass_2_NS}
(a) %$\mathbf{E}[\|\tilde{w}_i \tilde w_{1,i}'\|_F]< \infty$, where $\|\cdot\|_F$ denotes the Frobenius norm; (b)
$\mathbf{E}[\tilde w_i \tilde w_{1,i}']<\infty$ with rank $p$.
\end{ass}

Assumptions \ref{Ass_1_NS} and \ref{Ass_2_NS} are standard in the literature, see \textit{e.g.} \cite{Schennach2007}. %Assumption \ref{Ass_1_NS} \textit{(c)} guarantees that the asymptotic covariance matrix is invertible.
The following assumption instead is new and it is used to prove the approximation for the marginal likelihood. We denote by $\widetilde{w}_{i,k}$ the $k$-th component of $\widetilde{w}_i$. Moreover, for any $\delta>0$ and for some constant $C>0$, we denote by $B_{\delta}(\lambda_*) := \{\lambda\in\mathbb{R}^d; \|\lambda - \lambda_*(\theta_*)\|_2 \leq C\delta\}$ (resp. $B_{\delta}(\theta_*) := \{\theta\in\mathbb{R}^p; \|\theta - \theta_*\|_2 \leq C\delta\}$) a closed ball in $\mathbb{R}^d$ (resp. in $\mathbb{R}^p$) centered around $\lambda_* := \lambda_*(\theta_*)$ (resp. $\theta_*$) with radius $\delta$, where $\|\cdot\|_2$ denotes the Euclidean norm. For any $\delta >0$, $q\in\mathbb{N}_+$, a set $B$, and a function $\gamma:\mathbb{R}^{d+1}\times \mathbb{R}^p\rightarrow \mathbb{R}$, we introduce the following set of functions: $\mathcal{F}_{q,\delta}(B;\gamma) := \{f:B\times \mathbb{R}^{d+1}\times B_{\delta}(\theta_*) \rightarrow \mathbb{R}^q; \|g(u,v,\theta)\|_2 \leq \gamma(v,\theta),\,  \forall u\in B\}$. For a random vector $w$, we denote with $w_k$ its $k$-th component. Finally, $K_n\subset\mathbb{R}^p$ denotes a compact set centred on $0$.%Finally, we denote by $h$ a $p$-vector in a compact subset centred on $0$.\\

%NEW
\begin{ass}\label{Ass_3}
For any $\delta > 0$, there exist real-valued functions $\gamma_j(w,\theta)$, $j=0,1,2,3,4,5$, defined on $\mathbb{R}^{d+1}\times B_{\delta}(\theta_*) $ and %with values in $\mathbb{R}$
such that $\EE[\gamma_j(w,\theta)] < \infty$, for $j=0,1,2,3,4,5$ and $\forall \theta\in B_{\delta}(\theta_*)$, and such that:\\
(a) the function $(\lambda,w,\theta)\mapsto e^{\lambda'\widetilde w \varepsilon(\theta)} \widetilde w \varepsilon(\theta)\in \mathcal{F}_{d,\delta}(B_{\delta}(\lambda_*);\gamma_0)$;\\
(b) the function $(\lambda,w,\theta)\mapsto e^{\lambda' \widetilde w \varepsilon_i(\theta)} \in \mathcal{F}_{1,\delta}(B_{\delta}(\lambda_*);\gamma_1)$;\\
(c) the function $((\lambda',h')',w,\theta)\mapsto e^{\lambda'\widetilde w \varepsilon(\theta_*)} \varepsilon(\theta)^{j - 1}\wtl{w}_{k}^{\ell}(h'\wtl{w}_{1})^{\ell'} \wtl{w}_{k}^{\ell''} \in \mathcal{F}_{1,\delta}(B_{\delta}(\lambda_*)\times K_n;\gamma_2)$ for either $(j,\ell,\ell',\ell'')=(1,1,1,0)$, or $(j,\ell,\ell',\ell'')=(1,2,2,0)$, or $(j,\ell,\ell',\ell'')=(2,1,1,1)$;\\
%(d) the following operator norm
%$$
%\mathbf{E}\left[\sup_{\lambda\in B_{\delta}(\lambda_*(\theta_*))}\left\|e^{\ell \lambda'\widetilde w_i \varepsilon_i(\theta_*)}\varepsilon_{i}(\theta_{*})^{j}\widetilde w_{i,k}^{\ell} \wtl w_{i}\tilde w_{i}'\right\|_{op}\right]
%$$
%is bounded away from infinity for either $j=\ell=1$, or $j=3$, $\ell = j-2$, or $j=4$, $\ell = j-2$;\\
(d) the function $(\lambda,w,\theta) \mapsto e^{2\lambda'\widetilde w \varepsilon(\theta_*)}\varepsilon(\theta_*)^2 \left\|\widetilde w \right\|_2^2 \in \mathcal{F}_{1,\delta}(B_{\delta}(\lambda_*);\gamma_3)$;\\
(e) the function $(\lambda,w,\theta)\mapsto e^{\ell \lambda'g(w,\theta)}\wtl w_{j} \wtl w_{k}\varepsilon(\theta)^{\ell'} \widetilde w_{k'}^{i} \in \mathcal{F}_{1,\delta}(B_{\delta}(\lambda_*);\gamma_4)$ for every $j,k,k'=1,\ldots,d$, and for either $(\ell,\ell',i)=(1,1,1)$, or $(\ell,\ell',i)=(1,2,0)$, or $(\ell,\ell',i)=(1,3,1)$, or $(\ell,\ell',i)=(2,4,2)$, and where $\gamma_4$ may depend on all the previous indices;\\
(f) the function $((\lambda',h')',w,\theta)\mapsto e^{\lambda'\widetilde w \varepsilon(\theta_*)} \varepsilon(\theta)\wtl{w}\,\lambda'\wtl{w}\,h'\wtl{w}_{1}\in \mathcal{F}_{d,\delta}(B_{\delta}(\lambda_*)\times K_n;\gamma_5)$;\\
(g) the function $((\lambda',h')',w,\theta)\mapsto e^{\lambda'\widetilde w \varepsilon(\theta_*)} \varepsilon(\theta)^j\wtl{w}\,\lambda'\wtl{w}\,(h'\wtl{w}_{1})^{\ell} \wtl{w}_k\in \mathcal{F}_{d,\delta}(B_{\delta}(\lambda_*)\times K_n;\gamma_6)$ for every $k=1,\ldots,d$, and for either $(j,\ell)=(2,1)$, or $(j,\ell)=(1,2)$.
\end{ass}

This regularity assumption is a weak moment condition, requiring the existence of an upper bound -- dependent on the data and on $\theta$ -- with finite expectation for every $\theta\in B_{\delta}(\theta_*)$. This condition ensures that a uniform Law of Large Number holds (see \textit{e.g.} \cite[Lemma 2.4]{NeweyMcFadden1994}). In particular, we use this assumption in our proofs to establish the uniform convergence in probability of several terms, including: $\EE_n\left[e^{\wh\lambda'g(w_i,\theta_*)}\varepsilon_i(\theta_*)\wtl{w}_{i}\right]$, $\EE_n\left[e^{\wh\lambda'g(w_i,\theta_*)}\wtl{w}_{1,i}\wtl{w}_{i}\right]$, $h'\EE_n\left[e^{\wh\lambda'g(w_i,\theta_*)} h'\wtl{w}_{1,i}\wtl{w}_{i}'\wh\lambda\varepsilon_i(\theta_*)\wtl{w}_i\right]$, and $\EE_n\left[e^{\wh\lambda'g(w_i,\theta_*)}\varepsilon_i(\theta_*)\wtl{w}_{1,i}\wtl{w}_{i}\right]$. The uniform convergence must hold uniformly over $\wh\lambda$ in a closed ball $B_{\delta}(\lambda_*(\theta_*))$. Additionally, this assumption also guarantees that the dominance condition required for the application of the Dominated Convergence Theorem is satisfied, which is another result used in the proofs.

Part \textit{(e)} of the previous assumption implies that
$$
\mathbf{E}\left[\sup_{\lambda\in B_{\delta}(\lambda_*(\theta_*))}\left\|e^{\ell \lambda'\widetilde w \varepsilon(\theta_*)}\varepsilon(\theta_{*})^{\ell '}\widetilde w_{k}^{i} \wtl w\tilde w'\right\|_{op}\right]
$$
is bounded away from infinity for every $k'=\{1\,\ldots,d\}$ which is what we need in the proof because this operator norm is upper bounded by $d\,\max_{\{j,k\in\1,\ldots,d\}}|e^{\ell \lambda'g(w,\theta_*)}\wtl w_{j} \wtl w_{k}\varepsilon(\theta_*)^{\ell} \widetilde w_{k'}^{i} |$ for every $k'\in\{1,\ldots, d\}$.\\
\indent An assumption similar to Assumption \ref{Ass_3} is necessary to establish uniform convergences of $\EE_n\left[\varepsilon_i(\theta )\wtl{w}_{i}\right]$ and $\EE_n\left[\varepsilon_i(\theta)\wtl{w}_{i} \wtl{w_i}\right]$ uniformly over $\theta\in B_{\delta}(\theta_*)$. For this, we introduce the class
\begin{ass}\label{Ass_3_extended}
  For any $\delta > 0$, there exist real-valued functions $\gamma_{\circ,j}(w)$, $j=1,2$, defined on $\mathbb{R}^{d+1}$ and such that $\EE[\gamma_j(w)] < \infty$, for $j=1,2$ and :\\
  (a) $\|\widetilde w \varepsilon(\theta)\| \leq \gamma_{\circ,1}$ for every $\theta\in B_{\delta}(\theta_*)$;\\
  (b) $|\varepsilon(\theta)\widetilde w_j \widetilde w_k | \leq \gamma_{\circ,2}$ for every $\theta\in B_{\delta}(\theta_*)$;\\
  (c) $\|h'\wtl{w}_{1} \wtl{w}'\| \leq \gamma_{\circ,3}$ for every $h\in K_n$.
\end{ass}
%%%%%%%%%%%%%%%%%%%%%%%%
For the next assumption we denote by $\Theta_{n}:=\{\Vert \theta -\theta _{\ast }\Vert \leq M_{n}/\sqrt{n}\}$
%\begin{equation*}
%  \Theta_{n}:=\{\Vert \theta -\theta _{\ast }\Vert \leq M_{n}/\sqrt{n}\},
%\end{equation*}
a ball around $\theta_{*}$ with the radius at most $M_{n}/\sqrt{n}$, where $M_{n}$ is any sequence of positive constants diverging to $+\infty $. We denote by $\ell_{n,\theta}(w_i)$ the log-likelihood function for one observation $w_i$: $\ell_{n,\theta}(w_i) := \log \wh{q}_i(\theta|\mathcal{M}_b)$, and by $\ell_{n,\theta}(w_{1:n}) := \sum_{i=1}^n \ell_{n,\theta}(w_i) = \log \wh{q}(w_{1:n}|\theta,\mathcal{M}_b)$ the log-ETEL function. We recall that both $\ell_{n,\theta}(w_i)$ and $\ell_{n,\theta}(w_{1:n})$ are defined only for $\theta\in H_n$ and that under Assumption \ref{ass:feasibility} they are at least defined on $B_{*,n}$. The next assumption controls the behaviour of the ETEL function $\theta\mapsto \ell_{n,\theta }(w_{i})$ at a distance from $\theta_*$ and it ensures that $\theta_*$ is well-separated from the $\theta$s that are at a certain distance from it.
\begin{ass}[Base model.]\label{Ass_identification_base_model}
  Assume that there exists a constant $C>0$ such that
    \begin{equation}
      P\left(\sup_{\theta \in H_n\cap\Theta_{n}^{c}}\frac{1}{n}\sum_{i=1}^{n}\left(\ell_{n,\theta }(w_{i})-\ell _{n,\theta _{\ast }}(w_{i})\right) \leq  - \frac{CM_{n}^{2}}{n}\right) \rightarrow 1\;,\;\text{as}\;n\rightarrow \infty, \label{Ass_identification_rate_contraction_appendix}
    \end{equation}
  \noindent where $M_n$ is the same sequence used to define $\Theta_n$.
\end{ass}
A condition similar to Assumption \ref{Ass_identification_base_model} is in \cite[Lemma 4.2]{kleijn2012} and it is also related to the classical condition in \textit{e.g.} \cite[Assumption 6.B.3]{LehmanCasella1998} and \cite[Assumption 3]{ChernozhukovHong2003}. However, in our case the supremum in the assumption is taken over a smaller set, which is $H_n\cap\Theta_n^c$, instead of over $\Theta_n^c$ as in the mentioned literature. To better understand the meaning of this assumption, note that asymptotically the log-ETEL function is maximized at the pseudo-true value $\theta_*$. Hence, Assumption \eqref{Ass_identification_rate_contraction_appendix} requires that if the parameter $\theta $ is far from the pseudo-true value $\theta_{*}$,
that is $\|\theta -\theta_* \| > M_{n}/\sqrt{n}$, then the sum $\sum_{i=1}^{n} \ell_{n,\theta }(w_{i})$ evaluated at such a $\theta $ must be small relative to the sum $\sum_{i=1}^{n}\ell _{n,\theta_*}(w_{i})$, which is the maximum value. Controlling this behavior is important because the posterior involves integration over the whole support of $\theta$. Subsets of $\Theta$ that can be distinguished from $\theta_*$ uniformly (with probability approaching $1$ as $n\rightarrow \infty$) based on the ETEL function will receive a posterior probability that is asymptotically negligible. An alternative to this condition would be to require the existence of asymptotically consistent tests $\phi_n$ that are able to distinguish from the true distribution $P$ in a uniform way, that is, for every $\epsilon>0$, there exists a sequence of tests $\{\phi_n\}$ such that as $n\rightarrow$ 0,
\begin{equation}
  \EE[\phi_n] \rightarrow 0, \qquad \textrm{and} \qquad \sup_{\{\theta;\|\theta - \theta_*\| \geq \epsilon\}}\EE\left[e^{\ell_{n,\theta }(w_{i})-\ell _{n,\theta_*}(w_{i})}(1 - \phi_n)\right]\rightarrow 0.
\end{equation}
Similarly, for the extended model we denote by $\ell_{n,\psi}(w_i)$ the log-likelihood function for one observation $w_i$: $\ell_{n,\psi}(w_i) := \log \wh{q}_i(\psi|\mathcal{M}_e)$ and by $\ell_{n,\psi}(w_{1:n}) := \sum_{i=1}^n \ell_{n,\psi}(w_i) = \log \wh{q}(w_{1:n}|\psi,\mathcal{M}_e)$ the log-ETEL function. The next assumption has the same interpretation of Assumption \ref{Ass_identification_base_model} but for the extended model.
\begin{ass}[Extended model.]\label{Ass_identification_extended_model}
  Assume that there exists a constant $C>0$ such that as $n\rightarrow \infty $,
  \begin{equation}\label{Ass_identification_rate_contraction_extended}
    P\left( \sup_{\scriptstyle \|\psi -\psi_{\circ}\| > M_n/\sqrt{n};\atop\scriptstyle \psi\in H_n \times \mathcal{V}}\frac{1}{n}\sum_{i=1}^{n}\left(\ell_{n,\psi}(w_{i}) - \ell_{n,\psi_{\circ}}(w_{i})\right) \leq - \frac{C M_n^2}{n}\right) \rightarrow 1,
  \end{equation}
\noindent where $M_n$ is any sequence of positive constants diverging to infinity.
\end{ass}

\subsection{Experiments}

Consider the same generating process as in Section \ref{sec:numerical_illustration}, and suppose that $(\varepsilon _{i},u_{i},v_{i})$ have a joint distribution induced by a
Gaussian copula with covariance matrix $R=%
\begin{psmallmatrix}
    1 & \rho & 0 \\
    \rho & 1  & 0 \\
    0 & 0 & 1
\end{psmallmatrix}$. The parameter $\rho $ controls the degree of
endogeneity. We let $\rho $ take values in the set from -.5 to .5, in
increments of 0.1. For each value of $\rho $ in this set, we generate 100
samples of size $n$. For each sample, we compute the base and extended
models, and calculate the log-marginal likelihoods. We then count the number
of times the log-marginal likelihood of $\mathcal{M}_{e}$ exceeds that of $%
\mathcal{M}_{b}$. The results are given Table \ref{tab:sim_MegreaterthanMb}.
We can see from this table that even for small values of $\rho $, our test
of endogeneity correctly concludes that the correct model is $\mathcal{M}%
_{e} $.
\begin{table}[h]
\centering
\begin{tabular}{lccccccccccc}
\toprule $\rho$ & -0.5 & -0.4 & -0.3 & -0.2 & -0.1 & 0.0 & 0.1 & 0.2 & 0.3 &
0.4 & 0.5 \\
\midrule $n=250$ & 99 & 96 & 82 & 48 & 12 & 2 & 18 & 54 & 93 & 100 & 100 \\
$n=500$ & 100 & 100 & 98 & 76 & 17 & 1 & 29 & 87 & 99 & 100 & 100 \\
$n=1000$ & 100 & 100 & 100 & 96 & 46 & 1 & 46 & 100 & 100 & 100 & 100 \\
$n=2000$ & 100 & 100 & 100 & 100 & 80 & 1 & 70 & 100 & 100 & 100 & 100 \\
\bottomrule &  &  &  &  &  &  &  &  &  &  &  \\
&  &  &  &  &  &  &  &  &  &  &
\end{tabular}%
\caption{Model selection frequencies from 100 replications of data simulated
from the design in Example 1. For each combination of $n$ and $\mathrm{Cov}(%
\protect\varepsilon ,u)=\protect\rho $, the entries give the number of times
in 100 replications of the data that the log-marginal likelihood of $%
\mathcal{M}_{e}$ exceeds the log-marginal likelihood of $\mathcal{M}_{b}$.}
\label{tab:sim_MegreaterthanMb}
\end{table}

% ***I removed the sentence below as we don't have the actual simulation % results for it.
%In contrast, an incorrect test of endogeneity that compares the base
%without the $z_{2}$ restriction with the extended model without the $x$
%condition produces completely erroneous results, with the latter model
%selected 100\% of the times, even when $\rho =0$.

\section{Real data examples} \label{sec:examples}

\subsection{Causal effect of price on automobile demand}

We consider the classic problem of automobile demand studied in \citet{berry1995econometrica}. This problem has recently been revisited by %
\citet{chernozhukov2015post}, henceforth BLP and CHS, respectively. Apart
from its intrinsic value, this problem is worth analyzing because it
involves a realistically large number of controls and instruments.

To set up the problem, let $y_{ijt}$ denote the log of the ratio of the
market share of product $i$ in market $j$ at time $t$, relative to an
external option, and let $x_{ijt}$ denote the potentially endogenous
automobile price variable. In the sample data, this variable is demeaned.
For controls, let $z_{ijt}$ denote the observed characteristics of the
product. In BLP these are taken to be a constant, an air conditioning dummy (%
$air$), horsepower divided by weight ($hpwt$), miles per dollar ($mpd$), and
vehicle size ($space$). In our notation, $y_{ijt} = x_{ijt} \, \beta + {z}%
_{1ijt}^{\prime }{\gamma }+ \varepsilon _{i}$, where ${z}_{1ijt}=(1,
mpd_{ijt}, space_{ijt}, hpwt_{ijt}, air_{ijt})$. BLP used ten instruments,
five formed by summing the value of these five characteristics over other
automobiles produced by the same firm and five formed by summing the above
characteristics over automobiles produced by other firms. These form $%
z_{2ijt}$. In revisiting this analysis, CHS augment the original controls
with quadratics and cubics in $trend$, $mpd$, $space$, $hpwt$, and all
first order interactions, and then used sums of these characteristics as
potential instruments.

In our analysis, we consider both formulations, but in the augmented variant
we introduce nonlinear controls by transforming each of $trend$, $hpwt$, $%
mpd $ and $space$ by natural cubic spline basis functions, each centered at
five equally spaced quantile knots (the cubic spline basis functions are
taken from \cite{Chib2010}). We opt for this approach to avoid widely
different covariate values from parametric quadratic and cubic terms of
these covariates. After the imposition of an identification restriction on
the basis expansions, which reduces the number of nonlinear terms to four
for each continuous covariate, the right hand side of the augmented outcome model is
defined by $x$ (price) and $z_1$ (consisting of an intercept, sixteen
nonlinear covariates denoted by $trend{_Bj}$, $mpd_{Bj}$, $space_{Bj}$ and $%
hpwt_{Bj}$ for $j = 1,\ldots,4$, and the air-conditioning dummy). The set
of augmented instruments that form $z_2$ in this augmented model is then
constructed as in BLP.

We fit four models to these data: the base and extended models under the
controls and instruments in BLP, and the base and extended models under the
augmented set of controls and instruments. In the BLP version, the base and
extended models contain six and seven parameters, respectively, and ten
instruments, while in the augmented variant, the base and extended models
have nineteen and twenty parameters, respectively, and $53$ moment restrictions. We assume
that the $n = 2217$ observations on $(y_{ijt},x_{ijt},z_{1ijt})$ are a
random sample from the population of automobile products across markets and
time. Because it is difficult to formulate priors on the parameters by a
priori considerations, we randomly select $15\%$ of the sample to make
training sample priors. In particular, we used the GMM estimate and its
standard error fitted on the training data (model by model) as the prior
mean and twice the GMM standard error as the prior standard deviation (sd). The
ETEL is constructed from the remaining data and the posterior distribution
of each model is sampled by the single block M-H algorithm of %
\citet{ChibGreenberg1995}. This algorithm is fast and efficient despite the
relatively large numbers of parameters and instruments. The results show
that the posterior mean of the coefficient on $price$ is -0.14, and the 95\%
posterior credibility interval is (-.16,-.13). The posterior mean is larger
in magnitude than the OLS estimate originally reported by BLP. Note that the
posterior distribution of the covariance parameter, $v$, is concentrated to
the right of zero, indicating that the $price$ is likely endogenous.

For confirmation, we turn to our formal test of endogeneity. The results are
reported in Table \ref{table:blp2}. We can see that the marginal likelihood
is larger for the extended models in both the original BLP and the augmented
BLP specifications, supporting the conclusion that price is endogenous.
\begin{table}[t]
\centering
\begin{tabular}{lll}
\hline
& Original BLP (Linear) & Augmented BLP (Nonlinear) \\ \hline
Base model (price is exogenous) & -14386.81 & -14431.86 \\
Extended model (price is endogenous) & -14364.59 & -14397.67 \\ \hline
\end{tabular}%
\caption{Results from the proposed Bayesian test of endogeneity. The log-marginal likelihoods for the base and extended models under the original BLP
model and its augmented variant. Results based on a training sample prior,
(using randomly selected 15\% of the data) and 10,000 MCMC iterations
(beyond a burn-in of 1,000) of a tailored single block M-H algorithm.
Logarithm of marginal likelihoods are computed by the method of \protect\cite%
{Chib1995} and \protect\cite{ChibJeliazkov2001}.}
\label{table:blp2}
\end{table}

We conclude this analysis by plotting the posterior distributions of the
price coefficient from each model. The estimated effect of price on
automobile demand is larger (in absolute value) when endogeneity of price is
taken into account. Interestingly, the price effect is smaller and more
concentrated in the augmented models, suggesting that some of the excess
sensitivity to price observed in the original BLP model is due to the
omission of the nonlinear controls. In addition, it is worth noting that if
we were to only fit the base model (which the marginal likelihood confirms
is misspecified in this case), we would miss the fact that incorporating
nonlinearities impacts the posterior distribution.
\begin{figure}[h!]
\begin{center}
\includegraphics[width=5in]{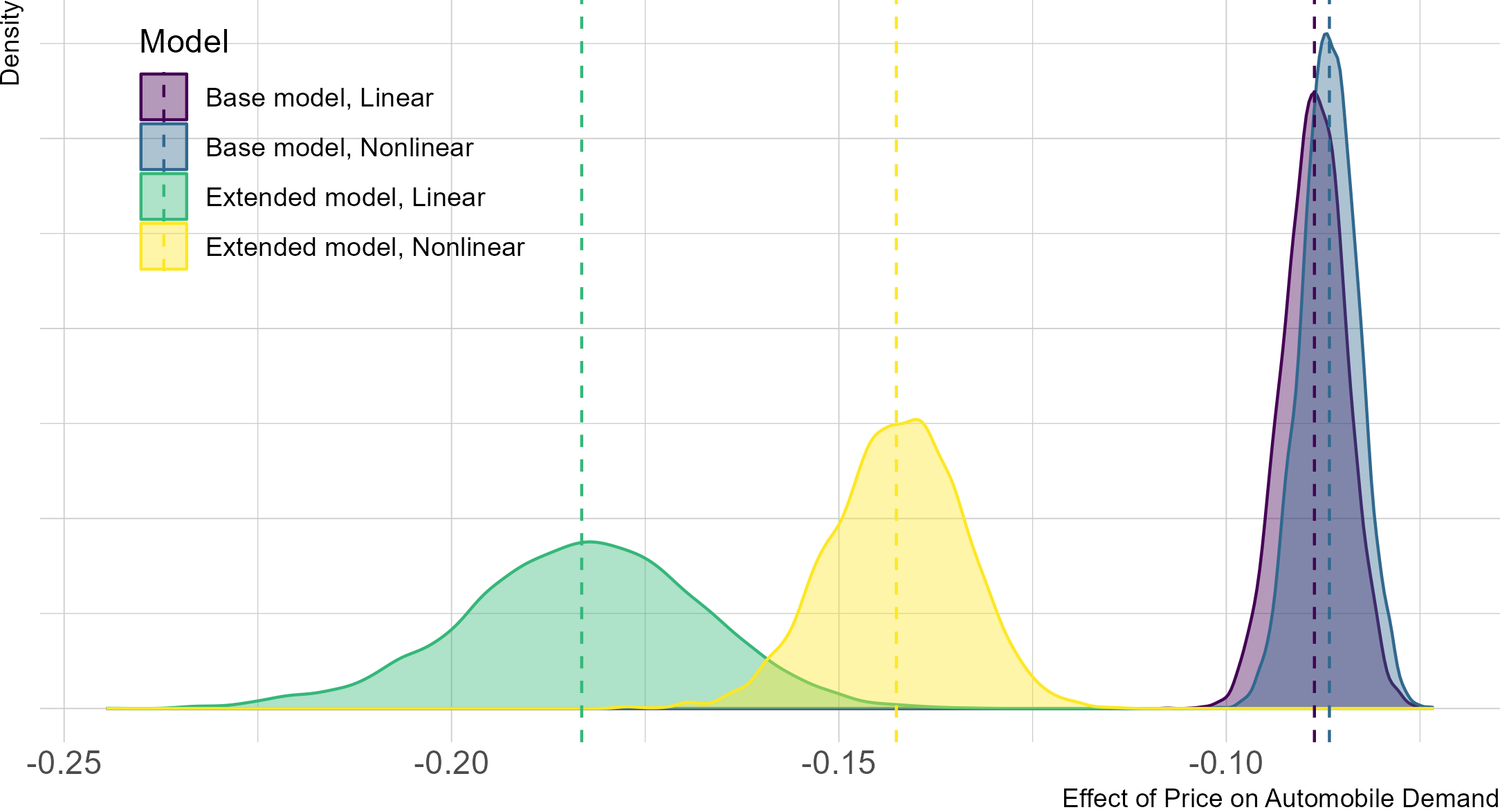}
\end{center}
\caption{BLP models: Marginal posterior distributions of the coefficient on
the price variable, $\protect\beta$. Posterior mean and standard deviation
of $\protect\beta$ are -0.089 and 0.004, respectively, for the base model with the original
BLP (linear) specification, while they are -0.087 and 0.004, respectively, with the
augmented BLP (nonlinear) specification. For the extended model, posterior
mean and standard deviation of $\protect\beta$ are -0.183 and 0.015, respectively, for the
linear specification and -0.143 and 0.009, respectively,  for the nonlinear specification.}
\label{figure:blp}
\end{figure}

% \begin{table}[ht]
% \centering
% \begin{tabular}{llll}
%   \hline
% Model & postmean & postsd \\
%   \hline
% Base, L & -0.089 & 0.004 \\
% Base, NL & -0.087 & 0.004 \\
% Extended, L & -0.183 & 0.015 \\
% Extended, NL & -0.143 & 0.009 \\
%    \hline
% \end{tabular}
% \caption{}
% \end{table}

\subsection{Effect of airfares on passenger traffic}

The emphasis of the theory and applications in this paper is on situations with a single outcome variable; however, our framework can be applied more broadly. An important example is clustered longitudinal data. Let $y_{i} = (y_{i1},\ldots,y_{iT})$ denote $T$ potentially correlated and heteroskedastic measurements on subject $i$. The outcome is thus a $T\times 1 $ vector, rather than a scalar. Adjusting the dimensions of the controls and instruments, respectively, suppose that independently across $i$ the clustered outcomes follow the linear model $y_{i}= X_{i} \beta +Z_{1,i} \gamma + \varepsilon_{i} $, where $X_{i}$ is $T \times d_{x}$, $Z_{1,i}$ is $T \times d_{z_{1}}$, $Z_{2,i}$ is $T \times d_{z_{2}}$, and $\varepsilon_{i}$ is $T \times 1$. Now assume that $Z_{1,i}$ and $Z_{2,i}$ satisfy the clustered data exogeneity restrictions $\EE[Z_{j,i}'\varepsilon_{i}(\theta)] = 0$, $j = 1,2$, but that the clustered data exogeneity restrictions $\EE[X_{i}^{\prime }\varepsilon_{i}(\theta)] = 0$ related to $X_i$ are in doubt. We can apply our framework to this problem by defining a base model in which the latter restrictions are imposed and an extended model that contains the inactive restrictions $\EE[X_{i}'\varepsilon_{i}(\theta)]=v$, where $v$ is now a $d_x \times 1$ vector of unknown parameters. In parallel to the approach developed above, the marginal likelihood comparison of these models is a test for the exogeneity of $X$.

As an illustration of this extended set-up, we consider a $T=4$ balanced
longitudinal data set on airfares and passenger traffic for the years 1997,
1998, 1999, and 2000 from \cite{wooldridge2010econometric}. For each year $t$%
, $t\leq 4$, the data is clustered by route $i$, $i\leq n=1149$. For each
flight route defined by the origin and destination cities, one has the log
of the average number of passengers per day ($lpassen$), the log of the
average one-way fare in dollars ($lfare$), the log of the distance in miles (%
$ldist$), and the fraction of the market corralled by the biggest carrier ($%
concen$). The model of interest is $lpassen_{it}=\beta\, lfare_{it} +\gamma
_{1}trend_{t}+\gamma _{2}ldist_{it}+\varepsilon _{it}$, where $trend$ is a
trend variable taking values $1,2,3,$ and $4$, and each of the variables in this
regression is mean centered. The goal is to estimate the price elasticity
parameter $\beta$, but one is concerned that $lfare$ is possibly endogenous.
In the estimation we assume that $concen$ is a valid instrument (it does not
directly appear in the outcome model and it affects $lfare$, both reasonable
assumptions).

Clustered by route $i$, we have
\begin{equation*}
\begin{pmatrix}
lpassen_{i1} \\
lpassen_{i2} \\
lpassen_{i3} \\
lpassen_{i4}%
\end{pmatrix}%
=\left(
\begin{array}{ccc}
lfare_{i1} & 1 & ldist_{i1} \\
lfare_{i2} & 2 & ldist_{i2} \\
lfare_{i3} & 3 & ldist_{i3} \\
lfare_{i4} & 4 & ldist_{i4}%
\end{array}%
\right) \left(
\begin{array}{c}
\beta \\
\gamma _{1} \\
\gamma _{2}%
\end{array}%
\right) +\quad \left(
\begin{array}{c}
\varepsilon _{i1} \\
\varepsilon _{i2} \\
\varepsilon _{i3} \\
\varepsilon _{i4}%
\end{array}%
\right) ,
\end{equation*}%
or compactly as $y_{i}=\widetilde{W}_{1,i}\theta +\varepsilon _{i}$, $%
i=1,2,\ldots,1149$, where $\theta:7\times 1$ is the unknown parameter of
interest. In this model, the distribution of $\varepsilon _{i}$ is not
specified. Moreover, the elements of $\varepsilon _{i}$ can be serially
correlated and heteroskedastic in an arbitrary, unknown way.

Now let $Z_{i} := \left(\widetilde{W}_{1,i},1,concen_{i}\right)$, $i \leq n$%
, be a $4\times 5$ matrix, where $1$ is a vector of ones, and $%
concen_{i}=\left(concen_{i1},\ldots,concen_{i4}\right) ^{\prime }:4\times 1$
is the vector of $concen$ values for route $i$. In the base model, $lfare$
is exogenous. The model is defined by the five moments
\begin{equation*}
\mathcal{M}_{b}:\;\;\;\mathbf{E}[Z_{i}^{\prime }(y_{i}-X_{i}\theta
)]=0_{5\times 1}.
\end{equation*}%
In the extended model, the $lfare$ moment condition is inactive.
Specifically,
\begin{equation*}
\mathcal{M}_{e}:\;\;\;\mathbf{E}[Z_{i}^{\prime }(y_{i}-X_{i}\theta )]
=\left(
\begin{array}{c}
v \\
0_{4\times 1}%
\end{array}%
\right).
\end{equation*}%
The ETEL-based estimation of these two models makes no assumption about the
joint distribution of the cluster-level errors.

We specify the prior from a training sample. We randomly split the sample
into a training sample (of say 115 clusters, equal to 10\% of the total
clusters) and an estimation sample (consisting of the remaining 1,034
clusters). We then estimate the base mode on the training sample with a
student-t prior centered on the system wide 2SLS estimate from the training
data, sd of 10 and 2.5 degrees of freedom. The posterior mean and sd is
calculated from these training data under this prior. We then take the
posterior mean and twice the sd from the training sample fit as the mean and
sd of the prior. This determination of the prior from the training sample is
helpful in the fitting, but, due to the thick tails of the prior, the
information brought in by the prior pales in comparison with the information
from the estimation sample.

We sample the posterior in each model by the one-block tailored MCMC
algorithm. In the base model, from 10,000 MCMC draws beyond a burn-in of
1,000, we find that the posterior mean of $\beta$ is -0.551 and its 95\%
posterior credibility interval is (-0.683, -0.419). Moreover, computation
shows that $\log (m(w_{1:n}|\mathcal{M}_{b})=-7190.222$ and $\log (m(w_{1:n}|%
\mathcal{M}_{e})=-7191.06$, signaling that $lprice$ in this problem can be
viewed as exogenous. %In both cases, these log
%marginal likelihood values strongly signal that $lpriice$ is endogenous,
%consistent with expectations based on a priori economic reasoning. Figure %
%\ref{fig:priceelasticity} shows the rather considerable difference in the
%marginal posterior density of $\beta $ from the two models (under the
%default prior), and the importance of selecting the model (and the posterior
%density) that is better supported by the data to avoid misleading
%conclusions.
%\begin{figure}[h]
%\centering
%\includegraphics[width = .95\textwidth]{priceelasticity.pdf}
%\caption{Panel model: price elasticity parameter $\protect\beta $. Marginal
%posterior densities from the extended (higher marginal likelihood) and base
%(lower marginal likelihood) models. Posterior mean in each case are
%indicated by dashed vertical lines.}
%\label{fig:priceelasticity}
%\end{figure}

\section{Concluding remarks} \label{sec:conclusion}

This paper develops a Bayesian test for exogeneity/endogeneity of the
treatment vector of interest in a linear mean regression model. This
endogeneity problem is generally assumed away in the Bayesian literature,
but this leads to a serious misspecification problem since endogeneity, in
practice, is the rule rather than the exception. In order to avoid the risk
of distributional misspecification, the framework we have developed relies
only on moment restrictions. The analysis in the paper revolves around the
study of two models: the base model, where the exogeneity assumption is
enforced, and an extended model, where the exogeneity moment is included but
is made inactive.\newline
\indent The testing procedure for exogeneity/endogeneity is based on Bayes
factor where the marginal ETEL of the base and the extended models are
compared. The procedure is validated from a frequentist point of view
because we establish the large sample consistency of the Bayes factor test.
In addition, we provide a comprehensive study of the log-marginal ETEL
function and determine which parts of it plays a role in the testing
procedure depending on whether the covariates $x$ are exogenous or
endogenous.\newline
\indent The real data examples discussed in the paper showcase the practical
relevance of the methods.

%The paper makes two key contributions. The first is in the study of the
%large sample behavior of the posterior distributions in the base and
%extended models in cases where the exogeneity assumption in the population
%is true or false. The second is in the development of a Bayesian test of
%endogeneity that is based on the marginal likelihoods of the base and
%extended models. In the former case, BvM theorems are established and, in
%the latter, the large sample consistency of the Bayes factor test is
%established.

It is important to mention that the approach proposed here can be extended
to situations where the controls are assumed to enter the model
nonparametrically. While the finite sample analysis of such models, after
approximating the unknown functions by (say) spline basis expansion methods,
would proceed in much the same way as discussed in this paper, the
specification of the prior and the large sample analysis would require new
developments to account for a growing number of basis function parameters
with sample size. We intend to describe the theory in a future paper.

Another interesting extension would involve relaxing our current strong identification assumption, which is based on a fixed number of relevant instruments, to allow for weak and many-instrumental variables. This would require substantial changes to our asymptotic results, as it would necessitate developing local-asymptotic Bayes factors for ETEL and incorporating instrument-growth penalties or shrinkage to prevent overfitting.

%\begin{flushright}
%$\square $
%\end{flushright}

\appendix

\section{Proofs}

\subsection{Proof of Theorem \ref{thm:equivalence}}
We first show that \textit{(i)} is equivalent to \textit{(ii)}. Let $H_{\infty}:=\{\theta\in\Theta; \mathcal{Q}_{b}(\theta) \neq \emptyset\}$. Suppose \textit{(i)} is true. Then, $P\notin \mathcal{Q}_{b}(\theta)$ for every $\theta\in\Theta$ and the $I$-projection of $P$ on the set $\mathcal{Q}_{b}(\theta)$ is different from $P$, that is: for every $\theta\in H_{\infty}$, $Q_b^*(\theta)\neq P$ and $KL(Q_b^*(\theta)||P)>0$, and it is not defined for every $\theta\in\Theta\setminus H_{\infty}$ because for these values of $\theta$, the set $\mathcal{Q}_b^*(\theta)$ is empty. It follows that also the reverse Kullback-Leibler divergence (where we have inverted the role played by the two probabilities) is strictly positive under Assumption \ref{Ass_absolute_continuity}: $\mathrm{KL}(P||Q_b^*(\theta)) >0$ for every $\theta\in H_{\infty}$. Since this is true for every $\theta\in H_{\infty}$, it is also true for $\theta_*$. Hence \textit{(ii)} holds.\\
\indent Now, suppose that \textit{(ii)} is true. Because $\mathrm{KL}(P||Q_b^*(\theta_*))>0$, then $P\neq Q_b^*(\theta_*)$ and so $P\notin \mathcal{Q}_{b}(\theta_*)$. Since $\theta_*$ minimizes $\mathrm{KL}(P||Q_b^*(\theta))$ over all the $\theta$s for which $\mathcal{Q}_{b}(\theta) \neq \emptyset$, then we also have that $P\notin \mathcal{Q}_{b}(\theta)$ for every $\theta\in\Theta$ such that $\mathcal{Q}_{b}(\theta)\neq \emptyset$. Since for the $\theta$s for which $\mathcal{Q}_{b}(\theta)= \emptyset$ the condition \textit{(i)} trivially holds then we have proved that \textit{(i)} holds for every $\theta\in\Theta$.\\
\indent Next, we show that \textit{(iii)} is equivalent to \textit{(iv)}. Suppose \textit{(iii)} holds. Then, there is a $\theta\in\Theta$, say $\theta_*$, for which $P\in \mathcal{Q}_{b}(\theta_*)$. Hence, $Q_b^*(\theta_*) = P$ and $\mathrm{KL}(P||Q_b^*(\theta)) = 0$. Hence \textit{(iv)} holds.\\
\indent Now, suppose that \textit{(iv)} holds. By the properties of the Kullback-Leibler divergence, $\mathrm{KL}(P||Q_b^*(\theta_*))$ $=$ $0$ if and only if $P=Q_b^*(\theta_*)$. It follows that $P\in\mathcal{Q}_{b}(\theta_*)$ because $Q_b^*(\theta_*)\in\mathcal{Q}_{b}(\theta_*)$ and therefore $P$ satisfy the moment restriction $\mathbf{E}^P[\varepsilon_{i}(\theta_*)x_{i}] = 0$. Hence \textit{(iii)} holds.
\begin{flushright}
$\square$
\end{flushright}

% ==============================================================================================
\subsection{Proof of Theorem \ref{lem_likelihood_ET_approx}}
%Here, we consider the case where there exists a $\theta$ such that $\EE[\varepsilon_i(\theta)\wtl w_i] = 0$.
Let us consider the expression for the likelihood evaluated at $\theta_*$:
\begin{multline}\label{eq_likelihood_ET}
  \log \wh{q}(w_{1:n}|\theta_*,\mathcal{M}_{b}) = - n\log n +\sum_{i=1}^n\wh{\lambda}(\theta_*)'g_b(w_i,\theta_*) - n\log \frac{1}{n}\sum_{j=1}^n e^{\wh{\lambda}(\theta_*)'g_b(w_j,\theta_*)}\\
    = - n\log n +\sum_{i=1}^n\wh{\lambda}(\theta_*)'\wtl{g}_b(w_i,\theta_*) - n\log \frac{1}{n}\sum_{j=1}^n e^{\wh{\lambda}(\theta_*)'\wtl{g}_b(w_j,\theta_*)},\\
    = - n\log n + \sum_{i=1}^n\wh{\lambda}(\theta_*)'\wtl{g}_b(w_i,\theta_*) + n \wh{\lambda}(\theta_*)'\EE[g_b(w_i,\theta_*)] - n\log \frac{1}{n}\sum_{j=1}^n e^{\wh{\lambda}(\theta_*)'g_b(w_j,\theta_*)},
\end{multline}
\noindent where $\wtl{g}_b(w_i,\theta_*) := g_b(w_i,\theta_*) - \EE[g_b(w_i,\theta_*)]$. We first deal with the second term on the right hand side of \eqref{eq_likelihood_ET}. By using the result of Lemma \ref{lem:1}: %with $\lambda_*(\theta_*) = \lambda_*(\theta_*) = 0$:
\begin{multline}
  \sum_{i=1}^n\wh{\lambda}(\theta_*)'\wtl g_b(w_i,\theta_*) = \sqrt{n}(\wh{\lambda}(\theta_*) - \lambda_*(\theta_*))'\frac{1}{\sqrt{n}}\sum_{i=1}^n \wtl g_b(w_i,\theta_*) + \lambda_*(\theta_*)'\sum_{i=1}^n \wtl g_b(w_i,\theta_*)\\
  = - \mathbb{G}_n\left[\tau_i^{\dagger}(\lambda_*,\theta_*) g_b(w_i,\theta_*)'\right]\Omega_*^{\dagger}(\theta_*)^{-1}\mathbb{G}_n\left[g_b(w_i,\theta_*)\right] + \lambda_*(\theta_*)'\sum_{i=1}^n \wtl g_b(w_i,\theta_*).\label{eq:1:lem_likelihood_ET_approx}
\end{multline}
\indent Next, we analyse the fourth term on the right hand side of \eqref{eq_likelihood_ET}. Let $\tilde\lambda$ be on the line joining $\lambda_*(\theta_*)$ and $\wh\lambda(\theta_*)$, then a second order Taylor expansion of $\lambda\mapsto \frac{1}{n}\sum_{j=1}^n e^{\lambda' g_b(w_j,\theta_*)}$ around $\lambda_*(\theta_*)$ gives
\begin{multline}\label{eq_Taylor_exp_1}
  \frac{1}{n}\sum_{j=1}^n e^{\wh\lambda(\theta_*)'g_b(w_j,\theta_*)} = \frac{1}{n}\sum_{j=1}^n e^{\lambda_*(\theta_*)' g_b(w_j,\theta_*)} + (\wh\lambda(\theta_*)' - \lambda_*(\theta_*)')\frac{1}{n}\sum_{i=1}^n e^{\lambda_*(\theta_*)' g_b(w_j,\theta_*)} g_b(w_i,\theta_*)\\\hfill
   + \frac{1}{2}(\wh\lambda(\theta_*)' - \lambda_*(\theta_*)')\frac{1}{n}\sum_{j=1}^n e^{\tilde\lambda' g_b(w_j,\theta_*)} g_b(w_j,\theta_*)g_b(w_j,\theta_*)'(\wh\lambda(\theta_*) - \lambda_*(\theta_*)).
\end{multline}
Under Assumption \ref{Ass_1_NS} and because $\|\tilde\lambda - \lambda_*(\theta_*)\|_2 = \mathcal{O}_p(n^{-1/2})$ (since by Lemma \ref{lem:2} $\|\wh\lambda(\theta_*) - \lambda_*(\theta_*)\|_2 = \mathcal{O}_p(n^{-1/2})$ and $\tilde\lambda = \tau(\wh\lambda(\theta_*) - \lambda_*(\theta_*)) + \lambda_*(\theta_*)$ for some $\tau\in[0,1]$) we can apply the same argument of the proof of Lemma \ref{lem:Omega} to get:
\begin{equation}
  \frac{1}{n}\sum_{j=1}^n e^{\tilde\lambda' g_b(w_j,\theta_*)} g_b(w_j,\theta_*)g_b(w_j,\theta_*)' \overset{p}{\to} \Omega_*^{\diamond}(\theta_*) := \EE[e^{\lambda_*(\theta_*)' g_b(w_j,\theta_*)} \varepsilon_i(\theta_*)^2\wtl w_i\wtl w_i'].
\end{equation}
By replacing this in \eqref{eq_Taylor_exp_1} and by using Lemma \ref{lem:2} to get the rate of the term $o_p(1/n)$ term, we obtain:
\begin{multline}\label{eq_Taylor_exp_2}
  \frac{1}{n}\sum_{j=1}^n e^{\wh\lambda(\theta_*)'g_b(w_j,\theta_*)} = \frac{1}{n}\sum_{j=1}^n e^{\lambda_*(\theta_*)'g_b(w_j,\theta_*)} + (\wh\lambda(\theta_*)' - \lambda_*(\theta_*)')\frac{1}{n}\sum_{i=1}^n e^{\lambda_*(\theta_*)' g_b(w_j,\theta_*)} g_b(w_i,\theta_*)\\\hfill
   + \frac{1}{2}(\wh\lambda(\theta_*)' - \lambda_*(\theta_*)')\Omega_*^{\diamond}(\theta_*)(\wh\lambda(\theta_*) - \lambda_*(\theta_*)) + o_p\left(\frac{1}{n}\right)%\\
  %= \EE_n\left[ e^{\lambda_*(\theta_*)' g_b(w_j,\theta_*)}\right] + (\wh\lambda(\theta_*)' - \lambda_*(\theta_*)')\frac{1}{n}\sum_{i=1}^n e^{\lambda_*(\theta_*)'g_b(w_j,\theta_*)} g_b(w_i,\theta_*)\\\hfill
  %+ \frac{1}{2}(\wh\lambda(\theta_*)' - \lambda_*(\theta_*)')\Omega_*^{\diamond}(\theta_*)(\wh\lambda(\theta_*) - \lambda_*(\theta_*)) + o_p\left(\frac{1}{n}\right).
\end{multline}
\noindent %where $\mathcal{A} := (\frac{1}{n}\sum_{j=1}^n e^{\lambda_*(\theta_*)'g_b(w_j,\theta_*)} - \EE\left[ e^{\lambda_*(\theta_*)' g_b(w_j,\theta_*)}\right]) = \mathcal{O}_p(1/\sqrt{n})$ by the Markov's inequality and under Assumption \ref{Ass_3} \textit{(b)}.
We now use the first order Taylor expansion of the function $u\mapsto \log(u)$ around $v$: $\log(u) = \log(v) +\frac{u-v}{v} + o(|u-v|)$, and plug \eqref{eq_Taylor_exp_2} in it to obtain:
\begin{multline}\label{eq_Taylor_exp_3}
  \log\lt(\frac{1}{n}\sum_{i=1}^n e^{\wh\lambda(\theta_*)' g_b(w_i,\theta_*)}\rt) = \log\lt(\EE_n\left[ e^{\lambda_*(\theta_*)' g_b(w_i,\theta_*)}\right]\rt) + (\wh\lambda(\theta_*)' - \lambda_*(\theta_*)')\frac{1}{n}\sum_{i=1}^n \tau_i(\lambda_*,\theta_*)g_b(w_i,\theta_*)\\\hfill
  + \frac{1}{2}(\wh\lambda(\theta_*)' - \lambda_*(\theta_*)')\Omega_*^{\dagger}(\theta_*)(\wh\lambda(\theta_*) - \lambda_*(\theta_*)) + o_p\lt(\frac{1}{n}\rt)\\
  + o\lt(\lt|\frac{1}{n}\sum_{i=1}^n e^{\wh\lambda(\theta_*
  )' g_b(w_i,\theta_*)} - \EE_n\left[ e^{\lambda_*(\theta_*)' g_b(w_i,\theta_*)}\right]\rt|\rt),
\end{multline}
\noindent where to get the term $o_p\lt(\frac{1}{n}\rt)$ we have used Lemma \ref{lem:2} and the fact that $|\EE_n\left[ e^{\lambda_*(\theta_*)' g_b(w_i,\theta_*)}\right] - \EE\left[ e^{\lambda_*(\theta_*)' g_b(w_i,\theta_*)}\right]| = \mathcal{O}_p(1/\sqrt{n})$.
%By using the result of Lemma \ref{lem:1} and by noting that $\Omega_\circ^{\diamond}(\theta_*) = \Omega_* := \EE[\varepsilon_i(\theta_*)^2\wtl w_i\wtl w_i']$ since we are in the correctly specified case, then
By using again the latter convergence, Lemma \ref{lem:1}, and the facts that $\EE[\tau_i^{\dagger}(\lambda_*,\theta_*) g_b(w_i,\theta_*)] = 0$ and
\begin{displaymath}
  \frac{1}{\sqrt{n}}\sum_{i=1}^n \tau_i^{\dagger}(\lambda_*,\theta_*) g_b(w_i,\theta_*) = \mathbb{G}_n\left[\tau_i^{\dagger}(\lambda_*,\theta_*) g_b(w_i,\theta_*)\right],%\\
%  = \mathbb{G}_n\left[\frac{1}{n}\sum_{i=1}^n e^{\lambda_*(\theta_*)'\wtl g_b(w_j,\theta_*)} g_b(w_i,\theta_*)\right],
\end{displaymath}
then
\begin{multline}
  n\left((\wh\lambda(\theta_*)' - \lambda_*(\theta_*)')\frac{1}{n}\sum_{i=1}^n \tau_i(\lambda_*,\theta_*)g_b(w_i,\theta_*) +  \frac{1}{2}(\wh\lambda(\theta_*)' - \lambda_*(\theta_*)')\Omega_*^{\dagger}(\theta_*)(\wh\lambda(\theta_*) - \lambda_*(\theta_*))\right)\\
  %= \sqrt{n}\wh\lambda(\theta_*)'\mathbb{G}_n\left[\tau_i^{\dagger}(\lambda_*,\theta_*)\wtl g_b(w_i,\theta_*)\right] + \frac{1}{2}\sqrt{n}\wh\lambda(\theta_*)' \Omega_\circ\sqrt{n}\wh\lambda(\theta_*)\\
  = - \mathbb{G}_n\left[\tau_i^{\dagger}(\lambda_*,\theta_*) g_b(w_i,\theta_*)'\right]\Omega_*^{\dagger}(\theta_*)^{-1}\mathbb{G}_n\left[\tau_i^{\dagger}(\lambda_*,\theta_*) g_b(w_i,\theta_*)\right]\\
  %+ n(\wh\lambda(\theta_*)' - \lambda_*(\theta_*)')\frac{1}{n}\sum_{i=1}^n \tau_i^{\dagger}(\lambda_*,\theta_*)\EE[g_b(w_i,\theta_*)]\\
  + \frac{1}{2}\mathbb{G}_n\left[\tau_i^{\dagger}(\lambda_*,\theta_*)g_b(w_i,\theta_*
  )'\right] \Omega_*^{\dagger}(\theta_*)^{-1} \mathbb{G}_n\left[\tau_i^{\dagger}(\lambda_*,\theta_*)g_b(w_i,\theta_*)\right] + o_p(1)\\
  =  -\frac{1}{2}\mathbb{G}_n\left[\tau_i^{\dagger}(\lambda_*,\theta_*) g_b(w_i,\theta_*)'\right] \Omega_*^{\dagger}(\theta_*)^{-1} \mathbb{G}_n\left[\tau_i^{\dagger}(\lambda_*,\theta_*) g_b(w_i,\theta_*)\right]
  + o_p(1).\label{eq:2:lem_likelihood_ET_approx}
\end{multline}
Finally, we have to deal with $\lt|\frac{1}{n}\sum_{i=1}^n e^{\wh\lambda(\theta_*)' g_b(w_i,\theta_*)} - \EE_n\left[ e^{\lambda_*(\theta_*)' g_b(w_i,\theta_*)}\right]\rt|$. By using \eqref{eq_Taylor_exp_2} and \eqref{eq:2:lem_likelihood_ET_approx}
%\begin{displaymath}
%  \frac{1}{\sqrt{n}}\sum_{i=1}^n e^{\lambda_*(\theta_*)' g_b(w_j,\theta_*)} g_b(w_i,\theta_*) = \mathbb{G}_n\left[e^{\lambda_*(\theta_*)'g_b(w_j,\theta_*)} g_b(w_i,\theta_*)\right],%\\
%  = \mathbb{G}_n\left[\frac{1}{n}\sum_{i=1}^n e^{\lambda_*(\theta_*)'\wtl g_b(w_j,\theta_*)} g_b(w_i,\theta_*)\right],
%\end{displaymath}
 we get that $\lt|\frac{1}{n}\sum_{i=1}^n e^{\wh\lambda(\theta_*)' g_b(w_i,\theta_*)} - \EE_n\left[ e^{\lambda_*(\theta_*)' g_b(w_i,\theta_*)}\right]\rt| = \mathcal{O}_p(1/n)$.

%By using exactly the same argument we have used in the proof of Lemma \ref{lem:A_2:A_3} with $\lambda_\circ(\theta_*) = 1$ we get that $\lt|\frac{1}{n}\sum_{i=1}^n e^{\wh\lambda(\theta_*)' \wtl g_b(w_i,\theta_*
%)} - \EE\left[ e^{\lambda_*(\theta_*)' \wtl g_b(w_i,\theta_*)}\right]\rt| = \mathcal{O}_p(1/\sqrt{n})$.\\
\indent By replacing this result, \eqref{eq:1:lem_likelihood_ET_approx}, \eqref{eq_Taylor_exp_3}, and \eqref{eq:2:lem_likelihood_ET_approx} in \eqref{eq_likelihood_ET} we get:
\begin{multline}
  \log \wh{q}(w_{1:n}|\theta_*,M_{b}) = - n\log n - \mathbb{G}_n\left[\tau_i^{\dagger}(\lambda_*,\theta_*) g_b(w_i,\theta_*)'\right]\Omega_*^{\dagger}(\theta_*)^{-1}\mathbb{G}_n\left[g_b(w_i,\theta_*)\right]\\
  + \lambda_*(\theta_*)'\sum_{i=1}^n \wtl g_b(w_i,\theta_*) +  n\wh\lambda(\theta_*)'\EE[g_b(w_i,\theta_*)] - n\log\lt(\EE_n\left[ e^{\lambda_*(\theta_*)' g_b(w_i,\theta_*)}\right]\rt)\\
  + \frac{1}{2}\mathbb{G}_n\left[\tau_i^{\dagger}(\lambda_*,\theta_*) g_b(w_i,\theta_*)'\right] \Omega_*^{\dagger}(\theta_*)^{-1} \mathbb{G}_n\left[\tau_i^{\dagger}(\lambda_*,\theta_*) g_b(w_i,\theta_*)\right]%\\
  %+ \mathbb{G}_n\left[\tau_i^{\dagger}(\lambda_*,\theta_*)g_b(w_i,\theta_*)'\right] \Omega_*^{\dagger}(\theta_*)^{-1} \frac{1}{\sqrt{n}}\sum_{i=1}^n \tau_i^{\dagger}(\lambda_*,\theta_*)\EE[g_b(w_i,\theta_*)
  + o_p(1),
\end{multline}
\noindent where $\mathbb{G}_n\left[\tau_i^{\dagger}(\lambda_*,\theta_*) g_b(w_i,\theta_*)\right] \xrightarrow{d} \mathcal{N}(0,\Omega_*^{\dagger}(\theta_*))$. By noticing that $\lambda_*(\theta_*)'\sum_{i=1}^n \wtl g_b(w_i,\theta_*) - n\log\lt(\EE_n\left[ e^{\lambda_*(\theta_*)' g_b(w_i,\theta_*)}\right]\rt) = \sum_{i=1}^n \log\left(\tau_i(\lambda_*,\theta_*)\right) - n\lambda_*(\theta_*)'\EE[g_b(w_i,\theta_*)]$, we prove \eqref{lem_likelihood_eq:1}. \\
%\indent Now, suppose that $\EE[\varepsilon_i(\theta_\circ) x_i] = 0$. Then, $\theta_*$ is replaced by $\theta_\circ$, $\lambda_*(\theta_*)=\lambda_*(\theta_\circ) = 0$, $\EE[g_b(w_i,\theta_\circ)] = 0$, $\tau_i^{\dagger}(\lambda_*,\theta_*) = 1$,  $\Omega_*^{\dagger}(\theta_*) = \Omega_\circ$, and so:
%\begin{multline}
%  \log \wh{q}(w_{1:n}|\theta_*,M_{b}) = - n\log n - \mathbb{G}_n\left[\tau_i^{\dagger}(\lambda_*,\theta_*) g_b(w_i,\theta_*)'\right]\Omega_*^{\dagger}(\theta_*)^{-1}\mathbb{G}_n\left[g_b(w_i,\theta_*)\right]\\
%  + \frac{1}{2}\mathbb{G}_n\left[\tau_i^{\dagger}(\lambda_*,\theta_*) g_b(w_i,\theta_*)'\right] \Omega_*^{\dagger}(\theta_*)^{-1} \mathbb{G}_n\left[\tau_i^{\dagger}(\lambda_*,\theta_*) g_b(w_i,\theta_*)\right] + o_p(1)\\
%   = - n\log n - \frac{1}{2}\mathbb{G}_n\left[g_b(w_i,\theta_*)'\right]\Omega_*^{\dagger}(\theta_*)^{-1}\mathbb{G}_n\left[g_b(w_i,\theta_*)\right] + o_p(1),
%\end{multline}
%\noindent where $\mathbb{G}_n\left[g(w_i,\theta_*)\right] \xrightarrow{d} \mathcal{N}(0,\Omega_\circ)$ and $\mathbb{G}_n\left[g(w_i,\theta_*)'\right] \Omega_\circ^{-1} \mathbb{G}_n\left[g(w_i,\theta_*)\right]\xrightarrow{d}\chi_d^2$. This proves \eqref{eq_lem_likelihood_ET_approx}.
The last part of the theorem follows from Lemma \ref{lem:1}
\begin{flushright}
$\square$
\end{flushright}

% ==============================================================================================
\subsection{Proof of Theorem \ref{lem_likelihood_ET_approx_extended}}
Since we are in the extended model, then there exists a $\psi_\circ := (\theta_\circ', v_\circ')'$ such that $\EE[\varepsilon_i(\theta_\circ)\wtl w_i] = (v_\circ',0')'$ and $\lambda_*(\psi_\circ) = 0$. Let us consider the expression for the likelihood evaluated at $\psi_\circ$:
\begin{multline}\label{eq_likelihood_ET_extended}
  \log \wh{q}(w_{1:n}|\psi_\circ,\mathcal{M}_{e}) = - n\log n +\sum_{i=1}^n\wh{\lambda}(\psi_\circ)'g_b(w_i,\theta_\circ)
   - n\log \frac{1}{n}\sum_{j=1}^n e^{\wh{\lambda}(\psi_\circ)'g_b(w_j,\theta_\circ)}\\
   = - n\log n +\sum_{i=1}^n\wh{\lambda}(\psi_\circ)'g_e(w_i,\psi_\circ)
   - n\log \frac{1}{n}\sum_{j=1}^n e^{\wh{\lambda}(\psi_\circ)'g_e(w_j,\psi_\circ)}.
\end{multline}
We start with dealing with the second term on the right hand side of \eqref{eq_likelihood_ET_extended}. By using the result of Lemma \ref{lem:1:extended}:
\begin{multline}
  \sum_{i=1}^n\wh{\lambda}(\psi_\circ)'g_e(w_i,\psi_\circ) = \sqrt{n}\wh{\lambda}(\psi_\circ)'\frac{1}{\sqrt{n}}\sum_{i=1}^n g_e(w_i,\psi_\circ)\\
  = - \mathbb{G}_n\left[g_e(w_i,\psi_\circ)'\right]\Omega_{\psi_{\circ}}^{-1}\mathbb{G}_n\left[g_e(w_i,\psi_\circ)\right] + o_p(1).\label{eq:1:lem_likelihood_ET_approx_extended}
\end{multline}
\indent Let $\tilde\lambda$ be on the line joining $0$ and $\wh\lambda(\psi_\circ)$, then a second order Taylor expansion of the function $\lambda\mapsto \frac{1}{n}\sum_{j=1}^n e^{\lambda'g_e(w_j,\theta_\circ)}$ around $0$ gives
\begin{multline}\label{eq_Taylor_exp_1_extended}
  \frac{1}{n}\sum_{j=1}^n e^{\wh\lambda(\psi_\circ)'g_e(w_j,\psi_\circ)} = 1 + \frac{1}{n}\sum_{i=1}^n\wh\lambda(\psi_\circ)' g_e(w_i,\psi_\circ)\\\hfill
   + \frac{1}{2}\wh\lambda(\psi_\circ)' \frac{1}{n}\sum_{j=1}^n e^{\tilde\lambda' g_e(w_j,\psi_\circ)} g_e(w_j,\psi_\circ)g_e(w_j,\psi_\circ)'\wh\lambda(\psi_\circ).
\end{multline}
Under Assumption \ref{Ass_1_NS} and because $\tilde\lambda = \mathcal{O}_p(n^{-1/2})$ (since by Lemma \ref{lem:2:extended} $\wh\lambda(\psi_\circ) = \mathcal{O}_p(n^{-1/2})$ and $\tilde\lambda$ is between $0$ and $\wh\lambda(\psi_\circ)$), we can apply the same argument of the proof of Lemma \ref{lem:Omega} to get:
\begin{equation}
  \frac{1}{n}\sum_{j=1}^n e^{\tilde\lambda' g_e(w_j,\theta_\circ)} g_e(w_j,\theta_\circ)g_e(w_j,\theta_\circ)' \overset{p}{\to} \Omega_{\psi_{\circ}} := \EE[g_e(w_i,\psi_\circ)g_e(w_i,\psi_\circ)'].
\end{equation}
By replacing this in \eqref{eq_Taylor_exp_1_extended} and by using Lemma \ref{lem:2:extended} to get the rate of the $o_p(1/n)$ term, we obtain:
\begin{equation}\label{eq_Taylor_exp_2_extended}
  \frac{1}{n}\sum_{j=1}^n e^{\wh\lambda(\psi_\circ)'g_e(w_j,\psi_\circ)} = 1 + \frac{1}{n}\sum_{i=1}^n\wh\lambda(\psi_\circ)' g_e(w_i,\psi_\circ) + \frac{1}{2}\wh\lambda(\psi_\circ)' \Omega_{\psi_{\circ}}\wh\lambda(\psi_\circ) + o_p(1/n).
\end{equation}
We now use the first order Taylor expansion of the function $u\mapsto\log(u)$ around $1$: $\log(u) = u-1 + o(|u-1|)$, and apply it to $u$ given by the left hand side of \eqref{eq_Taylor_exp_2_extended} to obtain:
\begin{multline}\label{eq_Taylor_exp_3_extended}
  \log\lt(\frac{1}{n}\sum_{i=1}^n e^{\wh\lambda(\psi_\circ)' g_e(w_i,\psi_\circ)}\rt) = \frac{1}{n}\sum_{i=1}^n \wh\lambda(\psi_\circ)' g_e(w_i,\psi_\circ)\\\hfill
   + \frac{1}{2}\wh\lambda(\psi_\circ)' \Omega_{\psi_\circ}\wh\lambda(\psi_\circ) + o_p\lt(1/n\rt) + o\lt(\lt|\frac{1}{n}\sum_{i=1}^n e^{\wh\lambda(\psi_\circ)' g_e(w_i,\psi_\circ)} - 1\rt|\rt).
\end{multline}
By using the result of Lemma \ref{lem:1:extended} and the fact that $\EE[g_e(w_i,\psi_\circ)] = 0$, then
\begin{multline}
  n\left(\frac{1}{n}\sum_{i=1}^n \wh\lambda(\psi_\circ)' g_e(w_i,\psi_\circ) + \frac{1}{2}\wh\lambda(\psi_\circ)' \Omega_{\psi_\circ}\wh\lambda(\psi_\circ)\right)\\
  = \sqrt{n}\wh\lambda(\psi_\circ)'\mathbb{G}_n\left[g_e(w_i,\psi_\circ)\right] + \frac{1}{2}\sqrt{n}\wh\lambda(\psi_\circ)' \Omega_{\psi_\circ}\sqrt{n}\wh\lambda(\psi_\circ)\\
  = - \mathbb{G}_n\left[g_e(w_i,\psi_\circ)'\right]\Omega_{\psi_\circ}^{-1}\mathbb{G}_n\left[g_e(w_i,\psi_\circ)\right] %+ 2\sqrt{n}\mathbb{G}_n\left[g_b(w_i,\theta_\circ)'\right]\Omega_{\psi_\circ}^{-1}\wtl{v}_{\circ}\\
%  - n\wtl{v}_\circ' \Omega_{\psi_\circ}^{-1}\wtl{v}_\circ
  + \frac{1}{2}\mathbb{G}_n\left[g_e(w_i,\psi_\circ)'\right] \Omega_{\psi_\circ}^{-1} \mathbb{G}_n\left[g_e(w_i,\psi_\circ)\right]  + o_p(1)\\
  =  -\frac{1}{2}\mathbb{G}_n\left[g_e(w_i,\psi_\circ)'\right] \Omega_*^{-1} \mathbb{G}_n\left[g_e(w_i,\psi_\circ)\right] + o_p(1).\label{eq:2:lem_likelihood_ET_approx_extended}
\end{multline}
Finally, we have to deal with $\lt|\frac{1}{n}\sum_{i=1}^n e^{\wh\lambda(\psi_\circ)' g_e(w_i,\psi_\circ)} - 1\rt|$. %By using exactly the same argument we have used in the proof of Lemma \ref{lem:A_2:A_3} with $\lambda_*(\theta_*)$ replaced by $\lambda_*(\psi_\circ) = 0$,
By \eqref{eq_Taylor_exp_3_extended} and \eqref{eq:2:lem_likelihood_ET_approx_extended} we get that $\lt|\frac{1}{n}\sum_{i=1}^n e^{\wh\lambda(\psi_\circ)' g_e(w_i,\psi_\circ)} - 1\rt| = \mathcal{O}_p(1/n)$.\\
\indent By replacing this result and \eqref{eq:2:lem_likelihood_ET_approx_extended} in \eqref{eq_Taylor_exp_3_extended}, and then by plugging \eqref{eq:1:lem_likelihood_ET_approx_extended} and \eqref{eq_Taylor_exp_3_extended} in \eqref{eq_likelihood_ET_extended} we get:
\begin{multline}
  \log \wh{q}(w_{1:n}|\psi_\circ,M_{b}) = - n\log n - \frac{1}{2}\mathbb{G}_n\left[g_e(w_i,\psi_\circ)'\right] \Omega_{\psi_\circ}^{-1} \mathbb{G}_n\left[g_e(w_i,\psi_\circ)\right] + o_p(1)\\
  = - n\log n - \frac{1}{2}\mathbb{G}_n\left[g_b(w_i,\theta_\circ)'\right] \Omega_{\psi_{\circ}}^{-1} \mathbb{G}_n\left[g_b(w_i,\theta_\circ)\right] - 2\sqrt{n}\mathbb{G}_n\left[g_b(w_i,\theta_\circ)'\right]\Omega_{\psi_\circ}^{-1}\wtl{v}_{\circ} + n\wtl{v}_\circ' \Omega_{\psi_\circ}^{-1}\wtl{v}_\circ + o_p(1).
\end{multline}
\noindent Moreover, by the central limit theorem,
\[\mathbb{G}_n\left[g_e(w_i,\psi_\circ)\right] \xrightarrow{d} \mathcal{N}(0,\Omega_{\psi_\circ})
\]
and
\[
\mathbb{G}_n\left[g_e(w_i,\psi_\circ)'\right] \Omega_{\psi_{\circ}}^{-1} \mathbb{G}_n\left[g_e(w_i,\psi_\circ)\right]\xrightarrow{d}\chi_d^2.
\]\\
%\indent If $\EE[\varepsilon_i(\theta_\circ) x_i] = 0$, then  $v_{\circ} = 0$, $\Omega_{\psi_\circ} = \Omega_\circ:= \EE[\varepsilon_i(\theta_\circ)^2\wtl w_i\wtl w_i']$, and the previous expression simplifies to
%\begin{equation}
%  \log \wh{q}(w_{1:n}|\psi_\circ,M_{b}) = - n\log n - \frac{1}{2}\mathbb{G}_n\left[g_b(w_i,\theta_\circ)'\right] \Omega_{\circ}^{-1} \mathbb{G}_n\left[g_b(w_i,\theta_\circ)\right].
%\end{equation}

\begin{flushright}
$\square$
\end{flushright}

% ==============================================================================================
\subsection{Proof of Corollary \ref{thm_limit_ETEL}}
By result \eqref{lem_likelihood_eq:1} in Theorem \ref{lem_likelihood_ET_approx} we have that
\begin{multline}
  \frac{1}{n}\log \wh{q}(w_{1:n}|\theta_*,\mathcal{M}_b) + \log(n) = \frac{1}{n}\sum_{i=1}^n \log\left(\frac{e^{\lambda_*(\theta_*)' g_b(w_i,\theta_*)}}{\EE_n[e^{\lambda_*(\theta_*)' g_b(w_j,\theta_*)}]}\right)\\
  + (\wh\lambda(\theta_*) - \lambda_*(\theta_*))'\EE[g_b(w_i,\theta_*)] + \mathcal{O}_p(1/n).
\end{multline}
 By Lemma \ref{lem:2} in the Online Appendix, $\|\wh\lambda(\theta_*) - \lambda_*(\theta_*)\|_2\xrightarrow{p} 0$. By the Law of Large Numbers
 $$\frac{1}{n}\sum_{i=1}^n \log\left(\frac{e^{\lambda_*(\theta_*)' g_b(w_i,\theta_*)}}{\EE_n[e^{\lambda_*(\theta_*)' g_b(w_j,\theta_*)}]}\right)\xrightarrow{p} \EE\left[\log\left(\frac{e^{\lambda_*(\theta_*)'g_b(w_i,\theta_*)}}{\EE[e^{\lambda_*(\theta_*)' g_b(w_j,\theta_*)}]}\right)\right].$$
 This concludes the proof.

\begin{flushright}
$\square$
\end{flushright}

\subsection{Proof of Corollary \ref{thm_limit_ETEL_extended_model}}
By result \eqref{eq_lem_likelihood_ET_approx_extended} in Theorem \ref{lem_likelihood_ET_approx_extended} we have that
$$\frac{1}{n}\log \wh{q}(w_{1:n}|\psi_\circ,\mathcal{M}_e) + \log(n) = \mathcal{O}_p(1/n).$$
Since $\lambda_*(\psi_\circ) = 0$ then, $\frac{e^{\lambda_*(\psi_\circ)'\sum_{i=1}^n g_b(w_i,\psi_\circ)}}{\EE[e^{\lambda_*(\psi_\circ)'\sum_{i=1}^n g_b(w_j,\psi_\circ)}]} = 1$ and so we can equivalently write:
\begin{displaymath}
  \frac{1}{n}\log \wh{q}(w_{1:n}|\psi_\circ,\mathcal{M}_e) + \log(n) = \frac{1}{n}\sum_{i=1}^n \log\left(\frac{e^{\lambda_*(\psi_\circ)'g_b(w_i,\psi_\circ)}}{\EE[e^{\lambda_*(\psi_\circ)' g_b(w_j,\psi_\circ)}]}\right) + \mathcal{O}_p(1/n).
\end{displaymath}
\noindent By the Law of Large Numbers
 $$\frac{1}{n}\sum_{i=1}^n \log\left(\frac{e^{\lambda_*(\psi_\circ)'g_b(w_i,\psi_\circ)}}{\EE[e^{\lambda_*(\psi_\circ)' g_b(w_j,\psi_\circ)}]}\right)\xrightarrow{p} \EE\left[\log\left(\frac{e^{\lambda_*(\psi_\circ)'g_b(w_i,\psi_\circ)}}{\EE[e^{\lambda_*(\psi_\circ)' g_b(w_j,\psi_\circ)}]}\right)\right].$$
 This concludes the proof.

\begin{flushright}
$\square$
\end{flushright}
% ==============================================================================================
\subsection{Proof of Theorem \ref{thm_consistency_selection_misspecification_2}}
\indent The proof is organized in two parts. In the first part we show that $\mathrm{KL}(P||Q_{e}^{*}(\psi_\circ)) < \mathrm{KL}(P||Q_b^{*}(\theta_{*}))$ if and only if there is no $\theta$ such that $\mathbf{E}[\varepsilon_{i}(\theta)x_i]=0$. In the second part we show that
$$P\left(\log m(w_{1:n}|\mathcal{M}_{e})> \log m(w_{1:n}|\mathcal{M}_{b})\right) \rightarrow 1$$ if and only if $\mathrm{KL}(P||Q_{e}^{*}(\psi_\circ)) < \mathrm{KL}(P||Q_b^{*}(\theta_{*}))$.

\paragraph{First part.} We start by proving that $\mathrm{KL}(P||Q_{e}^{*}(\psi_\circ)) < \mathrm{KL}(P||Q_b^{*}(\theta_{*}))$ if and only if there is no $\theta\in\Theta$ such that $\mathbf{E}[\varepsilon_{i}(\theta)x_i]=0$. Notice that $\mathrm{KL}(P||Q_{e}^{*}(\psi_\circ))=0$. Suppose that $\mathrm{KL}(P||Q_{e}^{\ast}(\psi_\circ))<\mathrm{KL}(P||Q_b^{\ast }(\theta_{*}))$ and suppose that there exists a $\theta\in\Theta$ such that $\mathbf{E}[\varepsilon_{i}(\theta)x_i] = 0$ so that $P\in \mathcal{Q}_{b}(\theta)$. By Assumption \ref{Ass_0_NS} with $\theta_*$ replaced by $\theta_\circ$, then this $\theta$ must be equal to $\theta_\circ$, which in turn equals $\theta_*$. It follows that $P\in \mathcal{Q}_{b}(\theta_*)$ and by definition of $Q_b^{\ast }(\theta_{*})$: $Q_b^{\ast }(\theta_{*})=P$ since $Q_b^{\ast }(\theta_{*})$ is the closest to $P$, in the KL sense, among all the distributions in $\mathcal{Q}_{b}(\theta_*)$. Hence, $\mathrm{KL}(P||Q_b^{\ast }(\theta_{*}))=0$. But this contradicts the assumption that $\mathrm{KL}(P||Q_{e}^{\ast }(\psi_\circ))<\mathrm{KL}(P||Q_b^{\ast }(\theta _{*}))$. Hence, there is no $\theta\in\Theta$ such that $\mathbf{E}[\varepsilon_{i}(\theta)x_i]=0$.\\
\indent We now prove the reverse implication. Suppose that there is no value $\theta\in\Theta$ such that $\mathbf{E}[\varepsilon_{i}(\theta)x_i]=0$. Hence, $P\notin \mathcal{Q}_{b}(\theta)$ for every $\theta\in\Theta$, which implies $P\notin \mathcal{Q}_{b}(\theta_*)$ and $\mathrm{KL}(P||Q_b^{*}(\theta_*)) > 0$. On the other hand, there exists a unique $\psi_\circ\in\mathbb{R}^{d_x}$ such that $P\in \mathcal{Q}_{e}(\psi_\circ)$ since $\mathcal{M}_e$ is always correctly specified. This implies that $\mathrm{KL}(P||Q_{e}^{*}(\psi_\circ)) = 0$ and so $\mathrm{KL}(P||Q_{e}^{*}(\psi_\circ)) < \mathrm{KL}(P||Q_b^{*}(\theta_*))$.\\

\paragraph{Second part.} We show that $P\left(\log m(w_{1:n}|\mathcal{M}_{e})> \log m(w_{1:n}|\mathcal{M}_{b})\right) \rightarrow 1$ if and only if $\mathrm{KL}(P||Q_{e}^{*}(\psi_\circ))$ $<$ $\mathrm{KL}(P||Q_b^{*}(\theta_{*}))$. By Theorems \ref{lem_likelihood_ET_approx} and \ref{lem_likelihood_ET_approx_extended}, and Theorems \ref{thm_BvM_misspecified} and \ref{thm_BvM:extended} in the Online Appendix and by \eqref{log_ML_base_initial}-\eqref{log_ML_extended_initial}, then \eqref{log_ML_base}-\eqref{log_ML_extended} hold. By \eqref{log_ML_base}, the $\log m(w_{1:n}|\mathcal{M}_b)$ is equal to
$$- n\log n + \sum_{i=1}^n \log\left(\frac{e^{\lambda_*(\theta_*)'g_b(w_i,\theta_*)}}{\EE_n[e^{\lambda_*(\theta_*)' g_b(w_j,\theta_*)}]}\right) + n(\wh\lambda(\theta_*) - \lambda_*(\theta_*))'\EE[g_b(w_i,\theta_*)] - \frac{p}{2}\log(n) + \mathcal{O}_p(1)$$
\noindent and by \eqref{log_ML_extended}, $\log m(w_{1:n}|\mathcal{M}_e) = -n\log(n) - \frac{p+d_x}{2}\log(n) + \mathcal{O}_p(1)$. Hence, since from the Law of Large Numbers
$$\frac{1}{n}\sum_{i=1}^n \log\left(\frac{e^{\lambda_*(\theta_*)'\sum_{i=1}^n g_b(w_i,\theta_*)}}{\EE_n[e^{\lambda_*(\theta_*)'\sum_{i=1}^n g_b(w_j,\theta_*)}]}\right)\xrightarrow{p} \EE\left[\log\left(\frac{e^{\lambda_*(\theta_*)'\sum_{i=1}^n g_b(w_i,\theta_*)}}{\EE[e^{\lambda_*(\theta_*)'\sum_{i=1}^n g_b(w_j,\theta_*)}]}\right)\right] \equiv - \mathrm{KL}(P||Q_b^*(\theta_*)),$$
it follows that
%, we have that for every $M_n$ such that $M_n^2/\sqrt{n} \rightarrow 0$ as $n\rightarrow \infty$:
%\begin{equation}
%  \frac{1}{n}\log m(w_{1:n}|\mathcal{M}_b) + \log(n) = \mathbf{E}^P\left[\log(dQ_b^*(\theta_*)/dP)\right] + o_p(1) = - KL(P||Q_b^*(\theta_*)) + o_p(1)
%\end{equation}
%\noindent since $-\mathbf{E}^P\left[\log(dQ_b^*(\theta_*)/dP)\right] = KL(P||Q_b^*(\theta_*))$. Similarly for the extended model, by the results of Theorems \ref{thm_approx_marginal_likelihood_extended_model} and \ref{thm_limit_ETEL_extended_model}, then we have that for every $M_n$ such that $M_n^2/\sqrt{n} \rightarrow 0$ as $n\rightarrow \infty$:
%\begin{equation}
%  \frac{1}{n}\log m(w_{1:n}|\mathcal{M}_e) + \log(n) = \mathbf{E}^P\left[\log(dQ_e^*(\psi_\circ)/dP)\right] + o_p(1) = o_p(1)
%\end{equation}
%\noindent since $-\mathbf{E}^P\left[\log(dQ_e^*(\psi_\circ)/dP)\right] = KL(P||Q_e^*(\psi_\circ))$. Hence, since $KL(P||Q_e^*(\psi_\circ)) = 0$:
\begin{multline*}
  P\left(\log m(w_{1:n}|\mathcal{M}_{e})> \log m(w_{1:n}|\mathcal{M}_{b})\right)
  = P\left(\frac{1}{n}\log m(w_{1:n}|\mathcal{M}_{e}) > \frac{1}{n}\log m(w_{1:n}|\mathcal{M}_{b}) \right)\\
  = P\left(0 > - \mathrm{KL}(P||Q_b^*(\theta_*))+ \mathcal{O}_p(\log(n)/\sqrt{n})\right),
\end{multline*}
\noindent where we have used Lemma \ref{lem:2} in the Online Appendix to control $\sqrt{n}(\wh\lambda(\theta_*) - \lambda_*(\theta_*))$.
Suppose that %there is no $\theta$ such that $\mathbf{E}[\varepsilon_{i}(\theta)x_i]=0$, then
$\mathrm{KL}(P||Q_b^*(\theta_*))>0$, then the previous probability converges to $1$. On the other hand, suppose that $P\left(0 > - \mathrm{KL}(P||Q_b^*(\theta_*)) + \mathcal{O}_p(\log(n)/\sqrt{n})\right) \rightarrow 1$ as $n\rightarrow \infty$. This is possible only if $\mathrm{KL}(P||Q_b^*(\theta_*))>0$. By the first part of the proof $\mathrm{KL}(P||Q_b^*(\theta_*))>0$ if and only if there is no $\theta$ such that $\mathbf{E}[\varepsilon_{i}(\theta)x_i]=0$.\\
\indent We now prove the last assertion of the theorem. In the case where there is a $\theta_{\circ}$ such that $\mathbf{E}[\varepsilon_{i}(\theta_\circ)x_i]=0$, then $\mathrm{KL}(P||Q_b^*(\theta_\circ))=0$ and the probability $P\left(0 > - \mathrm{KL}(P||Q_b^*(\theta_*)) + \mathcal{O}_p(\log(n)/\sqrt{n})\right)$ is equal to zero as $n\rightarrow\infty$. This concludes the proof.

\begin{flushright}
$\square$
\end{flushright}

% ==============================================================================================
\subsection{Proof of Theorem \ref{thm_consistency_selection_correct_spectification}}
We start by supposing that $\EE[\varepsilon_i(\theta_\circ) x_i] = 0$. In this case, $\theta_* = \theta_\circ$, $\lambda_*(\theta_*) = \lambda_*(\theta_\circ) = 0$ and by Theorems \ref{lem_likelihood_ET_approx} and \ref{lem_likelihood_ET_approx_extended}:
\begin{multline}
  \log \wh{q}(w_{1:n}|\theta_\circ,\mathcal{M}_b) - \log \wh{q}(w_{1:n}|\psi_\circ,\mathcal{M}_e)\\
  = -n\log n - \frac{1}{2}\mathbb{G}_n\left[g_b(w_i,\theta_\circ)'\right] \Omega_\circ^{-1} \mathbb{G}_n\left[g_b(w_i,\theta_\circ)\right] + n\log n\\
  + \frac{1}{2}\mathbb{G}_n\left[g_b(w_i,\theta_\circ)'\right] \Omega_\circ^{-1} \mathbb{G}_n\left[g_b(w_i,\theta_\circ)\right] + o_p(1) = o_p(1).\label{eq:12:proof:main}
\end{multline}
Let $\pi_{h_{\theta}}^n(\cdot|w_{1:n},\mathcal{M}_b)$ and $\pi_{h_{\psi}}^n(\cdot|w_{1:n},\mathcal{M}_e)$ denote the posterior density of $h_{\theta}$ and $h_{\psi}$, respectively. By Corollary \ref{Cor_2_1} below (which is valid if $\EE[\varepsilon_i(\theta_\circ) x_i] = 0$ holds)
\begin{multline}
  \left.\log \pi_{h_{\theta}}^n(\sqrt{n}(\theta - \theta_{\circ})|w_{1:n},\mathcal{M}_b)\right|_{\theta = \theta_\circ} = -\frac{p}{2} \log(2\pi) + \frac{1}{2}\log |V_{\theta_\circ}^{-1}|\\
  - \frac{1}{2}\mathbb{G}_n\left[\varepsilon_i(\theta_\circ)\wtl{w}_i'\right] \Omega_\circ^{-1} \EE[\wtl{w}_i\wtl{w}_{1,i}']V_{\theta_\circ}\EE[\wtl w_{1,i}\wtl{w}_i'] \Omega_\circ^{-1} \mathbb{G}_n\left[\varepsilon_i(\theta_\circ)\wtl{w}_i\right] + o_p(1), \label{eq:11:proof:main}
\end{multline}
and by Theorem \ref{thm_BvM:extended} below
\begin{multline}
  \left.\log \pi_{h_{\psi}}^n(\sqrt{n}(\psi - \psi_{\circ})|w_{1:n},\mathcal{M}_e)\right|_{\psi = \psi_\circ} = -\frac{(p + d_x)}{2} \log(2\pi) + \frac{1}{2}\log |V_{\psi_\circ}^{-1}|\\
  - \frac{1}{2}\mathbb{G}_n\left[\varepsilon_i(\theta_\circ)\wtl{w}_i'\right] \Omega_\circ^{-1} \left[\frac{dg_e(w_{i},\psi_{\circ})'}{d\psi }\right]V_{\psi_\circ}\left[\frac{dg_e(w_{i},\psi_{\circ})}{d\psi'}\right] \Omega_\circ^{-1} \mathbb{G}_n\left[\varepsilon_i(\theta_\circ)\wtl{w}_i\right] + o_p(1),\label{eq:10:proof:main}
\end{multline}
\noindent where $V_{\theta_\circ}$ and $V_{\psi_\circ}$ are defined in Corollary \ref{Cor_2_1} and Theorem \ref{thm_BvM:extended} below. Hence, by replacing \eqref{eq:10:proof:main}, \eqref{eq:11:proof:main} and \eqref{eq:12:proof:main} in $\log m(w_{1:n}|\mathcal{M}_b) - \log m(w_{1:n}|\mathcal{M}_e)$ by using the expressions for the log-marginal likelihoods given in \eqref{log_ML_base_initial}-\eqref{log_ML_extended_initial} with $\theta_*$ replaced by $\theta_\circ$, we obtain:
\begin{multline}
  P\left(\log m(w_{1:n}|\mathcal{M}_b) > \log m(w_{1:n}|\mathcal{M}_e)\right) = P\Big(\log \pi(\theta_\circ|\mathcal{M}_b) + \frac{p}{2} \log(2\pi) - \frac{1}{2}\log |V_{\theta_\circ}^{-1}|\\
  + \frac{1}{2}\mathbb{G}_n\left[\varepsilon_i(\theta_\circ)\wtl{w}_i'\right] \Omega_\circ^{-1} \EE[\wtl{w}_i\wtl{w}_{1,i}']V_{\theta_\circ}\EE[\wtl w_{1,i}\wtl{w}_i'] \Omega_\circ^{-1} \mathbb{G}_n\left[\varepsilon_i(\theta_\circ)\wtl{w}_i\right] - \frac{p}{2}\log(n)\\
  > \log \pi(\psi_\circ|\mathcal{M}_e) + \frac{(p + d_x)}{2} \log(2\pi) - \frac{1}{2}\log |V_{\psi_\circ}^{-1}|\\
  + \frac{1}{2}\mathbb{G}_n\left[\varepsilon_i(\theta_\circ)\wtl{w}_i'\right] \Omega_\circ^{-1} \left[\frac{dg(w_{i},\psi_{\circ})'}{d\psi }\right]V_{\psi_\circ}\left[\frac{dg(w_{i},\psi_{\circ})}{d\psi'}\right] \Omega_\circ^{-1} \mathbb{G}_n\left[\varepsilon_i(\theta_\circ)\wtl{w}_i\right] - \frac{p+d_x}{2}\log(n) + o_p(1)\Big).\label{eq:3_correct_specification_comparison}
\end{multline}
Because $\mathbb{G}_n\left[\varepsilon_i(\theta_\circ)\wtl{w}_i\right] = \mathcal{O}_p(1)$, $|V_{\theta_\circ}^{-1}| = \mathcal{O}(1)$ and $|V_{\psi_\circ}^{-1}| = \mathcal{O}(1)$ (since $V_{\theta_\circ}$ and $V_{\psi_\circ}$ are positive definite under Assumption \ref{Ass_2_NS}), then we can factorize $\log(n)$ in \eqref{eq:3_correct_specification_comparison} and get:
\begin{multline}
  P\left(\log m(w_{1:n}|\mathcal{M}_b) > \log m(w_{1:n}|\mathcal{M}_e)\right) = \\%P\Big(\log \pi(\theta_\circ|\mathcal{M}_b) - \frac{1}{2}\log |V_{\theta_\circ}^{-1}|\\
  %+ \frac{1}{2}\mathbb{G}_n\left[\varepsilon_i(\theta_\circ)\wtl{w}_i'\right] \Omega_\circ^{-1} \EE[\wtl{w}_i\wtl{w}_{1,i}']V_{\theta_\circ}\EE[\wtl w_{1,i}\wtl{w}_i'] \Omega_\circ^{-1} \mathbb{G}_n\left[\varepsilon_i(\theta_\circ)\wtl{w}_i\right]\\
  %> \log \pi(\psi_\circ|\mathcal{M}_e) + \frac{d_x}{2} \log(2\pi) - \frac{1}{2}\log |V_{\psi_\circ}^{-1}|\\
  %+ \frac{1}{2}\mathbb{G}_n\left[\varepsilon_i(\theta_\circ)\wtl{w}_i'\right] \Omega_\circ^{-1} \left[\frac{dg(w_{i},\psi_{\circ})'}{d\psi }\right]V_{\psi_\circ}\left[\frac{dg(w_{i},\psi_{\circ})}{d\psi'}\right] \Omega_\circ^{-1} \mathbb{G}_n\left[\varepsilon_i(\theta_\circ)\wtl{w}_i\right] - \frac{d_x}{2}\log(n) + o_p(1)\Big)\\
  = P\Big(0 > \log(n) \Big[\frac{1}{\log(n)}\log \frac{\pi(\psi_\circ|\mathcal{M}_e)}{\pi(\theta_\circ|\mathcal{M}_b)} + \frac{d_x\log(2\pi)}{2\log(n)}  - \frac{1}{2\log(n)}\log \frac{|V_{\theta_\circ}|}{|V_{\psi_\circ}|} - \frac{d_x}{2} \\
  + \frac{1}{2\log(n)}\mathbb{G}_n\left[\varepsilon_i(\theta_\circ)\wtl{w}_i'\right] \Omega_\circ^{-1} \left(\left[\frac{dg(w_{i},\psi_{\circ})'}{d\psi }\right]V_{\psi_\circ}\left[\frac{dg(w_{i},\psi_{\circ})}{d\psi'}\right] - \EE[\wtl{w}_i\wtl{w}_{1,i}']V_{\theta_\circ}\EE[\wtl w_{1,i}\wtl{w}_i'] \right) \\
  \times \Omega_\circ^{-1} \mathbb{G}_n\left[\varepsilon_i(\theta_\circ)\wtl{w}_i\right]\Big] + o_p(1)\Big) = P\left(0> \log(n)\left[o_p(1) - \frac{d_x}{2}\right] + o_p(1)\right)\rightarrow 1\label{eq:4_correct_specification_comparison}
\end{multline}
\noindent as $n\rightarrow \infty$. This proves the first implication.\\
%%%%
\indent We now prove the reverse implication. Suppose that $P\left(\log m(w_{1:n}|\mathcal{M}_b) > \log m(w_{1:n}|\mathcal{M}_e)\right) \rightarrow 1$. By \eqref{log_ML_base_initial}-\eqref{log_ML_extended_initial}:
\begin{multline}
  P\left(\log m(w_{1:n}|\mathcal{M}_b) > \log m(w_{1:n}|\mathcal{M}_e)\right) = P\Big( \log \pi(\theta_*|\mathcal{M}_b) + \log \wh{q}(w_{1:n}|\theta_*,\mathcal{M}_b)\\
   - \log \pi_{h_{\theta}}^n(0|w_{1:n},\mathcal{M}_b) - \frac{p}{2}\log(n) > \log \pi(\psi_\circ|\mathcal{M}_e) + \log \wh{q}(w_{1:n}|\psi_\circ,\mathcal{M}_e)\\
   - \log \pi_{h_{\psi}}^n(0|w_{1:n},\mathcal{M}_e) - \frac{p+d_x}{2}\log(n)\Big). \label{eq:13:proof:main}
\end{multline}
By using Theorems \ref{lem_likelihood_ET_approx} and \ref{lem_likelihood_ET_approx_extended}, we get:
\begin{multline}
  \log \wh{q}(w_{1:n}|\theta_*,\mathcal{M}_b) - \log \wh{q}(w_{1:n}|\psi_\circ,\mathcal{M}_e) = - \mathcal{A}_n'\Omega_*^{\dagger}(\theta_*)^{-1}\mathcal{B}_n\\
  + \sum_{i=1}^n \log\left(\frac{e^{\lambda_*(\theta_*)' g_b(w_i,\theta_*)}}{\EE_n[e^{\lambda_*(\theta_*)' g_b(w_j,\theta_*)}]}\right) + n(\wh\lambda(\theta_*) - \lambda_*(\theta_*))'\EE[g_b(w_i,\theta_*)]\\
  + \frac{1}{2}\mathbb{G}_n\left[\tau_i^{\dagger}(\lambda_*,\theta_*) g_b(w_i,\theta_*)'\right] \Omega_*^{\dagger}(\theta_*)^{-1} \mathbb{G}_n\left[\tau_i^{\dagger}(\lambda_*,\theta_*) g_b(w_i,\theta_*)\right]\\
  + \frac{1}{2}\mathbb{G}_n\left[g_e(w_i,\psi_\circ)'\right] \Omega_{\psi_\circ}^{-1} \mathbb{G}_n\left[g_e(w_i,\psi_\circ)\right] + o_p(1),\label{eq:12:proof:main:2}
\end{multline}
\noindent where $\mathcal{A}_n := \mathbb{G}_n\left[\tau_i^{\dagger}(\lambda_*,\theta_*) g_b(w_i,\theta_*)'\right] \xrightarrow{d} \mathcal{N}(0,\Omega_*^{\dagger}(\theta_*))$, $\mathcal{B}_n := \mathbb{G}_n\left[g_b(w_i,\theta_*)'\right] \xrightarrow{d} \mathcal{N}(0,\EE[\varepsilon_i(\theta_\circ)]\wtl{w}_i\wtl{w}_i')$ and $\mathbb{G}_n\left[g_e(w_i,\psi_\circ)'\right] \Omega_{\psi_\circ}^{-1} \mathbb{G}_n\left[g_e(w_i,\psi_\circ)\right]\xrightarrow{d} \chi_{d}^2$ and so they are bounded in probability. Moreover, by the Law of Large Numbers
$$\left|\frac{1}{n}\sum_{i=1}^n \log\left(\frac{e^{\lambda_*(\theta_*)'g_b(w_i,\theta_*)}}{\EE_n[e^{\lambda_*(\theta_*)' g_b(w_j,\theta_*)}]}\right) - \mathbf{E}\left[\log(dQ_b^*(\theta_*)/dP)\right]\right| = \mathcal{O}_p(1/\sqrt{n}),$$
\noindent where $\mathbf{E}\left[\log(dQ_b^*(\theta_*)/dP)\right] = - \mathrm{KL}(P||Q_b^*(\theta_*))$, and by Lemma \ref{lem:1} below,
$$\sqrt{n}(\wh\lambda(\theta_*) - \lambda(\theta_*))'\EE[g_b(w_i,\theta_*)] = -\mathbb{G}_n[\tau_i^{\dagger}(\lambda_*,\theta_*)\varepsilon_i(\theta_*)\wtl{w}_i']\Omega_*^{\dagger}(\theta_*)^{-1}\EE[g_b(w_i,\theta_*)] + o_p(1).$$
Therefore,
\begin{multline}
  \log \wh{q}(w_{1:n}|\theta_*,\mathcal{M}_b) - \log \wh{q}(w_{1:n}|\psi_\circ,\mathcal{M}_e) \\
  = \mathcal{O}_p(1) + n\left(\mathcal{O}_p(1/\sqrt{n}) - \mathrm{KL}(P||Q_b^*(\theta_*))\right)  + \sqrt{n}\mathcal{O}_p(1). \label{eq:12:proof:main:3}
\end{multline}

By replacing \eqref{eq:12:proof:main:3} in \eqref{eq:13:proof:main}, and by using Theorems \ref{thm_BvM_misspecified} and \ref{thm_BvM:extended} below to show that $\log \pi_{h_{\theta}}^n(0|w_{1:n},\mathcal{M}_b) = \mathcal{O}_p(1)$ and $\log \pi_{h_{\psi}}^n(0|w_{1:n},\mathcal{M}_e) = \mathcal{O}_p(1)$, the expression in \eqref{eq:13:proof:main} is equal to:
\begin{multline}
  P\left(\log m(w_{1:n}|\mathcal{M}_b) > \log m(w_{1:n}|\mathcal{M}_e)\right)\\
   = P\Big( \mathcal{O}_p(1)\sqrt{n} - n\mathrm{KL}(P||Q_b^*(\theta_*)) > \mathcal{O}_p(1) - \frac{d_x}{2}\log(n)\Big), \label{eq:14:proof:main}
\end{multline}
\noindent where $\mathcal{O}_p(1)\sqrt{n} - n\mathrm{KL}(P||Q_b^*(\theta_*))$ in the left hand side converges to $-\infty$ if $\mathrm{KL}(P||Q_b^*(\theta_*)) >0$ (since the term in $n$ is diverging faster than the term in $\sqrt{n}$) while the term on the right hand side also converges towards $-\infty$. The inequality is then satisfied with probability approaching $1$ only if $\mathrm{KL}(P||Q_b^*(\theta_*)) = 0$. This is equivalent to have $\EE[\varepsilon_i(\theta_\circ) x_i] = 0$ (by the first part of the proof of Theorem \ref{thm_consistency_selection_misspecification_2}) and we have proved the second part of the `if and only if' statement.

\begin{flushright}
$\square$
\end{flushright}

\bibliography{AnnaBib}

%%%%%%%%%%%%%%%%%%%%%%%%%%%%%%%%%%%%%%%%%%%%%%%%%%%%%%%%%%%%
%% Online Appendix
%%%%%%%%%%%%%%%%%%%%%%%%%%%%%%%%%%%%%%%%%%%%%%%%%%%%%%%%%%%%
\newpage

\begin{center}
{\Large \textbf{Online Appendix for}}\\[6pt]
{\Large \textbf{Testing for Endogeneity: A Moment-Based Bayesian Approach}}
\end{center}
\bigskip

\section{Comparison with GMM-based criteria}

\label{App:comparisonGMM}

There are frequentist approaches to the model (or moment) selection that
can be applied in our context. \cite{Andrews1999} develops a class of moment
selection criteria (MSC). Below are some popular criteria that fall into the
class:
\begin{equation}
\begin{split}
\text{GMM-BIC} &= J_{n}(c) -(|c|-p) \ln n \\
\text{GMM-AIC} &= J_{n}(c) - 2(|c|-p) \\
\text{GMM-HQIC} &= J_{n}(c) -2.01 (|c|-p) \ln \ln n,
\end{split}%
\end{equation}
where $c$ is a moment selection vector, $|c|$ is the number of moment
conditions selected by $c$, $p$ is the number of parameters to be estimated, and
$J_{n}(c)$ is the $J$ test statistic for overidentifying restrictions
constructed using $c$ with the optimal weighting matrix. Similar to the
traditional BIC, these criteria penalize model complexity based on the
number of parameters and the number of restrictions imposed. The model
complexity increases when the number of parameters increases or the number
of restrictions decreases. This idea was extended by \cite%
{HongPrestonShum2003} to GEL estimation.

We revisit the simulation exercise, originally presented in the main
text (Table \ref{tab:tab1_app}), and now report results for other
frequentist methods: GMM-BIC, GMM-AIC, and GMM-HQIC. From the table, we make the following observations. First, all methods exhibit model selection consistency,
meaning the probability of selecting the true model approaches one as the
number of observations increases. Second, our approach has stronger
discriminatory power when $\rho$ is close to zero compared to GMM-BIC.
Third, GMM-AIC and GMM-HQIC select the right model more often when $\rho$ is
not zero (no endogeneity). However, they seem to over-select the model with
endogeneity when there is no presence of endogeneity. In summary, under the
data-generating process considered in this example, our BETEL-based model
selection performs better than other alternatives, especially in a finite
sample.

\begin{table}[htbp]
\caption{Table \ref{tab:sim_MegreaterthanMb} with other frequentist approaches. Model selection
frequencies from 100 replications of data simulated from the design in
Section \ref{sec:numerical_illustration}. For each combination of $n$ and $\mathrm{Cov}(\protect\varepsilon %
,u)=\protect\rho $, the entries give the number of times in 100 replications
of the data that the log-marginal likelihood of $\mathcal{M}_{e}$ exceeds
the log-marginal likelihood of $\mathcal{M}_{b}$. The numbers for BETEL are
slightly different from those reported in the main text because they are
based on different sets of simulated data, i.e., the random number seed is
different. }
\label{tab:tab1_app}\centering
\begin{tabular}{rccccccccccc}
    \toprule
    \rowcolor[rgb]{ .949,  .949,  .949} \multicolumn{1}{l}{BETEL} & -0.5  & -0.4  & -0.3  & -0.2  & -0.1  & 0     & 0.1   & 0.2   & 0.3   & 0.4   & 0.5 \\
    \midrule
    250   & 99    & 96    & 82    & 48    & 12    & 2     & 18    & 54    & 93    & 100   & 100 \\
    500   & 100   & 100   & 98    & 76    & 17    & 1     & 29    & 87    & 99    & 100   & 100 \\
    1000  & 100   & 100   & 100   & 96    & 46    & 1     & 46    & 100   & 100   & 100   & 100 \\
    2000  & 100   & 100   & 100   & 100   & 80    & 1     & 70    & 100   & 100   & 100   & 100 \\
    \midrule
    \rowcolor[rgb]{ .949,  .949,  .949} \multicolumn{1}{l}{GMM-BIC} & -0.5  & -0.4  & -0.3  & -0.2  & -0.1  & 0     & 0.1   & 0.2   & 0.3   & 0.4   & 0.5 \\
    \midrule
    250   & 100   & 97    & 77    & 35    & 7     & 3     & 11    & 40    & 84    & 99    & 100 \\
    500   & 100   & 100   & 96    & 72    & 8     & 1     & 16    & 74    & 99    & 100   & 100 \\
    1000  & 100   & 100   & 100   & 92    & 29    & 1     & 25    & 99    & 100   & 100   & 100 \\
    2000  & 100   & 100   & 100   & 99    & 63    & 1     & 47    & 100   & 100   & 100   & 100 \\
    \midrule
    \rowcolor[rgb]{ .949,  .949,  .949} \multicolumn{1}{l}{GMM-AIC} & -0.5  & -0.4  & -0.3  & -0.2  & -0.1  & 0     & 0.1   & 0.2   & 0.3   & 0.4   & 0.5 \\
    \midrule
    250   & 100   & 100   & 96    & 74    & 28    & 15    & 37    & 79    & 98    & 100   & 100 \\
    500   & 100   & 100   & 100   & 94    & 46    & 11    & 60    & 98    & 100   & 100   & 100 \\
    1000  & 100   & 100   & 100   & 99    & 71    & 11    & 76    & 100   & 100   & 100   & 100 \\
    2000  & 100   & 100   & 100   & 100   & 95    & 12    & 94    & 100   & 100   & 100   & 100 \\
    \midrule
    \rowcolor[rgb]{ .949,  .949,  .949} \multicolumn{1}{l}{GMM-HQIC} & -0.5  & -0.4  & -0.3  & -0.2  & -0.1  & 0     & 0.1   & 0.2   & 0.3   & 0.4   & 0.5 \\
    \midrule
    250   & 100   & 99    & 85    & 54    & 17    & 5     & 20    & 62    & 94    & 99    & 100 \\
    500   & 100   & 100   & 100   & 82    & 23    & 3     & 31    & 89    & 100   & 100   & 100 \\
    1000  & 100   & 100   & 100   & 98    & 54    & 2     & 56    & 100   & 100   & 100   & 100 \\
    2000  & 100   & 100   & 100   & 100   & 79    & 1     & 73    & 100   & 100   & 100   & 100 \\
    \bottomrule
    \end{tabular}
\end{table}

\newpage

\section{Additional Monte Carlo illustrations}\label{sec:additional-mc}

\subsection{Monte Carlo illustration 1: two potentially endogenous regressors}\label{sec:mc-illustration-1}
We now illustrate the above logic using a Monte Carlo design with two potentially endogenous regressors, where only \(x_1\) is endogenous in the DGP. Data are generated as follows:
\[
y_i = \beta_1 x_{1i} + \beta_2 x_{2i} + \gamma_0 + \gamma_1 z_{1i} + \varepsilon_i,
\]
with \((\beta_1,\beta_2,\gamma_0,\gamma_1)=(1.0,0.8,1.0,0.6)\), sample size \(n=500\), and true endogeneity parameters
\((v_1,v_2)=\left(\operatorname{Cov}(x_{1i},\varepsilon_i),\operatorname{Cov}(x_{2i},\varepsilon_i)\right)=(0.5,0)\).

We consider four models that differ by which of \(x_1\) and \(x_2\) are treated as exogenous or endogenous; Table~\ref{tab:mc1} lists them and their log marginal likelihoods. The entry \(v_i\) corresponds to the moment condition \(\mathbf{E}[\varepsilon(\theta)x_i]=0\): setting \(v_i=0\) imposes that moment (so \(x_i\) is exogenous), while leaving \(v_i\) free relaxes it (so \(x_i\) is endogenous). All four models use the same moment vector based on
\(Z_i=(x_{1i},x_{2i},1,z_{1i},z_{2ai},z_{2bi})'\), with restrictions imposed only through which components of \(v\) are set to zero.

More explicitly, let \(\theta=(\beta_1,\beta_2,\gamma_0,\gamma_1)\) and define the regression residual
\[
\varepsilon_i(\theta) = y_i - \beta_1 x_{1i} - \beta_2 x_{2i} - \gamma_0 - \gamma_1 z_{1i}.
\]
The moment vector is
\[
\mathbf{E}\!\left[\varepsilon_i(\theta) Z_i\right]
=
\begin{pmatrix}
\mathbf{E}[\varepsilon_i(\theta)x_{1i}]\\
\mathbf{E}[\varepsilon_i(\theta)x_{2i}]\\
\mathbf{E}[\varepsilon_i(\theta)]\\
\mathbf{E}[\varepsilon_i(\theta)z_{1i}]\\
\mathbf{E}[\varepsilon_i(\theta)z_{2ai}]\\
\mathbf{E}[\varepsilon_i(\theta)z_{2bi}]
\end{pmatrix}
=
\begin{pmatrix}
v_1\\
v_2\\
0\\
0\\
0\\
0
\end{pmatrix}.
\]
Model \(\mathcal{M}_b\) imposes \(v_1=v_2=0\) (both orthogonality conditions \(\mathbf{E}[\varepsilon_i(\theta)x_{1i}]=0\) and \(\mathbf{E}[\varepsilon_i(\theta)x_{2i}]=0\)), model \(\mathcal{M}_{e_2}\) imposes only \(v_2=0\) (relaxing the \(x_1\) moment), model \(\mathcal{M}_{e_1}\) imposes only \(v_1=0\) (relaxing the \(x_2\) moment), and model \(\mathcal{M}_e\) leaves \((v_1,v_2)\) unrestricted (relaxing both moments).

\begin{table}[H]
\centering
\caption{Monte Carlo 1: log marginal likelihoods}
\label{tab:mc1}
\begin{tabular}{llcc}
\hline
Model & \(v\) set to zero & Endogeneity status & \(\log p(y\mid\mathcal{M})\) \\
\hline
\(\mathcal{M}_b\) & \(v_1=0\), \(v_2=0\) & \(x_1\) exo, \(x_2\) exo & \(-3147.5009\) \\
\(\mathcal{M}_{e_2}\) & \(v_2=0\) & \(x_1\) endo, \(x_2\) exo & \(-3121.9751\) \\
\(\mathcal{M}_{e_1}\) & \(v_1=0\) & \(x_1\) exo, \(x_2\) endo & \(-3147.6291\) \\
\(\mathcal{M}_e\) & none & \(x_1\) endo, \(x_2\) endo & \(-3122.7344\) \\
\hline
\end{tabular}
\end{table}

The largest marginal likelihood is attained by \(\mathcal{M}_{e_2}\), which matches the true DGP (\(x_1\) endogenous, \(x_2\) exogenous). The model \(\mathcal{M}_e\) is second-best but slightly lower, while \(\mathcal{M}_b\) and \(\mathcal{M}_{e_1}\) are much less supported. The ranking \(\mathcal{M}_{e_2} \succ \mathcal{M}_e \gg \{\mathcal{M}_b,\mathcal{M}_{e_1}\}\) indicates strong evidence that \(x_1\) should be treated as endogenous, and weaker support for treating \(x_2\) as endogenous.

\subsection{Monte Carlo illustration 2: endogeneity and model specification}\label{sec:mc-illustration-2}
Next, we consider a Monte Carlo design to jointly test for endogeneity and functional form.
Data are generated as follows
\[
y_i = \gamma_0 + \beta_1 x_i + \beta_2 x_i^2 + \gamma_1 z_{1i} + \varepsilon_i,
\]
with \((\gamma_0,\beta_1,\beta_2,\gamma_1)=(0.5,1.0,1.0,0.8)\), \(n=500\), and \(\operatorname{Cov}(x_i,\varepsilon_i)=0.5\).
We consider four models that vary in functional form (linear vs.\ quadratic) and in whether the orthogonality conditions involving \(x\) and \(x^2\) are imposed or relaxed; Table~\ref{tab:mc2} lists them and their log marginal likelihoods. The linear specification imposes \(\beta_2=0\); the quadratic leaves \(\beta_2\) free. All models use the common instrument vector
\(Z_i=(x_i,x_i^2,1,z_{1i},z_{2i},z_{2i}^2)'\), so differences in marginal likelihoods reflect only
(i) linear versus quadratic (i.e., whether \(x^2\) enters the regression), and (ii) which orthogonality restrictions are imposed.

Let \(\varepsilon_i^L(\theta_L)=y_i-\gamma_0-\beta_1 x_i-\gamma_1 z_{1i}\) denote the linear-specification residual and \(\varepsilon_i^Q(\theta_Q)=y_i-\gamma_0-\beta_1 x_i-\beta_2 x_i^2-\gamma_1 z_{1i}\) the quadratic-specification residual. The corresponding moment vector is always built from \(Z_i\). In exogenous models, all moments are set to zero:
\(\mathbf{E}[\varepsilon_i(\theta)Z_i]=0\). In endogenous models, the moments involving \(x\) and \(x^2\) are relaxed (regardless of whether the regression includes \(x^2\)):
\[
\mathbf{E}\!\left[\varepsilon_i(\theta) Z_i\right]
=
\begin{pmatrix}
v_x\\
v_{x^2}\\
0\\
0\\
0\\
0
\end{pmatrix},
\]
where \(v_x=0\) (resp.\ \(v_{x^2}=0\)) imposes the moment \(\mathbf{E}[\varepsilon_i(\theta)x_i]=0\) (resp.\ \(\mathbf{E}[\varepsilon_i(\theta)x_i^2]=0\)).
Thus, both \(\mathcal{M}_{L,e}\) and \(\mathcal{M}_{Q,e}\) relax the \(x\) and \(x^2\) orthogonality conditions (allowing \(v_x\neq 0\) and \(v_{x^2}\neq 0\)).

\begin{table}[H]
\centering
\caption{Monte Carlo 2: log marginal likelihoods}
\label{tab:mc2}
\begin{tabular}{lllcc}
\hline
Model & Relaxed moments & \(\beta_2\) set to zero & Specification & \(\log p(y\mid\mathcal{M})\) \\
\hline
\(\mathcal{M}_{L,b}\) & none & yes & linear, exogenous & \(-7426.6056\) \\
\(\mathcal{M}_{Q,b}\) & none & no & quadratic, exogenous & \(-3146.1535\) \\
\(\mathcal{M}_{L,e}\) & \(x,x^2\) & yes & linear, endogenous & \(-3138.3612\) \\
\(\mathcal{M}_{Q,e}\) & \(x,x^2\) & no & quadratic, endogenous & \(-3123.3361\) \\
\hline
\end{tabular}
\end{table}

Again, the highest marginal likelihood is assigned to the true model \(\mathcal{M}_{Q,e}\). The next-best model is \(\mathcal{M}_{L,e}\), followed by \(\mathcal{M}_{Q,b}\), with \(\mathcal{M}_{L,b}\) far behind. This ranking shows that accounting for endogeneity is crucial in both specifications, and that allowing for the quadratic term further improves fit.

\section{Computation of the ETEL Posterior and Marginal Likelihood}\label{sec:mcmc-ml-computation}

For the base model $\mathcal{M}_b$, the parameter is $\theta\in\Theta$. For the extended model $\mathcal{M}_e$, the parameter is $\psi:=(\theta,v)\in\Psi$. Let $H_n$ denote the ETEL feasibility set for the model at hand. The ETEL posterior for the base model is
\begin{equation}
\pi^n(\theta\mid w_{1:n},\mathcal{M}_b)
\propto
\pi_b(\theta)\,\wh q(w_{1:n}\mid\theta,\mathcal{M}_b)\,I[\theta\in H_n],
\label{eq:posterior-truncated-base}
\end{equation}
where $\int \pi_{b}(\theta) d \theta = 1$. The extended-model posterior $\pi^n(\psi\mid w_{1:n},\mathcal{M}_e)$ is defined analogously with $\psi=(\theta,v)$.

\subsection{Tailored Single-Block Metropolis--Hastings}\label{sec:tailored-mh}

We sample from \eqref{eq:posterior-truncated-base} using a tailored one-block Metropolis--Hastings (M-H) algorithm:
\begin{enumerate}
\item Compute the posterior mode $\wh\theta$ by numerical optimization of
\[
\log \pi_b(\theta)+\log \wh q(w_{1:n}\mid\theta,\mathcal{M}_b),
\]
subject to $\theta\in H_n$.
\item Compute the inverse Hessian at the mode,
\[
\wh V
:=
\left[-\nabla_{\theta}^{2}\log \pi^n(\theta\mid w_{1:n},\mathcal{M}_b)\right]^{-1}_{\theta=\wh\theta},
\]
and use it as the scale matrix of the proposal.
\item Use an independence proposal with multivariate Student-$t$ density
\[
q(\theta')
=
 t_{\kappa}\!\left(\theta';\,\wh\theta,\wh V\right).
\]
\item If $\theta'\notin H_n$, reject immediately.
\item If $\theta'\in H_n$, accept with probability
\[
\alpha(\theta,\theta')
=
\min\!\left\{
1,\,
\frac{\wh q(w_{1:n}\mid\theta',\mathcal{M}_b)\,\pi_b(\theta')\,q(\theta)}
{\wh q(w_{1:n}\mid\theta,\mathcal{M}_b)\,\pi_b(\theta)\,q(\theta')}
\right\}.
\]
\end{enumerate}

The feasibility indicator in \eqref{eq:posterior-truncated-base} is enforced by direct rejection of proposals outside $H_n$. Equivalently, Step~4 can be implemented implicitly by defining the ETEL likelihood to be zero outside $H_n$, so that the acceptance probability vanishes for $\theta'\notin H_n$. The $\psi=(\theta,v)$ case is defined in the same way, with $\theta$ replaced by $\psi$, $\pi_b$ by $\pi_e$, and $\mathcal{M}_b$ by $\mathcal{M}_e$.

\subsection{Log Marginal Likelihood Calculation}\label{sec:ml-computation}

The ETEL marginal likelihood for the base model is
\begin{equation}
 m(w_{1:n}\mid\mathcal{M}_b)
=
\int_{H_n}
\wh q(w_{1:n}\mid\theta,\mathcal{M}_b)\,
\pi_b(\theta)\,
 d\theta
=
\int_{\Theta}
\wh q(w_{1:n}\mid\theta,\mathcal{M}_b)\,
\pi_b(\theta)\,
I[\theta\in H_n]\, d\theta.
\label{eq:ml-with-indicator-base}
\end{equation}
For any fixed $\theta^*\in H_n$, Chib's identity gives
\begin{equation}
\log m(w_{1:n}\mid\mathcal{M}_b)
=
\log \wh q(w_{1:n}\mid\theta^*,\mathcal{M}_b)
+\log \pi_b(\theta^*)
-\log \pi^n(\theta^*\mid w_{1:n},\mathcal{M}_b).
\label{eq:log-ml-chib-base}
\end{equation}
We calculate the marginal likelihood based on this identity using the output from the M-H sampler and the posterior mode as $\theta^{*}$ \citep{ChibJeliazkov2001}.

\subsection{Alternative Prior Specification}

Unlike in the main text, some researchers may prefer to work directly with the prior restricted to the feasible set. This specification can be achieved by attaching the feasibility indicator function to the prior. Although this approach makes the prior explicit in the sense that we place prior mass on $\theta$ only where the likelihood is well-defined, one has to accept that the support of the prior for $\theta$ is sample-dependent through $H_n$.

For this prior specification, the MCMC algorithm is largely unchanged, except for the marginal likelihood calculation, which is modified as follows. Define the prior mass on the feasibility set and the normalized prior by
\begin{equation}
\begin{split}
 p_{H_n}(\mathcal{M}_b)
& :=
\int_{\Theta}\pi_b(\theta)\,I[\theta\in H_n] \, d\theta
\label{eq:ph-def-base}
\\
\pi_{b,H_n}(\theta)
&:=
\frac{\pi_b(\theta)\,I[\theta\in H_n]}{p_{H_n}(\mathcal{M}_b)}. \\
\end{split}
\end{equation}
Under this prior specification, the marginal likelihood is
\[
m_{H_n}(w_{1:n}\mid\mathcal{M}_b)
:=
\int_{H_n} \wh q(w_{1:n}\mid\theta,\mathcal{M}_b)\,\pi_{b,H_n}(\theta)\,d\theta,
\]
so that $m(w_{1:n}\mid\mathcal{M}_b) = p_{H_n}(\mathcal{M}_b)\,m_{H_n}(w_{1:n}\mid\mathcal{M}_b)$ and
\begin{equation}
\log m_{H_n}(w_{1:n}\mid\mathcal{M}_b)
=
\log m(w_{1:n}\mid\mathcal{M}_b)-\log p_H(\mathcal{M}_b).
\label{eq:log-ml-decomp-base}
\end{equation}
Thus $p_{H_n}(\mathcal{M}_b)$ appears only when the prior is normalized on $H_n$; it is a constant shift relative to the unrestricted-prior definition in \eqref{eq:log-ml-chib-base}.

We can estimate $p_{H_n}(\mathcal{M}_b)$ by prior simulation: draw $\{\tilde\theta^{(b)}\}_{b=1}^B$ i.i.d.\ from $\pi_b(\theta)$ and compute
\[
\wh p_{H_n}(\mathcal{M}_b)
:=
\frac{1}{B}\sum_{b=1}^B I[\tilde\theta^{(b)}\in H_n].
\]

\section{Proofs of the main results}\label{sec:proofs-main-results}

The following notation will be used in the proofs of this section. Further notation will be introduced in section \ref{ss:notation:Appendix_C} and will be used in the proofs there. When we omit $y_i$ from the vector of the $i$-th observation we use the notation $\wtl w_{i}:=(x_i',z_i')'$, and when in addition we omit $z_{2,i}$ we use the notation $\wtl w_{1,i}:=(x_i',z_{1,i}')'$. We use the notation $\EE_n[\cdot]:=\frac{1}{n}\sum_{i=1}^n[\cdot]$ for the empirical mean. For a probability $Q$ we use the notation $\EE^{Q}[\cdot]$ to denote the expectation with respect to $Q$ and $\mathbb{V}ar_Q$ the variance with respect to $Q$. For the true distribution $P$: $\EE[\cdot] := \EE^P[\cdot]$. We use standard notation in empirical process theory: $\mathbb{P}_n: = \EE_n[\delta_{w_i}]$ where $\delta_{x}$ is the Dirac measure at $x$, and $\mathbb{G}_n g:= \sqrt{n}(\mathbb{P}_n f - \EE f)$ for every function $f$.\\
\indent For a function $\lambda(\theta)$ of $\theta$, define $\tau_i(\lambda,\theta):= \frac{e^{\lambda(\theta)'g_i(\theta)}}{\EE_n[ e^{\lambda(\theta)'g_j(\theta)}]}$, $\tau_i^{\dagger}(\lambda,\theta):= \frac{e^{\lambda(\theta)'g_i(\theta)}}{\EE[ e^{\lambda(\theta)'g_j(\theta)}]}$, $\tau_i^{\diamond}(\lambda,\theta) := e^{\lambda(\theta)'g_i(\theta)}$, so that $\tau_i(\wh\lambda,\theta) = n \wh{p}_i(\theta)$ and $\tau_i^{\dagger}(\lambda_*,\theta) = dQ^*(\theta)/dP$. We also use the notation: $\Omega_*^{\diamond}(\lambda, \theta) := \EE[\tau_i^{\diamond}(\lambda,\theta)\varepsilon_i(\theta)\wtl{w}_i \wtl{w}_i']$, $\Omega_*^{\diamond}(\theta) := \Omega_*^{\diamond}(\lambda_*, \theta)$, $\Omega_*^{\dagger}(\theta) := \EE[\tau_i^{\dagger}(\lambda_*,\theta)\varepsilon_i(\theta)\wtl{w}_i \wtl{w}_i'] = \EE^{Q^*(\theta)}[\varepsilon_i(\theta)\wtl{w}_i \wtl{w}_i']$, and $\Omega_\circ := \EE[\varepsilon_i(\theta_\circ)^2\wtl{w}_i \wtl{w}_i']$. Moreover, $\Omega_* \equiv \Omega_*^{\diamond}(\theta_\circ)$.

% Proofs_manuscript.tex is already included in the main paper's Appendix above
%\input{Proofs_manuscript}

%\section{Appendix}
\section{Proofs of useful auxiliary results}
\subsection{Notation}\label{ss:notation:Appendix_C}
In the following we suppress the subindexes ``b'' and ``e'' in the moment function and simply write $g(w,\theta)$ for both the base and the extended model. We denote by $g_i(\theta) := g(w_i,\theta)$ the moment function evaluated at $w_i$ and by $\wh\lambda(\theta):=\wh\lambda(w_{1:n},\theta)$ the tilting parameter. Moreover, $\varepsilon_i := \varepsilon_i(\theta_\circ)$ denotes the model error term evaluated at the true parameter. When we omit $y_i$ from the vector of the $i$-th observation we use the notation $\wtl w_{i}:=(x_i',z_i')'$, and when in addition we omit $z_{2,i}$ we use the notation $\wtl w_{1,i}:=(x_i',z_{1,i}')'$. We denote $\wtl{v}:=(v',0)$ the vector that contains the $d_x$-auxiliary parameter and a $d_z$-vector of zeros. For every $\theta\in\mathbb{R}^p$ let $h := \sqrt{n}(\theta - \theta_*)$ and denote by $\pi_h^n(\cdot|w_{1:n})$ its posterior distribution.\\
\indent We use the notation $\EE_n[\cdot]:=\frac{1}{n}\sum_{i=1}^n[\cdot]$ for the empirical mean. Moreover, $\wh g(\theta) := \EE_n[g_i(\theta)]$, $dg_i(\theta)/d\theta' = - \wtl w_{i} \wtl w_{1,i}'$. For a probability $Q$ we use the notation $\EE^{Q}[\cdot]$ to denote the expectation with respect to $Q$ and $\mathbb{V}ar_Q$ the variance with respect to $Q$. For the true distribution $P$: $\EE[\cdot] := \EE^P[\cdot]$. The log-likelihood function for one observation $w_i$ is denoted by $\ell_{n,\theta}(w_i)$:
$$\ell_{n,\theta}(w_i) := \log \wh{q}_i(\theta) = \log \frac{e^{\wh{\lambda}(\theta)'g(w_i,\theta)}}{\sum_{k=1}^n e^{\wh{\lambda}(\theta)'g_k(\theta)}}
= -\log n + \log \frac{e^{\wh \lambda(\theta)'g(w_i,\theta)}}{\frac{1}{n}\sum_{k=1}^n e^{\wh\lambda(\theta)'g_k(\theta)}}$$
so that the log-ETEL function is $\ell_{n,\theta}(w_{1:n}) = \sum_{i=1}^n \ell_{n,\theta}(w_i)  = \log \prod_{i=1}^n \wh{q}_i(\theta) = \log \wh q(w_{1:n}|\theta)$.

By replacing $\wh \lambda(\theta)$ with its true value $\lambda_*(\theta)$ we define:
$$\ell_{*,\theta}(w) := \log \frac{e^{\lambda_*(\theta)'g(w,\theta)}}{\sum_{k=1}^n e^{\lambda_*(\theta)'g_k(\theta)}} =: \log p_w^*(\theta) \qquad \textrm{ and } \qquad \ell_{*,\theta}(w_{1:n}) := \sum_{i=1}^n\ell_{*,\theta}(w_i).$$
%%%
The first (resp. second) derivative of $\theta\mapsto \ell_{n,\theta}(w_{1:n})$ evaluated at a point $\theta_1$ is denoted by $\dot{\ell}_{n,\theta_1}(w_{1:n})$ (resp. $\ddot{\ell}_{n,\theta_1}(w_{1:n})$). Moreover, for a function $\lambda(\theta)$ of $\theta$, define $\tau_i(\lambda,\theta):= \frac{e^{\lambda(\theta)'g_i(\theta)}}{\EE_n[ e^{\lambda(\theta)'g_j(\theta)}]}$, $\tau_i^{\dagger}(\lambda,\theta):= \frac{e^{\lambda(\theta)'g_i(\theta)}}{\EE[ e^{\lambda(\theta)'g_j(\theta)}]}$, $\tau_i^{\diamond}(\lambda,\theta) := e^{\lambda(\theta)'g_i(\theta)}$. So, $\tau_i(\wh\lambda,\theta) = n \wh{p}_i(\theta)$ and $\tau_i^{\dagger}(\lambda_*,\theta) = dQ^*(\theta)/dP$.

%%%
\indent We also use the notation: $\check{\Omega}(\lambda,\theta) := \EE_n[\tau_i(\lambda,\theta)\varepsilon_i(\theta)^2\wtl{w}_i\wtl{w}_i']$, $\check{\Omega}^{\diamond}(\lambda,\theta) := \EE_n[\tau_i^{\diamond}(\lambda,\theta) \varepsilon_i(\theta)^2 \wtl{w}_i \wtl{w}_i']$, $\check{\Omega}^{\dagger}(\lambda,\theta) := \EE_n[\tau_i^{\dagger}(\lambda,\theta) \varepsilon_i(\theta)^2 \wtl{w}_i \wtl{w}_i']$, $\Omega_*^{\diamond}(\lambda, \theta) := \EE[\tau_i^{\diamond}(\lambda,\theta)\varepsilon_i(\theta)^2\wtl{w}_i \wtl{w}_i']$, $\Omega_*^{\diamond}(\theta) := \Omega_*^{\diamond}(\lambda_*, \theta)$ and $\Omega_*^{\dagger}(\theta) := \EE[\tau_i^{\dagger}(\lambda_*,\theta)\varepsilon_i(\theta)^2\wtl{w}_i \wtl{w}_i'] = \EE^{Q^*(\theta)}[\varepsilon_i(\theta)^2\wtl{w}_i \wtl{w}_i']$. Moreover, $\Omega_* \equiv \Omega_*^{\diamond}(\theta_\circ)$ and $\Omega_\circ := \EE[\varepsilon_i(\theta_\circ)^2\wtl{w}_i \wtl{w}_i']$. For the extended model we use the notation $\Omega_{\psi_0} := \EE[g_e(w_i,\psi_\circ)g_e(w_i,\psi_\circ)'] = \EE[\varepsilon_i^2 \wtl w_i \wtl w_i'] - \wtl{v}_\circ\wtl{v}_\circ'$.\\
%%%
\indent Let $\|\cdot\|_2$ denotes the Euclidean norm and $\|\cdot\|_{op}$ the operator norm. We use the notation $C$ for a generic positive constant. We denote by $B_{r_n}(a)$ a closed ball centered on a vector $a$ with radius $r_n$: $B_{r_n}(a) := \{b; \|b-a\|_2\leq C r_n\}$. We denote $\Theta_{n}:=\{\Vert \theta -\theta _{\ast }\Vert \leq M_{n}/\sqrt{n}\}$. We denote by $L_2(P)$ the space of square integrable functions with respect to $P$ and by $\|\cdot\|_{P,2}$ the norm in the $L_2(P)$ space. The Total Variation norm is denoted by $\|\cdot\|_{TV}$. Finally, let CS, T, J, MVT, CMT and DCT refer to the Cauchy-Schwartz, triangular, Jensen's inequalities, Mean Value Theorem, continuous mapping theorem and dominated convergence theorem, respectively.\\
%%%
\indent For a set $\mathcal{A}\subset\mathbb{R}^m$, we denote by $int(\mathcal{A})$ its interior relative to $\mathbb{R}^m$. We use standard notation in empirical process theory: $\mathbb{P}_n: = \EE_n[\delta_{w_i}]$ where $\delta_{x}$ is the Dirac measure at $x$, and $\mathbb{G}_n g:= \sqrt{n}(\mathbb{P}_n f - \EE f)$ for every function $f$. Further notations are introduced as required. \\
\indent Moreover, we make use of the following identities that are derived in Lemmas \ref{lem:continuity:sample} and \ref{lem:continuity} below under some assumptions: $\forall\theta\in B_{*,n}$,
\begin{multline}\label{App_derivative_lambda_sample}
  \frac{d\wh\lambda(\theta)'}{d\theta} = - \EE_n \left[\tau_i(\wh\lambda,\theta)\frac{d g_i(\theta)'}{d\theta}(I + \wh\lambda(\theta) g_i(\theta)')\right]\check\Omega(\wh\lambda,\theta)^{-1}\\
  = \EE_n \left[\tau_i^{\diamond}(\wh\lambda,\theta) \wtl w_{1,i}\wtl w_i' (I + \wh\lambda(\theta) g_i(\theta)')\right]\check\Omega^{\diamond}(\wh\lambda,\theta)^{-1},
\end{multline}
\noindent and
\begin{multline}\label{App_BvM_Lem_eq_7_miss:paper}
  \frac{d\lambda_*(\theta)'}{d\theta} = \EE^{Q^*(\theta)} \left[\wtl w_{1,i}\wtl w_i'(I + \lambda_*(\theta) g_i(\theta)')\right]\Omega_*^{\dagger}(\theta)^{-1}\\
  =  \EE \left[e^{\lambda_{*}(\theta)'g_i(\theta)}\wtl w_{1,i}\wtl w_i'(I + \lambda_*(\theta) g_i(\theta)')\right]\left(\EE[e^{\lambda_{*}(\theta)'g_i(\theta)}\varepsilon_i(\theta)^2\wtl w_{i}\wtl w_{i}']\right)^{-1},
\end{multline}
\noindent respectively, where $\frac{d g_i(\theta)'}{d\theta} = w_{1,i}\wtl w_i'$.

Finally, under Assumption \ref{Ass_absolute_continuity} the first order condition for $\theta_{*}$ is
\begin{equation}\label{FOC_theta_circ}
  \frac{d \lambda_{*}(\theta_{*})'}{d\theta} \EE[g_i(\theta_{*})]  - \EE[\wtl{w}_{1i}\wtl{w}_i']\lambda_{*}(\theta_{*}) + \EE[\tau_i^{\dagger}(\lambda_*,\theta_{*})\wtl{w}_{1i}\wtl{w}_i']\lambda_*(\theta_{*}) = 0,
\end{equation}
\noindent and $\EE[\tau_i(\lambda_*,\theta_{*}) \varepsilon_i(\theta_*) \wtl{w}_i] = 0$ since it is the first order condition for $\lambda_{*}$.\\

\indent Under Assumptions \ref{Ass_absolute_continuity} and \ref{ass:feasibility}, the first derivative of $\ell_{*,\theta}(w_{1:n})$ with respect to $\theta$, denoted by $\dot\ell_{*,\theta_{*}}(w_{1:n})$, is given by: $\forall \theta \in B_{*,n}$:
\begin{displaymath}
  \frac{1}{n}\dot\ell_{*,\theta}(w_{1:n}) = \frac{d\lambda_*(\theta)'}{d\theta}\EE_n[(1 - \tau_i(\lambda_*,\theta))g_i(\theta)]
  + \EE_n\left[(1 - \tau_i(\lambda_*,\theta)) \wtl w_{1,i} \wtl w_{i}'\right]'\lambda_*(\theta)
\end{displaymath}
\noindent with probability approaching $1$ as $n\rightarrow \infty$.

\subsection{Stochastic Local Asymptotic Normality (LAN) for the base and the extended models}
%=============================================================================================================================================================================
We now provide three theorems that establish stochastic LAN for the base and the extended model. We provide below each theorem the corresponding proof. Proofs are novel and substantially differ from proofs of similar results in \cite{ChibShinSimoni2018}.
\begin{thm}[Stochastic LAN in the base model.]\label{thm:stochasticLAN:endogeneity}
    Let $V_{\theta_*}$ be a positive definite matrix whose expression is given in the statement of Theorem \ref{thm_BvM_misspecified}. Let Assumptions \ref{Ass_absolute_continuity} - \ref{Ass_3} hold. For every $\theta\in\mathbb{R}^p$ let $h := \sqrt{n}(\theta - \theta_*)$. Then, for every closed ball $K_n\subset \mathbb{R}^p$ centred on zero with radius $M_n\rightarrow\infty$ such that $M_n=o(\sqrt{n})$,
\begin{equation}\label{eq_stochastic_LAN_neglected_endogeneity}
  \sup_{h\in K_n}\left|\log \frac{\wh q(w_{1:n}|\theta_* + h/\sqrt{n})}{\wh q(w_{1:n}|\theta_*)} - h'V_{\theta_*}^{-1}\Delta_{n,\theta_*} + \frac{1}{2}h'V_{\theta_*}^{-1}h\right|\overset{p}{\to} 0\qquad \textrm{as }n\rightarrow \infty,
\end{equation}
\noindent where $\theta_*$ is as defined in \eqref{eq:thetastar}, $h'V_{\theta_*}^{-1}\Delta_{n,\theta_*} := \frac{h'}{\sqrt{n}}\dot\ell_{n,\theta_{*}}(w_{1:n}) \xrightarrow{d}\mathcal{N}(0,h'H_*h)$ is bounded in probability and $H_*$ is a positive definite matrix defined in Lemma \ref{lem:asymp:normality:FOC} below.
\end{thm}

\proof
%The proof of this theorem differs from the strategy of proofs used in \cite{ChibShinSimoni2018} and it is specific for the setting considered in this paper. Therefore, we detail it.
First, compact sets $K_n := \{h\in\mathbb{R}^p ; \|h\| \leq M_n\}$ are such that the corresponding $\theta := \theta_* + h/\sqrt{n}$ belongs to $B_{*,n}$ and so, under Assumption \ref{ass:feasibility}, there exists a $N\geq 1$ such that for every $n>N$ the $\log$-ETEL function $\ell_{n,\theta}(w_{1:n})$ is well-defined on $B_{*,n}$.\\
\indent We use a second order MVT expansion applied to $\theta\mapsto\ell_{n,\theta}(w_{1:n})$ around $\theta_*$, the first order condition of $\wh \lambda(\theta)$ which is $\EE_n\left[e^{\wh \lambda(\theta)'g_i(\theta)}g_i(\theta)\right] = 0$, and Lemma \ref{lem:continuity:sample} which guarantees $\theta\mapsto \wh{\lambda}(\theta)$ is $\mathcal{C}^2$ on $B_{*,n}$, to get: $\forall \theta\in B_{*,n}$, with probability approaching $1$,
\begin{multline*}
    \ell_{n,\theta}(w_{1:n}) - \ell_{n,\theta_*}(w_{1:n})
    = (\theta - \theta_*)'\dot{\ell}_{n,\theta_*}(w_{1:n}) + \frac{1}{2}(\theta - \theta_*)'\ddot{\ell}_{n,\wtl\theta}(w_{1:n})(\theta - \theta_*)\\
    = (\theta - \theta_*)' \frac{d\wh\lambda(\theta_*)'}{d\theta}n\wh g(\theta_*) + (\theta - \theta_*)'\frac{d\wh g(\theta_*)'}{d\theta}\wh\lambda(\theta_*)n + \frac{n}{2}(\theta - \theta_*)'\frac{d^2[\wh\lambda(\wtl\theta)'\wh g(\wtl\theta)]}{d\theta d\theta'}(\theta - \theta_*)\\
    - n(\theta - \theta_*)'\EE_n\left[\tau_i(\wh\lambda,\theta_*)\frac{dg_i(\theta_*)'}{d\theta}\right] \wh\lambda(\theta) - \frac{n}{2}(\theta - \theta_*)'\EE_n\left[\tau_i(\wh\lambda,\wtl\theta)\frac{dg_i(\wtl\theta)'}{d\theta}\right]\frac{d\wh\lambda(\theta)}{d\theta'}(\theta - \theta_*)\\
     - \frac{n}{2}(\theta - \theta_*)'\EE_n\left[\tau_i(\wh\lambda,\wtl\theta)\frac{dg_i(\wtl\theta)'}{d\theta}\wh\lambda(\wtl\theta) \wh\lambda(\wtl\theta)'\frac{dg_i(\wtl\theta)}{d\theta'}\right](\theta - \theta_*)\\
     - \frac{n}{2}(\theta - \theta_*)'\frac{d\wh\lambda(\wtl\theta)'}{d\theta}\EE_n\left[\tau_i(\wh\lambda,\wtl\theta)g_i(\wtl\theta)\wh\lambda(\theta_*)'\frac{dg_i(\wtl\theta)}{d\theta'}\right](\theta - \theta_*)\\
     + \frac{n}{2}(\theta - \theta_*)'\EE_n\left[\tau_i(\wh\lambda,\wtl\theta)\frac{dg_i(\wtl\theta)'}{d\theta}\right]\wh\lambda(\wtl \theta) \wh\lambda(\wtl \theta)'\EE_n\left[\tau_i(\wh\lambda,\wtl\theta)\frac{dg_i(\wtl\theta)}{d\theta'}\right](\theta - \theta_*)
\end{multline*}
\noindent for $\wtl \theta = \tau\theta + (1 - \tau)\theta_*$ and some $\tau \in [0,1]$. By replacing $\theta$ with $\theta_* + h/\sqrt{n}$, so that $\wtl \theta = \theta_* + \tau h/\sqrt{n}$, and by using the expression for $g_i(\theta)$ and its first derivative with respect to $\theta$, the previous expression simplifies as: $\forall h \in K_n$,
\begin{multline}\label{eq:3:neglected}
    \ell_{n,\theta_* + h/\sqrt{n}}(w_{1:n}) - \ell_{n,\theta_*}(w_{1:n}) %= \\
    = h' \frac{d\wh\lambda(\theta_*)'}{d\theta}\sqrt{n}\wh g(\theta_*) - h'\sqrt{n}\EE_n[\wtl{w}_{1,i}\wtl{w}_{i}']\wh\lambda(\theta_*) + \frac{1}{2}h'\frac{d^2[\wh\lambda(\wtl\theta)'\wh g(\wtl\theta)]}{d\theta d\theta'}h\\
    + \sqrt{n}h'\EE_n\left[\tau_i(\wh\lambda,\theta_*)\wtl{w}_{1,i}\wtl{w}_{i}'\right] \wh\lambda(\theta_*) + \frac{1}{2}h'\EE_n\left[\tau_i(\wh\lambda,\wtl\theta)\wtl{w}_{1,i}\wtl{w}_{i}'\right]\frac{d\wh\lambda(\wtl\theta)}{d\theta'}h\\
     - \frac{1}{2}\EE_n\left[\tau_i(\wh\lambda,\wtl\theta)\left(h'\wtl{w}_{1,i} \wtl{w}_i'\wh\lambda(\wtl\theta)\right)^2\right] + \frac{1}{2}h'\frac{d\wh\lambda(\wtl\theta)'}{d\theta}\EE_n\left[\tau_i(\wh\lambda,\wtl\theta)g_i(\wtl\theta)\wh\lambda(\wtl\theta)'\wtl{w}_i\wtl{w}_{1,i}'\right]h\\
     + \frac{1}{2}h'\EE_n\left[\tau_i(\wh\lambda,\wtl\theta)\wtl w_{1,i} \wtl w_i'\right]\wh\lambda(\wtl \theta) \wh\lambda(\wtl \theta)'\EE_n\left[\tau_i(\wh\lambda,\wtl\theta)\wtl w_i \wtl w_{1,i}'\right]h.
\end{multline}
\noindent Let us start by considering the terms of first order in \eqref{eq:3:neglected}, to which we add and subtract the first order condition for $\theta_*$ given in \eqref{FOC_theta_circ}: $\forall h \in K_n$, %(which is $\frac{d \lambda_{*}(\theta_{*})'}{d\theta} \EE[g_i(\theta_{*})]  - \EE[\wtl{w}_{1i}\wtl{w}_i']\lambda_{*}(\theta_{*}) + \EE[\tau_i^{\dagger}(\lambda_*,\theta_{*})\wtl{w}_{1i}\wtl{w}_i']\lambda_*(\theta_{*}) = 0$):
\begin{multline*}
    h' \frac{d\wh\lambda(\theta_*)'}{d\theta}\sqrt{n}\wh g(\theta_*) - h'\sqrt{n}\EE_n[\wtl{w}_{1,i}\wtl{w}_{i}']\wh\lambda(\theta_*) + \sqrt{n}h'\EE_n\left[\tau_i(\wh\lambda,\theta_*)\wtl{w}_{1,i}\wtl{w}_{i}'\right] \wh\lambda(\theta) = \\
    h' \sqrt{n}\left(\frac{d\wh\lambda(\theta_*)'}{d\theta}\wh g(\theta_*) - \frac{d \lambda_{*}(\theta_{*})'}{d\theta} \EE[g_i(\theta_{*})]\right) - h'\sqrt{n}\left(\EE_n[\wtl{w}_{1,i}\wtl{w}_{i}']\wh\lambda(\theta_*) - \EE[\wtl{w}_{1i}\wtl{w}_i']\lambda_{*}(\theta_{*})\right)\\
    + \sqrt{n}h'\left(\EE_n\left[\tau_i(\wh\lambda,\theta_*)\wtl{w}_{1,i}\wtl{w}_{i}'\right] \wh\lambda_*(\theta) - \EE[\tau_i^{\dagger}(\lambda_*,\theta_{*})\wtl{w}_{1i}\wtl{w}_i']\lambda_*(\theta_{*})\right) =: h'V_{\theta_*}^{-1} \Delta_{n,\theta_*}.
\end{multline*}
By Lemma \ref{lem:asymp:normality:FOC} below, $h'V_{\theta_*}^{-1} \Delta_{n,\theta_*}$ is asymptotically normal with zero mean and variance equal to the non-singular matrix $h'H_*h$ whose expression is given in the statement of Lemma \ref{lem:asymp:normality:FOC}.

Now, let us consider the terms of second order in \eqref{eq:3:neglected}. First, by Lemma \ref{lem:technical:endogeneity:second:derivative}, $h'\frac{d^2[\wh\lambda(\wtl\theta)'\wh g(\wtl\theta)]}{d\theta d\theta'}h$ is bounded in probability uniformly in $h\in K_n$.
%notice that $h'\frac{d^2[\wh\lambda(\wtl\theta)'\wh g(\wtl\theta)]}{d\theta d\theta'}h = h'\sum_{l=1}^d \frac{d^2\wh\lambda_l(\wtl\theta)}{d\theta d\theta'} \wh g_{l}(\wtl\theta)h - 2h'\frac{d\wh\lambda(\wtl\theta)'}{d\theta} \EE_n\left[\wtl{w}_i \wtl{w}_{1,i}'\right]h$, where $\wh\lambda_l(\wtl\theta)$ and $\wh g_{l}(\theta)$ denote the $l$-th components of the vectors $\wh\lambda(\wtl\theta)$ and $\wh g(\theta)$, respectively.
Because $h/\sqrt{n}\rightarrow 0$ since $\|h\|\leq M_n$ and $M_n=o(\sqrt{n})$, then $\wtl\theta\rightarrow \theta_*$ uniformly in $h\in K_n$ as $n\rightarrow \infty$. Moreover, we use the following limits as $n\rightarrow \infty$: (1) By continuity of $\theta\mapsto \lambda_*(\theta)$ (by Lemma \ref{lem:continuity} in the Supplementary Material), and continuity of $\theta\mapsto g_i(\theta)$, we have: $\lambda_*(\wtl\theta)\rightarrow\lambda_*(\theta_*)$, and $g_i(\wtl\theta)\rightarrow g_i(\theta_*)$. (2) By Lemma \ref{lem:technical:endogeneity:1} then $\EE_n\left[\tau_i(\wh\lambda,\wtl\theta)\wtl{w}_{1,i}\wtl{w}_{i}'\right]$ converges in probability to $\EE\left[\tau_i^{\dagger}(\lambda_*,\theta_*)h'\wtl{w}_{1,i}\wtl{w}_{i}'\right]$ uniformly in $h\in K_n$ as $n\rightarrow \infty$. (3) By Lemma \ref{lem:technical:endogeneity:2} then $\EE_n\left[\tau_i(\wh\lambda,\wtl\theta)\left(h'\wtl{w}_{1,i} \wtl{w}_i'\wh\lambda(\wtl\theta)\right)^2\right]$ converges in probability to $\EE\left[\tau_i^{\dagger}(\lambda_*,\theta_*)\left(h'\wtl{w}_{1,i} \wtl{w}_i' \lambda_*(\theta_*)\right)^2\right]$ uniformly in $h\in K_n$ as $n\rightarrow \infty$. (4) By Lemma \ref{lem:technical:endogeneity:3} then $$\EE_n\left[\tau_i(\wh\lambda,\wtl\theta)h'\wtl{w}_{1,i} \wtl{w}_i'\wh\lambda(\wtl\theta) \varepsilon_i(\wtl\theta)'\wtl{w}_i'\right]\xrightarrow{p} \EE\left[\tau_i^{\dagger}(\lambda_*,\theta_*)h'\wtl{w}_{1,i} \wtl{w}_i' \lambda_*(\theta_*) \varepsilon_i(\theta_*)'\wtl{w}_i'\right]$$
uniformly in $(\wh\lambda(\wtl\theta),h) \in B_{1/\sqrt{n}}(\lambda_*(\theta_*))\times K$. %, (5) by Assumption \ref{Ass_3_extended} \textit{(c)}: $\EE_n\left[\wtl{w}_i \wtl{w}_{1,i}'\right]h \xrightarrow{p} \EE\left[\wtl{w}_i \wtl{w}_{1,i}'\right]h$ uniformly in $h\in K_n$.
By combining the convergences in (2) and (4) above, Lemma \ref{lem:Omega}, and the expression of $\frac{d\wh\lambda(\wtl\theta)'}{d\theta}$ we have that $\forall \theta\in B_{*,n}$
\[
h'\frac{d\wh\lambda(\wtl\theta)'}{d\theta}\xrightarrow{p} h'\EE^{Q^*(\theta_*)}[\wtl{w}_{1,i}\wtl{w}_{i}'(I + \lambda_*(\theta_*)g_i(\theta_*)')]\Omega_*^{\diamond}(\theta_*)^{-1}
\]
uniformly in $(\wh\lambda(\wtl\theta),h) \in B_{1/\sqrt{n}}(\lambda_*(\theta_*))\times K$. Hence, by using these limits, the term of second order in \eqref{eq:3:neglected} is equal to:
\begin{multline}
  \frac{1}{2}h'\frac{d^2\lambda_*(\theta_*)}{d\theta d\theta'} \EE [g_i(\theta_*)]h
  - h'\frac{d\lambda_*(\theta_*)'}{d\theta} \EE\left[\wtl{w}_i \wtl{w}_{1,i}'\right]h + \frac{1}{2}h'\EE\left[\tau_i^{\dagger}(\lambda_*,\theta_*)\wtl{w}_{1,i}\wtl{w}_{i}'\right]\frac{d\lambda_*(\theta_*)}{d\theta'}h\\
  - \frac{1}{2}\EE\left[\tau_i^{\dagger}(\lambda_*,\theta_*)\left(h'\wtl{w}_{1,i} \wtl{w}_i'\lambda_*(\theta_*)\right)^2 \right] + \frac{h'}{2}\frac{d\lambda_*(\theta_*)'}{d\theta}\EE\left[\tau_i^{\dagger}(\lambda_*,\theta_*)g_i(\theta_*)\lambda_*(\theta_*)'\wtl{w}_i\wtl{w}_{1,i}'\right]h\\
  + \frac{1}{2} h'\EE\left[\tau_i^{\dagger}(\lambda_*,\theta_*)\wtl w_{1,i} \wtl w_i'\right]\lambda_*(\theta_*) \lambda_*(\theta_*)'\EE\left[\tau_i^{\dagger}(\lambda_*,\theta_*)\wtl w_i \wtl w_{1,i}'\right]h + o_p(1).
\end{multline}
\noindent where the $o_p(1)$ is uniform in $(\wh\lambda(\wtl\theta),h) \in B_{1/\sqrt{n}}(\lambda_*(\theta_*))\times K$. By remarking that $\tau_i^{\dagger}(\lambda_*,\theta_*) = dQ^*(\theta_*)/dP$, and that $\EE\left[ g_{i,l}(\theta_*)\right] = 0$ for every $l>d_x$, the previous expression can be simplified as
\begin{multline*}
    \frac{1}{2}h'\frac{d^2\lambda_*(\theta_*)}{d\theta d\theta'} \EE [g_i(\theta_*)]h
    - h'\frac{d\lambda_*(\theta_*)'}{d\theta} \EE\left[\wtl{w}_i \wtl{w}_{1,i}'\right]h - \frac{1}{2}h'\mathbb{V}ar_{Q^*(\theta_*)}\left[\wtl{w}_{1,i} \wtl{w}_i'\lambda_*(\theta_*) \right]h \\
    + \frac{1}{2}h'\EE^{Q^*(\theta_*)}\left[\wtl{w}_{1,i}\wtl{w}_{i}'\left(I + \lambda_*(\theta_*)g_i(\theta_*)'\right)\right]\frac{d\lambda_*(\theta_*)}{d\theta'}h + o_p(1) =: -\frac{1}{2}h'V_{\theta_*}^{-1}h + o_p(1).
\end{multline*}
By putting all these elements together we get:
\begin{displaymath}
    \ell_{n,\theta_* + h/\sqrt{n}}(w_{1:n}) - \ell_{n,\theta_*}(w_{1:n}) = h'V_{\theta_*}^{-1} \Delta_{n,\theta_*} - \frac{1}{2}h'V_{\theta_*}^{-1}h  + o_p(1),
\end{displaymath}
\noindent where $h'V_{\theta_*}^{-1}\Delta_{n,\theta_*} \xrightarrow{d} \mathcal{N}(0,h'H_*h)$, $\frac{h'}{\sqrt{n}}\dot{\ell}_{n,\theta_*} = h'V_{\theta_*}^{-1}\Delta_{n,\theta_*}$ and $V_{\theta_*}^{-1} = \plim\ddot{\ell}_{n,\theta_*}/n + o_p(1)$. Thus, we obtain the result of the theorem.
\begin{flushright}
  $\square$
\end{flushright}

%===========================================================================================================
\begin{thm}[Stochastic LAN in the base model under exogeneity]\label{thm:stochasticLAN:exogeneity}
    Suppose Assumptions \ref{Ass_0_NS}-\ref{Ass_3} hold and for every $\theta\in\mathbb{R}^p$, let $h:= \sqrt{n}(\theta - \theta_\circ)$. Then, for every closed ball $K_n\subset \mathbb{R}^p$ centred on zero with radius $M_n\rightarrow\infty$ such that $M_n=o(\sqrt{n})$,
\begin{equation}\label{eq_stochastic_LAN}
  \sup_{h\in K_n}\left|\log \frac{\wh q(w_{1:n}|\theta_\circ + h/\sqrt{n})}{\wh q(w_{1:n}|\theta_\circ)} - h'V_{\theta_\circ}^{-1}\Delta_{n,\theta_\circ} + \frac{1}{2}h'V_{\theta_\circ}^{-1}h\right|\overset{p}{\to} 0\qquad \textrm{as }n\rightarrow \infty
\end{equation}
\noindent where $\theta_\circ$ is the true value of $\theta$, $V_{\theta_\circ}^{-1} := \mathbf{E}\left[ \wtl w_{1,i}\wtl w_{i}'\right] \left(\EE[\varepsilon_i^2 \wtl w_{i} \wtl w_{i}']\right)^{-1}\mathbf{E}\left[ \wtl w_{i} \wtl w_{1,i}'\right]$ assumed to be nonsingular and $V_{\theta_\circ}^{-1}\Delta_{n,\theta_\circ} := \frac{1}{\sqrt{n}}\sum_{i=1}^{n}\mathbf{E}\left[ \wtl w_{1,i} \wtl w_{i}'\right] \left(\EE[\varepsilon_i^2 \wtl w_{i} \wtl w_{i}']\right)^{-1} \varepsilon_{i}\tilde{w}_{i}$ is bounded in probability.
\end{thm}

\proof
First, compact sets $K_n := \{h\in\mathbb{R}^p ; \|h\| \leq M_n\}$ are such that the corresponding $\theta := \theta_{\circ} + h/\sqrt{n}$ belongs to $B_{\circ,n}$ and so, under Assumption \ref{ass:feasibility}, there exists a $N\geq 1$ such that for every $n>N$ the $\log$-ETEL function $\ell_{n,\theta}(w_{1:n})$ is well-defined on $B_{\circ,n}$.\\
\indent The proof proceeds as the proof of Theorem \ref{thm:stochasticLAN:endogeneity} above but now the tilting parameter evaluated at the true $\theta_{\circ}$ is zero: $\lambda_*(\theta_\circ) = 0$. Moreover, $\EE[g_i(\theta_\circ)] = 0$, $\tau_i(\lambda_*,\theta_\circ) = 1$ and $\frac{d \lambda_{*}(\theta_{\circ})'}{d\theta} = \EE[\wtl{w}_{1,i}\wtl{w}_i']\left(\EE[\varepsilon_i^2\wtl{w}_{i}\wtl{w}_{i}']\right)^{-1}$. Therefore, the terms of first order in the proof of Theorem \ref{thm:stochasticLAN:endogeneity} simplify. To see this we treat the different terms separately.
We start with the term $h' \frac{d\wh\lambda(\theta_{\circ})'}{d\theta}$. The following decomposition holds: by using the expression of $\frac{d\wh\lambda(\theta_{\circ})'}{d\theta}$ in \eqref{App_derivative_lambda_sample},
\begin{multline}
  h' \frac{d\wh\lambda(\theta_{\circ})'}{d\theta}\check{\Omega}^{\diamond}(\wh\lambda,\theta_{\circ}) - h' \EE[\wtl{w}_{1i}\wtl{w}_i'] = h'\EE_n\left[\tau_i^{\diamond}(\wh\lambda,\theta_\circ) \wtl{w}_{1,i}\wtl{w}_i'(I + \wh\lambda(\theta_\circ)g_i(\theta_\circ)')\right]  - h' \EE[\wtl{w}_{1i}\wtl{w}_i']\\
  = h'\EE_n\left[\tau_i^{\diamond}(\wh\lambda,\theta_\circ) \wtl{w}_{1,i}\wtl{w}_i'\right]  - h' \EE[\tau_i^{\diamond}(\wh\lambda,\theta_\circ) \wtl{w}_{1i}\wtl{w}_i']\\
  + h' \EE[(\tau_i^{\diamond}(\wh\lambda,\theta_\circ) - 1) \wtl{w}_{1i}\wtl{w}_i'] + h'\EE_n\left[\tau_i^{\diamond}(\wh\lambda,\theta_\circ) \wtl{w}_{1,i}\wtl{w}_i'\wh\lambda(\theta_\circ)g_i(\theta_{\circ})'\right].\label{eq:LAN:exogeneity:1}
\end{multline}
\noindent By Lemma \ref{lem:2}, for every $\eta > 0$ there exists a finite $C$ and a finite $N(C,\eta)$ such that for every $n> N(C,\eta)$, $\|\wh{\lambda}(\theta_\circ) - \lambda_*(\theta_\circ)\|_2 < C/\sqrt{n}$ with probability at least $\eta$. Hence, the following upper bound for the first two terms on the right hand side of \eqref{eq:LAN:exogeneity:1} holds:
\begin{multline*}
  \sup_{h\in K_n} \left\|h'\EE_n\left[\tau_i^{\diamond}(\wh\lambda,\theta_\circ) \wtl{w}_{1,i}\wtl{w}_{i}\right]  - h' \EE[\tau_i^{\diamond}(\wh\lambda,\theta_\circ) \wtl{w}_{1,i}\wtl{w}_{i}]\right\|\\
  \leq \sup_{h\in K_n} \sup_{\lambda\in B_{1/\sqrt{n}}(0)}\left\|h'\EE_n\left[e^{\lambda'g_i(\theta_\circ)} \wtl{w}_{1,i}\wtl{w}_{i}\right]  - h' \EE[e^{\lambda'g_i(\theta_\circ)} \wtl{w}_{1i}\wtl{w}_{i}]\right\|
\end{multline*}
\noindent which converges to zero under Assumption \ref{Ass_3} \textit{(c)} (with $(j,\ell,\ell') = (1,1,1)$), by compactness of $B_{1/\sqrt{n}}(0)$ and by \cite[Lemma 2.4]{NeweyMcFadden1994}. Next, we control term $h' \EE[(\tau_i^{\diamond}(\wh\lambda,\theta_\circ) - 1) \wtl{w}_{1i}\wtl{w}_{i}]$ in \eqref{eq:LAN:exogeneity:1}. By the CMT and Lemma \ref{lem:2}: $\sup_{h\in K_n}\|e^{\wh\lambda(\theta_\circ)'g_i(\theta_\circ)}h'\wtl{w}_{1,i}\wtl{w}_i - h'\wtl{w}_{1,i}\wtl{w}_i\|\xrightarrow{p} 0$ for every $i=1,\ldots,n$. By this and the DCT, which is valid under Assumption \ref{Ass_3} \textit{(c)} with $(j,\ell,\ell') = (1,1,1)$, it holds that $\EE[(\tau_i^{\diamond}(\wh\lambda,\theta_\circ) - 1) h' \wtl{w}_{1i}\wtl{w}_i']\xrightarrow{p} 0$ uniformly in $h\in K_n$. The last term in the right hand side of \eqref{eq:LAN:exogeneity:1} can be treated in a similar way by using the upper bound:
\begin{multline*}
  \sup_{h\in K_n} \left\|h'\EE_n\left[\tau_i^{\diamond}(\wh\lambda,\theta_\circ) \wtl{w}_{1,i}\wtl{w}_i'\wh\lambda(\theta_\circ)g_{i,}(\theta_\circ)\right]\right\|\\
  \leq \sup_{h\in K_n} \sup_{\lambda\in B_{1/\sqrt{n}}(0)}\left\|h'\EE_n\left[\tau_i^{\diamond}(\lambda,\theta_\circ) \wtl{w}_{1,i}\wtl{w}_i'\lambda g_{i}(\theta_\circ)\right] - h'\EE\left[\tau_i^{\diamond}(\lambda,\theta_\circ) \wtl{w}_{1,i}\wtl{w}_i'\lambda g_{i}(\theta_\circ)\right]\right\|\\
  + \sup_{h\in K_n} \left\|h'\EE\left[\tau_i^{\diamond}(\wh\lambda,\theta_\circ) \wtl{w}_{1,i}\wtl{w}_i'\wh\lambda(\theta_\circ)g_{i}(\theta_\circ)\right]\right\|.
\end{multline*}
The two terms in the right hand side of the previous expression converge to zero in probability under Assumption \ref{Ass_3} \textit{(f)}, \cite[Lemma 2.4]{NeweyMcFadden1994}, Lemma \ref{lem:2}, the CMT, the DMT, and the fact that $\lambda_*(\theta_\circ) = 0$. By putting all these elements together and by Lemma \ref{lem:Omega} to control $\check\Omega^{\diamond}(\wh\lambda,\theta_\circ)$ in the expression of $\frac{d\wh\lambda(\theta_{\circ})'}{d\theta}$ we get $h' \frac{d\wh\lambda(\theta_{\circ})'}{d\theta}\sqrt{n}\wh g(\theta_{\circ})$ is equal to $h' \EE[\wtl{w}_{1i}\wtl{w}_i']\EE[\varepsilon_i^2 \wtl w_i \wtl{w}_i']^{-1}\mathbb{G}_n[g(w_i,\theta_{\circ})] + o_p(1)$, where the $o_p(1)$ term is uniform in $h\in K_n$.\\
\indent Next, we analyse the other terms of first order in the proof of Theorem \ref{thm:stochasticLAN:endogeneity}, namely, $ - h'\EE_n[\wtl{w}_{1,i}\wtl{w}_{i}']\sqrt{n}\wh\lambda(\theta_{\circ}) + h'\EE_n[\tau_i(\wh\lambda,\theta_{\circ})\wtl{w}_{1i}\wtl{w}_i'] \sqrt{n}\wh \lambda(\theta_{\circ})$. The factor multiplying $\sqrt{n}\wh\lambda(\theta_{\circ})$ can be decomposed as follows:
\begin{multline}
   - h'\EE_n[\wtl{w}_{1,i}\wtl{w}_{i}'] + h'\EE_n[\tau_i(\wh\lambda,\theta_{\circ})\wtl{w}_{1i}\wtl{w}_i'] =\\ %h'\EE_n\left[\left(\tau_i(\wh\lambda,\theta_{\circ}) - 1\right)\wtl{w}_{1i}\wtl{w}_i'\right] \sqrt{n}\wh \lambda(\theta_{\circ})\\
   = h'\EE_n\left[\tau_i(\wh\lambda,\theta_\circ)\wtl{w}_{1,i}\wtl{w}_i' - \EE\left(\tau_i(\wh\lambda,\theta_\circ)\wtl{w}_{1,i}\wtl{w}_i'\right)\right] + \EE\left[\left(\tau_i(\wh\lambda,\theta_\circ) - 1\right)h'\wtl{w}_{1,i}\wtl{w}_i'\right]\\
   - \EE_n\left[h' \wtl{w}_{1,i}\wtl{w}_i' - \EE\left(h' \wtl{w}_{1,i}\wtl{w}_i'\right)\right].\label{eq:LAN:exogeneity:2}
\end{multline}
The first two terms in the right hand side of the previous expression are the same as the first three terms in \ref{eq:LAN:exogeneity:1} that have been shown to converge to zero uniformly in $h\in K_n$. The last term converges to zero by the uniform Law of Large Numbers under Assumption \ref{Ass_3_extended} \textit{(c)}. Next, because by Lemma \ref{lem:1}, $\sqrt{n}\wh \lambda(\theta_{\circ}) = - \EE[\varepsilon_i^2 \wtl w_i \wtl{w}_i']^{-1}\mathbb{G}_n[\varepsilon_i\wtl w_i ] + o_p(1)$ and $\mathbb{G}_n[\varepsilon_i\wtl w_i ]\xrightarrow{d} \mathcal{N}(0,\EE[\varepsilon_i^2 \wtl w_i \wtl{w}_i'])$, then $\sqrt{n}\wh \lambda(\theta_{\circ}) = O_p(1)$ and so, $h'\EE_n[(\tau_i(\wh\lambda,\theta_{\circ}) - 1)\wtl{w}_{1i}\wtl{w}_i'] \sqrt{n}\wh \lambda(\theta_{\circ}) \xrightarrow{p} 0$ under Assumption \ref{Ass_3} \textit{(c)} with $(j,\ell,\ell') = (1,1,1)$.

By putting all these elements together, we get that the terms of first order in the proof of Theorem \ref{thm:stochasticLAN:endogeneity} are equal to the following simplified expression: $\forall h\in K_n$,
\begin{multline*}
    h' \frac{d\wh\lambda(\theta_{\circ})'}{d\theta}\sqrt{n}\wh g(\theta_{\circ}) - h'\EE_n[\wtl{w}_{1,i}\wtl{w}_{i}']\sqrt{n}\wh\lambda(\theta_{\circ}) + h'\EE_n[\tau_i(\wh\lambda,\theta_{\circ})\wtl{w}_{1i}\wtl{w}_i'] \sqrt{n}\wh \lambda(\theta_{\circ})\\
    = h' \EE[\wtl{w}_{1i}\wtl{w}_i']\EE[\varepsilon_i^2 \wtl w_i \wtl{w}_i']^{-1}\mathbb{G}_n[g(w_i,\theta_{\circ})] + o_p(1)
\end{multline*}
\noindent where the $o_p(1)$ term is uniform in $h\in K_n$.
%\noindent since the second and third terms of the left hand side are $o_p(1)$ because $\sqrt{n}\wh \lambda(\theta_{\circ}) = - \EE[\varepsilon_i(\theta_{\circ})^2 \wtl w_i \wtl{w}_i']^{-1}\mathbb{G}_n[\tau_i^{\dagger}(\lambda_*,\theta_{\circ})\varepsilon_i(\theta_{\circ})\wtl w_i ] + o_p(1)$ and $\mathbb{G}_n[\tau_i^{\dagger}(\lambda_{\circ},\theta_{\circ})\varepsilon_i(\theta_{\circ})\wtl w_i ]\xrightarrow{d} \mathcal{N}(0,\EE[\varepsilon_i(\theta_{\circ})^2 \wtl w_i \wtl{w}_i'])$ so that $\sqrt{n}\wh \lambda(\theta_{\circ}) = O_p(1)$.

%In addition, the first term in the right hand side is obtained because each component of the vector $h' \frac{d\wh\lambda(\theta_{\circ})'}{d\theta}$ converges in probability to the corresponding element of the vector $h' \EE[\wtl{w}_{1i}\wtl{w}_i']\EE[\varepsilon_i(\theta_{\circ})^2 \wtl w_i \wtl{w}_i']^{-1}$ under Assumptions \ref{Ass_absolute_continuity}, \ref{Ass_1_NS} (a), \ref{Ass_3} (d)-(f). This holds by Lemma \ref{lem:Omega_star} and by the following argument that we now explain.

Finally, since $\mathbb{G}_n[g(w_i,\theta_{\circ})]\xrightarrow{d} \mathcal{N}(0,\EE[\varepsilon_i(\theta_{\circ})^2 \wtl w_i \wtl{w}_i'])$ by the Lindberg-Levy central limit theorem under Assumption \ref{Ass_1_NS} \textit{(a)}, we conclude that the previous term converges in distribution to the $\mathcal{N}(0,h'\EE[\wtl{w}_{1i}\wtl{w}_i']\EE[\varepsilon_i(\theta_{\circ})^2 \wtl w_i \wtl{w}_i']^{-1}\EE[\wtl{w}_{i}\wtl{w}_{1i}']h)$ distribution for every $h\in K_n$.
%$$\mathcal{A}_1 + \mathcal{A}_2 + \mathcal{A}_3 \xrightarrow{d} \mathcal{N}(0,h'\EE[\wtl{w}_{1i}\wtl{w}_i']\EE[\varepsilon_i(\theta_{\circ})^2 \wtl w_i \wtl{w}_i']^{-1}\EE[\wtl{w}_{i}\wtl{w}_{1i}']h).$$

\indent Now, let us consider the terms of second order in \eqref{eq:3:neglected} where $\wtl \theta := \theta_\circ + \tau h/\sqrt{n}\in B_{\circ,n}$ for some $\tau\in[0,1]$ and $h\in K_n$:
\begin{multline}
  \frac{1}{2}h'\sum_{l=1}^d \frac{d^2\wh\lambda_l(\wtl\theta)}{d\theta d\theta'} \wh g_{l}(\wtl\theta)h - h'\frac{d\wh\lambda(\wtl\theta)'}{d\theta} \EE_n\left[\wtl{w}_i \wtl{w}_{1,i}'\right]h + \frac{1}{2}h'\EE_n\left[\tau_i(\wh\lambda,\wtl\theta)\wtl{w}_{1,i}\wtl{w}_{i}'\right]\frac{d\wh\lambda(\wtl\theta)}{d\theta'}h\\
     - \frac{1}{2}\EE_n\left[\tau_i(\wh\lambda,\wtl\theta)\left(h'\wtl{w}_{1,i} \wtl{w}_i'\wh\lambda(\wtl\theta)\right)^2 \right] + \frac{1}{2}h'\frac{d\wh\lambda(\wtl\theta)'}{d\theta}\EE_n\left[\tau_i(\wh\lambda,\wtl\theta)g_i(\wtl\theta)\wh\lambda(\wtl\theta)'\wtl{w}_i\wtl{w}_{1,i}'\right]h\\
     + \frac{1}{2}h'\EE_n\left[\tau_i(\wh\lambda,\wtl\theta)\wtl w_{1,i} \wtl w_i'\right]\wh\lambda(\wtl \theta) \wh\lambda(\wtl \theta)'\EE_n\left[\tau_i(\wh\lambda,\wtl\theta)\wtl w_i \wtl w_{1,i}'\right]h, \label{eq:LAN:exogeneity:3}
\end{multline}
\noindent where $\wh\lambda_l(\wtl\theta)$ (resp. $\wh g_{l}(\wtl\theta)$) denote the $l$-th element of the vector $\wh\lambda(\wtl\theta)$ (resp. $\wh g(\wtl\theta)$) which is $\mathcal{C}^2$ by Lemma \ref{lem:continuity:sample}. Because $h$ is bounded, then the sequence $\wtl\theta \rightarrow \theta_\circ$ as $n\rightarrow \infty$ uniformly in $h\in K_n$. To treat term \eqref{eq:LAN:exogeneity:3} we use the following limits that are established by using similar arguments as before. (1) $\wh{g}(\wtl\theta) \xrightarrow{p} 0$ uniformly in $h\in K_n$ under Assumption \ref{Ass_3_extended} \textit{(a)}, by \cite[Lemma 2.4]{NeweyMcFadden1994}, the CMT and the DCT. (2) By Lemma \ref{lem:technical:exogeneity:1}, $h'\EE_n[\tau_i(\wh\lambda,\wtl\theta)\wtl{w}_{1,i}\wtl{w}_i] \xrightarrow{p} h'\EE_n[\wtl{w}_{1,i}\wtl{w}_{i}] $ uniformly in $h\in K_n$. (3) By Lemma \ref{lem:technical:exogeneity:2}, $\EE_n[\tau_i(\wh\lambda,\wtl\theta)\left(h'\wtl{w}_{1,i} \wtl{w}_i'\wh\lambda(\wtl\theta)\right)^2] \xrightarrow{p} 0 $ uniformly in $h\in K_n$. (4) By Lemma \ref{lem:technical:exogeneity:3}, $\EE_n\left[\tau_i(\wh\lambda,\wtl\theta)g_i(\wtl\theta)\wh\lambda(\wtl\theta)'\wtl{w}_i\wtl{w}_{1,i}'\right]h\xrightarrow{p} 0$ uniformly in $h\in K_n$. (5) By combining (2), (4) and Lemma \ref{lem:Omega} with $\theta_*$ replaced by $\theta_\circ$ we have that $h'\frac{d\wh\lambda(\wtl\theta)'}{d\theta}\xrightarrow{p} h'\EE[\wtl{w}_{1,i}\wtl{w}_{i}']\Omega_*^{\diamond}(\theta_\circ)^{-1}$ uniformly in $h\in K_n$. Finally $\wh\lambda(\wtl\theta)\xrightarrow{p} 0$ by Lemmas \ref{lem:2} and \ref{lem:continuity}, and the fact that $\wtl\theta\rightarrow\theta_\circ$ uniformly in $h\in K$. By replacing these limits in \eqref{eq:LAN:exogeneity:3} we get that \eqref{eq:LAN:exogeneity:3} is equal to: $- \frac{1}{2}h'\EE\left[\wtl{w}_{1,i}\wtl{w}_{i}'\right]\left(\EE[\varepsilon_i^2\wtl{w}_{1,i}\wtl{w}_{i}]\right)^{-1}\EE[\wtl{w}_{1,i}\wtl{w}_{i}'] + o_p(1) =: -\frac{1}{2}h'V_{\theta_\circ}^{-1}h + o_p(1)$ uniformly in $h\in K_n$, and $V_{\theta_\circ}^{-1}$ is well-defined under Assumption \ref{Ass_1_NS} \textit{(c)} with $\lambda = \lambda_*(\theta_*) = 0$ and Assumption \ref{Ass_2_NS}.

\begin{flushright}
  $\square$
\end{flushright}

%===========================================================================================================
\begin{thm}[Stochastic LAN in the extended model.]\label{thm:stochasticLAN:extended}
  Suppose Assumptions \ref{Ass_0_NS} - \ref{Ass_3} hold and for every $\psi\in\Psi$, let $h := \sqrt{n}(\psi - \psi_\circ)$. Then, for every closed ball $K_n\subset \mathbb{R}^{p + d_x}$ centred on zero with radius $M_n\rightarrow\infty$ such that $M_n=o(\sqrt{n})$,
\begin{equation}\label{eq_stochastic_LAN_extended}
  \sup_{h\in K_n}\left|\log \frac{\wh q(w_{1:n}|\psi_\circ + h/\sqrt{n})}{\wh q(w_{1:n}|\psi_\circ)} - h'V_{\psi_\circ}^{-1}\Delta_{n,\psi_\circ} + \frac{1}{2}h'V_{\psi_\circ}^{-1}h\right|\overset{p}{\to} 0\qquad \textrm{as }n\rightarrow \infty
\end{equation}
\noindent where $\psi_\circ$ is the true value of $\psi$, $V_{\psi_\circ}^{-1} := V_{\psi_{\circ}}^{-1} := \mathbf{E}\left[ \frac{dg_e(w_{i},\psi_{\circ})'}{d\psi }\right] \Omega_{\psi_{\circ}}^{-1}\mathbf{E}\left[ \frac{dg_e(w_{i},\psi_{\circ})}{d\psi'}\right]$ assumed to be nonsingular and $V_{\psi_\circ}^{-1}\Delta_{n,\psi_\circ} := \frac{1}{\sqrt{n}}\sum_{i=1}^{n}\mathbf{E}\left[ \frac{dg_e(w_i,\psi_{\circ})'}{d\psi }\right] \Omega_{\psi_{\circ}}^{-1}\left(\varepsilon_{i}\tilde{w}_{i}-\wtl{v}_\circ\right)$ is bounded in probability.
\end{thm}

\proof
First, compact sets $K_n := \{h\in\mathbb{R}^p ; \|h\| \leq M_n\}$ are such that the corresponding $\theta := \theta_{\circ} + h/\sqrt{n}$ belongs to $B_{\circ,n}$ and so, under Assumption \ref{ass:feasibility}, there exists a $N\geq 1$ such that for every $n>N$ the $\log$-ETEL function $\ell_{n,\theta}(w_{1:n})$ is well-defined on $B_{\circ,n}$.\\
\indent The proof proceeds as the proof of Theorem \ref{thm:stochasticLAN:exogeneity} by replacing $\lambda_*(\theta_\circ) = 0$ by $\lambda_*(\psi_\circ) = 0$, $\wh{\lambda}(\theta)$ by $\wh\lambda(\psi)$, $\EE[g_i(\theta_\circ)] = 0$ by $\EE[g_i(\theta_\circ)] = v_0$. Let $g_{e,i}(\psi):= g_e(w_i,\psi)$, $\tau_{e,i}^{\diamond}(\lambda,\psi):= e^{\lambda(\psi)'g_{e,i}(\psi)}$, $\check\Omega_{\psi}^{\diamond}(\lambda,\psi) := \EE_n[\tau_{e,i}^{\diamond}(\lambda,\psi)g_{e,i}(\psi)g_{e,i}(\psi)'] = \check\Omega^{\diamond}(\lambda,\theta) - \wtl v \wtl v'$. Moreover, notice that

\begin{multline*}
  h'\frac{d \wh\lambda(\psi)'}{d\psi}\check\Omega_{\psi}^{\diamond}(\wh\lambda,\psi) = h_{\theta}'\EE_n[\tau_{e,i}^{\diamond}(\wh\lambda,\psi)\wtl{w}_{1,i}\wtl{w}_i']I_d + h_{\theta}'\EE_n[\tau_{e,i}^{\diamond}(\wh\lambda,\psi)\wtl{w}_{1,i}\wtl{w}_i'\wh\lambda(\theta)\varepsilon_i(\theta)\wtl{w}_i']\\
  + (h_v'\EE_n[\tau_{e,i}^{\diamond}(\wh\lambda,\psi)],0')'
\end{multline*}
\noindent where $h:=(h_{\theta}',h_{v}')'\in K_n :=K_{n,\theta}\times K_{n,v}$ and $0$ has a conformable dimension. The following decomposition holds:
\begin{multline}
  h' \frac{d\wh\lambda(\psi_{\circ})'}{d\psi}\check{\Omega}(\wh\lambda,\psi_{\circ}) - h_{\theta}' \EE[\wtl{w}_{1i}\wtl{w}_i'] = %- h_{\theta}'\EE_n[\tau_{e,i}^{\diamond}(\wh\lambda,\psi)\wtl{w}_{1,i}\wtl{w}_i']\wh\lambda(\theta)(v',0')
  (h_v'\EE_n[\tau_{e,i}^{\diamond}(\wh\lambda,\psi)],0')'\\
  + h_{\theta}'\EE_n\left[\tau_{e,i}^{\diamond}(\wh\lambda,\psi_\circ) \wtl{w}_{1,i}\wtl{w}_i'\right] - h_{\theta}' \EE[\tau_{e,i}^{\diamond}(\wh\lambda,\psi_\circ) \wtl{w}_{1i}\wtl{w}_i']\\
  + h_{\theta}' \EE[(\tau_{e,i}^{\diamond}(\wh\lambda,\psi_\circ) - 1) \wtl{w}_{1i}\wtl{w}_i'] + h_{\theta}'\EE_n\left[\tau_{e,i}^{\diamond}(\wh\lambda,\psi_\circ) \wtl{w}_{1,i}\wtl{w}_i'\wh\lambda(\psi_\circ)\varepsilon_i(\theta)\wtl{w}_i'\right].\label{eq:LAN:extended:1}
\end{multline}
\noindent By Lemma \ref{lem:2}, for any $\eta > 0$ there exists a finite $\delta>0$ and a finite $N(\delta,\eta)>0$ such that for every $n> N(\delta,\eta)$, $\|\wh{\lambda}(\psi_\circ)\|_2 \leq C/\sqrt{n}$ with probability larger than $\eta$. Hence,
\begin{multline*}
  \sup_{h_{\theta}\in K_{n,\theta}}\left\|h_{\theta}'\EE_n\left[\tau_{e,i}^{\diamond}(\wh\lambda,\psi_\circ) \wtl{w}_{1,i}\wtl{w}_{i}\right]  - h_{\theta}' \EE[\tau_{e,i}^{\diamond}(\wh\lambda,\psi_\circ) \wtl{w}_{1,i}\wtl{w}_{i}]\right\|\\
  \leq \sup_{h_{\theta}\in K_{n,\theta}}\sup_{\lambda\in B_{1/\sqrt{n}}(0)}\left\|h_{\theta}'\EE_n\left[e^{\lambda'g_{e,i}(\psi_\circ)} \wtl{w}_{1,i}\wtl{w}_{i}\right]  - h_{\theta}' \EE[e^{\lambda'g_{e,i}(\psi_\circ)} \wtl{w}_{1i}\wtl{w}_{i}]\right\|
\end{multline*}
\noindent which converges to zero under Assumption \ref{Ass_3} \textit{(c)} (with $(j,\ell,\ell') = (1,1,1)$), by compactness of $B_{1/\sqrt{n}}(0)$ and by \cite[Lemma 2.4]{NeweyMcFadden1994}. Next, we control term $h_{\theta}' \EE[(\tau_{e,i}^{\diamond}(\wh\lambda,\psi_\circ) - 1) \wtl{w}_{1i}\wtl{w}_{i}]$ in \eqref{eq:LAN:extended:1}. By the CMT and Lemma \ref{lem:2}: $\|e^{\wh\lambda(\psi_\circ)'g_{e,i}(\psi_\circ)}h'\wtl{w}_{1,i}\wtl{w}_i - h'_{\theta}\wtl{w}_{1,i}\wtl{w}_i\|\xrightarrow{p} 0$ uniformly in $h\in K_{n,\theta}$, for every $i=1,\ldots,n$. By this and the DCT, which is valid under Assumption \ref{Ass_3} \textit{(c)} with $(j,\ell,\ell') = (1,1,1)$, it holds that $\EE[(\tau_{e,i}^{\diamond}(\wh\lambda,\psi_\circ) - 1) h' \wtl{w}_{1i}\wtl{w}_i']\xrightarrow{p} 0$ uniformly in $h\in K_{n,\theta}$. The last term in the right hand side of \eqref{eq:LAN:extended:1} can be treated in a similar way by using the upper bound:
\begin{multline*}
  \sup_{h_{\theta}\in K_{n,\theta}}\left\|h_{\theta}'\EE_n\left[\tau_{e,i}^{\diamond}(\wh\lambda,\psi_\circ) \wtl{w}_{1,i}\wtl{w}_i'\wh\lambda(\psi_\circ)\varepsilon_i(\theta_{\circ})\wtl{w}_{i}\right]\right\|
  \leq
  \sup_{h_{\theta}\in K_{n,\theta}}\sup_{\lambda\in B_{1/\sqrt{n}}(0)}\Big\|h_{\theta}'\EE_n\left[\tau_{e,i}^{\diamond}(\lambda,\psi_\circ) \wtl{w}_{1,i}\wtl{w}_i'\lambda(\psi_\circ)\varepsilon_i(\theta_{\circ})\wtl{w}_{i}\right]\\
  - h_{\theta}'\EE\left[\tau_{e,i}^{\diamond}(\lambda,\psi_\circ) \wtl{w}_{1,i}\wtl{w}_i'\lambda(\psi_\circ)\varepsilon_i(\theta_{\circ})\wtl{w}_{i}\right]\Big\|\\
  + \sup_{h_{\theta}\in K_{n,\theta}}\sup_{\lambda\in B_{1/\sqrt{n}}(0)}\left\|h_{\theta}'\EE\left[\tau_{e,i}^{\diamond}(\wh\lambda,\psi_\circ) \wtl{w}_{1,i}\wtl{w}_i'\wh\lambda(\psi_\circ)\varepsilon_i(\theta_{\circ})\wtl{w}_{i}\right]\right\|
\end{multline*}
(where we have used the fact that $\lambda_*(\psi_\circ) = 0$). The two terms in the right hand side of the previous expression converge to zero in probability under Assumption \ref{Ass_3} \textit{(f)}, by \cite[Lemma 2.4]{NeweyMcFadden1994}, Lemma \ref{lem:2}, the CMT and the DMT. By putting all these elements together and by Lemma \ref{lem:Omega} to control $\check\Omega_{\psi}^{\diamond}(\wh\lambda,\psi_\circ)$ in the expression of $\frac{d\wh\lambda(\psi_{\circ})'}{d\psi}$ we get $h' \frac{d\wh\lambda(\psi_{\circ})'}{d\psi}\sqrt{n}\wh g_e(\psi_{\circ})$ is equal to $(h_{\theta}' \EE[\wtl{w}_{1i}\wtl{w}_i'] + (h_v',0'))\Omega_{\psi_\circ}^{-1}\mathbb{G}_n[g_e(w_i,\psi_{\circ})] + o_p(1)$, where the $o_p(1)$ term is uniform in $h\in K_n$.\\
\indent For the other terms of first order in the MVT expansion of $\ell_{n,\psi_* + h/\sqrt{n}}(w_{1:n}) - \ell_{n,\psi_*}(w_{1:n})$ we proceed in a similar way as in the proof of Theorem \ref{thm:stochasticLAN:exogeneity} and, by using similar arguments as above, it follows that these terms converge to zero in probability uniformly in $h\in K_n$. Therefore, the terms of first order in the MVT expansion of $\ell_{n,\psi_* + h/\sqrt{n}}(w_{1:n}) - \ell_{n,\psi_*}(w_{1:n})$ are equal to: $(h_{\theta}' \EE[\wtl{w}_{1i}\wtl{w}_i'] + (h_v',0'))\Omega_{\psi_\circ}^{-1}\mathbb{G}_n[g_e(w_i,\psi_{\circ})] + o_p(1)$, where the $o_p(1)$ is uniform in $h\in K_n$. Finally, since $\mathbb{G}_n[g_e(w_i,\psi_{\circ})]\xrightarrow{d} \mathcal{N}(0,\Omega_{\psi_\circ})$ by the Lindberg-Levy central limit theorem under Assumption \ref{Ass_1_NS} \textit{(a)}, we conclude that the term of first order in the MVT expansion of $\ell_{n,\psi_* + h/\sqrt{n}}(w_{1:n}) - \ell_{n,\psi_*}(w_{1:n})$ is bounded in probability.\\

\indent Now, let us consider the terms of second order in \eqref{eq:3:neglected} (adapted to the extended model) where $\wtl \psi := \psi_\circ + \tau h/\sqrt{n}$ for some $\tau\in[0,1]$ and $h\in K_n$:
\begin{multline}
  \frac{1}{2}h'\sum_{l=1}^d \frac{d^2\wh\lambda_l(\wtl\psi)}{d\psi d\psi'} \wh g_{e,l}(\wtl\psi)h - h'\frac{d\wh\lambda(\wtl\psi)'}{d\psi} \EE_n\left[\frac{dg_e(\wtl\psi)}{d\psi'}\right]h + \frac{1}{2}h'\EE_n\left[\tau_{e,i}(\wh\lambda,\wtl\psi)\frac{dg_e(\wtl\psi)'}{d\psi}\right]\frac{d\wh\lambda(\wtl\psi)}{d\psi'}h\\
     - \frac{1}{2}\EE_n\left[\tau_{e,i}(\wh\lambda,\wtl\psi)\left(h'\frac{dg_e(\wtl\psi)'}{d\psi}\wh\lambda(\wtl\psi)\right)^2 \right] + \frac{1}{2}h'\frac{d\wh\lambda(\wtl\psi)'}{d\psi}\EE_n\left[\tau_{e,i}(\wh\lambda,\wtl\psi)g_{e,i}(\wtl\psi)\wh\lambda(\wtl\psi)'\frac{dg_e(\wtl\psi)}{d\psi'}\right]h\\
     + \frac{1}{2}h'\EE_n\left[\tau_{e,i}(\wh\lambda,\wtl\psi)\frac{dg_e(\wtl\psi)'}{d\psi}\right]\wh\lambda(\wtl \psi) \wh\lambda(\wtl \psi)'\EE_n\left[\tau_{e,i}(\wh\lambda,\wtl\psi)\frac{dg_e(\wtl\psi)}{d\psi'}\right]h, \label{eq:LAN:extended:3}
\end{multline}
\noindent where $\wh\lambda_l(\wtl\psi)$ (resp. $\wh g_{e,l}(\wtl\psi)$) denote the $l$-th element of the vector $\wh\lambda(\wtl\psi)$ (resp. $\wh g_e(\wtl\psi)$) which is $\mathcal{C}^2$ by Lemma \ref{lem:continuity:sample}. Because $h$ is bounded then the sequence $\wtl\psi \rightarrow \psi_\circ$ as $n\rightarrow \infty$ uniformly in $h\in K_n$. To treat term \eqref{eq:LAN:extended:3} we use the following limits that are established by using similar arguments as before. (1) $\wh{g}_e(\wtl\psi) \xrightarrow{p} 0$ uniformly in $h\in K_n$under Assumption \ref{Ass_3_extended} \textit{(a)}, by compactness of $B_{1/\sqrt{n}}(\psi_\circ)$, \cite[Lemma 2.4]{NeweyMcFadden1994}, the CMT and the DCT. (2) By Lemma \ref{lem:technical:exogeneity:1}, $h'\EE_n[\tau_{e,i}(\wh\lambda,\wtl\psi)\wtl{w}_{1i}\wtl{w}_i] \xrightarrow{p} h'\EE_n[\wtl{w}_{1,i}\wtl{w}_{i}]$ uniformly in $h\in K_n$. (3) By Lemma \ref{lem:technical:exogeneity:2}, $\EE_n[\tau_{e,i}(\wh\lambda,\wtl\psi)\left(h'\wtl{w}_{1,i} \wtl{w}_i'\wh\lambda(\wtl\psi)\right)^2] \xrightarrow{p} 0 $ uniformly in $h\in K_n$. (4) By Lemma \ref{lem:technical:exogeneity:3},
\[
\EE_n\left[\tau_i(\wh\lambda,\wtl\psi)g_{e,i}(\wtl\psi)\wh\lambda(\wtl\psi)'\wtl{w}_i\wtl{w}_{1,i}'\right]h\xrightarrow{p} 0
\]
uniformly in $h\in K_n$. (5) By combining (2), (4) and Lemma \ref{lem:Omega} with $\theta_*$ replaced by $\psi_\circ$ we have that $h'\frac{d\wh\lambda(\wtl\psi)'}{d\psi}\xrightarrow{p} (h_{\theta}'\EE[\wtl{w}_{1,i}\wtl{w}_{i}'] + (h_v',0'))\Omega_{\psi_\circ}^{-1}$ uniformly in $h\in K_n$. Finally, $\wh\lambda(\wtl\psi)\xrightarrow{p} 0$ uniformly in $h\in K_n$ by Lemma \ref{lem:2}, \ref{lem:continuity} and $\wtl\psi\rightarrow\psi_\circ$. By replacing these limits in \eqref{eq:LAN:extended:3} we get that \eqref{eq:LAN:extended:3} is equal to: $- \frac{1}{2}h'\EE\left[\frac{dg_e(\wtl\psi)'}{d\psi}\right]\Omega_{\psi_{\circ}}^{-1}\EE_n[\frac{dg_e(\wtl\psi)}{d\psi'}]h + o_p(1) =: -\frac{1}{2}h'V_{\psi_\circ}^{-1}h + o_p(1)$ uniformly in $h\in K_n$, and $V_{\psi_\circ}^{-1}$ is well-defined under Assumption \ref{Ass_1_NS} \textit{(c)} with $\lambda = \lambda_*(\psi_{\circ}) = 0$ and Assumption \ref{Ass_2_NS}.

\begin{flushright}
  $\square$
\end{flushright}

%=============================================================================================================================================================================
\subsection{Posterior consistency}
\begin{thm}[Posterior Consistency - base model]\label{thm_post_consistency_misspecified}
Let Assumptions \ref{Ass_absolute_continuity} - \ref{Ass_identification_base_model} hold. Let the prior on $\theta$ be a continuous probability measure that admits a density with respect to the Lebesgue measure and that is positive on a neighborhood of $\theta_*$. %Let
%\begin{equation*}
%  \Theta_{n}:=\{\Vert \theta -\theta _{\ast }\Vert \leq M_{n}/\sqrt{n}\},
%\end{equation*}
%denote a ball around $\theta _{\ast }$ with radius at most $M_{n}/\sqrt{n}$, where $M_{n}$ is any sequence of positive constants diverging to $\infty $. Assume that there exists a constant $C>0$ such that
%\begin{equation}
%  P\left(\sup_{\theta \in \Theta _{n}^{c}}\frac{1}{n}\sum_{i=1}^{n}\left(\ell_{n,\theta }(w_{i})-\ell _{n,\theta _{\ast }}(w_{i})\right) \leq  - \frac{CM_{n}^{2}}{n}\right) \rightarrow 1\;,\;\text{as}\;n\rightarrow %\infty.
%\label{Ass_identification_rate_contraction_appendix}
%\end{equation}
Then, for any $M_n\rightarrow \infty$,
\begin{equation}
  \pi^n\left(\left. \theta\in H_n;\,\sqrt{n}\|\theta -\theta _{\ast }\| >M_{n}\right\vert w_{1:n}\right) \overset{p}{\rightarrow }0\;,\;\text{as}\;n\rightarrow \infty \;.  \label{eq_posterior_consistency}
\end{equation}
\end{thm}
\proof
Define the events
$A_{n,1} := \left\{\sup_{\theta\in H_n\cap\Theta_n^{c}}\frac{1}{n}\sum_{i=1}^n (\ell_{n,\theta}(w_i) - \ell_{n,\theta_{*}}(w_i)) \leq - CM_n^2/n\right\}$
\noindent and
$$A_{n,2} := \left\{\bigintssss_{\Theta\cap H_n}\frac{\wh p(w_{1:n}|\theta)}{\wh p(w_{1:n}|\theta_*)}\pi(\theta)d(\theta) \geq e^{-CM_n^2/2}\right\}.$$
By Assumption \ref{Ass_identification_base_model}, $P(A_{n,1}^c)\rightarrow 0$ and by \cite[Lemma E.3]{ChibShinSimoni2018}, $P(A_{n,2}^c) \rightarrow 0$. Therefore, by the Law of Total Expectations
\begin{multline}
  \EE\left[\pi^n\left(\left.H_n\cap\Theta_n^c\right|w_{1:n}\right)\right] \leq \\
  \EE\left[ \left.\pi\left(\left.\theta\in H_n; \sqrt{n}\|\theta - \theta_*\| > M_n\right|w_{1:n}\right)\right|A_{n,1}\cap A_{n,2}\right]\,P(A_{n,1}\cap A_{n,2}) + o(1)\\
  \hfill= \EE\left[ \left.\frac{\bigintssss_{H_n\cap\Theta_n^{c}}e^{\sum_{i=1}^n (\ell_{n,\theta}(w_i) - \ell_{n,\theta_{*}}(w_i))}\pi(\theta)d\theta}
  {\bigintssss_{\Theta\cap H_n}e^{\sum_{i=1}^n (\ell_{n,\theta}(w_i) - \ell_{n,\theta_{*}}(w_i))}\pi(\theta)d\theta}\right|A_{n,1}\cap A_{n,2}\right]P(A_{n,1}\cap A_{n,2}) + o(1)\\
  \hfill \leq e^{-CM_n^2}\pi(\Theta_n^{c})\EE\left[\left.\left(\bigintssss_{\Theta \cap H_n}\frac{\wh p(w_{1:n}|\theta)}{\wh p(w_{1:n}|\theta_*)}\pi(\theta)d\theta\right)^{-1}\right|A_{n,1}\cap A_{n,2}\right] + o(1)\\
  \leq e^{-CM_n^2}e^{CM_n^2/2}\pi(\Theta_n^c) + o(1) = o(1),
\end{multline}
\noindent where the convergence to zero follows from $M_n \rightarrow \infty$. This proves the result of the theorem.
\begin{flushright}
  $\square$
\end{flushright}

\begin{thm}[Posterior Consistency - extended model]\label{thm_post_consistency_extended_model}
Let Assumptions \ref{Ass_absolute_continuity} - \ref{Ass_3} and Assumption \ref{Ass_identification_extended_model} hold. Let the prior on $\psi$ be a continuous probability measure that admits a density with respect to the Lebesgue measure and that is positive on a neighborhood of $\psi_\circ$. %Let
%\begin{equation*}
%  \Theta_{n}:=\{\Vert \theta -\theta _{\ast }\Vert \leq M_{n}/\sqrt{n}\},
%\end{equation*}
%denote a ball around $\theta _{\ast }$ with radius at most $M_{n}/\sqrt{n}$, where $M_{n}$ is any sequence of positive constants diverging to $\infty $. Assume that there exists a constant $C>0$ such that
%\begin{equation}
%  P\left(\sup_{\theta \in \Theta _{n}^{c}}\frac{1}{n}\sum_{i=1}^{n}\left(\ell_{n,\theta }(w_{i})-\ell _{n,\theta _{\ast }}(w_{i})\right) \leq  - \frac{CM_{n}^{2}}{n}\right) \rightarrow 1\;,\;\text{as}\;n\rightarrow %\infty.
%\label{Ass_identification_rate_contraction_appendix}
%\end{equation}
Then, for any $M_n\rightarrow \infty$,
\begin{equation}
  \pi^n\left( \left. \psi\in H_n\times \mathcal{V} ;\sqrt{n}\|\psi -\psi_{\circ}\| >M_{n}\right\vert w_{1:n}\right) \overset{p}{\rightarrow }0\;,\;\text{as}\;n\rightarrow \infty \;.  \label{eq_posterior_consistency_extended}
\end{equation}
\end{thm}
\proof
The proof proceeds as the proof of Theorem \ref{thm_post_consistency_misspecified} and then it is omitted.
\begin{flushright}
  $\square$
\end{flushright}

\subsection{Bernstein-von Mises theorems for the base and the extended models}
%=============================================================================================================================================================================
\begin{thm}[Bernstein-von Mises in the base model]\label{thm_BvM_misspecified}
Assume that the conditions of Theorems \ref{thm:stochasticLAN:endogeneity} and \ref{thm_post_consistency_misspecified} hold. Then, the sequence of posteriors of $h:=\sqrt{n}(\theta - \theta_*)$ converge in total variation towards a Normal distribution, that is,
\begin{equation}
  \left\|\pi_h^n(\sqrt{n}(\theta -\theta_*)|w_{1:n}) - \mathcal{N}_{\Delta_{n,\theta_*},V_{\theta_*}}\right\|_{TV} \overset{p}{\rightarrow }0,
\end{equation}
where %$B\subseteq \Theta $ is any Borel set,
$\Delta_{n,\theta_*} := \frac{1}{\sqrt{n}}\sum_{i=1}^{n}V_{\theta _*}\dot{\ell}_{n,\theta_{\ast}}(w_{i}) + o_p(1) = \mathcal{O}_{p}(1)$, and $V_{\theta_*}$ is a positive definite matrix equal to
the inverse of:
\begin{multline*}
  V_{\theta_*}^{-1} = \mathbf{E}^{Q^*(\theta_*)}\left[\tilde{w}_{1,i}\tilde{w}_{i}'(I+\lambda_*(\theta_{*})\varepsilon _{i}\tilde{w}_{i}')\right] \left(\mathbf{E}^{Q^*(\theta_*)}\left[ \varepsilon_{i}^{2}\tilde{w}_{i}\tilde{w}_{i}'\right] \right)^{-1}\\ \hfill
  \times\left(2\mathbf{E}[\tilde{w}_{i}\tilde{w}_{1,i}']-\mathbf{E}^{Q^*(\theta_*)}\left[\tilde{w}_{1,i}\tilde{w}_{i}'(I+\lambda_*(\theta_{*})\varepsilon _{i}\tilde{w}_{i}')\right] \right) \\
  -\sum_{j=1}^{d_x}\ddot{\lambda}_{* ,j}(\theta_*)\mathbf{E}[\varepsilon_{i}x_{i,j}]+ \mathbb{V}ar_{Q^*(\theta_*)}[\tilde{w}_{1,i}\tilde{w}_{i}'\lambda_*(\theta_*)],
\end{multline*}
\noindent where $\mathbb{V}ar_{Q^{*}(\theta_*)}$ denotes the variance taken with respect to the distribution $Q^{*}(\theta_*)$.
\end{thm}

\proof
The proof of this theorem proceeds as the proof of \cite[Theorem 2.2]{ChibShinSimoni2018}. It depends on two intermediate results: the posterior consistency result of Theorem \ref{thm_post_consistency_misspecified} and the stochastic LAN expansion \eqref{eq_stochastic_LAN_neglected_endogeneity} established in Theorem \ref{thm:stochasticLAN:endogeneity} below.
\begin{flushright}
  $\square$
\end{flushright}

%==============================================================================================================================================================
\begin{cor}[Bernstein-von Mises in the base model under exogeneity]\label{Cor_2_1}
Let $\theta_\circ$ denote the true value of $\theta $ and assume $\EE[\varepsilon_i(\theta_\circ) x_i] = 0$. Let Assumptions \ref{Ass_0_NS} - \ref{Ass_3} hold with $\theta_*$ replaced by $\theta_{\circ}$, $\lambda_*(\theta_\circ)$ replaced by zero, and the matrix in
Assumption \ref{Ass_1_NS} \textit{(c)} replaced by the matrix $\mathbf{E}[\varepsilon_i(\theta_\circ)^2\widetilde w_i \widetilde w_i']$. Suppose that Assumption \ref{Ass_identification_base_model} holds and that the prior on $\theta$ is a continuous probability measure that admits a density with respect to the Lebesgue measure and that is positive on a neighborhood of $\theta_\circ$. Then the sequence of posterior distributions of $h:=\sqrt{n}(\theta - \theta_*)$ converge in total variation towards a Normal distribution, that is,
\begin{equation}
\left\|\pi_h^n(\sqrt{n}(\theta - \theta_{\circ})|w_{1:n})-\mathcal{N}_{\Delta_{n,\theta_{\circ}},V_{\theta_{\circ}}}\right\|_{TV}\overset{p}{\rightarrow }0,
\end{equation}%
where %$B\subseteq \Theta$ is any Borel set,
$\Delta_{n,\theta_{\circ}} := \frac{1}{\sqrt{n}}\sum_{i=1}^{n}V_{\theta_{\circ}}\mathbf{E}\left[ \widetilde w_{1,i} \widetilde w_{i}'\right]\left(\mathbf{E}[\varepsilon_i^2 \widetilde w_{i} \widetilde w_{i}]\right)^{-1} \varepsilon_{i}\tilde{w}_{i}
$, and $V_{\theta_{\circ}}$ is the inverse of
$V_{\theta_{\circ}}^{-1} := \mathbf{E}\left[ \widetilde w_{1,i}\widetilde w_{i}'\right] \left(\mathbf{E}[\varepsilon_i^2 \widetilde w_{i}\widetilde w_{i}']\right)^{-1}\mathbf{E}\left[ \widetilde w_{i} \widetilde w_{1,i}'\right]$.
\end{cor}

\proof
The proof of this Corollary proceeds as the proof of Theorem \ref{thm_BvM_misspecified}. The only differences are: the posterior consistency result, where we have to replace the pseudo-true value $\theta_*$ by the true value $\theta_{\circ}$, and the stochastic LAN expansion which we provide, for the correctly specified case, in Theorem \ref{thm:stochasticLAN:exogeneity} above.

\begin{flushright}
  $\square$
\end{flushright}

%==============================================================================================================================================================
\begin{thm}[Bernstein-von Mises in the extended model]\label{thm_BvM:extended}
Let $\psi_{\circ} := (\theta_{\circ}',v_{\circ}')'$ denote the true value of $\psi$ and let Assumptions \ref{Ass_0_NS} - \ref{Ass_3} hold with $\theta_*$ replaced by $\psi_{\circ}$, $\lambda_*(\psi_\circ)$ replaced by zero, and the matrix in Assumption \ref{Ass_1_NS} \textit{(c)} replaced by the matrix $\mathbf{E}[\varepsilon_i(\psi_\circ)^2\widetilde w_i \widetilde w_i']$. Suppose that Assumption \ref{Ass_identification_extended_model} holds. Let the prior on $\psi$ be a continuous probability measure that admits a density with respect to the Lebesgue measure and that is positive on a neighborhood of $\psi_\circ$. Then the sequence of posterior distributions converge in total variation towards a Normal distribution, that is,
\begin{equation}
  \left\| \pi ^{n}(\sqrt{n}(\psi -\psi _{\circ})|w_{1:n})-\mathcal{N}_{\Delta _{n,\psi _{\circ}},V_{\psi_{\circ}}}\right\|_{TV} \overset{p}{\rightarrow }0,
\end{equation}%
where %$B\subseteq \Psi $ is any Borel set,
$\Delta_{n,\psi_{\circ}}=\frac{1}{\sqrt{n}}\sum_{i=1}^{n}V_{\psi _{\circ}}\mathbf{E}\left[ \frac{dg_e(w_i,\psi_{\circ})'}{d\psi }\right] \Omega_{\psi_{\circ}}^{-1}\left(\varepsilon_{i}\tilde{w}_{i}-\wtl{v}_\circ\right)$,
with $\wtl v_\circ := (v_\circ',0',0')'$ and $v_\circ := \EE[\varepsilon_i x_i]$, and $V_{\psi_\circ}$ is the inverse of
\begin{equation*}
V_{\psi_\circ}^{-1} = \mathbf{E}\left[ \frac{dg_e(w_{i},\psi_{\circ})'}{d\psi }\right] \Omega_{\psi_{\circ}}^{-1}\mathbf{E}\left[ \frac{dg_e(w_{i},\psi_{\circ})}{d\psi'}\right]
\end{equation*}%
\noindent with $\Omega_{\psi_{\circ}} := \mathbf{E}\left[ g_e(w_{i},\psi_{\circ})g_e(w_{i},\psi_{\circ})'\right] = \mathbf{E}[\varepsilon_i^2\widetilde w_i \widetilde w_i'] - \wtl v_\circ \wtl v_\circ'$ and $\mathbf{E}\left[\frac{dg_e(w_{i},\psi_{\circ})'}{d\psi }\right] = -\begin{pmatrix}
  \EE[\wtl{w}_{1i}x_i'] & \EE[\wtl{w}_{1i}z_i']\\
  I_{d_x} & 0
\end{pmatrix}$.
\end{thm}

\proof
The proof of this theorem proceeds as the proof of Corollary \ref{Cor_2_1} with $\theta_\circ$ replaced by $\psi_\circ$ and it uses the stochastic LAN expansion in Theorem \ref{thm:stochasticLAN:extended} above.
\begin{flushright}
  $\square$
\end{flushright}

\subsection{Technical results for the base model}
In this section we shorten the notation and denote $g_i(\theta) := g_{b}(w_i,\theta)$. The result of \ref{lem:continuity:sample} is valid asymptotically, that is for $n \geq N$ where $N$ is a suitably large positive integer.

%===========================================================================================================
\begin{lem}\label{lem:continuity:sample}
  Suppose Assumptions \ref{Ass_absolute_continuity}, \ref{ass:feasibility} and \ref{Ass_1_NS} \textit{(c)} hold. Let $\wh\lambda_*$ be a solution in $\lambda$ of the equation $F_n(\lambda,\theta_*) = 0$, where $F_n(\lambda,\theta) := - \sum_{i=1}^n \frac{e^{\lambda' \widetilde w_{i} \varepsilon_i(\theta)}}{\sum_{j=1}^n e^{\lambda' \widetilde w_{j} \varepsilon_j(\theta)}}\varepsilon_i(\theta) \widetilde w_i$. Then, there is a neighborhood $\mathcal{U}\subset\mathbb{R}^d$ of $\wh\lambda_*$ and a $N\geq 1$ such that $\forall n\geq N$ there exists a unique continuous mapping $\wh\lambda$ of $B_{*,n}$ into $\mathcal{U}$ that satisfies:
  \begin{itemize}
    \item[\textit{(i)}] $\wh{\lambda}(\theta_*) = \wh\lambda_*$;
    \item[\textit{(ii)}] $F_n(\wh\lambda(\theta),\theta) = 0$ for every $\theta\in B_{*,n}$;
%    \item[\textit{(i)}] $F(\wh\lambda(\theta),\theta) = 0$ for every $\theta\in \bigcap_{\{M_n; M_n\rightarrow \infty\}} B(\theta_*,M_n n^{-1/2})$;
%    \item[\textit{(ii)}] $\wh{\lambda}(\theta)$ is the unique solution to $F(\lambda,\theta) = 0$ lying in $\mathcal{U}$, for every $\theta\in \bigcap_{\{M_n; M_n\rightarrow \infty\} } B(\theta_*,M_n n^{-1/2})$;
    \item[\textit{(iii)}] $\wh\lambda\in\mathcal{C}^2$ on $B_{*,n}$ and $\forall \theta\in B_{*,n}$,
    \begin{equation}
      \frac{\partial \wh{\lambda}(\theta)}{\partial \theta'} = \check{\Omega}^{\diamond}(\wh\lambda,\theta)^{-1}\EE_n \left[\tau_i^{\diamond}(\wh\lambda,\theta) w_{1,i}\wtl w_i' (I + \wh\lambda(\theta) g_i(\theta)')\right].
    \end{equation}
  \end{itemize}
\end{lem}
\proof
We intend to apply the Implicit Function Theorem (\textit{e.g.} \cite[Theorem 10.2.3]{Dieudonne1969}) and for this we only need to check that its conditions are verified. First, the function $(\lambda,\theta) \mapsto F(\lambda,\theta)$ is at least two times continuously differentiable on $\mathbb{R}^{d+p}$. Second, under Assumptions \ref{Ass_absolute_continuity} and \ref{ass:feasibility}, there exists a $\wh\lambda$ such that $F(\wh\lambda,\theta) = 0$ with probability approaching $1$ for every $\theta\in B_{*,n}$. Moreover, the first partial derivative $\partial F(\wh\lambda,\theta)/\partial \lambda$ evaluated at the previous $(\theta,\wh\lambda)$ is, with probability approaching $1$:
\begin{equation}
  \frac{\partial{F(\wh\lambda,\theta)}}{\partial \wh\lambda} %:= -\left(\sum_{i=1}^n \frac{e^{\wh\lambda' \widetilde w_{i} \varepsilon_i(\theta_*)}}{\sum_{j=1}^n[e^{\wh\lambda' \widetilde w_{j} \varepsilon_j(\theta_*)}]} \varepsilon_i(\theta_*)\widetilde w_{i} \widetilde w_{i}'  - \sum_{i=1}^n \frac{e^{\wh\lambda' \widetilde w_{i} \varepsilon_i(\theta_*)}}{\sum_{j=1}^n[e^{\wh\lambda' \widetilde w_{j} \varepsilon_j(\theta_*)}]} \widetilde w_{i} \varepsilon_i(\theta_*) \sum_{j=1}^n \frac{e^{\wh\lambda'\widetilde w_{j} \varepsilon_j(\theta_*)}}{\sum_{k=1}^n[e^{\wh\lambda' \widetilde w_{k} \varepsilon_k(\theta_*)}]}\widetilde w_{i} \varepsilon_i(\theta_*)\right)\\
  = -\sum_{i=1}^n \frac{e^{\wh\lambda' \widetilde w_{i} \varepsilon_i(\theta)}}{\sum_{j=1}^n[e^{\wh\lambda' \widetilde w_{j} \varepsilon_j(\theta)}]} \varepsilon_i(\theta)^2\widetilde w_{i} \widetilde w_{i} '.
\end{equation}
\noindent where we have used the fact that $P(F(\wh\lambda,\theta) = 0)\rightarrow 1$. Under Assumption \ref{Ass_1_NS} \textit{(c)}, the matrix $\frac{\partial{F(\wh\lambda,\theta)}}{\partial \wh\lambda}$ is not zero and so it is invertible.\\
\indent Then, by the Implicit Function Theorem, there exists an open neighborhood $\mathcal{U}$ of $\wh\lambda(\theta_*)$ and a $N\geq 1$ such that $\forall n\geq N$ :\\
(i) to every $\theta\in B_{*,n}$ it corresponds a unique $\wh\lambda\in\mathcal{U}$ such that $F(\wh\lambda,\theta) = 0$;\\
(ii) this $\wh\lambda$ can be written as $\wh\lambda = \wh\lambda(\theta)$ and it is such that $F(\wh\lambda(\theta),\theta) = 0$, with $\theta\in B_{*,n}$;\\
(iii) $\wh\lambda\in\mathcal{C}^2$ on $B_{*,n}$ and it holds:
$$\frac{\partial \wh{\lambda}(\theta)}{\partial \theta'} =  \left( \frac{\partial F(\wh\lambda,\theta)}{\partial \lambda'} \right)^{-1} \frac{\partial F(\wh\lambda,\theta)}{\partial \theta'} \equiv \check{\Omega}^{\diamond}(\wh\lambda,\theta)^{-1}\EE_n \left[\tau_i^{\diamond}(\wh\lambda,\theta) w_{1,i}\wtl w_i' (I + \wh\lambda(\theta) g_i(\theta)')\right].$$
\begin{flushright}
$\square$
\end{flushright}
%------------------------------------------------------------------------------------------------------------------------------------------------------------------------------
\begin{lem}[Continuity]\label{lem:continuity}
  Let $\lambda_* := \argmin \EE\left[e^{\lambda'g_i(\theta_*)}\right]$ and suppose Assumption \ref{ass:feasibility} holds true. Then, there is a neighborhood $\mathcal{U}\subset\mathbb{R}^d$ of $\lambda_*$ and a $N\geq 1$ such that $\forall n\geq N$ there exists a unique continuous mapping $\lambda_*$ of $B_{*,n}$ into $\mathcal{U}$ that satisfies:
  \begin{itemize}
    \item[\textit{(i)}] $\lambda_*(\theta_*) = \lambda_*$;
    \item[\textit{(ii)}] $\EE[e^{\lambda_*(\theta)' \widetilde w_{i} \varepsilon_i(\theta)}\varepsilon_i(\theta) \widetilde w_i] = 0$ for every $\theta\in B_{*,n}$;
    \item[\textit{(iii)}] $\lambda_*\in\mathcal{C}^2$ on $B_{*,n}$ and $\forall \theta\in B_{*,n}$,
    \begin{equation}
      \frac{\partial \lambda_*(\theta)}{\partial \theta'} = \Omega_*^{\diamond}(\theta_*)^{-1}\EE_n \left[\tau_i^{\diamond}(\lambda_*,\theta) w_{1,i}\wtl w_i' (I + \lambda_*(\theta) g_i(\theta)')\right].
    \end{equation}
    \item[\textit{(iv)}] $\theta\mapsto \EE\left[e^{\lambda_*(\theta)'g(w_{i},\theta )}\right]$ is continuous on $B_{*,n}$.
  \end{itemize}
%    \begin{enumerate}
%    \item[(i).] $\theta\mapsto \lambda_*(\theta)$ is continuous on $\bigcap_{\{M_n; M_n\rightarrow \infty\}} B(\theta_*,M_n n^{-1/2})$;
%    \item[(ii).] $\theta\mapsto \EE\left[e^{\lambda_*(\theta)'g(w_{i},\theta )}\right]$ is continuous on $\bigcap_{\{M_n; M_n\rightarrow \infty\}} B(\theta_*,M_n n^{-1/2})$.
%  \end{enumerate}
\end{lem}

\proof
The proof of \textit{(i)-(iii)} proceeds similarly to the proof of Lemma \ref{lem:continuity:sample} by applying the Implicit Function Theorem, and then it is omitted. Result \textit{(iv)} follows from continuity of $\theta \mapsto \lambda_*(\theta)$ on $B_{*,n}$ and continuity of the function $\theta\mapsto \EE\left[e^{\lambda'g_i(\theta)} g_i(\theta)\right]$.
%By strict convexity of $\lambda\mapsto \EE\left[e^{\lambda'g_i(\theta)}\right]$, the minimizer $\lambda_*(\theta)$ is unique and it is characterized as the implicit solution of $\EE\left[e^{\lambda'g_i(\theta)} g_i(\theta)\right] = 0$. Since the function $(\lambda,\theta)\mapsto \EE\left[e^{\lambda'g_i(\theta)} g_i(\theta)\right]$ is continuously differentiable then by the implicit function theorem $\theta \mapsto \lambda_*(\theta)$ is continuous. This shows \textit(i). By this and continuity of the function $\theta\mapsto \EE\left[e^{\lambda'g_i(\theta)} g_i(\theta)\right]$, statement \textit{(ii)} follows.
\begin{flushright}
  $\square$
\end{flushright}

%=========================================================================================================
\begin{lem}\label{lem:2}
    Suppose Assumptions \ref{Ass_absolute_continuity}, \ref{ass:feasibility}, \ref{Ass_1_NS} (a), and \ref{Ass_3} (d)-(e) hold. Then, for every given $\theta\in B_{*,n}$,
    $$\sqrt{n}\|\wh\lambda(\theta) - \lambda_*(\theta)\|_2 = \mathcal{O}_p(1)$$
    or, equivalently, for any $\eta >0$ there exists a finite $\delta>0$ and a finite $N(\delta,\eta)>0$ such that: $\forall n > N(\delta,\eta)$,
    $$P\left(\sqrt{n}\|\wh\lambda(\theta) - \lambda_*(\theta)\|_2 < \delta\right) > \eta.$$
\end{lem}
\proof To prove the lemma we intend to apply \cite[Corollary 3.2.6]{VanDerVaartWellner1996} with $\mathbb{M}_n = -\EE_n\left[e^{\lambda'g_i(\theta)}\right]$. Here, $\theta\in B_{*,n}$ is kept fixed and we see $\mathbb{M}_n$ as a function of the data $w_i$ and of $\lambda$. For $\lambda_* = \arg\max_{\lambda\in\mathbb{R}^d}\mathbb{M}$, where $\mathbb{M} :=  -\EE\left[e^{\lambda'g_i(\theta)}\right]$, define the set of functions $\mathcal{M}_{\delta}(\theta) := \left\{w_i\mapsto e^{\lambda_*'g_i(\theta)} - e^{\lambda'g_i(\theta)}; \lambda \in B_{\delta}(\lambda_*) \right\}$ indexed by $\lambda\in B_{\delta}(\lambda_*)$. Because the functions in $\mathcal{M}_{\delta}(\theta)$ are Lipschitz in the index parameter $\lambda\in B_{\delta}(\lambda_*)$ then, by the MVT and the CS there exists a $\tau\in(0,1)$ such that $\wtl\lambda_{12} := \lambda_1 + \tau(\lambda_2 - \lambda_1)$ with $\lambda_1,\lambda_2\in B_{\delta}(\lambda_*)$ satisfies
$$|e^{\lambda_2'g_i(\theta)} - e^{\lambda_1'g_i(\theta)}| \leq 2\delta \|g_i(\theta) e^{\wtl \lambda_{12}'g_i(\theta)}\|_2$$
for every $w_i$ and every fixed $\theta\in B_{*,n}$. Hence, an envelope function of $\mathcal{M}_{\delta}(\theta)$, denoted by $M_{\delta}(\theta)$, is given by
\[
M_{\delta}(\theta) = 2\delta \|g_i(\theta)\|_2 \sup_{\lambda_1,\lambda_2\in B_{\delta}(\lambda_*)} e^{\wtl \lambda_{12}'g_i(\theta)}
\]
(see the discussion before \cite[Theorem 2.7.11]{VanDerVaartWellner1996} about envelope functions). Moreover, there exists a constant $K$ such that the bracketing number $N_{[\,]}(\varepsilon \|M_{\delta}(\theta)\|_{P,2},\mathcal{M}_{\delta}(\theta), \|\cdot\|_{P,2})$ of $\mathcal{M}_{\delta}(\theta)$ is upper bounded by $K(2\delta)^d/\varepsilon^d$ by \cite[Example 19.7]{VanDerVaart2000} for every $0<\varepsilon < 2\delta $ if $\|M_{\delta}(\theta)\|_{P,2}<C/\sqrt{n}$, which is guaranteed for $\delta = 1/\sqrt{n}$ under Assumption \ref{Ass_3} (d). Hence, the bracketing integral $J_{[\,]}(1,\mathcal{M}_{\delta}(\theta),\|\cdot\|_{P,2}) := \bigintsss_0^{1}\sqrt{1 + \log N_{[\,]}(\varepsilon \|M_{\delta}(\theta)\|_{P,2},\mathcal{M}_{\delta}(\theta), \|\cdot\|_{P,2})}d\varepsilon$ is bounded and the class of functions $\mathcal{M}_{\delta}(\theta)$ is $P$-Glivenko-Cantelli (see \textit{e.g.} \cite[Theorem 2.5.6]{VanDerVaartWellner1996} and its simplified version discussed on \cite[pages 243-244]{VanDerVaartWellner1996}). This ensure that $|\mathbb{M}_n - \mathbb{M}| \rightarrow 0$ in $P$-probability uniformly over $\mathcal{M}_{\delta}(\theta)$ for every $\theta\in B_{*,n}$.\\
\indent Therefore, by \cite[Corollary 3.2.3]{VanDerVaartWellner1996}, $\wh\lambda(\theta) \xrightarrow{p} \lambda_*(\theta)$ for every fixed $\theta\in B_{*,n}$. Finally, to apply \cite[Corollary 3.2.6]{VanDerVaartWellner1996} we need to upper bound $\EE (\sup_{f\in \mathcal{M}_{1/\sqrt{n}}}\left|\sqrt{n}\EE_n[f - \EE(f)]\right|)$. According with the discussion following \cite[Corollary 3.2.6]{VanDerVaartWellner1996} (and because $J_{[\,]}(1,\mathcal{M}_{\delta},\|\cdot\|_{P,2})$ is bounded by the previous argument), the latter is upper bounded by a constant times $\|M_{1/\sqrt{n}}\|_{P,2}$. Since $\|M_{1/\sqrt{n}}\|_{P,2} \leq C/\sqrt{n}$, then $r_n^4 \|M_{1/r_n}\|_{P,2}^2 \leq n$ yields $r_n \asymp \sqrt{n}$. Therefore, the result of \cite[Corollary 3.2.6]{VanDerVaartWellner1996} holds with $r_n \leq \sqrt{n}$. This proves the result of the Lemma.

\begin{flushright}
  $\square$
\end{flushright}

%=========================================================================================================
\begin{lem}\label{lem:1}
    Let Assumptions \ref{Ass_absolute_continuity}, \ref{ass:feasibility}, \ref{Ass_1_NS} (a), \ref{Ass_3} (e) hold. Then,
    $$\sqrt{n}(\wh\lambda(\theta_*) - \lambda_*(\theta_*)) = -\Omega_*^{\dagger}(\theta_*)^{-1}\mathbb{G}_n[\tau_i^{\dagger}(\lambda_*,\theta_*)\varepsilon_i(\theta_*)\wtl{w}_i] + o_p(1).$$
\end{lem}
\proof
Let $\wh\lambda_*\equiv \wh\lambda(\theta_*)$ and $\lambda_*\equiv \lambda_*(\theta_*)$. By a second order MVT expansion of the function $\lambda \mapsto \EE_n[e^{\lambda' \widetilde w_{i} \varepsilon_i(\theta_*)}\varepsilon_i(\theta_*) \widetilde w_i]$  around $\lambda_*$, evaluated at $\wh\lambda$, we get: $\wh\lambda_* - \lambda_* = -\Omega_*^{\diamond}(\theta_*)^{-1}\EE_n[\tau_i^{\diamond}(\lambda_*,\theta_*)\varepsilon_i(\theta_*)\wtl{w}_i] + \mathcal{O}_p(1/n)$, where the convergence in probability of $\check{\Omega}_*^{\diamond}(\lambda_*,\theta_*)$ towards $\Omega_*^{\diamond}(\theta_*)$ follows from the Law of Large Number, and the $\mathcal{O}_p(1/n)$ term follows from Lemma \ref{lem:2} and Assumption \ref{Ass_3} \textit{(e)} (with $(\ell,\ell',i) = (1,3,1)$) which allow to control the quadratic term in the MVT expansion. Moreover, $\EE[\tau_i^{\diamond}(\lambda_*,\theta_*)\varepsilon_i(\theta_*)\wtl{w}_i] = 0$. Therefore,
\begin{multline*}
  \sqrt{n} \left(\wh\lambda(\theta_*) - \lambda_*(\theta_*)\right) = -\Omega_*^{\diamond}(\theta_*)^{-1}\mathbb{G}_n[\tau_i^{\diamond}(\lambda_*,\theta_*)\varepsilon_i(\theta_*)\wtl{w}_i] + o_p(1)\\
  = -\Omega_*^{\dagger}(\theta_*)^{-1}\mathbb{G}_n[\tau_i^{\dagger}(\lambda_*,\theta_*)\varepsilon_i(\theta_*)\wtl{w}_i] + o_p(1).
\end{multline*}
\begin{flushright}
  $\square$
\end{flushright}

\begin{lem}\label{lem:asymp:normality:FOC}
    Let Assumptions \ref{Ass_absolute_continuity} - \ref{Ass_3} hold. Then, for every $h$ belonging to a closed ball $K_n\subset \mathbb{R}^p$ centred on zero with radius $M_n\rightarrow\infty$,
\begin{multline*}
  h' \sqrt{n}\left(\frac{d\wh\lambda(\theta_*)'}{d\theta}\wh g(\theta_*) - \frac{d \lambda_{*}(\theta_{*})'}{d\theta} \EE[g_i(\theta_{*})]\right) - h'\sqrt{n}\left(\EE_n[\wtl{w}_{1,i}\wtl{w}_{i}']\wh\lambda(\theta_*) - \EE[\wtl{w}_{1i}\wtl{w}_i']\lambda_{*}(\theta_{*})\right)\\
  + \sqrt{n}h'\left(\EE_n\left[\tau_i(\wh\lambda,\theta_*)\wtl{w}_{1,i}\wtl{w}_{i}'\right] \wh\lambda_*(\theta) - \EE[\tau_i^{\dagger}(\lambda_*,\theta_{*})\wtl{w}_{1i}\wtl{w}_i']\lambda_*(\theta_{*})\right)
\end{multline*}
\noindent is asymptotically normal with zero mean and variance $h'H_*h$ defined as
\begin{multline*}
  h'H_*h := h'A\EE^{Q^*(\theta_*)}\left[\tau_i^{\dagger}(\lambda_*,\theta_*)\varepsilon_{i*}^2 \wtl w_i\wtl w_i'\right] A'h\\
  + Var \Big(\tau_i^{\dagger}(\lambda_*,\theta_*)h'\wtl w_{1,i}\wtl w_i'(I + \lambda_*(\theta_*) g_i(\theta_*)')B + h'a' (1 - \tau_i^{\dagger}(\lambda_*,\theta_*))\varepsilon_{i*}(\theta_*)\wtl{w}_i\\
  + h'a'\tau_i^{\dagger}(\lambda_*,\theta_*)\varepsilon_i(\theta_*)\wtl{w}_i'\Omega_*^{\dagger}(\theta_*)^{-1}\EE^{Q^*(\theta_*)}\left[\varepsilon_{i*}^3 \wtl w_i \wtl w_i \wtl w_i'  \right]B\\
  - h'\tau_i^{\dagger}(\lambda_*,\theta_*)\varepsilon_{i*}\wtl{w}_i'\Omega_*^{\dagger}(\theta_*)^{-1}\EE^{Q^*(\theta_*)}\left[g_i(\theta_*)\wtl w_{1,i} \wtl w_i'(I + \lambda_*(\theta_*) g_i(\theta_*)')\right]B\\
  - h'a' \tau_i^{\dagger}(\lambda_*,\theta_*)\varepsilon_{i*}^2 \wtl w_{i} \wtl w_i'B - \tau_i^{\dagger}(\lambda_*,\theta_*)\EE^{Q_*(\theta_*)}\left[h'\wtl{w}_{1,i}\wtl{w}_i'\right] \lambda_*(\theta_*)\\
  - \left(1 - \tau_i^{\dagger}(\lambda_*,\theta_*)\right)h'\wtl{w}_{1,i}\wtl{w}_{i}'\lambda_*(\theta_*)\Big)\\
  + 2 h'A\EE^{Q^*(\theta_*)}\left[\tau_i^{\dagger}(\lambda_*,\theta_*)\varepsilon_{i*} \wtl w_i \wtl w_i'(I + \lambda_*(\theta_*) g_i(\theta_*)')B\wtl w_{1,i}'h\right]\\
  + 2 h'A\EE^{Q^*(\theta_*)}\left[\left(1 - \tau_i^{\dagger}(\lambda_*,\theta_*)\right)\varepsilon_{i*}^2  \wtl w_i \wtl w_i'\right]a h\\
  + 2 h'A \EE^{Q^*(\theta_*)}\left[\tau_i^{\dagger}(\lambda_*,\theta_*)\varepsilon_{i*}^2 \wtl w_i \wtl{w}_i'\right]\Omega_*^{\dagger}(\theta_*)^{-1}\EE^{Q^*(\theta_*)}\left[\varepsilon_{i*}(\theta_*)^3 \wtl w_i \wtl w_i'  B\wtl w_i'\right]a h \\
  - 2 h'A\EE^{Q^*(\theta_*)}\left[\tau_i^{\dagger}(\lambda_*,\theta_*)\varepsilon_{i*}^2 \wtl w_i \wtl{w}_i'\right]\Omega_*^{\dagger}(\theta_*)^{-1}\EE^{Q^*(\theta_*)}\left[g_i(\theta_*)\wtl w_i'(I + \lambda_*(\theta_*) g_i(\theta_*)')B\wtl w_{1,i}'\right]h\\
  - 2 h'A\EE^{Q^*(\theta_*)}\left[\tau_i^{\dagger}(\lambda_*,\theta_*)\varepsilon_{i*}^3 \wtl w_i \wtl w_i' B\wtl w_{i}'\right]  a h \\
  - 2 h'A \EE^{Q^*(\theta_*)}\left[\tau_i^{\dagger}(\lambda_*,\theta_*) \varepsilon_{i*} \wtl w_i \right] \lambda_*(\theta_*)'\EE^{Q^*(\theta_*)}[\wtl{w}_i\wtl{w}_{1,i}'h]\\
  - 2 h'A \EE^{Q^*(\theta_*)}\left[\left(1 - \tau_i^{\dagger}(\lambda_*,\theta_*)\right) \varepsilon_{i*} \wtl w_i \lambda_*(\theta_*)' \wtl{w}_{i} \wtl{w}_{1,i}'\right] h,
\end{multline*}
\noindent where $H_*$ is a positive definite matrix, $A := \EE[\wtl w_{1i} \wtl w_i'] \Omega_*^{\dagger}(\theta_*)^{-1}$ is a $(p\times d)$ matrix, $B := \Omega_*^{\dagger}(\theta_*)^{-1}\EE[g_i(\theta_*)]$ is a $d$-vector, and $a := \Omega_*^{\dagger}(\theta)^{-1}\EE^{Q^*(\theta)} \left[(I + g_i(\theta)\lambda_*(\theta) ')\wtl w_i \wtl w_{1,i}'\right]$.
\end{lem}

\proof
In this proof we denote $\varepsilon_{i*} := \varepsilon_i(\theta_*)$ to simplify the expressions. We start by defining the terms $\mathcal{A}_1$, $\mathcal{A}_2$ and $\mathcal{A}_3$ as follows:
\begin{multline*}
  h' \sqrt{n}\left(\frac{d\wh\lambda(\theta_*)'}{d\theta}\wh g(\theta_*) - \frac{d \lambda_{*}(\theta_{*})'}{d\theta} \EE[g_i(\theta_{*})]\right) - h'\sqrt{n}\left(\EE_n[\wtl{w}_{1,i}\wtl{w}_{i}']\wh\lambda(\theta_*) - \EE[\wtl{w}_{1,i}\wtl{w}_i']\lambda_{*}(\theta_{*})\right)\\
    + \sqrt{n}h'\left(\EE_n\left[\tau_i(\wh\lambda,\theta_*)\wtl{w}_{1,i}\wtl{w}_{i}'\right] \wh\lambda(\theta_*) - \EE[\tau_i^{\dagger}(\lambda_*,\theta_{*})\wtl{w}_{1,i}\wtl{w}_i']\lambda_*(\theta_{*})\right) =: \mathcal{A}_1 + \mathcal{A}_2 + \mathcal{A}_3.
\end{multline*}

First, we analyse $\mathcal{A}_1$:
\begin{displaymath}
    \mathcal{A}_1 = h' \sqrt{n}\left(\frac{d\wh\lambda(\theta_*)'}{d\theta}- \frac{d \lambda_{*}(\theta_{*})'}{d\theta}\right)\wh g(\theta_*) + h'\frac{d \lambda_{*}(\theta_{*})'}{d\theta} \sqrt{n}\EE_n\left[g_i(\theta_*) - \EE[g_i(\theta_{*})]\right]
\end{displaymath}
\noindent and by using the expression for $\frac{d\wh\lambda(\theta_*)'}{d\theta}$ and $\frac{d \lambda_{*}(\theta_{*})'}{d\theta}$ given in Lemmas \ref{lem:continuity:sample} and \ref{lem:continuity} we get:
\begin{multline*}
    \sqrt{n}h'\left(\frac{d\wh\lambda(\theta_*)'}{d\theta}- \frac{d \lambda_{*}(\theta_{*})'}{d\theta}\right) = \sqrt{n}\Big(\EE_n \left[\tau_i^{\diamond}(\wh\lambda,\theta_*)h'\wtl w_{1,i}\wtl w_i'(I + \wh\lambda(\theta_*) g_i(\theta_*)')\right]\check\Omega^{\diamond}(\wh\lambda,\theta_*)^{-1}\\
    - \EE \left[e^{\lambda_{*}(\theta_*)'g_i(\theta_*)}h'\wtl w_{1,i}\wtl w_i'(I + \lambda_*(\theta_*) g_i(\theta_*)')\right]\left(\EE[e^{\lambda_{*}(\theta_*)'g_i(\theta_*)}\varepsilon_{i*}^2\wtl w_{i}\wtl w_{i}']\right)^{-1}\Big).
\end{multline*}
\noindent  By a first order MVT expansion of $\lambda \mapsto\tau_i^{\diamond}(\lambda,\theta_*)$ around $\lambda_*(\theta_*)$, evaluated at $\wh\lambda(\theta_*)$, there exists a $\tau\in[0,1]$ such that: $\tau_i^{\diamond}(\wh\lambda,\theta_*) = \tau_i^{\diamond}(\lambda_*,\theta_*) + \tau_i^{\diamond}(\wtl\lambda_*,\theta_*)g_i(\theta_*)'(\wh\lambda(\theta_*) - \lambda_*(\theta_*))$, where $\wtl\lambda_*:=\tau \wh\lambda(\theta_*) + (1 - \tau)\lambda_*(\theta_*)$. By replacing this in the previous expression we get:
\begin{multline}
    \sqrt{n}h'\left(\frac{d\wh\lambda(\theta_*)'}{d\theta}- \frac{d \lambda_{*}(\theta_{*})'}{d\psi}\right) = \mathbb{G}_n\left[\tau_i^{\diamond}(\lambda_*,\theta_*)h'\wtl w_{1,i}\wtl w_i'(I + \lambda_*(\theta_*) g_i(\theta_*)')\right]\check\Omega^{\diamond}(\wh\lambda,\theta_*)^{-1}\\
    + \sqrt{n}\EE_n\left[\tau_i^{\diamond}(\wtl\lambda_*,\theta_*)(\wh\lambda(\theta_*) - \lambda_*(\theta_*))'g_i(\theta_*)h'\wtl w_{1,i} \wtl w_i'(I + \wh\lambda(\theta_*) g_i(\theta_*)')\right]\check\Omega^{\diamond}(\wh\lambda,\theta_*)^{-1}\\
    %+ \sqrt{n}\EE_n\left[\tau_i^{\diamond}(\lambda_*,\theta_*)\wtl w_{1,i} \wtl w_i'(\wh\lambda(\theta_*) - \lambda_*(\theta_*)) g_i(\theta_*)'\right]\check\Omega^{\diamond}(\wh\lambda,\theta_*)^{-1}\\
    + \sqrt{n}\EE \left[e^{\lambda_{*}(\theta_*)'g_i(\theta_*)}h'\wtl w_{1,i}\wtl w_i'(I + \lambda_*(\theta_*) g_i(\theta_*)')\right]\left(\check\Omega^{\diamond}(\wh\lambda,\theta_*)^{-1} - \Omega_*^{\diamond}(\theta_*)^{-1}\right). \label{eq:5}
\end{multline}
The second term of \eqref{eq:5} postmultiplied by $\wh g(\theta_*)$ is equal to (by using Lemma \ref{lem:1} above and the LLN)
\begin{multline*}
  \sqrt{n}\EE_n\left[\tau_i^{\diamond}(\wtl\lambda_*,\theta_*)(\wh\lambda(\theta_*) - \lambda_*(\theta_*))'g_i(\theta_*) h'\wtl w_{1,i} \wtl w_i'(I + \wh\lambda(\theta_*) g_i(\theta_*)')\right]\check\Omega^{\diamond}(\wh\lambda,\theta_*)^{-1}\wh g(\theta_*)\\
  = \sqrt{n}(\wh\lambda(\theta_*) - \lambda_*(\theta_*))'\EE_n\left[\tau_i^{\dagger}(\wtl\lambda_*,\theta_*)g_i(\theta_*)h'\wtl w_{1,i} \wtl w_i'(I + \wh\lambda(\theta_*) g_i(\theta_*)')\right]\check\Omega^{\dagger}(\wh\lambda,\theta_*)^{-1}\wh g(\theta_*)\\
  = - \mathbb{G}_n[\tau_i^{\dagger}(\lambda_*,\theta_*)\varepsilon_{i*}\wtl{w}_i']\Omega_*^{\dagger}(\theta_*)^{-1}\EE^{Q^*(\theta_*)}\left[g_i(\theta_*)h'\wtl w_{1,i} \wtl w_i'(I + \lambda_*(\theta_*) g_i(\theta_*)')\right]\Omega_*^{\dagger}(\theta_*)^{-1}\EE[g_i(\theta_*)].
\end{multline*}
We then analyse the third term in \eqref{eq:5} postmultiplied by $\wh g(\theta_*)$:
\begin{multline}
    \sqrt{n}\EE \left[e^{\lambda_{*}(\theta_*)'g_i(\theta_*)}h'\wtl w_{1,i}\wtl w_i'(I + \lambda_*(\theta_*) g_i(\theta_*)')\right]\left(\check\Omega^{\diamond}(\wh\lambda,\theta_*)^{-1} - \Omega_*^{\diamond}(\theta_*)^{-1}\right) \wh g(\theta_*)\\
    = h'\sqrt{n}\EE^{Q^*(\theta_*)} \left[\wtl w_{1,i}\wtl w_i'(I + \lambda_*(\theta_*) g_i(\theta_*)')\right]\Omega_*^{\dagger}(\theta_*)^{-1} \left(\Omega_*^{\diamond}(\theta_*) - \check\Omega^{\diamond}(\wh\lambda,\theta_*)\right)\check\Omega^{\diamond}(\wh\lambda,\theta_*)^{-1} \wh g(\theta_*)\\
    =: h'a'\sqrt{n}\left(\Omega_*^{\diamond}(\theta_*) - \check\Omega^{\diamond}(\wh\lambda,\theta_*)\right)\check\Omega^{\diamond}(\wh\lambda,\theta_*)^{-1} \wh g(\theta_*),\label{eq:6}
\end{multline}
which we intend to control by using \cite[Theorem 6.15]{VVaartSaintFlour}. Hence, we verify the conditions of that theorem:\\
\indent (1) because $\wh\lambda(\theta_*)\xrightarrow{p} \lambda_*(\theta_*)$ as $n\rightarrow \infty$ by Lemma \ref{lem:2} and by the continuous mapping theorem, $|h'a'\tau_i^{\diamond}(\wh\lambda,\theta_*)\varepsilon_{i*}^2\wtl w_i \wtl{w}_{i,k} - h'a'\tau_i^{\diamond}(\lambda_*,\theta_*)\varepsilon_{i*}^2\wtl w_i \wtl{w}_{i,k}|\xrightarrow{p} 0$ for every $k=1,\ldots,d$. Moreover, for every $k=1,\ldots,d$, for any $\delta >0$ and $\eta > 0$ there exists a number $N(\delta,\eta)$ such that for every $n> N(\delta, \eta)$,
$$|a'\tau_i^{\diamond}(\wh\lambda,\theta_*)\varepsilon_{i*}^2\wtl w_i \wtl{w}_{i,k} - a'\tau_i^{\diamond}(\lambda_*,\theta_*)\varepsilon_{i*}^2\wtl w_i \wtl{w}_{i,k}|^2\leq 4 \sup_{\lambda\in B_{\delta}(\lambda_*(\theta_*))} e^{2\lambda'g_i(\theta_*)} \varepsilon_{i*}^4 h'a'\wtl w_i \wtl w_i' a h \wtl{w}_{ik}^2$$
\noindent with probability $\geq \eta$. By this and Assumption \ref{Ass_3} \textit{(e)} (with $(\ell,\ell',i) = (2,4,2)$) which ensures integrability of the previous upper bound, we can apply the dominated convergence theorem which implies that $\EE[(h'a'\tau_i^{\diamond}(\wh\lambda,\theta_*)\varepsilon_{i*}^2\wtl w_i \wtl{w}_{i,k} - h'a'\tau_i^{\diamond}(\lambda_*,\theta_*)\varepsilon_{i*}^2\wtl w_i \wtl{w}_{i,k})^2] \xrightarrow{p} 0$ for every $k=1,\ldots,d$. The latter is one of the conditions of \cite[Theorem 6.15]{VVaartSaintFlour}.\\
\indent (2) $\EE\left[\left(h'a'\tau_i^{\diamond} (\lambda_*,\theta_*) \varepsilon_{i*}^2\wtl w_i \wtl w_{ik}\right)^2\right]$ is bounded for every $k=1,\ldots,d$ by Assumption \ref{Ass_3} \textit{(e)} (with $(\ell,\ell',i) = (2,4,2)$).\\
\indent (3) For every $k=1,\ldots,d$ and any $\delta > 0$, let us define the class of functions $\mathcal{F}_{k,\delta} := \{h'a'e^{\lambda'g_i(\theta_*)} \varepsilon_{i*}^2\wtl w_i \wtl w_{ik}; \lambda\in B_{\delta}(\lambda_*(\theta_*))\}$ which is a Donsker class under Assumption \ref{Ass_3} \textit{(e)} (with $(\ell,\ell',i) = (2,4,2)$) for some $\delta >0$ (by using \cite[Theorem 2.7.11 and Corollary 2.7.10]{VanDerVaartWellner1996})\footnote{By following the discussion above \cite[Theorem 2.7.11]{VanDerVaartWellner1996} we use as the envelope function
$$\sup_{\|\lambda - \lambda_*\|_2 \leq \delta}\left|h'a'\tau_i^{\diamond}(\lambda,\theta_*)\varepsilon_{i*}^2\wtl w_i\wtl{w}_{ik}\right|.$$}.
By the result of Lemma \ref{lem:2}, for any $\delta>0$ and any $\eta>0$, there exists a number $N(\delta,\eta)$ such that for every $n>N(\delta,\eta)$, $h'a'\tau_i^{\diamond} (\wh\lambda,\theta_*) \varepsilon_{i*}^2\wtl w_i \wtl w_{ik}\in\mathcal{F}_{k,\delta}$ with probability $\geq \eta$ for every $k=1,\ldots,d$.\\
\indent Hence, \cite[Theorem 6.15]{VVaartSaintFlour} guarantees that
$$\sqrt{n}h'a'\left(\check\Omega^{\diamond}(\wh\lambda,\theta_*) - \check\Omega^{\diamond}(\lambda_*,\theta_*) - \left[\Omega_*^{\diamond}(\wh\lambda,\theta_*) - \Omega_*^{\diamond}(\theta_*)\right]\right) \xrightarrow{p} 0$$
\noindent for every non-random matrix $a\in\mathbb{R}^{d\times p}$. To exploit this result we have to add and subtract $\check\Omega^{\diamond}(\lambda_*,\theta_*) + \Omega_*^{\diamond}(\wh\lambda,\theta_*) - \Omega_*^{\diamond}(\theta_*)$ in \eqref{eq:6} so that we get:
\begin{multline}
    \sqrt{n}\EE \left[e^{\lambda_{*}(\theta_*)'g_i(\theta_*)}h'\wtl w_{1,i}\wtl w_i'(I + \lambda_*(\theta_*) g_i(\theta_*)')\right]\left(\check\Omega^{\diamond}(\wh\lambda,\theta_*)^{-1} - \Omega_*^{\diamond}(\theta_*)^{-1}\right) \wh g(\theta_*)\\
    = h'a'\sqrt{n}\left(\Omega_*^{\diamond}(\theta_*) - \check\Omega^{\diamond}(\lambda_*,\theta_*) - \Omega_*^{\diamond}(\wh\lambda,\theta_*) + \Omega_*^{\diamond}(\theta_*)\right)\check\Omega^{\diamond}(\wh\lambda,\theta_*)^{-1}\wh g(\theta_*) + o_p(1). \label{eq:4}
\end{multline}
We start by analyzing the first two terms inside the brackets in \eqref{eq:4}. Since $\check\Omega^{\diamond}(\wh\lambda,\theta_*) \xrightarrow{p} \Omega_*^{\diamond}(\theta_*)$ by Lemma \ref{lem:Omega} and $\wh g(\theta_*) \xrightarrow{p}\EE[g_i(\theta_*)]$ by the law of large number, we get:
\begin{multline*}
    -h'a'\sqrt{n}\left(\check\Omega^{\diamond}(\lambda_*,\theta_*) - \Omega_*^{\diamond}(\theta_*)\right)\check\Omega^{\diamond}(\wh\lambda,\theta_*)^{-1}\wh g(\theta_*)\\
    = -h'a' \mathbb{G}_n\left[e^{\lambda_*(\theta_*)'g_i(\theta_*)}\varepsilon_{i*}^2 \wtl w_{i} \wtl w_i'\right]\Omega_*^{\diamond}(\theta_*)^{-1}\EE[g_i(\theta_*)] + o_p(1),
\end{multline*}
\noindent where, as it will be shown below, the term $\mathbb{G}_n\left[e^{\lambda_*(\theta_*)'g_i(\theta_*)}\varepsilon_{i*}^2 \wtl w_{i} \wtl w_i'\right]$ is bounded in probability. Next, we analyse the last two terms inside the brackets in \eqref{eq:4}. %By the MVT, there exists a $\tau\in[0,1]$ such that $\wtl\lambda_* := \tau\lambda_*(\theta_*) + (1- \tau)\wh \lambda(\theta_*)$ satisfies $e^{\wh\lambda(\theta_*)'g_i(\theta_*)} = e^{\lambda_*(\theta_*)'g_i(\theta_*)} + e^{\wtl\lambda_*'g_i(\theta_*)}g_i(\theta_*)'(\wh\lambda(\theta_*) - \lambda_*(\theta_*))$. By this and Lemma \ref{lem:1}, for any $\delta >0$ and $\eta > 0$ there exists a number $N(\delta,\eta)$ such that for every $n> N(\delta, \eta)$: $e^{\wh\lambda(\theta_*)'g_i(\theta_*)} = e^{\lambda_*(\theta_*)'g_i(\theta_*)} + e^{\lambda_*'g_i(\theta_*)}g_i(\theta_*)'(\wh\lambda(\theta_*) - \lambda_*(\theta_*))$ with probability $\geq \eta$.
By Lemma \ref{lem:Omega_star}, the fact that $\check\Omega^{\diamond}(\wh\lambda,\theta_*) \xrightarrow{p} \Omega_*^{\diamond}(\theta_*)$ by Lemma \ref{lem:Omega}, $\wh g(\theta_*) \xrightarrow{p}\EE[g_i(\theta_*)]$ by the law of large numbers, and by Lemma \ref{lem:1} we obtain:
\begin{multline*}
     \sqrt{n} h'a'\left[- \Omega_*^{\diamond}(\wh\lambda,\theta_*) + \Omega_*^{\diamond}(\lambda_*)\right]\check\Omega^{\diamond}(\wh\lambda,\theta_*)^{-1}\wh g(\theta_*)\\
     = - \sqrt{n}(\wh\lambda(\theta_*) - \lambda_*(\theta_*))'\EE^{Q^*(\theta_*)}\left[\varepsilon_{i*}^3 h'a'\wtl w_i \wtl w_i \wtl w_i'  \right]\check\Omega^{\dagger}(\wh\lambda,\theta_*)^{-1}\wh g(\theta_*) + o_p(1)\\
     = \mathbb{G}_n[\tau_i^{\dagger}(\lambda_*,\theta_*)\varepsilon_{i*}\wtl{w}_i']\Omega_*^{\dagger}(\theta_*)^{-1}\EE^{Q^*(\theta_*)}\left[\varepsilon_{i*}^3 h'a'\wtl w_i \wtl w_i \wtl w_i'  \right]\Omega_*^{\dagger}(\theta_*)^{-1}\EE[g_i(\theta_*)] + o_p(1)
\end{multline*}
\noindent where, as it will be shown below, the term $\mathbb{G}_n[\tau_i^{\dagger}(\lambda_*,\theta_*)\varepsilon_{i*}\wtl{w}_i']$ is bounded in probability. By putting all these elements together, term $\mathcal{A}_1$ is equal to:
\begin{multline*}
  \mathcal{A}_1 = \mathbb{G}_n\left[\tau_i^{\diamond}(\lambda_*,\theta_*)h'\wtl w_{1,i}\wtl w_i'(I + \lambda_*(\theta_*) g_i(\theta_*)')\right]\check\Omega^{\diamond}(\wh\lambda,\theta_*)^{-1}\EE[g_i(\theta_{*})]\\
   + h'\frac{d \lambda_{*}(\theta_{*})'}{d\theta} \sqrt{n}\EE_n\left[g_i(\theta_*) - \EE[g_i(\theta_{*})]\right]\\
    - \mathbb{G}_n[\tau_i^{\dagger}(\lambda_*,\theta_*)\varepsilon_{i*}\wtl{w}_i']\Omega_*^{\dagger}(\theta_*)^{-1}\EE^{Q^*(\theta_*)}\left[\varepsilon_{i*}\wtl{w}_i h'\wtl w_{1,i} \wtl w_i'(I + \lambda_*(\theta_*) g_i(\theta_*)')\right]\Omega_*^{\dagger}(\theta_*)^{-1}\EE[g_i(\theta_*)]\\
    %+ \sqrt{n}\EE_n\left[\tau_i^{\diamond}(\lambda_*,\theta_*)\wtl w_{1,i} \wtl w_i'(\wh\lambda(\theta_*) - \lambda_*(\theta_*)) g_i(\theta_*)'\right]\check\Omega^{\diamond}(\wh\lambda,\theta_*)^{-1}\\
    - h'a' \mathbb{G}_n\left[e^{\lambda_*(\theta_*)'g_i(\theta_*)}\varepsilon_{i*}^2 \wtl w_{i} \wtl w_i'\right]\Omega_*^{\diamond}(\theta_*)^{-1}\EE[g_i(\theta_*)]\\
    + h'a'\mathbb{G}_n[\tau_i^{\dagger}(\lambda_*,\theta_*)\varepsilon_{i*}\wtl{w}_i']\Omega_*^{\dagger}(\theta_*)^{-1}\EE^{Q^*(\theta_*)}\left[\varepsilon_{i*}^3 \wtl w_i \wtl w_i \wtl w_i'  \right]\Omega_*^{\dagger}(\theta_*)^{-1}\EE[g_i(\theta_*)] + o_p(1).
\end{multline*}
Next, we analyse $\mathcal{A}_2 + \mathcal{A}_3$ which can be written as
\begin{multline*}
    \mathcal{A}_2 + \mathcal{A}_3 = h' \sqrt{n} \EE_n\left[\left(\tau_i(\wh{\lambda},\theta_*) - \tau_i^{\dagger}(\lambda_*,\theta_*)\right)\wtl{w}_{1,i}\wtl{w}_i'\right] \wh\lambda(\theta_*)\\
    - h'\EE_n[\left(1 - \tau_i^{\dagger}(\lambda_*,\theta_*)\right)\wtl{w}_{1,i}\wtl{w}_{i}']\sqrt{n}\left(\wh\lambda(\theta_*) - \lambda_*(\theta_*)\right)\\
    - \mathbb{G}_n[\left(1 - \tau_i^{\dagger}(\lambda_*,\theta_*)\right)h'\wtl{w}_{1,i}\wtl{w}_{i}'\lambda_*(\theta_*)].
\end{multline*}
\noindent We use Lemmas \ref{lem:1} and \ref{lem:A_2:A_3}, and the fact that $h'\EE_n[\left(1 - \tau_i^{\dagger}(\lambda_*,\theta_*)\right)\wtl{w}_{1,i}\wtl{w}_{i}] \xrightarrow{p} h'\EE[\left(1 - \tau_i^{\dagger}(\lambda_*,\theta_*)\right)\wtl{w}_{1,i}\wtl{w}_{i}]$ by the Law of Large Numbers to get:
\begin{multline*}
  \mathcal{A}_2 + \mathcal{A}_3 = -\mathbb{G}_n[\tau_i^{\dagger}(\lambda_*,\theta_*)\varepsilon_{i*}\wtl{w}_i']\Omega_*^{\dagger}(\theta_*)^{-1}\EE\left[\tau_i^{\dagger}(\lambda_*,\theta_*) g_i(\theta_*)h'\wtl{w}_{1,i}\wtl{w}_i'\right] \lambda_*(\theta_*)\\
  - \mathbb{G}_n\left(\tau_i^{\dagger}(\lambda_*,\theta_*)\right)\EE\left[\tau_i^{\dagger}(\lambda_*,\theta_*) h'\wtl{w}_{1,i}\wtl{w}_i'\right] \lambda_*(\theta_*)\\
  + h'\EE[\left(1 - \tau_i^{\dagger}(\lambda_*,\theta_*)\right)\wtl{w}_{1,i}\wtl{w}_{i}']\Omega_*^{\dagger}(\theta_*)^{-1}\mathbb{G}_n[\tau_i^{\dagger}(\lambda_*,\theta_*)\varepsilon_{i*}\wtl{w}_i]\\
  - \mathbb{G}_n[\left(1 - \tau_i^{\dagger}(\lambda_*,\theta_*)\right)h'\wtl{w}_{1,i}\wtl{w}_{i}'\lambda_*(\theta_*)] + o_p(1)\\
  = -\mathbb{G}_n[\tau_i^{\dagger}(\lambda_*,\theta_*)\varepsilon_{i*}\wtl{w}_i']\frac{d\lambda_*(\theta_*)}{d\theta'} - \mathbb{G}_n\left(\tau_i^{\dagger}(\lambda_*,\theta_*)\right)\EE^{Q_*(\theta_*)}\left[h'\wtl{w}_{1,i}\wtl{w}_i'\right] \lambda_*(\theta_*)\\ + h'\EE[\wtl{w}_{1,i}\wtl{w}_{i}']\Omega_*^{\dagger}(\theta_*)^{-1}\mathbb{G}_n[\tau_i^{\dagger}(\lambda_*,\theta_*)\varepsilon_{i*}\wtl{w}_i] - \mathbb{G}_n[\left(1 - \tau_i^{\dagger}(\lambda_*,\theta_*)\right)h'\wtl{w}_{1,i}\wtl{w}_{i}'\lambda_*(\theta_*)] + o_p(1)
\end{multline*}
\noindent Finally, we put together $\mathcal{A}_1$, $\mathcal{A}_2$ and $\mathcal{A}_3$ to get:
\begin{multline*}
  \mathcal{A}_1 + \mathcal{A}_2 + \mathcal{A}_3 = \\
  \mathbb{G}_n\left[\tau_i^{\dagger}(\lambda_*,\theta_*)h'\wtl w_{1,i}\wtl w_i'(I + \lambda_*(\theta_*) g_i(\theta_*)')\right]\Omega_*^{\dagger}(\theta_*)^{-1}\EE[g_i(\theta_*)] + h'\frac{d \lambda_{*}(\theta_{*})'}{d\theta} \mathbb{G}_n\left[(1 - \tau_i^{\dagger}(\lambda_*,\theta_*))\varepsilon_{i*}\wtl{w}_i\right]\\
  + h'a'\mathbb{G}_n[\tau_i^{\dagger}(\lambda_*,\theta_*)\varepsilon_{i*}\wtl{w}_i']\Omega_*^{\dagger}(\theta_*)^{-1}\EE^{Q^*(\theta_*)}\left[\varepsilon_{i*}^3 \wtl w_i \wtl w_i \wtl w_i'  \right]\Omega_*^{\dagger}(\theta_*)^{-1}\EE[g_i(\theta_*)]\\
  - h'\mathbb{G}_n[\tau_i^{\dagger}(\lambda_*,\theta_*)\varepsilon_{i*}\wtl{w}_i']\Omega_*^{\dagger}(\theta_*)^{-1}\EE^{Q^*(\theta_*)}\left[g_i(\theta_*)\wtl w_{1,i} \wtl w_i'(I + \lambda_*(\theta_*) g_i(\theta_*)')\right]\Omega_*^{\dagger}(\theta_*)^{-1}\EE[g_i(\theta_*)]\\
  - h'a' \mathbb{G}_n\left[\tau_i^{\dagger}(\lambda_*,\theta_*)\varepsilon_{i*}^2 \wtl w_{i} \wtl w_i'\right]\Omega_*^{\dagger}(\theta_*)^{-1}\EE[g_i(\theta_*)] - \mathbb{G}_n\left(\tau_i^{\dagger}(\lambda_*,\theta_*)\right)\EE^{Q_*(\theta_*)}\left[h'\wtl{w}_{1,i}\wtl{w}_i'\right] \lambda_*(\theta_*)\\
  + h'\EE[\wtl{w}_{1,i}\wtl{w}_{i}']\Omega_*^{\dagger}(\theta_*)^{-1}\mathbb{G}_n[\tau_i^{\dagger}(\lambda_*,\theta_*)\varepsilon_{i*}\wtl{w}_i]
  - h'\mathbb{G}_n[\left(1 - \tau_i^{\dagger}(\lambda_*,\theta_*)\right)\wtl{w}_{1,i}\wtl{w}_{i}'\lambda_*(\theta_*)] + o_p(1).%\\
%  = h'\mathbb{G}_n\left[\tau_i^{\dagger}(\lambda_*,\theta_*)\wtl w_{1,i}\wtl w_i'(I + \lambda_*(\theta_*) g_i(\theta_*)')\right]\Omega_*^{\dagger}(\theta_*)^{-1}\EE[g_i(\theta_*)] + h'\frac{d \lambda_{*}(\theta_{*})'}{d\theta} \mathbb{G}_n\left[(1 - \tau_i^{\dagger}(\lambda_*,\theta_*))\varepsilon_i(\theta_*)\wtl{w}_i\right]\\
%  + h'\mathbb{G}_n[\tau_i^{\dagger}(\lambda_*,\theta_*)\varepsilon_i(\theta_*)\wtl{w}_i']\Omega_*^{\dagger}(\theta_*)^{-1}b\Omega_*^{\dagger}(\theta_*)^{-1}\EE[g_i(\theta_*)] + h'\EE[\wtl{w}_{1,i}\wtl{w}_{i}']\Omega_*^{\dagger}(\theta_*)^{-1}\mathbb{G}_n[\tau_i^{\dagger}(\lambda_*,\theta_*)\varepsilon_i(\theta_*)\wtl{w}_i]\\
%  - h'a' \mathbb{G}_n\left[\tau_i^{\dagger}(\lambda_*,\theta_*)\varepsilon_i^2 \wtl w_{i} \wtl w_i'\right]\Omega_*^{\dagger}(\theta_*)^{-1}\EE[g_i(\theta_*)]\\
%  - \mathbb{G}_n[\left(1 - \tau_i^{\dagger}(\lambda_*,\theta_*)\right)h'\wtl{w}_{1,i}\wtl{w}_{i}'\lambda_*(\theta_*)] + o_p(1).
\end{multline*}
By the Lindberg-Levy central limit theorem, $\mathcal{A}_1 + \mathcal{A}_2 + \mathcal{A}_3$ is asymptotically normal with zero mean and variance given by
\begin{multline*}
  h'\EE[\wtl w_{1i} \wtl w_i'] \Omega_*^{\dagger}(\theta_*)^{-1}\EE^{Q^*(\theta_*)}\left[\tau_i^{\dagger}(\lambda_*,\theta_*)\varepsilon_{i*}^2 \wtl w_i\wtl w_i'\right] \Omega_*^{\dagger}(\theta_*)^{-1} \EE[\wtl w_{i} \wtl w_{1i}'] h\\
  + Var \Big(\tau_i^{\dagger}(\lambda_*,\theta_*)h'\wtl w_{1,i}\wtl w_i'(I + \lambda_*(\theta_*) g_i(\theta_*)')\Omega_*^{\dagger}(\theta_*)^{-1}\EE[g_i(\theta_*)] + h'\frac{d \lambda_{*}(\theta_{*})'}{d\theta} (1 - \tau_i^{\dagger}(\lambda_*,\theta_*))\varepsilon_{i*}(\theta_*)\wtl{w}_i\\
  + h'a'\tau_i^{\dagger}(\lambda_*,\theta_*)\varepsilon_i(\theta_*)\wtl{w}_i'\Omega_*^{\dagger}(\theta_*)^{-1}\EE^{Q^*(\theta_*)}\left[\varepsilon_{i*}^3 \wtl w_i \wtl w_i \wtl w_i'  \right]\Omega_*^{\dagger}(\theta_*)^{-1}\EE[g_i(\theta_*)]\\
  - h'\tau_i^{\dagger}(\lambda_*,\theta_*)\varepsilon_{i*}\wtl{w}_i'\Omega_*^{\dagger}(\theta_*)^{-1}\EE^{Q^*(\theta_*)}\left[g_i(\theta_*)\wtl w_{1,i} \wtl w_i'(I + \lambda_*(\theta_*) g_i(\theta_*)')\right]\Omega_*^{\dagger}(\theta_*)^{-1}\EE[g_i(\theta_*)]\\
  - h'a' \tau_i^{\dagger}(\lambda_*,\theta_*)\varepsilon_{i*}^2 \wtl w_{i} \wtl w_i'\Omega_*^{\dagger}(\theta_*)^{-1}\EE[g_i(\theta_*)] - \tau_i^{\dagger}(\lambda_*,\theta_*)\EE^{Q_*(\theta_*)}\left[h'\wtl{w}_{1,i}\wtl{w}_i'\right] \lambda_*(\theta_*)\\
  - \left(1 - \tau_i^{\dagger}(\lambda_*,\theta_*)\right)h'\wtl{w}_{1,i}\wtl{w}_{i}'\lambda_*(\theta_*)\Big)\\
  + 2 h'\EE[\wtl w_{1i} \wtl w_i'] \Omega_*^{\dagger}(\theta_*)^{-1}\EE^{Q^*(\theta_*)}\left[\tau_i^{\dagger}(\lambda_*,\theta_*)\varepsilon_{i*} \wtl w_i \wtl w_i'(I + \lambda_*(\theta_*) g_i(\theta_*)')\Omega_*^{\dagger}(\theta_*)^{-1}\EE[g_i(\theta_*)]\wtl w_{1,i}'h\right]\\
  + 2 h'\EE[\wtl w_{1i} \wtl w_i'] \Omega_*^{\dagger}(\theta_*)^{-1}\EE^{Q^*(\theta_*)}\left[\left(1 - \tau_i^{\dagger}(\lambda_*,\theta_*)\right)\varepsilon_{i*}^2  \wtl w_i \wtl w_i'\right]\frac{d \lambda_{*}(\theta_{*})}{d\theta'}h\\
  + 2 h'\EE[\wtl w_{1i} \wtl w_i'] \Omega_*^{\dagger}(\theta_*)^{-1}\EE^{Q^*(\theta_*)}\left[\tau_i^{\dagger}(\lambda_*,\theta_*)\varepsilon_{i*}^2 \wtl w_i \wtl{w}_i'\right]\Omega_*^{\dagger}(\theta_*)^{-1}\EE^{Q^*(\theta_*)}\left[\varepsilon_{i*}(\theta_*)^3 \wtl w_i \wtl w_i'  \Omega_*^{\dagger}(\theta_*)^{-1}\EE[g_i(\theta_*)]\wtl w_i'\right]a h \\
  - 2 h'\EE[\wtl w_{1i} \wtl w_i'] \Omega_*^{\dagger}(\theta_*)^{-1}\EE^{Q^*(\theta_*)}\left[\tau_i^{\dagger}(\lambda_*,\theta_*)\varepsilon_{i*}^2 \wtl w_i \wtl{w}_i'\right]\Omega_*^{\dagger}(\theta_*)^{-1}\EE^{Q^*(\theta_*)}\left[g_i(\theta_*)\wtl w_i'(I + \lambda_*(\theta_*) g_i(\theta_*)')\Omega_*^{\dagger}(\theta_*)^{-1}\EE[g_i(\theta_*)]\wtl w_{1,i}'\right]h\\
  - 2 h'\EE[\wtl w_{1i} \wtl w_i'] \Omega_*^{\dagger}(\theta_*)^{-1}\EE^{Q^*(\theta_*)}\left[\tau_i^{\dagger}(\lambda_*,\theta_*)\varepsilon_{i*}^3 \wtl w_i \wtl w_i' \Omega_*^{\dagger}(\theta_*)^{-1}\EE[g_i(\theta_*)]\wtl w_{i}'\right]  a h \\
  - 2 h'\EE[\wtl w_{1i} \wtl w_i'] \Omega_*^{\dagger}(\theta_*)^{-1}\EE^{Q^*(\theta_*)}\left[\tau_i^{\dagger}(\lambda_*,\theta_*) \varepsilon_{i*} \wtl w_i \right] \lambda_*(\theta_*)'\EE^{Q^*(\theta_*)}[\wtl{w}_i\wtl{w}_{1,i}'h]\\
  - 2 h'\EE[\wtl w_{1i} \wtl w_i'] \Omega_*^{\dagger}(\theta_*)^{-1}\EE^{Q^*(\theta_*)}\left[\left(1 - \tau_i^{\dagger}(\lambda_*,\theta_*)\right) \varepsilon_{i*} \wtl w_i \lambda_*(\theta_*)' \wtl{w}_{i} \wtl{w}_{1,i}'\right] h =: h' H_* h.
\end{multline*}
\noindent The matrix $H_*$ is non-singular under Assumptions \ref{Ass_2_NS} \textit{(b)} and \ref{Ass_1_NS} \textit{(c)}.
\begin{flushright}
  $\square$
\end{flushright}

%------------------------------------------------------------------------------------------------------------------------------------------------------------------------------
\begin{lem}\label{lem:Omega}
  Let Assumptions \ref{Ass_absolute_continuity}, \ref{Ass_1_NS} (a), \ref{Ass_3} \textit{(e)-(f)} hold. Then,
    $$\left\|\check\Omega^{\diamond}(\wh\lambda,\theta_*) - \Omega_*^{\diamond}(\theta_*)\right\| \xrightarrow{p} 0$$
    as $n\rightarrow \infty$.
\end{lem}

\proof
Let us denote $\varepsilon_{i*} := \varepsilon_i(\theta_*)$. %\textcolor{red}{For every $j,k\in\{1,\ldots,d\}$ and every $\delta >0$ let us consider the class of functions $\mathcal{F}_{j,k} := \{\lambda \mapsto e^{\lambda' g_i(\theta_*)}\varepsilon_i^2 \wtl{w}_{i,j}\wtl{w}_{i,k}; \, \|\lambda - \lambda_*(\theta_*)\| \leq \delta \}$. By Lemma \ref{lem:2}, $\|\wh\lambda(\theta_*) - \lambda_*(\theta_*)\|_2 \xrightarrow{p} 0$ and so for any $\delta >0$ and $\eta>0$ there exists a number $N(\delta,\eta)$ such that for every $n> N(\delta,\eta)$, $\|\wh\lambda(\theta_*) - \lambda_*(\theta_*)\|_2 \leq \delta$ with probability $\geq \eta$. Therefore, for every $j,k= 1,\ldots,d$, and every $n > N(\delta,\eta)$, $e^{\wh{\lambda}(\theta_*)' g_i(\theta_*)}\varepsilon_i^2 \wtl{w}_{i,j}\wtl{w}_{i,k} \in \mathcal{F}_{j,k}$ with probability $\geq \eta$.}
By Lemma \ref{lem:2}, $\|\wh\lambda(\theta_*) - \lambda_*(\theta_*)\|_2 \xrightarrow{p} 0$ and so for any $\delta >0$ and $\eta>0$ there exists a number $N(\delta,\eta)$ such that for every $n> N(\delta,\eta)$, $\|\wh\lambda(\theta_*) - \lambda_*(\theta_*)\|_2 \leq \delta$ with probability $\geq \eta$. Therefore, for every $n > N(\delta,\eta)$, $\wh\lambda(\theta_*)\in B_{\delta}(\lambda_*(\theta_*))$. By Assumption \ref{Ass_3} \textit{(e)} with $(\ell,\ell',i) = (1,2,0)$, compactness of $B_{\delta}(\lambda_*(\theta_*))$ and \cite[Lemma 2.4]{NeweyMcFadden1994} we have that for every $n > N(\delta,\eta)$,
\begin{multline}\label{eq:res:1}
  \left|\EE_n\left[e^{\wh\lambda(\theta_*)' g_i(\theta_*)}\varepsilon_{i*}^2 \wtl{w}_{i,j}\wtl{w}_{i,k}\right] - \EE\left[e^{\wh\lambda(\theta_*)' g_i(\theta_*)}\varepsilon_{i*}^2 \wtl{w}_{i,j}\wtl{w}_{i,k}\right]\right|\\
  \leq \sup_{\lambda\in B_{\delta}(\lambda_*(\theta_*))}\left|\EE_n\left[e^{\lambda' g_i(\theta_*)}\varepsilon_{i*}^2 \wtl{w}_{i,j}\wtl{w}_{i,k}\right] - \EE\left[e^{\lambda' g_i(\theta_*)}\varepsilon_{i*}^2 \wtl{w}_{i,j}\wtl{w}_{i,k}\right]\right| \xrightarrow{p} 0
\end{multline}
\noindent for every sequence of random variables $\wh\lambda(\theta_*)$ for which the result of Lemma \ref{lem:2} holds.\\
\indent Moreover, by the continuous mapping theorem and Lemma \ref{lem:2}, $e^{\wh{\lambda}(\theta_*)' g_i(\theta_*)}\varepsilon_{i*}^2 \wtl{w}_{i,j}\wtl{w}_{i,k} \xrightarrow{p} e^{\lambda_*(\theta_*)' g_i(\theta_*)}\varepsilon_{i*}^2 \wtl{w}_{i,j}\wtl{w}_{i,k}$ for every $j,k =1,\ldots,d$. It then follows from the dominated convergence theorem (applicable by Assumption \ref{Ass_3} \textit{(e)} with $(\ell,\ell',i) = (1,2,0)$) that
\begin{equation}\label{eq:res:2}
  \left|\EE\left[e^{\wh{\lambda}(\theta_*)' g_i(\theta_*)}\varepsilon_{i*}^2 \wtl{w}_{i,j}\wtl{w}_{i,k}\right] - \EE\left[e^{\lambda_*(\theta_*)' g_i(\theta_*)}\varepsilon_{i*}^2 \wtl{w}_{i,j}\wtl{w}_{i,k}\right]\right| \xrightarrow{p} 0
\end{equation}
\noindent for every $j,k = 1,\ldots,d$. By putting together \ref{eq:res:1} and \eqref{eq:res:2}, and by T we have:
\begin{multline*}
  \left\|\check\Omega^{\diamond}(\wh\lambda,\theta_*) - \Omega_*^{\diamond}(\theta_*)\right\| \leq d\max_{j,k\in\{1,\ldots,d\}} \left|\EE_n\left[e^{\wh\lambda(\theta_*)' g_i(\theta_*)}\varepsilon_{i*}^2 \wtl{w}_{i,j}\wtl{w}_{i,k}\right] - \EE\left[e^{\lambda_*(\theta_*)' g_i(\theta_*)}\varepsilon_{i*}^2 \wtl{w}_{i,j}\wtl{w}_{i,k}\right]\right|\\
  \leq d\max_{j,k\in\{1,\ldots,d\}} \left|\EE_n\left[e^{\wh\lambda(\theta_*)' g_i(\theta_*)}\varepsilon_{i*}^2 \wtl{w}_{i,j}\wtl{w}_{i,k}\right] - \EE\left[e^{\wh{\lambda}(\theta_*)' g_i(\theta_*)}\varepsilon_{i*}^2 \wtl{w}_{i,j}\wtl{w}_{i,k}\right]\right|\\ \hfill
  + d\max_{j,k\in\{1,\ldots,d\}} \left|\EE\left[e^{\wh{\lambda}(\theta_*)' g_i(\theta_*)}\varepsilon_{i*}^2 \wtl{w}_{i,j}\wtl{w}_{i,k}\right] - \EE\left[e^{\lambda_*(\theta_*)' g_i(\theta_*)}\varepsilon_{i*}^2 \wtl{w}_{i,j}\wtl{w}_{i,k}\right]\right|
  = o_p(1).
\end{multline*}

\begin{flushright}
  $\square$
\end{flushright}

%%------------------------------------------------------------------------------------------------------------------------------------------------------------------------------
\begin{lem}\label{lem:Omega_star}
  Let Assumptions \ref{Ass_absolute_continuity}, \ref{ass:feasibility}, \ref{Ass_1_NS} (a), \ref{Ass_3} (d)-(e) hold. Then, as $n\rightarrow \infty$,
  \begin{equation}
    \left[\Omega_*^{\diamond}(\wh\lambda,\theta_*) - \Omega_*^{\diamond}(\theta_*)\right] = (\wh\lambda(\theta_*) - \lambda_*(\theta_*))'\EE\left[e^{\lambda_*(\theta_*)'g_i(\theta_*)}\varepsilon_i(\theta_*)^3 \wtl w_i \wtl w_i \wtl w_i'  \right] + o_p(1).
  \end{equation}
\end{lem}
\proof
By a first order MVT expansion, there exists a $\tau\in[0,1]$ such that $\wtl\lambda_* := \tau\lambda_*(\theta_*) + (1- \tau)\wh \lambda(\theta_*)$ satisfies $e^{\wh\lambda(\theta_*)'g_i(\theta_*)} = e^{\lambda_*(\theta_*)'g_i(\theta_*)} + e^{\wtl\lambda_*'g_i(\theta_*)}g_i(\theta_*)'(\wh\lambda(\theta_*) - \lambda_*(\theta_*))$. By replacing this in the expression of $\Omega_*^{\diamond}(\wh\lambda,\theta_*)$ we get
\begin{equation}
  \left[\Omega_*^{\diamond}(\wh\lambda,\theta_*) - \Omega_*^{\diamond}(\theta_*)\right] = (\wh\lambda(\theta_*) - \lambda_*(\theta_*))'\EE\left[e^{\wtl\lambda_*'g_i(\theta_*)}\varepsilon_i(\theta_*)^3 \wtl w_i \wtl w_i \wtl w_i'  \right].
\end{equation}
\noindent Lemma \ref{lem:2} %, for any $\eta > 0$ there exists a finite $\delta >0$ and a finite $N(\delta,\eta)$ such that for every $n> N(\delta, \eta)$: $e^{\wh\lambda(\theta_*)'g_i(\theta_*)} = e^{\lambda_*(\theta_*)'g_i(\theta_*)} + e^{\lambda_*'g_i(\theta_*)}g_i(\theta_*)'(\wh\lambda(\theta_*) - \lambda_*(\theta_*))$ with probability $\geq \eta$. This
and the continuous mapping theorem imply that for every $k = 1,\ldots,d$,
$$e^{\wtl\lambda_*'g_i(\theta_*)}\varepsilon_i(\theta_*)^3 \wtl w_{i,k} \wtl w_i \wtl w_i' \xrightarrow{p} e^{\lambda_*(\theta_*)'g_i(\theta_*)}\varepsilon_i(\theta_*)^3 \wtl w_{i,k} \wtl w_i \wtl w_i'.$$
By using Assumption \ref{Ass_3} \textit{(e)} (with $(\ell,\ell',i) = (1,3,1)$) we can apply the DCT which guarantees:
$$\left\|\EE\left[e^{\wtl\lambda_*'g_i(\theta_*)}\varepsilon_i(\theta_*)^3 \wtl w_{i,k} \wtl w_i \wtl w_i'  \right] - \EE\left[e^{\lambda_*(\theta_*)'g_i(\theta_*)}\varepsilon_i(\theta_*)^3 \wtl w_{i,k} \wtl w_i \wtl w_i'  \right]\right\| \rightarrow 0.$$
This concludes the proof.

\begin{flushright}
  $\square$
\end{flushright}

%%------------------------------------------------------------------------------------------------------------------------------------------------------------------------------
\begin{lem}\label{lem:A_2:A_3}
  Let Assumptions \ref{Ass_absolute_continuity}, \ref{ass:feasibility}, \ref{Ass_1_NS} (a), \ref{Ass_3} (a)-(c), (e) hold. Then, as $n\rightarrow \infty$,
  \begin{multline*}
    \sqrt{n} \EE_n\left[\left(\tau_i(\wh\lambda,\theta_*) - \tau_i^{\dagger}(\lambda_*,\theta_*)\right)h'\wtl{w}_{1,i}\wtl{w}_i'\right] \wh{\lambda}(\theta_*) \\
    = -\mathbb{G}_n[\tau_i^{\dagger}(\lambda_*,\theta_*)\varepsilon_i(\theta_*)\wtl{w}_i']\Omega_*^{\dagger}(\theta_*)^{-1}\EE\left[\tau_i^{\dagger}(\lambda_*,\theta_*) g_i(\theta_*)h'\wtl{w}_{1,i}\wtl{w}_i'\right] \lambda_*(\theta_*)\\
    - \EE\left[\tau_i^{\dagger}(\lambda_*,\theta_*) h'\wtl{w}_{1,i}\wtl{w}_i'\right] \lambda_*(\theta_*)\mathbb{G}_n\left(\tau_i^{\dagger}(\lambda_*,\theta_*)\right) + o_p(1),
  \end{multline*}
   where the $o_p(1)$ is uniform in $h\in K$ for any compact set $K\subset \mathbb{R}^p$.
\end{lem}
\proof
By a first order MVT expansion of the function $\lambda \mapsto e^{\lambda'g_i(\theta_*)}$ around $\lambda_*(\theta_*)$, evaluated at $\wh\lambda(\theta_*)$, there exists a $\tau\in[0,1]$ such that: $e^{\wh\lambda(\theta_*)'g_i(\theta_*)} = e^{\lambda_*(\theta_*)'g_i(\theta_*)} + e^{\wtl \lambda_*(\theta_*)'g_i(\theta_*)}g_i(\theta_*)'(\wh{\lambda}(\theta_*) - \lambda_*(\theta_*))$, where $\wtl \lambda_*(\theta_*) := \tau\wh{\lambda}(\theta_*) + (1 - \tau) \lambda_*(\theta_*)$. We use this result to get the second equality below:
  \begin{multline}
    \sqrt{n} \EE_n\left[\left(\tau_i(\wh\lambda,\theta_*) - \tau_i^{\dagger}(\lambda_*,\theta_*)\right)h'\wtl{w}_{1,i}\wtl{w}_i'\right] \wh{\lambda}(\theta_*)\\
    = \sqrt{n} \EE_n\left[\left(e^{\wh\lambda(\theta_*)'g_i(\theta_*)} - e^{\lambda_*(\theta_*)'g_i(\theta_*)}\right) h'\wtl{w}_{1,i}\wtl{w}_i'\right] \wh{\lambda}(\theta_*)\frac{1}{\EE_n\left[e^{\wh\lambda(\theta_*)' g_i(\theta_*)}\right]}\\ \hfill
    + \sqrt{n} \EE_n\left[e^{\lambda_*(\theta_*)'g_i(\theta_*)} h'\wtl{w}_{1,i}\wtl{w}_i'\right] \wh{\lambda}(\theta_*)\left(\frac{1}{\EE_n\left[e^{\wh\lambda(\theta_*)' g_i(\theta_*)}\right]} - \frac{1}{\EE\left[e^{\lambda_*(\theta_*)' g_i(\theta_*)}\right]}\right)\\
    = \sqrt{n} (\wh{\lambda}(\theta_*) - \lambda_*(\theta_*))'\EE_n\left[e^{\wtl \lambda_*(\theta_*)'g_i(\theta_*)}g_i(\theta_*)h'\wtl{w}_{1,i}\wtl{w}_i'\right] \wh{\lambda}(\theta_*)\frac{1}{\EE_n\left[e^{\wh\lambda(\theta_*)' g_i(\theta_*)}\right]} + \\
    \sqrt{n} \EE_n\left[e^{\lambda_*(\theta_*)'g_i(\theta_*)} h'\wtl{w}_{1,i}\wtl{w}_i'\right] \frac{\wh{\lambda}(\theta_*)}{\EE_n\left[e^{\wh\lambda(\theta_*)' g_i(\theta_*)}\right]}\left(\EE\left[e^{\lambda_*(\theta_*)' g_i(\theta_*)}\right] - \EE_n\left[e^{\wh\lambda(\theta_*)' g_i(\theta_*)}\right]\right)\frac{1}{\EE\left[e^{\lambda_*(\theta_*)' g_i(\theta_*)}\right]}.\label{eq:2:Lem_A_8}
  \end{multline}

Let us consider the factor $\sqrt{n}\left(\EE\left[e^{\lambda_*(\theta_*)' g_i(\theta_*)}\right] - \EE_n\left[e^{\wh\lambda(\theta_*)' g_i(\theta_*)}\right]\right)$ in the second term of the right hand side of \eqref{eq:2:Lem_A_8}:
\begin{multline}
  \sqrt{n}\left(\EE\left[e^{\lambda_*(\theta_*)' g_i(\theta_*)}\right] - \EE_n\left[e^{\wh\lambda(\theta_*)' g_i(\theta_*)}\right]\right)\\
  = -\sqrt{n}\EE_n\left[e^{\lambda_*(\theta_*)' g_i(\theta_*)} - \EE[e^{\lambda_*(\theta_*)' g_i(\theta_*)}]\right] - \sqrt{n}\left(\EE_n\left[e^{\wh\lambda(\theta_*)' g_i(\theta_*)}\right]  - \EE_n\left[e^{\lambda_*(\theta_*)' g_i(\theta_*)}\right]\right)\\
  = -\mathbb{G}_n\left(e^{\lambda_*(\theta_*)' g_i(\theta_*)}\right) - \sqrt{n}\EE_n\left[e^{\wtl{\lambda}_*(\theta_*)'g_i(\theta_*)}g_i(\theta_*)'\right]\left(\wh{\lambda}(\theta_*) - \lambda_*(\theta_*)\right),\label{eq:1:Lem_A_8}
\end{multline}
\noindent where we have used again the first order MVT expansion of the function $\lambda \mapsto e^{\lambda'g_i(\theta_*)}$ around $\lambda_*(\theta_*)$ evaluated at $\wh\lambda(\theta_*)$. By Lemma \ref{lem:2}, for any $\eta>0$ there exists a finite $\delta > 0$ and a finite $N(\delta,\eta)$ such that for every $n > N(\delta,\eta)$, $\wtl{\lambda}_*(\theta_*) \in B_{\delta}(\lambda_*(\theta_*))$ with probability $\geq\eta$. By this and since $\EE\left[e^{\lambda_*(\theta_*)'g_i(\theta_*)} g_i(\theta_*)\right] = 0$ then by T: $\forall n > N(\delta,\eta)$,
\begin{multline*}
  \left\|\EE_n\left[e^{\wtl{\lambda}_*(\theta_*)'g_i(\theta_*)}g_{i}(\theta_*)\right]\right\|_2 = \left\|\EE_n\left[e^{\wtl{\lambda}_*(\theta_*)'g_i(\theta_*)}g_{i}(\theta_*)\right] - \EE\left[e^{\wtl{\lambda}_*(\theta_*)'g_i(\theta_*)}g_{i}(\theta_*)\right] \right\|_2\\
  + \left\|\EE\left[\left(e^{\wtl{\lambda}_*(\theta_*)'g_i(\theta_*)} - e^{\lambda_*(\theta_*)'g_i(\theta_*)}\right)g_{i}(\theta_*)\right] \right\|_2\\
  \leq \sup_{\lambda\in B_{\delta}(\lambda_*(\theta_*))}\left\|\EE_n\left[e^{\lambda'g_i(\theta_*)}g_{i}(\theta_*)\right] - \EE\left[e^{\lambda'g_i(\theta_*)}g_{i,k}(\theta_*)\right]\right\|_2 \\
  + \left\|\EE\left[\left(e^{\wtl{\lambda}_*(\theta_*)'g_i(\theta_*)} - e^{\lambda_*(\theta_*)'g_i(\theta_*)}\right)g_{i}(\theta_*)\right]\right\|_2
\end{multline*}
\noindent with probability $\geq \eta$. By Assumption \ref{Ass_3} \textit{(a)}, compactness of $B_{\delta}(\lambda_*(\theta_*))$ and \cite[Lemma 2.4]{NeweyMcFadden1994}, the first term in the right hand side of the previous expression converges to zero in probability. To deal with the second term we use the CMT and Lemma \ref{lem:2} that guarantee that $\left|\left(e^{\wtl{\lambda}_*(\theta_*)'g_i(\theta_*)} - e^{\lambda_*(\theta_*)'g_i(\theta_*)}\right)g_{i,k}(\theta_*)\right| \xrightarrow{p} 0$ for every $k=1,\ldots,d$. By Assumption \ref{Ass_3} \textit{(a)} and the DCT we conclude that
\begin{displaymath}
  \left\|\EE\left[\left(e^{\wtl{\lambda}_*(\theta_*)'g_i(\theta_*)} - e^{\lambda_*(\theta_*)'g_i(\theta_*)}\right)g_{i}(\theta_*)'\right] \right\|\xrightarrow{p}0.
\end{displaymath}
\noindent We then conclude that $\left\|\EE_n\left[e^{\wtl{\lambda}_*(\theta_*)'g_i(\theta_*)}g_{i}(\theta_*)\right]\right\|_2 = o_p(1)$ and since $\sqrt{n}\left(\wh{\lambda}(\theta_*) - \lambda_*(\theta_*)\right) = \mathcal{O}_p(1)$ by Lemma \ref{lem:2}, then term $\sqrt{n}\EE_n\left[e^{\wtl{\lambda}_*(\theta_*)'g_i(\theta_*)}g_i(\theta_*)'\right]\left(\wh{\lambda}(\theta_*) - \lambda_*(\theta_*)\right)$ in \eqref{eq:1:Lem_A_8} converges to zero in probability as well.\\
We now analyse the first term in the right hand side of \eqref{eq:2:Lem_A_8}. By Lemma \ref{lem:1}:
$$\sqrt{n} (\wh{\lambda}(\theta_*) - \lambda_*(\theta_*))' = -\mathbb{G}_n[\tau_i^{\dagger}(\lambda_*,\theta_*)\varepsilon_i(\theta_*)\wtl{w}_i']\Omega_*^{\dagger}(\theta_*)^{-1} + o_p(1),$$ by the uniform law of large numbers, which is valid under Assumption \ref{Ass_3} \textit{(c)} (with $(j,\ell,\ell') = (1,1,1)$), we have that $\EE_n\left[e^{\lambda_*(\theta_*)'g_i(\theta_*)} h'\wtl{w}_{1,i}\wtl{w}_i'\right] \xrightarrow{p} \EE\left[e^{\lambda_*(\theta_*)'g_i(\theta_*)} h'\wtl{w}_{1,i}\wtl{w}_i'\right] $ uniformly in $h\in K$, and by Lemma \ref{lem:2}: $\wh\lambda(\theta_*)\xrightarrow{p}\lambda_*(\theta_*)$. It remains to consider the terms $\EE_n\left[e^{\wtl \lambda_*(\theta_*)'g_i(\theta_*)}g_i(\theta_*)h'\wtl{w}_{1,i}\wtl{w}_i'\right]$ and $\EE_n\left[e^{\wh\lambda(\theta_*)' g_i(\theta_*)}\right]$. We start with the first one, which is a $(d\times d)$ matrix, for which we analyse every element of the matrix: by T, for every $j,k=1,\ldots,d$, and $\forall n > N(\delta,\eta)$,

\begin{multline*}
  \sup_{h\in K}\left|\EE_n\left[e^{\wtl \lambda_*(\theta_*)'g_i(\theta_*)}g_{i,j}(\theta_*)h'\wtl{w}_{1,i}\wtl{w}_{i,k}\right] - \EE\left[e^{\lambda_*(\theta_*)'g_i(\theta_*)}g_{i,j}(\theta_*)h'\wtl{w}_{1,i}\wtl{w}_{i,k}\right]\right|\\
  \leq \sup_{h\in K}\sup_{\lambda\in B_{\delta}(\lambda_*(\theta_*))}\left|\EE_n\left[e^{\lambda'g_i(\theta_*)}g_{i,j}(\theta_*)h'\wtl{w}_{1,i}\wtl{w}_{i,k}\right] - \EE\left[e^{\lambda'g_i(\theta_*)}g_{i,j}(\theta_*)h'\wtl{w}_{1,i}\wtl{w}_{i,k}\right]\right|\\ \hfill
  + \sup_{h\in K}\left|\EE\left[\left(e^{\wtl \lambda_*(\theta_*)'g_i(\theta_*)} - e^{\lambda_*(\theta_*)'g_i(\theta_*)}\right)g_{i,j}(\theta_*)h'\wtl{w}_{1,i}\wtl{w}_{i,k}\right]\right|
\end{multline*}
\noindent which converges to zero in probability under Assumption \ref{Ass_3} \textit{(f)} by compactness of $B_{\delta}(\lambda_*(\theta_*))$, \cite[Lemma 2.4]{NeweyMcFadden1994}, the CMT and the DCT. Similarly,  $\forall n > N(\delta,\eta)$,
\begin{multline*}
  \left|\EE_n\left[e^{\wh\lambda(\theta_*)' g_i(\theta_*)}\right] - \EE\left[e^{\lambda_*(\theta_*)' g_i(\theta_*)}\right]\right| \\
  \leq \sup_{\lambda \in B_{\delta}(\lambda_*(\theta_*))} \left|\EE_n\left[e^{\lambda' g_i(\theta_*)}\right] - \EE\left[e^{\lambda' g_i(\theta_*)}\right]\right| + \left|\EE\left[e^{\wh\lambda(\theta_*)' g_i(\theta_*)}  -e^{\lambda_*(\theta_*)' g_i(\theta_*)}\right]\right|
\end{multline*}
\noindent which converges to zero in probability under Assumption \ref{Ass_3} \textit{(b)}. By putting all these elements back in \eqref{eq:2:Lem_A_8} we obtain:
\begin{multline*}
  \sqrt{n} \EE_n\left[\left(\tau_i(\wh\lambda,\theta_*) - \tau_i^{\dagger}(\lambda_*,\theta_*)\right)h'\wtl{w}_{1,i}\wtl{w}_i'\right] \wh{\lambda}(\theta_*) \\
  = -\mathbb{G}_n[\tau_i^{\dagger}(\lambda_*,\theta_*)\varepsilon_i(\theta_*)\wtl{w}_i']\Omega_*^{\dagger}(\theta_*)^{-1}\EE\left[e^{\lambda_*(\theta_*)'g_i(\theta_*)}g_i(\theta_*)h'\wtl{w}_{1,i}\wtl{w}_i'\right] \lambda_*(\theta_*)\frac{1}{\EE\left[e^{\lambda_*(\theta_*)' g_i(\theta_*)}\right]}\\
  - \EE\left[e^{\lambda_*(\theta_*)'g_i(\theta_*)} h'\wtl{w}_{1,i}\wtl{w}_i'\right] \frac{\lambda_*(\theta_*)}{\EE\left[e^{\lambda_*(\theta_*)' g_i(\theta_*)}\right]}\mathbb{G}_n\left(e^{\lambda_*(\theta_*)' g_i(\theta_*)}\right)\frac{1}{\EE\left[e^{\lambda_*(\theta_*)' g_i(\theta_*)}\right]} + o_p(1),
\end{multline*}
\noindent where the $o_p(1)$ is uniform in $h\in K$.
\begin{flushright}
  $\square$
\end{flushright}

%%------------------------------------------------------------------------------------------------------------------------------------------------------------------------------
\begin{lem}\label{lem:technical:endogeneity:1}
  Let Assumptions \ref{Ass_absolute_continuity}-\ref{ass:feasibility}, \ref{Ass_1_NS} (a)-(b), \ref{Ass_3} (b)-(c) hold and let $\wtl\theta := \theta_* + \frac{\tau}{\sqrt{n}} h$ for some $\tau\in[0,1]$ and any $h$ in a compact set $K\subset \mathbb{R}^p$. Then, as $n\rightarrow \infty$,
  \begin{displaymath}
    \EE_n\left[\tau_i(\wh\lambda,\wtl\theta) h'\wtl{w}_{1,i}\wtl{w}_i\right] \xrightarrow{p} \EE\left[\tau_i^{\dagger}(\lambda_*,\theta_*)h'\wtl{w}_{1,i}\wtl{w}_i\right]
  \end{displaymath}
  uniformly in $h\in K$.
\end{lem}
\proof
Notice that $\wtl\theta\in B_{*,n}$. The following decomposition holds:
  \begin{multline}
    \EE_n\left(\tau_i(\wh\lambda,\wtl\theta)h'\wtl{w}_{1,i}\wtl{w}_i\right) - \EE\left[\tau_i^{\dagger}(\lambda_*,\theta_*)h'\wtl{w}_{1,i}\wtl{w}_i\right]\\
    = \EE_n\left[e^{\wh\lambda(\wtl\theta)'g_i(\wtl\theta)} h'\wtl{w}_{1,i}\wtl{w}_i  - \EE\left[e^{\wh\lambda(\wtl\theta)'g_i(\wtl\theta)} h'\wtl{w}_{1,i}\wtl{w}_i\right]\right] \frac{1}{\EE_n\left[e^{\wh\lambda(\wtl\theta)' g_i(\wtl\theta)}\right]}\\
    + \EE\left[\left(e^{\wh\lambda(\wtl\theta)'g_i(\wtl\theta)} - e^{\lambda_*(\theta_*)'g_i(\theta_*)}\right)h'\wtl{w}_{1,i}\wtl{w}_i\right]\frac{1}{\EE_n\left[e^{\wh\lambda(\wtl\theta)' g_i(\wtl\theta)}\right]}\\ \hfill
    + \EE\left[e^{\lambda_*(\theta_*)'g_i(\theta_*)}h'\wtl{w}_{1,i}\wtl{w}_i\right]\frac{1}{\EE_n\left[e^{\wh\lambda(\wtl\theta)' g_i(\wtl\theta)}\right]}\left(\EE[e^{\lambda_*(\theta_*)'g_i(\theta_*)}] - \EE_n\left[e^{\wh\lambda(\wtl\theta)' g_i(\wtl\theta)}\right]\right)\frac{1}{\EE[e^{\lambda_*(\theta_*)'g_i(\theta_*)}]}.\label{eq:2:Lem_A_9:endogeneity}
  \end{multline}

\noindent By Lemma \ref{lem:2} which guarantees that $\wh\lambda(\wtl\theta) \in B_{1/\sqrt{n}}(\lambda_*(\theta_*))$ for $n$ large, by compactness of $B_{1/\sqrt{n}}(\lambda_*(\theta_*))$ and of $B_{1/\sqrt{n}}(\theta_*)$, and by \cite[Lemma 2.4]{NeweyMcFadden1994} which is valid under Assumption \ref{Ass_3} \textit{(c)} with $j=\ell = \ell' = 1$, we have that $$\EE_n\left[e^{\wh\lambda(\wtl\theta)'g_i(\wtl\theta)} h'\wtl{w}_{1,i}\wtl{w}_i'  - \EE\left[e^{\wh\lambda(\wtl\theta)'g_i(\wtl\theta)} h'\wtl{w}_{1,i}\wtl{w}_i'\right]\right] \xrightarrow{p}0$$
uniformly in $(\wh\lambda(\wtl\theta),h) \in B_{1/\sqrt{n}}(\lambda_*(\theta_*))\times K$. By Lemma \ref{lem:2}, $\wtl\theta\rightarrow \theta_*$, by the CMT and the DCT (which is valid under Assumption \ref{Ass_3} \textit{(c)} with $j=\ell = \ell' = 1$) we have that $\EE\left[\left(e^{\wh\lambda(\wtl\theta)'g_i(\wtl\theta)} - e^{\lambda_*(\theta_*)'g_i(\theta_*)}\right)h'\wtl{w}_{1,i}\wtl{w}_i'\right]\xrightarrow{p} 0$ uniformly in $h\in K$. Similarly, $\EE_n\left[e^{\wh\lambda(\wtl\theta)' g_i(\wtl\theta)}\right]\xrightarrow{p} \EE[e^{\lambda_*(\theta_*)'g_i(\theta_*)}]$ uniformly in $h\in K$ under Assumption \ref{Ass_3} \textit{(b)}. Hence, \eqref{eq:2:Lem_A_9:endogeneity} converges to zero in probability uniformly in $h\in K$.
\begin{flushright}
  $\square$
\end{flushright}

%%------------------------------------------------------------------------------------------------------------------------------------------------------------------------------
\begin{lem}\label{lem:technical:endogeneity:2}
  Let Assumptions \ref{Ass_absolute_continuity}-\ref{ass:feasibility}, \ref{Ass_1_NS} (a)-(b), \ref{Ass_3} (b)-(c) hold and let $\wtl\theta := \theta_* + \frac{\tau}{\sqrt{n}} h$ for some $\tau\in[0,1]$ and any $h$ in a compact set $K\subset \mathbb{R}^p$. Then, as $n\rightarrow \infty$,
  \begin{displaymath}
    \EE_n\left[\tau_i(\wh\lambda,\wtl\theta) \left(h'\wtl{w}_{1,i}\wtl{w}_i'\wh\lambda(\wtl\theta)\right)^2\right] \xrightarrow{p} \EE\left[\tau_i^{\dagger}(\lambda_*,\theta_*) \left(h'\wtl{w}_{1,i}\wtl{w}_i'\lambda_*(\theta_*)\right)^2\right]
  \end{displaymath}
  uniformly in $h\in K$.
\end{lem}
\proof
Notice that $\wtl\theta\in B_{*,n}$. The following decomposition holds:
  \begin{multline}
    \EE_n\left[\tau_i(\wh\lambda,\wtl\theta)\left(h'\wtl{w}_{1,i}\wtl{w}_i'\wh\lambda(\wtl\theta)\right)^2\right] - \EE\left[\tau_i^{\dagger}(\lambda_*,\theta_*) \left(h'\wtl{w}_{1,i}\wtl{w}_i'\lambda_*(\theta_*)\right)^2\right]\\
    = \EE_n\left[e^{\wh\lambda(\wtl\theta)'g_i(\wtl\theta)} \left(h'\wtl{w}_{1,i}\wtl{w}_i'\wh\lambda(\wtl\theta)\right)^2  - \EE\left[e^{\wh\lambda(\wtl\theta)'g_i(\wtl\theta)} \left(h'\wtl{w}_{1,i}\wtl{w}_i'\wh\lambda(\wtl\theta)\right)^2\right]\right] \frac{1}{\EE_n\left[e^{\wh\lambda(\wtl\theta)' g_i(\wtl\theta)}\right]}\\
    + \left(\EE\left[e^{\wh\lambda(\wtl\theta)'g_i(\wtl\theta)}\left(h'\wtl{w}_{1,i}\wtl{w}_i'\wh\lambda(\wtl\theta)\right)^2\right] - \EE\left[e^{\lambda_*(\theta_*)'g_i(\theta_*)} \left(h'\wtl{w}_{1,i}\wtl{w}_i'\lambda_*(\theta_*)\right)^2\right]\right)\frac{1}{\EE_n\left[e^{\wh\lambda(\wtl\theta)' g_i(\wtl\theta)}\right]}\\ \hfill
    + \EE\left[e^{\lambda_*(\theta_*)'g_i(\theta_*)}\left(h'\wtl{w}_{1,i}\wtl{w}_i'\lambda_*(\theta_*)\right)^2\right]\frac{1}{\EE_n\left[e^{\wh\lambda(\wtl\theta)' g_i(\wtl\theta)}\right]}\left(\EE\left[e^{\lambda_*(\theta_*)' g_i(\theta_*)}\right] - \EE_n\left[e^{\wh\lambda(\wtl\theta)' g_i(\wtl\theta)}\right]\right)\frac{1}{\EE\left[e^{\lambda_*(\theta_*)' g_i(\theta_*)}\right]}.
\label{eq:2:Lem_A_10:endogeneity}
  \end{multline}

\noindent By Lemma \ref{lem:2} which guarantees that $\wh\lambda(\wtl\theta) \in B_{1/\sqrt{n}}(0)$ for $n$ large, by compactness of $B_{1/\sqrt{n}}(0)$ and of $B_{1/\sqrt{n}}(\theta_\circ)$, and by \cite[Lemma 2.4]{NeweyMcFadden1994} which is valid under Assumption \ref{Ass_3} \textit{(c)} with $j = 1$ and $\ell = \ell' = 2$, we have that
$$\EE_n\left[e^{\wh\lambda(\wtl\theta)'g_i(\wtl\theta)} \left(h'\wtl{w}_{1,i}\wtl{w}_i'\wh\lambda(\wtl\theta)\right)^2  - \EE\left[e^{\wh\lambda(\wtl\theta)'g_i(\wtl\theta)} \left(h'\wtl{w}_{1,i}\wtl{w}_i'\wh\lambda(\wtl\theta)\right)^2\right]\right] \xrightarrow{p}0$$
uniformly in $(\wh\lambda(\wtl\theta),h) \in B_{1/\sqrt{n}}(\lambda_*(\theta_*))\times K$. By Lemma \ref{lem:2}, by the fact that $\wtl\theta\rightarrow \theta_*$, by the CMT and the DCT (which is valid under Assumption \ref{Ass_3} \textit{(c)} with $j = 1$ and $\ell = \ell' = 2$) we have that $\EE\left[e^{\wh\lambda(\wtl\theta)'g_i(\wtl\theta)}\left(h'\wtl{w}_{1,i}\wtl{w}_i'\wh\lambda(\wtl\theta)\right)^2\right]\xrightarrow{p} \EE\left[e^{\lambda_*(\theta_*)'g_i(\theta_*)} \left(h'\wtl{w}_{1,i}\wtl{w}_i'\lambda_*(\theta_*)\right)^2\right]$ uniformly in $h\in K$. Similarly, $\EE_n\left[e^{\wh\lambda(\wtl\theta)' g_i(\wtl\theta)}\right]\xrightarrow{p} \EE[e^{\lambda_*(\theta_*)' g_i(\theta_*)}]$ under Assumption \ref{Ass_3} \textit{(b)} uniformly in $h\in K$. Hence, \eqref{eq:2:Lem_A_10:endogeneity} converges to zero in probability uniformly in $h\in K$.
\begin{flushright}
  $\square$
\end{flushright}

%%------------------------------------------------------------------------------------------------------------------------------------------------------------------------------
\begin{lem}\label{lem:technical:endogeneity:3}
  Let Assumptions \ref{Ass_absolute_continuity}-\ref{ass:feasibility}, \ref{Ass_1_NS} (a)-(b), \ref{Ass_3} (b),(g) hold and let $\wtl\theta := \theta_* + \frac{\tau}{\sqrt{n}} h$ for some $\tau\in[0,1]$ and any $h$ in a compact set $K\subset \mathbb{R}^p$. Then, as $n\rightarrow \infty$,
  \begin{displaymath}
    \EE_n\left[\tau_i(\wh\lambda,\wtl\theta) g_i(\wtl\theta)\wh\lambda(\wtl\theta)'\wtl{w}_i\wtl{w}_{1,i}'\right]h \xrightarrow{p} \EE\left[\tau_i^{\dagger}(\lambda_*,\theta_*) g_i(\theta_*)\lambda_*(\theta_*)'\wtl{w}_i\wtl{w}_{1,i}'\right]h
  \end{displaymath}
  uniformly in $h\in K$.
\end{lem}
\proof
Notice that $\wtl\theta\in B_{*,n}$. The following decomposition holds:
  \begin{multline}
    \EE_n\left[\tau_i(\wh\lambda,\wtl\theta)g_i(\wtl\theta)\wh\lambda(\wtl\theta)'\wtl{w}_i\wtl{w}_{1,i}'\right]h - \EE\left[\tau_i(\lambda_*,\theta_*) g_i(\theta_*)\lambda_*(\theta_*)'\wtl{w}_i\wtl{w}_{1,i}'\right]h\\
    = \EE_n\left[e^{\wh\lambda(\wtl\theta)'g_i(\wtl\theta)} g_i(\wtl\theta)\wh\lambda(\wtl\theta)'\wtl{w}_i\wtl{w}_{1,i}'h  - \EE\left[e^{\wh\lambda(\wtl\theta)'g_i(\wtl\theta)} g_i(\wtl\theta)\wh\lambda(\wtl\theta)'\wtl{w}_i\wtl{w}_{1,i}'h\right]\right] \frac{1}{\EE_n\left[e^{\wh\lambda(\wtl\theta)' g_i(\wtl\theta)}\right]}\\
    + \left(\EE\left[e^{\wh\lambda(\wtl\theta)'g_i(\wtl\theta)} g_i(\wtl\theta)\wh\lambda(\wtl\theta)'\wtl{w}_i\wtl{w}_{1,i}'h\right] - \EE\left[\tau_i^{\diamond}(\lambda_*,\theta_*) g_i(\theta_*)\lambda_*(\theta_*)'\wtl{w}_i\wtl{w}_{1,i}'\right]h\right)\frac{1}{\EE_n\left[e^{\wh\lambda(\wtl\theta)' g_i(\wtl\theta)}\right]}\\\hfill
    + \EE\left[\tau_i^{\diamond}(\lambda_*,\theta_*) g_i(\theta_*)\lambda_*(\theta_*)'\wtl{w}_i\wtl{w}_{1,i}'\right]h\frac{1}{\EE_n\left[e^{\wh\lambda(\wtl\theta)' g_i(\wtl\theta)}\right]}\frac{\left(\EE\left[e^{\lambda_*(\theta_*)' g_i(\theta_*)}\right] - \EE_n\left[e^{\wh\lambda(\wtl\theta)' g_i(\wtl\theta)}\right]\right)}{\EE\left[e^{\lambda_*(\theta_*)' g_i(\theta_*)}\right]}.\label{eq:2:Lem_A_11:endogeneity}
  \end{multline}
\noindent By Lemma \ref{lem:2} which guarantees that $\wh\lambda(\wtl\theta) \in B_{1/\sqrt{n}}(\lambda_*(\theta_*))$ for $n$ large, by compactness of $B_{1/\sqrt{n}}(\lambda_*(\theta_*))$ and of $B_{1/\sqrt{n}}(\theta_*)$, and by \cite[Lemma 2.4]{NeweyMcFadden1994} which is valid under Assumption \ref{Ass_3} \textit{(f)}, we have that
$$\EE_n\left[e^{\wh\lambda(\wtl\theta)'g_i(\wtl\theta)} g_i(\wtl\theta)\wh\lambda(\wtl\theta)'\wtl{w}_i\wtl{w}_{1,i}'h  - \EE\left[e^{\wh\lambda(\wtl\theta)'g_i(\wtl\theta)} g_i(\wtl\theta)\wh\lambda(\wtl\theta)'\wtl{w}_i\wtl{w}_{1,i}'h\right]\right] \xrightarrow{p}0$$
uniformly in $(\wh\lambda(\wtl\theta),h) \in B_{1/\sqrt{n}}(\lambda_*(\theta_*))\times K$. By Lemma \ref{lem:2}, by $\wtl\theta\rightarrow \theta_*$ uniformly in $h\in K$, by the CMT and the DCT (which is valid under Assumption \ref{Ass_3} \textit{(f)}) we have that $\EE\left[e^{\wh\lambda(\wtl\theta)'g_i(\wtl\theta)}g_i(\wtl\theta)\wh\lambda(\wtl\theta)'\wtl{w}_i\wtl{w}_{1,i}'h\right]\xrightarrow{p} \EE\left[e^{\lambda_*(\theta_*)'g_i(\theta_*)}g_i(\theta_*)\lambda_*(\theta_*)'\wtl{w}_i\wtl{w}_{1,i}'h\right]$ uniformly in $h\in K$. Similarly, $\EE_n\left[e^{\wh\lambda(\wtl\theta)' g_i(\wtl\theta)}\right]\xrightarrow{p} \EE_n\left[e^{\lambda_*(\theta_*)' g_i(\theta_*)}\right]$ under Assumption \ref{Ass_3} \textit{(b)} uniformly in $h \in K$. Hence, \eqref{eq:2:Lem_A_11:endogeneity} converges to zero in probability uniformly in $h\in K$.
\begin{flushright}
  $\square$
\end{flushright}

%%------------------------------------------------------------------------------------------------------------------------------------------------------------------------------
\begin{lem}\label{lem:technical:endogeneity:second:derivative}
  Let Assumptions \ref{Ass_absolute_continuity}-\ref{ass:feasibility}, \ref{Ass_1_NS} (a)-(b), \ref{Ass_3} (b)-(g) hold and let $\wtl\theta := \theta_* + \frac{\tau}{\sqrt{n}} h$ for some $\tau\in[0,1]$ and any $h$ in a compact set $K\subset \mathbb{R}^p$. Then, as $n\rightarrow \infty$,
  \begin{displaymath}
    h'\frac{d^2[\wh\lambda(\wtl\theta)'\wh g(\wtl\theta)]}{d\theta d\theta'}h = \mathcal{O}_p(1)
  \end{displaymath}
  uniformly in $h\in K$.
\end{lem}
\proof
Notice that $\wtl\theta\in B_{*,n}$. First, since $d^2 \wh g(\theta)/d\theta d\theta' = 0$, then $h'\frac{d^2[\wh\lambda(\wtl\theta)'\wh g(\wtl\theta)]}{d\theta d\theta'}h = h'\frac{d^2\wh\lambda(\wtl\theta)'}{d\theta d\theta'} \wh g(\wtl\theta)h - 2h'\frac{d\wh\lambda(\wtl\theta)'}{d\theta} \EE_n\left[\wtl{w}_i \wtl{w}_{1,i}'\right]h$. Under Assumption \ref{Ass_3_extended} \textit{(c)}, and Assumption \ref{Ass_3} \textit{(c)} (with $j=\ell=\ell' = 1$, $\ell'' = 0$) and \textit{(g)}, we apply the Uniform Law of Large Numbers and Lemma \ref{lem:Omega} to show that: $\EE_n\left[\wtl{w}_i \wtl{w}_{1,i}'\right]h \xrightarrow{p} \EE\left[\wtl{w}_i \wtl{w}_{1,i}'\right]h$ and $h'\frac{d\wh\lambda(\wtl\theta)'}{d\theta} \xrightarrow{p}h'\frac{d\lambda_*(\theta_*)'}{d\theta}$ uniformly in $h\in K_n$.\\
\indent We then consider the first term on the right hand side:
  \begin{multline}
    h'\frac{d^2\wh\lambda(\wtl\theta)}{d\theta d\theta'} \wh g(\wtl\theta)h = \frac{d\EE_n \left[\tau_i^{\diamond}(\wh\lambda,\wtl\theta) h'\wtl w_{1,i}\wtl w_i' (I + \wh\lambda(\wtl\theta) g_i(\wtl\theta)')\right]\check\Omega^{\diamond}(\wh\lambda,\wtl\theta)^{-1}}{d\theta'} \wh g(\wtl\theta)h\\
    = h' \EE_n\left(\frac{d}{d \theta'} \left( e^{\wh\lambda(\wtl\theta)' g_i(\wtl\theta)} \right)h \wtl w_{1,i} \wtl w_i' (I + \wh\lambda(\wtl\theta) g_i(\wtl\theta)') + \tau_i^{\diamond}(\wh\lambda,\wtl\theta)\wtl w_{1,i} \wtl w_i' \frac{d}{d \theta'} \left(\wh\lambda(\wtl\theta) g(\wtl\theta)' \right)h \right) \check\Omega^{\diamond}(\wh\lambda,\wtl\theta)^{-1} \wh g(\wtl\theta)\\
    - \EE_n\left[ \tau_i^{\diamond}(\wh\lambda,\wtl\theta) h'\wtl{w}_{1,i} \wtl w_i' (I + \wh \lambda(\wtl\theta) g_i(\wtl\theta)') \right]\check\Omega^{\diamond}(\wh\lambda,\wtl\theta)^{-1} \\
    \quad \times \EE_n\left( \frac{d}{d \theta} \left( e^{\wh\lambda(\wtl\theta)' g_i(\wtl\theta)} \right) \varepsilon_i(\wtl\theta)^2 \wtl w_{i} \wtl w_i' + 2\tau_i^{\diamond}(\wh\lambda,\wtl\theta)\varepsilon_i(\wtl\theta) \wtl{w}_{1,i}'h\wtl w_i \wtl w_i' \right) \check\Omega^{\diamond}(\wh\lambda,\wtl\theta)^{-1}  \wh g(\wtl\theta).\label{eq:2nd_derivative:proof:1}
  \end{multline}
We now analyse the asymptotic behaviour of each of these terms. First, by Lemma \ref{lem:Omega}: $\left\|\check\Omega^{\diamond}(\wh\lambda,\theta_*) - \Omega_*^{\diamond}(\theta_*)\right\| \xrightarrow{p} 0$ as $n\rightarrow \infty$ and by the continuous mapping theorem: $\check\Omega^{\diamond}(\wh\lambda,\theta_*)^{-1} \xrightarrow{p} \Omega_*^{\diamond}(\theta_*)^{-1}$. Moreover,
\begin{displaymath}
  \frac{d}{d \theta'} \left( e^{\wh\lambda(\wtl\theta)' g_i(\wtl\theta)} \right)h = \tau_i^{\diamond}(\wh\lambda,\wtl{\theta})\left[g_i(\wtl\theta)'\frac{d\wh\lambda(\wtl\theta)}{d\theta'}h + \wh\lambda(\wtl\theta)'\wtl w_i \wtl w_{1,i}'h\right].
\end{displaymath}
We replace this expression in \eqref{eq:2nd_derivative:proof:1}. Hence the first factor becomes:
\begin{multline*}
  h' \EE_n\left(\frac{d}{d \theta'} \left( e^{\wh\lambda(\wtl\theta)' g_i(\wtl\theta)} \right)h \wtl w_{1,i} \wtl w_i' (I + \wh\lambda(\wtl\theta) g_i(\wtl\theta)') + \tau_i^{\diamond}(\wh\lambda,\wtl\theta)\wtl w_{1,i} \wtl w_i' \frac{d}{d \theta'} \left(\lambda_*(\wtl\theta) g(\wtl\theta)' \right)h \right)\\
  = h'\frac{d\wh\lambda(\wtl\theta)'}{d\theta}\EE_n\left(\tau_i^{\diamond}(\wh\lambda,\wtl{\theta})\left[g_i(\wtl\theta) + \wh\lambda(\wtl\theta)'\wtl w_i \wtl w_{1,i}'h\right] h'\wtl w_{1,i} \wtl w_i' (I + \wh\lambda(\wtl\theta) g_i(\wtl\theta)')\right)\\
   + \EE_n\left(\tau_i^{\diamond}(\wh\lambda,\wtl\theta)h'\wtl w_{1,i} \wtl w_i' \left[\frac{d}{d \theta'} \wh\lambda(\wtl\theta)h g(\wtl\theta)' + \wh\lambda(\wtl\theta)\wtl w_{i} \wtl w_{1,i}'h\right]\right)
\end{multline*}
By Assumption \ref{Ass_3} \textit{(c)} with $j=2,\ell = \ell' = \ell'' =1$ and with $j=1,\ell = \ell'' =1, \ell'=2$, Assumption \ref{Ass_3} \textit{(g)} with $(j,\ell)=(1,2)$ and $(j,\ell)=(2,1)$, \textit{(c)} which allows us to use the uniform Law of Large Numbers and the DCT, the previous expression converges in probability towards:
\begin{multline}
  h'\frac{d\wh\lambda(\wtl\theta)'}{d\theta}\EE\left(\tau_i^{\diamond}(\lambda_*,\theta_*)\left[g_i(\theta_*) + \lambda_*(\theta_*)'\wtl w_i \wtl w_{1,i}'h\right] h'\wtl w_{1,i} \wtl w_i' (I + \lambda_*(\theta_*) g_i(\theta_*)')\right)\\
   + \EE\left(\tau_i^{\diamond}(\lambda_*,\theta_*)h'\wtl w_{1,i} \wtl w_i' \left[\frac{d}{d \theta'} \wh\lambda(\theta_*)h g(\theta_*)' + \lambda_*(\theta_*)\wtl w_{i} \wtl w_{1,i}'h\right]\right)\\
   =  h'\frac{d\wh\lambda(\wtl\theta)'}{d\theta}\EE\left(\tau_i^{\diamond}(\lambda_*,\theta_*)\left[g_i(\theta_*) + \lambda_*(\theta_*)'\wtl w_i \wtl w_{1,i}'h\right] h'\wtl w_{1,i} \wtl w_i' (2I + \lambda_*(\theta_*) g_i(\theta_*)')\right)
\end{multline}
uniformly in $h\in K_n$. Finally, $h'\frac{d\wh\lambda(\wtl\theta)'}{d\theta}$ converges in probability towards $\EE \left[e^{\lambda_{*}(\theta)'g_i(\theta)}\wtl w_{1,i}\wtl w_i'(I + \lambda_*(\theta) g_i(\theta)')\right]$ uniformly in $h\in K_n$ under Assumption \ref{Ass_3} \textit{(c)} with $j=2,\ell=\ell'=1,\ell''=0$ and Assumption \ref{Ass_3} \textit{(f)}.\\
\indent We now analyse the second term in \eqref{eq:2nd_derivative:proof:1}:

\begin{multline}
  \EE_n\left( \frac{d}{d \theta} \left( e^{\wh\lambda(\wtl\theta)' g_i(\wtl\theta)} \right) \varepsilon_i(\wtl\theta)^2 \wtl w_{i} \wtl w_i' + 2\tau_i^{\diamond}(\wh\lambda,\wtl\theta)\varepsilon_i(\wtl\theta) \wtl{w}_{1,i}'h\wtl w_i \wtl w_i' \right)\\
  = \EE_n\left( \tau_i^{\diamond}(\wh\lambda,\wtl{\theta})\left[g_i(\wtl\theta)'\frac{d\wh\lambda(\wtl\theta)}{d\theta'}h + \wh\lambda(\wtl\theta)'\wtl w_i \wtl w_{1,i}'h\right]\varepsilon_i(\wtl\theta)^2 \wtl w_{i} \wtl w_i' + 2\tau_i^{\diamond}(\wh\lambda,\wtl\theta)\varepsilon_i(\wtl\theta) \wtl{w}_{1,i}'h\wtl w_i \wtl w_i' \right)
\end{multline}
which, under Assumption \ref{Ass_3} \textit{(e)} with $i=\ell =1$, $\ell' = 3$, and Assumption \ref{Ass_3} \textit{(c)} with $(j,\ell,\ell',\ell'') = (2,1,1,1)$, converges in probability towards
\begin{multline}
  \EE\left(\tau_i^{\diamond}(\lambda_*,\theta_*)\left[g_i(\theta_*)'\frac{d\lambda_*(\theta_*)}{d\theta'}h + \lambda_*(\theta_*)'\wtl w_i \wtl w_{1,i}'h\right]\varepsilon_i(\theta_*)^2 \wtl w_{i} \wtl w_i' + 2\tau_i^{\diamond}(\lambda_*,\theta_*)\varepsilon_i(\theta_*) \wtl{w}_{1,i}'h\wtl w_i \wtl w_i' \right)
\end{multline}
uniformly in $h\in K_n$.\\
By putting all this together we find that
\begin{multline*}
  h'\frac{d^2[\wh\lambda(\wtl\theta)'\wh g(\wtl\theta)]}{d\theta d\theta'}h \xrightarrow{p} \\
  h'\frac{d\lambda_*(\theta_*)'}{d\theta}\Omega_*^{\diamond}(\theta_*)\EE\left(\tau_i^{\diamond}(\lambda_*,\theta_*)\left[g_i(\theta_*) + \lambda_*(\theta_*)'\wtl w_i \wtl w_{1,i}'h\right] h'\wtl w_{1,i} \wtl w_i' (2I + \lambda_*(\theta_*) g_i(\theta_*)')\right)\Omega_*^{\diamond}(\theta_*)\EE[g_i(\theta_*)]\\
  -h'\frac{d\lambda_*(\theta_*)'}{d\theta} \Omega_*^{\diamond}(\theta_*)\EE\Big(\tau_i^{\diamond}(\lambda_*,\theta_*)\left[g_i(\theta_*)'\frac{d\lambda_*(\theta_*)}{d\theta'}h + \lambda_*(\theta_*)'\wtl w_i \wtl w_{1,i}'h\right]\varepsilon_i(\theta_*)^2 \wtl w_{i} \wtl w_i'\\
  + 2\tau_i^{\diamond}(\lambda_*,\theta_*)\varepsilon_i(\theta_*) \wtl{w}_{1,i}'h\wtl w_i \wtl w_i' \Big) \Omega_*^{\diamond}(\theta_*)\EE[g_i(\theta_*)] - h'\frac{d\lambda_*(\theta_*)'}{d\theta}\EE\left[\wtl{w}_i \wtl{w}_{1,i}'\right]h\\
  = h'\frac{d^2\lambda_*(\theta_*)}{d\theta d\theta'} \EE [g_i(\theta_*)]h  - h'\frac{d\lambda_*(\theta_*)'}{d\theta}\EE\left[\wtl{w}_i \wtl{w}_{1,i}'\right]h
\end{multline*}
\noindent uniformly in $h\in K_n$. We conclude that $h'\frac{d^2[\wh\lambda(\wtl\theta)'\wh g(\wtl\theta)]}{d\theta d\theta'}h $ is bounded in probability uniformly in $h\in K_n$.
\begin{flushright}
  $\square$
\end{flushright}

%%------------------------------------------------------------------------------------------------------------------------------------------------------------------------------
%%% -------------------------------------------------------------------------------- EXOGENEITY --------------------------------------------------------------------------------
%%------------------------------------------------------------------------------------------------------------------------------------------------------------------------------

%%------------------------------------------------------------------------------------------------------------------------------------------------------------------------------
%%------------------------------------------------------------------------------------------------------------------------------------------------------------------------------
\subsubsection{Technical results for the base model under exogeneity}
%=========================================================================================================
\begin{lem}\label{lem:2:base:exogeneity}
    Suppose Assumptions \ref{ass:feasibility}, \ref{Ass_1_NS} (a), and \ref{Ass_3} (d)-(e) hold. Then, for every given $\theta\in B_{\circ,n}$,
    $$\sqrt{n}\|\wh\lambda(\theta)\|_2 = \mathcal{O}_p(1)$$
    or, equivalently, for any $\eta >0$ there exists a finite $\delta>0$ and a finite $N(\delta,\eta)$ such that: $\forall n > N(\delta,\eta)$,
    $$P\left(\sqrt{n}\|\wh\lambda(\theta)\|_2 < \delta\right) > \eta.$$
\end{lem}
\proof The proof proceeds as the proof of Lemma \ref{lem:2} and then it is omitted.

\begin{flushright}
  $\square$
\end{flushright}

%%------------------------------------------------------------------------------------------------------------------------------------------------------------------------------
\begin{lem}\label{lem:technical:exogeneity:1}
  Let Assumptions \ref{Ass_0_NS}-\ref{ass:feasibility}, \ref{Ass_1_NS} (a)-(b), \ref{Ass_3} (b)-(c) hold and let $\wtl\theta := \theta_\circ + \frac{\tau}{\sqrt{n}} h$ for some $\tau\in[0,1]$ and any $h$ in a compact set $K\subset \mathbb{R}^p$. Then, as $n\rightarrow \infty$,
  \begin{displaymath}
    \EE_n\left[\tau_i(\wh\lambda,\wtl\theta) h'\wtl{w}_{1,i}\wtl{w}_i\right] \xrightarrow{p} \EE\left[h'\wtl{w}_{1,i}\wtl{w}_i\right]
  \end{displaymath}
  uniformly in $h\in K$.
\end{lem}
\proof
Notice that $\wtl\theta\in B_{\circ,n}$. Since $\lambda_*(\theta_\circ) = 0$ so that $\tau_i^{\dagger}(\lambda_*,\theta_\circ) = 1$, the following decomposition holds:
  \begin{multline}
    \EE_n\left(\tau_i(\wh\lambda,\wtl\theta)h'\wtl{w}_{1,i}\wtl{w}_i\right) - \EE\left[h'\wtl{w}_{1,i}\wtl{w}_i\right]\\
    %= \EE_n\left[e^{\wh\lambda(\theta_\circ)'g_i(\theta_\circ)} h'\wtl{w}_{1,i}\wtl{w}_i'  - \EE\left[e^{\lambda_*(\theta_\circ)'g_i(\theta_\circ)}h'\wtl{w}_{1,i}\wtl{w}_i'\right]\right] \frac{1}{\EE_n\left[e^{\wh\lambda(\theta_\circ)' g_i(\theta_\circ)}\right]}\\ \hfill
    %+ \EE\left[e^{\lambda_*(\theta_\circ)'g_i(\theta_\circ)} h'\wtl{w}_{1,i}\wtl{w}_i'\right]\left(\frac{1}{\EE_n\left[e^{\wh\lambda(\theta_\circ)' g_i(\theta_*)}\right]} - \frac{1}{\EE\left[e^{\lambda_*(\theta_\circ)' g_i(\theta_*)}\right]}\right)\\
    = \EE_n\left[e^{\wh\lambda(\wtl\theta)'g_i(\wtl\theta)} h'\wtl{w}_{1,i}\wtl{w}_i  - \EE\left[e^{\wh\lambda(\wtl\theta)'g_i(\wtl\theta)} h'\wtl{w}_{1,i}\wtl{w}_i\right]\right] \frac{1}{\EE_n\left[e^{\wh\lambda(\wtl\theta)' g_i(\wtl\theta)}\right]}\\
    + \EE\left[\left(e^{\wh\lambda(\wtl\theta)'g_i(\wtl\theta)} - 1\right)h'\wtl{w}_{1,i}\wtl{w}_i\right]\frac{1}{\EE_n\left[e^{\wh\lambda(\wtl\theta)' g_i(\wtl\theta)}\right]}\\ \hfill
    + \EE\left[h'\wtl{w}_{1,i}\wtl{w}_i\right]\frac{1}{\EE_n\left[e^{\wh\lambda(\wtl\theta)' g_i(\wtl\theta)}\right]}\left(1 - \EE_n\left[e^{\wh\lambda(\wtl\theta)' g_i(\wtl\theta)}\right]\right).\label{eq:2:Lem_A_9}
  \end{multline}

\noindent By Lemma \ref{lem:2:base:exogeneity} which guarantees that $\wh\lambda(\wtl\theta) \in B_{1/\sqrt{n}}(0)$ for $n$ large, by compactness of $B_{1/\sqrt{n}}(0)$ and of $B_{1/\sqrt{n}}(\theta_\circ)$, and by \cite[Lemma 2.4]{NeweyMcFadden1994} which is valid under Assumption \ref{Ass_3} \textit{(c)} with $j=\ell = \ell' = 1$, we have that
\[
\EE_n\left[e^{\wh\lambda(\wtl\theta)'g_i(\wtl\theta)} h'\wtl{w}_{1,i}\wtl{w}_i'  - \EE\left[e^{\wh\lambda(\wtl\theta)'g_i(\wtl\theta)} h'\wtl{w}_{1,i}\wtl{w}_i'\right]\right] \xrightarrow{p}0
\]
uniformly in $(\wh\lambda(\wtl\theta),h) \in B_{1/\sqrt{n}}(0)\times K$. By Lemma \ref{lem:2:base:exogeneity}, by the convergence $\wtl\theta\rightarrow \theta_\circ$ uniform on $h\in K$, by the CMT and the DCT (which is valid under Assumption \ref{Ass_3} \textit{(c)} with $j=\ell = \ell' = 1$) we have that $\EE\left[\left(e^{\wh\lambda(\wtl\theta)'g_i(\wtl\theta)} - 1\right)h'\wtl{w}_{1,i}\wtl{w}_i'\right]\xrightarrow{p} 0$ uniformly in $h\in K$. Similarly, $\EE_n\left[e^{\wh\lambda(\wtl\theta)' g_i(\wtl\theta)}\right]\xrightarrow{p} 1$ uniformly in $h\in K$ under Assumption \ref{Ass_3} \textit{(b)}. Hence, \eqref{eq:2:Lem_A_9} converges to zero in probability uniformly in $h \in K$.
\begin{flushright}
  $\square$
\end{flushright}
%%------------------------------------------------------------------------------------------------------------------------------------------------------------------------------
\begin{lem}\label{lem:technical:exogeneity:2}
  Let Assumptions \ref{Ass_0_NS}, \ref{Ass_1_NS} (a)-(b), \ref{Ass_3} (b)-(c) hold and let $\wtl\theta := \theta_\circ + \frac{\tau}{\sqrt{n}} h$ for some $\tau\in[0,1]$ and any $h$ in a compact set $K\subset \mathbb{R}^p$. Then, as $n\rightarrow \infty$,
  \begin{displaymath}
    \EE_n\left[\tau_i(\wh\lambda,\wtl\theta) \left(h'\wtl{w}_{1,i}\wtl{w}_i'\wh\lambda(\wtl\theta)\right)^2\right] \xrightarrow{p} 0
  \end{displaymath}
  uniformly in $h\in K$.
\end{lem}
\proof
Notice that $\wtl\theta\in B_{\circ,n}$. The following decomposition holds:
  \begin{multline}
    \EE_n\left[\tau_i(\wh\lambda,\wtl\theta)\left(h'\wtl{w}_{1,i}\wtl{w}_i'\wh\lambda(\wtl\theta)\right)^2\right]\\
    = \EE_n\left[e^{\wh\lambda(\wtl\theta)'g_i(\wtl\theta)} \left(h'\wtl{w}_{1,i}\wtl{w}_i'\wh\lambda(\wtl\theta)\right)^2  - \EE\left[e^{\wh\lambda(\wtl\theta)'g_i(\wtl\theta)} \left(h'\wtl{w}_{1,i}\wtl{w}_i'\wh\lambda(\wtl\theta)\right)^2\right]\right] \frac{1}{\EE_n\left[e^{\wh\lambda(\wtl\theta)' g_i(\wtl\theta)}\right]}\\
    + \EE\left[e^{\wh\lambda(\wtl\theta)'g_i(\wtl\theta)}\left(h'\wtl{w}_{1,i}\wtl{w}_i'\wh\lambda(\wtl\theta)\right)^2\right]\frac{1}{\EE_n\left[e^{\wh\lambda(\wtl\theta)' g_i(\wtl\theta)}\right]}.%\\ \hfill
%    + \EE\left[\left(h'\wtl{w}_{1,i}\wtl{w}_i'\wh\lambda(\wtl\theta)\right)^2\right]\frac{1}{\EE_n\left[e^{\wh\lambda(\wtl\theta)' g_i(\wtl\theta)}\right]}\left(1 - \EE_n\left[e^{\wh\lambda(\wtl\theta)' g_i(\wtl\theta)}\right]\right).
\label{eq:2:Lem_A_10}
  \end{multline}

\noindent By Lemma \ref{lem:2:base:exogeneity} which guarantees that $\wh\lambda(\wtl\theta) \in B_{1/\sqrt{n}}(0)$ for $n$ large, by compactness of $B_{1/\sqrt{n}}(0)$ and of $B_{1/\sqrt{n}}(\theta_\circ)$, and by \cite[Lemma 2.4]{NeweyMcFadden1994} which is valid under Assumption \ref{Ass_3} \textit{(c)} with $j = 1$ and $\ell = \ell' = 2$, we have that
$$\EE_n\left[e^{\wh\lambda(\wtl\theta)'g_i(\wtl\theta)} \left(h'\wtl{w}_{1,i}\wtl{w}_i'\wh\lambda(\wtl\theta)\right)^2  - \EE\left[e^{\wh\lambda(\wtl\theta)'g_i(\wtl\theta)} \left(h'\wtl{w}_{1,i}\wtl{w}_i'\wh\lambda(\wtl\theta)\right)^2\right]\right] \xrightarrow{p}0$$
uniformly in $(\wh\lambda(\wtl\theta),h) \in B_{1/\sqrt{n}}(0)\times K$. By Lemma \ref{lem:2:base:exogeneity}, by the fact that $\wtl\theta\rightarrow \theta_\circ$ uniformly in $h\in K$, by the CMT and the DCT (which is valid under Assumption \ref{Ass_3} \textit{(c)} with $j = 1$ and $\ell = \ell' = 2$) we have that $\EE\left[e^{\wh\lambda(\wtl\theta)'g_i(\wtl\theta)}\left(h'\wtl{w}_{1,i}\wtl{w}_i'\wh\lambda(\wtl\theta)\right)^2\right]\xrightarrow{p} 0$ uniformly in $h\in K$. Similarly, $\EE_n\left[e^{\wh\lambda(\wtl\theta)' g_i(\wtl\theta)}\right]\xrightarrow{p} 1$ under Assumption \ref{Ass_3} \textit{(b)} uniformly in $h\in K$. Hence, \eqref{eq:2:Lem_A_10} converges to zero in probability uniformly in $h \in K$.
\begin{flushright}
  $\square$
\end{flushright}

%%------------------------------------------------------------------------------------------------------------------------------------------------------------------------------
\begin{lem}\label{lem:technical:exogeneity:3}
  Let Assumptions \ref{Ass_0_NS}, \ref{Ass_1_NS} (a)-(b), \ref{Ass_3} (b),(g) hold and let $\wtl\theta := \theta_\circ + \frac{\tau}{\sqrt{n}} h$ for some $\tau\in[0,1]$ and any $h$ in a compact set $K\subset \mathbb{R}^p$. Then, as $n\rightarrow \infty$,
  \begin{displaymath}
    \EE_n\left[\tau_i(\wh\lambda,\wtl\theta) g_i(\wtl\theta)\wh\lambda(\wtl\theta)'\wtl{w}_i\wtl{w}_{1,i}'\right]h \xrightarrow{p} 0
  \end{displaymath}
  \noindent uniformly in $h\in K$.
\end{lem}
\proof
The following decomposition holds:
  \begin{multline}
    \EE_n\left[\tau_i(\wh\lambda,\wtl\theta)g_i(\wtl\theta)\wh\lambda(\wtl\theta)'\wtl{w}_i\wtl{w}_{1,i}'\right]h\\
    = \EE_n\left[e^{\wh\lambda(\wtl\theta)'g_i(\wtl\theta)} g_i(\wtl\theta)\wh\lambda(\wtl\theta)'\wtl{w}_i\wtl{w}_{1,i}'h  - \EE\left[e^{\wh\lambda(\wtl\theta)'g_i(\wtl\theta)} g_i(\wtl\theta)\wh\lambda(\wtl\theta)'\wtl{w}_i\wtl{w}_{1,i}'h\right]\right] \frac{1}{\EE_n\left[e^{\wh\lambda(\wtl\theta)' g_i(\wtl\theta)}\right]}\\
    + \EE\left[e^{\wh\lambda(\wtl\theta)'g_i(\wtl\theta)} g_i(\wtl\theta)\wh\lambda(\wtl\theta)'\wtl{w}_i\wtl{w}_{1,i}'h\right]\frac{1}{\EE_n\left[e^{\wh\lambda(\wtl\theta)' g_i(\wtl\theta)}\right]}.\label{eq:2:Lem_A_11}
  \end{multline}
\noindent By Lemma \ref{lem:2:base:exogeneity} which guarantees that $\wh\lambda(\wtl\theta) \in B_{1/\sqrt{n}}(0)$ for $n$ large, by compactness of $B_{1/\sqrt{n}}(0)$ and of $B_{1/\sqrt{n}}(\theta_\circ)$, and by \cite[Lemma 2.4]{NeweyMcFadden1994} which is valid under Assumption \ref{Ass_3} \textit{(f)}, we have that
$$\EE_n\left[e^{\wh\lambda(\wtl\theta)'g_i(\wtl\theta)} g_i(\wtl\theta)\wh\lambda(\wtl\theta)'\wtl{w}_i\wtl{w}_{1,i}'h  - \EE\left[e^{\wh\lambda(\wtl\theta)'g_i(\wtl\theta)} g_i(\wtl\theta)\wh\lambda(\wtl\theta)'\wtl{w}_i\wtl{w}_{1,i}'h\right]\right] \xrightarrow{p}0$$
uniformly in $(\wh\lambda(\wtl\theta),h) \in B_{1/\sqrt{n}}(0)\times K$. By Lemma \ref{lem:2:base:exogeneity}, by the fact that $\wtl\theta\rightarrow \theta_\circ$ uniformly in $h\in K$, by the CMT and the DCT (which is valid under Assumption \ref{Ass_3} \textit{(f)}) we have that $\EE\left[e^{\wh\lambda(\wtl\theta)'g_i(\wtl\theta)}g_i(\wtl\theta)\wh\lambda(\wtl\theta)'\wtl{w}_i\wtl{w}_{1,i}'h\right]\xrightarrow{p} 0$ uniformly in $h\in K$. Similarly, $\EE_n\left[e^{\wh\lambda(\wtl\theta)' g_i(\wtl\theta)}\right]\xrightarrow{p} 1$ under Assumption \ref{Ass_3} \textit{(b)} uniformly in $h\in K$. Hence, \eqref{eq:2:Lem_A_11} converges to zero in probability uniformly in $h\in K$.
\begin{flushright}
  $\square$
\end{flushright}

%%------------------------------------------------------------------------------------------------------------------------------------------------------------------------------
%%------------------------------------------------------------------------------------------------------------------------------------------------------------------------------
\subsection{Technical results for the extended model}

%=========================================================================================================
\begin{lem}\label{lem:1:extended}
    Let Assumptions \ref{ass:feasibility}, \ref{Ass_1_NS} (a) - (b) with $\Theta$ replaced by $\Psi$ and , \ref{Ass_3} (d)-(e) with $\theta_*$ and $\lambda_*(\theta_*)$ replaced by $\theta_\circ$ and $0$, respectively, hold. Then,
    $$\sqrt{n}\wh\lambda(\psi_\circ) = -\Omega_{\psi_\circ}^{-1}\mathbb{G}_n\left[\varepsilon_i(\theta_\circ)\wtl{w}_i - \left(\begin{array}{c}
      \EE[\varepsilon_i(\theta_\circ)x_i]\\
      0\\
      0
    \end{array}\right)\right] + o_p(1),$$
    where $\Omega_{\psi_\circ} := \EE[g_e(w_i,\psi_\circ)g_e(w_i,\psi_\circ)']$.
\end{lem}
\proof
By a second order MVT expansion of the function $\lambda\mapsto \EE_n[e^{\wh\lambda'g_e(w_i,\psi_{\circ})} g_e(w_i,\psi_{\circ})]$ around $\lambda_*(\psi_\circ) = 0$, evaluated at $\wh\lambda(\psi_{\circ})$ we can write : $\wh\lambda(\psi_\circ) = -\EE[g_e(w_i,\psi_\circ)g_e(w_i,\psi_\circ)']^{-1}\EE_n[g_e(w_i,\psi_\circ)] + \mathcal{O}_p(1/n)$, where the $\mathcal{O}_p(1/n)$ follows from the Law of Large Numbers (which implies the convergence in probability of $\EE_n[g_e(w_i,\psi_\circ)g_e(w_i,\psi_\circ)']$ towards $\EE[g_e(w_i,\psi_\circ)g_e(w_i,\psi_\circ)']$) and from Lemma \ref{lem:2:extended} and Assumption \ref{Ass_3} \textit{(e)} (with $(\ell,\ell',i) = (1,3,1)$ and with $\theta_*$ replaced with $\theta_\circ$) that allows to control the quadratic term in the MVT expansion. Moreover, $\EE[g_e(w_i,\psi_\circ)] = 0$. Therefore,
\begin{equation}
  \sqrt{n} \wh\lambda(\psi_\circ) = -\EE[g_e(w_i,\psi_\circ)g_e(w_i,\psi_\circ)']^{-1}\mathbb{G}_n[g_e(w_i,\psi_\circ)] + o_p(1).
\end{equation}
\begin{flushright}
  $\square$
\end{flushright}

%=========================================================================================================
\begin{lem}\label{lem:2:extended}
    Suppose Assumptions \ref{ass:feasibility}, \ref{Ass_1_NS} (a) - (b) with $\Theta$ replaced by $\Psi$ and \ref{Ass_3} (e) with $\theta_*$ and $\lambda_*(\theta_*)$ replaced with $\theta_\circ$ and $0$, respectively, hold. Then, for every given $\psi\in\Psi$,
    $$\sqrt{n}\|\wh\lambda(\psi) - \lambda_*(\psi)\|_2 = \mathcal{O}_p(1).$$
\end{lem}
\proof To prove the lemma we intend to apply \cite[Corollary 3.2.6]{VanDerVaartWellner1996} with $\mathbb{M}_n = -\EE_n\left[e^{\lambda'g_{e,i}(\psi)}\right]$. We keep $\psi$ fixed and we see $\mathbb{M}_n$ as a function of the data $w_i$ and of $\lambda$. For $\lambda_* = \arg\max_{\lambda\in\mathbb{R}^d}\mathbb{M}$, where $\mathbb{M} :=  -\EE\left[e^{\lambda'g_{e,i}(\psi)}\right]$, define the set of functions $\mathcal{M}_{\delta} := \left\{w_i\mapsto e^{\lambda_*'g_{e,i}(\psi)} - e^{\lambda'g_{e,i}(\psi)}; \lambda \in B_{\delta}(\lambda_*) \right\}$. Let $M_{\delta}$ denote an envelope function of $\mathcal{M}_{\delta}$. Because the functions in $\mathcal{M}_{\delta}$ are Lipschitz in the index parameter $\lambda\in B_{\delta}(\lambda_*)$ and
$$|e^{\lambda_2'g_{e,i}(\psi)} - e^{\lambda_1'g_{e,i}(\psi)}| \leq 2\delta \|g_{e,i}(\psi) e^{\wtl \lambda_{12}'g_{e,i}(\psi)}\|_2$$
for every $w_i$ and every $\lambda_1,\lambda_2\in B_{\delta}(\lambda_*)$ by the MVT for a given $\tau\in(0,1)$ such that $\wtl\lambda_{12} := \lambda_1 + \tau(\lambda_2 - \lambda_1)$ and by CS, then an envelope function is given by
\[
M_{\delta} = 2\delta \|g_{e,i}(\psi)\|_2 \sup_{\lambda_1,\lambda_2\in B_{\delta}(\lambda_*)} e^{\wtl \lambda_{12}'g_{e,i}(\psi)}
\]
(see the discussion before \cite[Theorem 2.7.11]{VanDerVaartWellner1996}). Moreover, there exists a constant $K$ such that the bracketing number $N_{[\,]}(\varepsilon \|M_{\delta}\|_{P,2},\mathcal{M}_{\delta}, \|\cdot\|_{P,2})$ of $\mathcal{M}_{\delta}$ is upper bounded by $K(2\delta)^d/\varepsilon^d$ by \cite[Example 19.7]{VanDerVaart2000} for every $0<\varepsilon < 2\delta $ if $\|M_{\delta}\|_{P,2}<C/\sqrt{n}$, which is guaranteed for $\delta = 1/\sqrt{n}$ under Assumption \ref{Ass_3} (d) and Assumption \ref{Ass_1_NS} \textit{(b)} with $\Theta$ replaced by $\Psi$. Hence, the bracketing integral $J_{[\,]}(1,\mathcal{M}_{\delta},\|\cdot\|_{P,2}) := \bigintsss_0^{1}\sqrt{1 + \log N_{[\,]}(\varepsilon \|M_{\delta}\|_{P,2},\mathcal{M}_{\delta}, \|\cdot\|_{P,2})}d\varepsilon$ is bounded and the class of functions $\mathcal{M}_{\delta}$ is $P$-Glivenko-Cantelli (see \textit{e.g.} \cite[Theorem 2.5.6]{VanDerVaartWellner1996} and its simplified version discussed on \cite[pages 243-244]{VanDerVaartWellner1996}).\\
\indent Therefore, by \cite[Corollary 3.2.3]{VanDerVaartWellner1996}, $\wh\lambda(\psi) \xrightarrow{p} \lambda_*(\psi)$ for every given $\psi\in B(\circ,n)$. Finally, to apply \cite[Corollary 3.2.6]{VanDerVaartWellner1996} we need to upper bound $\EE (\sup_{f\in \mathcal{M}_{1/\sqrt{n}}}\left|\sqrt{n}\EE_n[f - \EE(f)]\right|)$. According with the discussion following \cite[Corollary 3.2.6]{VanDerVaartWellner1996} (and because $J_{[\,]}(1,\mathcal{M}_{\delta},\|\cdot\|_{P,2})$ is bounded by the previous argument), the latter is upper bounded by a constant times $\|M_{1/\sqrt{n}}\|_{P,2}$. Since $\|M_{1/\sqrt{n}}\|_{P,2} \leq C/\sqrt{n}$, then $r_n^4 \|M_{1/r_n}\|_{P,2}^2 \leq n$ yields $r_n \asymp \sqrt{n}$. Therefore, the result of \cite[Corollary 3.2.6]{VanDerVaartWellner1996} holds with $r_n \leq \sqrt{n}$. This proves the result of the Lemma.

\begin{flushright}
  $\square$
\end{flushright}

\section{Details of duality theory for the ETEL problem in the IV setting}
In this section, we first consider a generic moment function $g(w,\theta ):\mathbb{R}^{d+1}\times \mathbb{R}^{p}\rightarrow \mathbb{R}^{d}$ supposed to be known. Then, the conditions on $g(w,\theta)$ will be written in terms of conditions on $g_b(w,\theta)$ and $g_e(w,\theta)$ when necessary. The primal optimization problem $\mathcal{P}_s$ that defines the ETEL is given by:
\begin{align}
  (\mathcal{P}_s)\qquad & \min_{q_{1},\ldots ,q_{n}>0}\sum_{i=1}^{n}\left[q_{i}\log (nq_{i})\right]  \notag \\
  & \text{subject to }\sum_{i=1}^{n}q_{i}= 1,\qquad \qquad \text{and }\qquad \sum_{i=1}^{n}q_{i}g(w_{i},\theta )=0,\label{eq:KLbasic:duality}
\end{align}
which depends on the parameter vector $\theta\in\Theta$, or more simply
\begin{equation}
  (\mathcal{P}_s)\qquad \qquad \min_{q_{1},\ldots ,q_{n}\in \mathcal{Q}_n(\theta)}\sum_{i=1}^{n}\left[q_{i}\log (nq_{i})\right],
\end{equation}
\noindent where $\mathcal{Q}_n(\theta)$ is the feasible set of the optimisation problem $\mathcal{P}_s$ and is defined as:
\begin{displaymath}
  \mathcal{Q}_n(\theta) := \left\{\{q_i\}_{i=1}^n; \sum_{i=1}^n q_i = 1, \; q_1>0,\ldots,q_n>0, \textrm{ and }\sum_{i=1}^n g(w_i,\theta)q_i = 0\right\}.
\end{displaymath}
\noindent We notice that $\mathcal{Q}_n(\theta)$ is the convex hull of $\bigcup_{i=1}^{n}g(w_{i},\theta )$
The subindex in $\mathcal{P}_s$ stresses that this is the primal problem in the sample. The function $u\mapsto \phi(u) \equiv \log(u) u$, which is the summand in the first line of $\mathcal{P}_s$, is a strictly convex function defined on $\mathbb{R}_+$. Denote by $\phi_+(\cdot)$ the extension of $\phi$ to $\mathbb{R}$ defined as : $\phi_+(u) = \phi(u)$ if $u>0$, $\phi_+(u) = 0$ if $u=0$, and $\phi_+(u) = \infty$ if $u<0$. The function $u\mapsto \phi_+(u)$ is strictly convex, essentially smooth, and satisfy $(0,\infty)\subset dom(\phi_+)$ (see \cite[Example 5.6]{BorweinLewis1991}). Strict convexity of $\phi_+$ on $\mathbb{R}_+$ and \cite[Proposition 2.11]{BorweinLewis1991} imply that any optimal solution to $(\mathcal{P}_s)$ is unique.\\
\indent To reduce the dimensionality from $n$ to $d$ it is convenient to work with the dual problem (denoted $\mathcal{D}_s$). The Lagrangian dual problem is
\begin{align}\label{eq:2:sample}
  (\mathcal{D}_s)\qquad & \sup_{\gamma\in\mathbb{R},\lambda\in\mathbb{R}^d} \gamma - \frac{1}{n}\sum_{i=1}^n \phi_+^{*}(\gamma + \lambda' g(w_i,\theta)),
\end{align}
\noindent where $\phi_+^{*}(t)$ is the convex conjugate of $\phi_+:\mathbb{R} \rightarrow (-\infty,\infty]$ and it is given by $t\mapsto \phi_+^*(t) = e^{t-1}$. By solving the optimization problem in \eqref{eq:2:sample} with respect to $\gamma$, we get a dual optimal value $\wh\gamma(\lambda,\theta) = 1 - \log(\EE_n[\exp\{\lambda^{'}g(w_i,\theta) \}])$ and if we replace it back in \eqref{eq:2:sample}, the resulting dual problem $\mathcal{D}_s$ is:
\begin{align}\label{eq:3:sample}
  (\mathcal{D}_s)\qquad & \sup_{\lambda\in\mathbb{R}^d} - \log\left(\frac{1}{n}\sum_{i=1}^n e^{\lambda' g(w_i,\theta)}\right).
\end{align}\\
\indent The next lemma establishes the relation between the primal and the dual ETEL problems and recall the assumptions that guarantee duality results.

\begin{lem}\label{lem:1:duality}
  Suppose that $\sum_{i=1}^n |g_j(w_i,\theta)| < \infty$, $P$-a.s. for all $j=1,\ldots,d$ and every $\theta\in\Theta$. Then, if for a given $\theta$ there exists a feasible solution $\{\wh q_i(\theta)\}_{i=1}^n$ to $\mathcal{P}_s$, that is, $\{\wh q_i(\theta)\}_{i=1}^n\in\mathcal{Q}_n(\theta)$ and $\{\wh q_i(\theta)\}_{i=1}^n = \arg\min_{q_{1},\ldots ,q_{n}\in \mathcal{Q}_n(\theta)}\sum_{i=1}^{n}\left[q_{i}\log (nq_{i})\right]$, then the values of $(\mathcal{P}_s)$ and $(\mathcal{D}_s)$ are equal for this particular $\theta$.
\end{lem}
This lemma implies that under Assumption \ref{ass:feasibility} and because of strict convexity of the function $\phi_+$, which together guarantee existence of a feasible solution with probability approaching $1$ for every $\theta\in \bigcap_{\{M_n; M_n\rightarrow \infty\} }B(\theta_*, M_n n^{-1/2})$, then the probability that the values of $(\mathcal{P}_s)$ and $(\mathcal{D}_s)$ are equal converges to $1$ for every $\theta \in \bigcap_{\{M_n; M_n\rightarrow \infty\} }B(\theta_*, M_n n^{-1/2})$.\\

\proof
The result of the theorem follows from \cite[Corollary 2.6]{BorweinLewis1991} once we show that the conditions of that corollary are satisfied. The first condition requires that the domain of integration is a finite measure space. In our case the integrals in the minimization problem are with respect to the empirical measure $\frac{1}{n}\sum_{i=1}^n \delta_{w_i}$, where $\delta_{w_i}$ denotes a Dirac measure on $w_i$. Then, we can take the domain of integration to be $[\min_i w_i, \max_i w_i]$ which is a finite measure space.\\
\indent The second condition of \cite[Corollary 2.6]{BorweinLewis1991} requires that $\phi_+$ is closed and convex. Convexity is satisfied. For closedness, it is sufficient to notice that the set $\{u\in\mathbb{R}; \phi_+(u)\leq \alpha\}$ is a closed set for every $\alpha\in\mathbb{R}$ because $0$ is the limit point of $\phi_+(u)$ as $u\rightarrow 0^+$. Finally, we have to check the qualification constraint which, since $(0,\infty)\subset dom(\phi_+)$, it is equivalent to require that there exists a feasible solution $\{\wh q_i\}_{i=1}^n$ such that $\wh{q}_i >0$, for every $i=1,\ldots,n$. The latter is satisfied because it is one of the conditions that define the feasible set $\mathcal{Q}_n(\theta)$.
\begin{flushright}
$\square$
\end{flushright}

The condition $\sum_{i=1}^n|g(w_i,\theta)| < \infty$, $P$-a.s. for every $\theta\in\Theta$ in Lemma \ref{lem:1:duality} becomes the condition $\sum_{i=1}^n|(y_i - \theta'\wtl{w}_{1,i})\wtl{w}_i| < \infty$, $P$-a.s. for every $\theta\in\Theta$ in both the base and the extended models. This is because the auxiliary parameter $v$ does not enter this condition as it is evident from \cite{BorweinLewis1991}.\\
\indent The next lemma characterizes the solution of the primal problem $\mathcal{P}_s$ in terms of the solution of the dual problem $\mathcal{D}_s$.
\begin{lem}\label{lem:2:duality}
  Assume the functions $w\mapsto g_j(w,\theta)$ are continuously differentiable for all $j=1,\ldots,d$ and every $\theta\in\Theta$ and suppose the conditions of Lemma \ref{lem:1:duality} are satisfied. For every $\theta$ for which $\mathcal{Q}_n(\theta)$ is non-empty, let $(\wh\gamma(\theta),\wh\lambda(\theta))$ be a dual optimal solution. Then, the unique optimal solution of the primal problem is given by: $\forall \theta$ such that $\mathcal{Q}_n(\theta)\neq \varnothing$,
    \begin{eqnarray}
      \wh q_i(\theta) & = & \frac{1}{n}((\phi_+)^*)'(\wh\gamma + \wh\lambda^{'}g(w_i,\theta)),\\
      & = & \frac{e^{\wh{\lambda}'g(w_i,\theta)}}{n\EE_n[e^{\wh{\lambda}'g(w_j,\theta)}]} \qquad \textrm{ for } i=1,\ldots,n.
    \end{eqnarray}
\end{lem}
\proof
Under the assumptions of the lemma, the result follows from \cite[Theorem 5.5]{BorweinLewis1991} with $D:= \lim_{u\rightarrow\infty}(\phi(u)/u) = \infty$. % where $D$ is such that there exists $\lambda^*\in\mathbb{R}^d$ such that
%$$\lambda^{*'}g(w_i,\theta) < D,\qquad \textrm{ for all }w_i\in [\min_i w_i,\max_i w_i].$$
\begin{flushright}
$\square$
\end{flushright}
\indent Continuous differentiability is satisfied for both the base and the extended model. The result of Lemma \ref{lem:2:duality} is important and means that if there exists a feasible solution for the dual problem then, we can find from it a solution for the primal problem for every $\theta$ for which $\mathcal{Q}_n(\theta)$ is non-empty.

Existence of a solution of $\mathcal{P}_s$ is not guaranteed even if the function $\phi_+(\cdot)$ is strictly convex. In fact, it might be that the feasible set $\mathcal{Q}_n(\theta)$ is empty for some $\theta$. Under Assumption \ref{Ass_absolute_continuity}, then in order to justify any estimation procedure based on the ETEL we need the following assumption. %For a constant $C>0$, and any $\wtl\theta\in\Theta$, define the $n^{-1/2}$-ball around $\wtl\theta$ as $B(\wtl\theta,C n^{-1/2}) := \{\theta\in\Theta:\,\|\theta - \wtl\theta\| \leq C n^{-1/2}\}.$ Denote by $int \Delta_n := \{q := (q_1,\ldots, q_n)'; \sum_{i=1}^n q_i = 1, \; q_1>0,\ldots,q_n>0\}$ the interior of the $n-1$ simplex.
\begin{ass}\label{ass:feasibility:Appendix:duality}
  For any constant $C>0$ and any $\theta\in B(\theta_*, C n^{-1/2})$:
  \begin{equation}
    P\left(\sum_{i=1}^n g(w_i,\theta)q_i = 0 \textrm{ for at least one } q \in int \Delta_n \right) \rightarrow 1.
  \end{equation}
\end{ass}
This assumption is slightly weaker than Assumption \ref{ass:feasibility} that we require because the condition has to hold for a $\theta$ in a slightly smaller set $B(\theta_*, C n^{-1/2})$.
%\indent \textcolor{red}{To add: results in Schennach (2007) are based on arguments built on the first derivative of the tilting parameter. This of course assumes... }

\end{document}